\documentclass[a4paper,10pt]{article}
\let\oldmarginpar\marginpar
\renewcommand\marginpar[1]{\-\oldmarginpar[\raggedleft\footnotesize #1]%
{\raggedright\footnotesize #1}}
\usepackage{graphicx}
\usepackage{graphics}
\usepackage{color}
\usepackage{amsmath, amsthm, slashed, glossary}
\usepackage{amssymb}
\usepackage{rotating}
\usepackage{epsfig}
\usepackage{verbatim}
\usepackage{rotating}
\usepackage{graphicx}
\usepackage{nomencl}

\makeglossary

\newtheorem{definition}{Definition}[section]
\newtheorem{theorem}[definition]{Theorem}
\newtheorem{conjecture}[definition]{Conjecture}
\newtheorem{corollary}[definition]{Corollary}
\newtheorem{proposition}[definition]{Proposition}
\newtheorem{lemma}[definition]{Lemma}
\newtheorem{remark}[definition]{Remark}

\title{Ultimately Schwarzschildean Spacetimes and the Black Hole Stability Problem}
\author{Gustav Holzegel\thanks{Princeton University,
Department of Mathematics, Fine Hall, Washington Road,
Princeton, NJ 08544 United States}}
\begin{document}

\maketitle

\begin{abstract}
In this paper, we introduce a class of spacetimes $\left(\mathcal{M},g\right)$ which satisfy the vacuum Einstein equations and dynamically approach a Schwarzschild solution of mass $M$, a class we shall call \emph{ultimately Schwarzschildean spacetimes}. The approach is captured in terms of boundedness and decay assumptions on appropriate spacetime-norms of the Ricci-coefficients and spacetime curvature. Given such assumptions at the level of $k$ derivatives of the Ricci-coefficients (and hence $k-1$ derivatives of curvature), we prove boundedness and decay estimates for $k$ derivatives of \emph{curvature}. The proof employs the framework of vectorfield multipliers and commutators for the Bel-Robinson tensor, pioneered by Christodoulou-Klainerman in the context of the stability of the Minkowski space. We provide multiplier analogues capturing the essential decay mechanisms (which have been identified previously for the scalar wave equation on black hole backgrounds) for the Bianchi equations. In particular, a formulation of the redshift-effect near the horizon is obtained. Morever, we identify a certain hierarchy in the Bianchi equations, which leads to the control of strongly $r$-weighted spacetime curvature-norms near infinity. This allows to avoid the use the classical conformal Morawetz multiplier $K$, therby generalizing recent work of Dafermos and Rodnianski in the context of the wave equation. Finally, the proof requires a detailed understanding of the structure of the error-terms in the interior. This is particularly intricate in view of both the phenomenon of trapped orbits and the fact that, unlike in the stability of Minkowski space, not all curvature components decay to zero.
\end{abstract}
\tableofcontents

\section{Introduction}
A major open problem in general relativity is to establish the non-linear stability of the Kerr family of solutions. Considerable mathematical progress towards this goal has been achieved in recent years, mainly by studying the linear wave equation,
\begin{align} \label{waveeq}
\Box_g \psi = 0 \, ,
\end{align}
on fixed Schwarzschild \cite{DafRod2, Sterbenz, Toha1} or Kerr-black hole backgrounds \cite{DafRodKerr, Toha2, AndBlue}.
 
 In particular, a precise and robust understanding of the role of the horizon for the decay mechanism of linear waves is now available within the framework of vectorfield multipliers and commutators, which is the common ground on which 
 almost all of the currently available results stand. In the same context, there is also a good understanding of the obstruction to decay associated with the phenomenon of trapped orbits and the associated loss of derivatives in the estimates, at least for Schwarzschild- and Kerr black holes. We refer the reader to \cite{Mihalisnotes} for a detailed discussion of these characteristic properties of black hole backgrounds.
 
 While the study of (\ref{waveeq}) provides important insights regarding the role of the geometry for the decay problem, its connection with the black hole stability problem is a-priori rather remote: (\ref{waveeq}) is certainly not the linearization of the Einstein equations with respect to a fixed black hole background $g$, but merely a  ``poor man's" linearization, which forgets entirely about the tensorial character of the original equations.

Encouraged by the rapid progress regarding (\ref{waveeq}), we initiate  in this paper a new approach to address the linear stability problem for black holes. This novel point of view will turn out to be much more intimately connected to the non-linear stability problem, as it is based entirely on a study of the Einstein and the Bianchi equations governing the metric evolution. 

The setting we suggest is the following. Consider a non-stationary black hole spacetime $\left(\mathcal{M}, g\right)$, which satisfies the vacuum Einstein equations 
\begin{equation} \label{vacE}
R_{\mu \nu} \left(g\right)=0 \, ,
\end{equation} 
and in addition settles down to a stationary solution for late times. In the familiar geometrical framework, the latter assumption of approach will manifest itself in appropriate decay assumptions on deformation tensors and Ricci rotation coefficients, once an appropriate coordinate system in which decay is measured, has been fixed. 

Clearly, in view of (\ref{vacE}) being satisfied, the Weyl-curvature tensor $W$ of such a spacetime is equal to the Riemann-tensor  and obeys the Bianchi equations
\begin{align} \label{BiI}
D^\alpha W_{\alpha \beta \gamma \delta} \left(g\right) = 0 \, .
\end{align}
Exploiting the rich structure of (\ref{BiI}) one can then try to prove estimates for the curvature tensor from appropriate assumptions on the rates of approach for the metric. 

Now by the very definition of curvature, decay assumptions on the future asymptotic behavior of the metric and its derivatives will imply certain decay of the curvature tensor.\footnote{When we talk about \emph{decay} of the curvature tensor here and in the future, we always understand this as decay to the curvature tensor of the spacetime we are approaching.}  Hence in order for the estimates arising from (\ref{BiI}) to be useful, one should be able to establish \emph{more} decay on the curvature than what immediately follows from the assumptions on the metric coefficients and deformation tensors. More precisely, in order to be valuable in a potential non-linear application, the decay proven for $W$ should be sufficiently strong to stand a chance to eventually \emph{improve}, via the null-structure equations, the assumptions initially made on the spacetime metric.
This type of argument forms an integral part in the context of the bootstrap setting in which non-linear stability results are typically proven. 

In the present paper, we address the simplest example, for which the ideas outlined above can be successfully carried out: A spacetime approaching a Schwarzschild solution for late times. While this example is non-generic in some sense, a large class of such spacetimes is expected to exist. Moreover, as we will see, the problem introduced here already exhibits many of the difficulties that the analogous problem for an ultimately Kerr spacetime would do. 

Roughly speaking, the statement we are going to prove in this context is the following. \emph{For sufficiently large $k$, boundedness of the $L^2$-energy and the $L^2$-spacetime-norm of $k$-derivatives of the Ricci rotation coefficients (and appropriate decay for lower order norms) implies boundedness and a degenerate version of integrated decay of the $L^2$-norm of $k$ derivatives of the curvature tensor.} While in this paper we do not concern ourselves with the problem of improving the bounds on the Ricci coefficients from the curvature bounds established (such estimates should follow closely the ones in \cite{ChristKlei}), we will in addition prove decay estimates for lower order curvature norms in terms of norms on the Ricci-coefficients. These estimates will reveal consistency of the approach in the sense that one does not have to assume stronger decay-rates on the Ricci-coefficients than the ones one expects to eventually derive from the curvature bounds obtained via multiplier estimates from the Bianchi equations.\footnote{As discussed below, the boundedness result for $k$ derivatives of curvature would be much easier to prove if one were to assume strong decay in $t$ for $k$ derivatives of the Ricci-coefficients, say $\frac{1}{t^{2+\epsilon}}$-decay of the $L^2$-norm. However, this would be inconsistent in the above sense: Because of the trapping, multiplier estimates from the Bianchi equations can only produce $\frac{1}{t}$-decay for $k-1$ derivatives of curvature, which fails to match the decay-rate assumed on $k$-derivatives of the Ricci-coefficients. Cf.~the discussion in section \ref{Tintro}.}
\newline

Before we turn to a more precise statement of the theorem and the main difficulties associated with its proof, let us briefly outline the basic technique of obtaining estimates for $W$ from (\ref{BiI}).

There exist natural energy currents for (\ref{BiI}) or, more generally, (\ref{BiI}) with an inhomogeneity $\mathcal{J}_{\beta \gamma \delta}$ on the right hand side. These currents arise from the so-called Bel-Robinson tensor
\begin{equation} 
 Q\left[\mathcal{W}\right]_{\alpha \beta \gamma \delta} = \mathcal{W}_{\alpha \rho \gamma \sigma} \mathcal{W}_{\beta\phantom{\rho}\delta\phantom{\sigma}}^{\phantom{\beta}\rho\phantom{\delta}\sigma} + \phantom{}^\star \mathcal{W}_{\alpha \rho \gamma \sigma} \phantom{}^\star \mathcal{W}_{\beta\phantom{\rho}\delta\phantom{\sigma}}^{\phantom{\beta}\rho\phantom{\delta}\sigma} \, .
\end{equation}
The tensor $Q$ is symmetric and traceless and satisfies the divergence identity
\begin{align} \label{bintro}
D^\alpha  \left(Q_{\alpha \beta \gamma \delta } \mathcal{X}^\beta \mathcal{Y}^\gamma \mathcal{Z}^\delta\right) &= K_1^{\mathcal{X}\mathcal{Y}\mathcal{Z}}  \left[\mathcal{W}\right] + K_2^{\mathcal{X}\mathcal{Y}\mathcal{Z}}  \left[\mathcal{W}\right] \, ,
\end{align}
\begin{align}
K_1^{\mathcal{X}\mathcal{Y}\mathcal{Z}} \left[\mathcal{W}\right] &= Q^{\alpha \beta \gamma \delta} \left(\phantom{}^{(\mathcal{X})}\pi_{\alpha \beta} \mathcal{Y}_\gamma \mathcal{Z}_{\delta} + \phantom{}^{(\mathcal{Y})}\pi_{\alpha \beta} \mathcal{Z}_\gamma \mathcal{X}_{\delta} + \phantom{}^{(\mathcal{Z})}\pi_{\alpha \beta} \mathcal{X}_\gamma \mathcal{Y}_{\delta} \right) \, , \nonumber \\
K_2^{\mathcal{X}\mathcal{Y}\mathcal{Z}} \left[\mathcal{W}\right] &= \left[\mathcal{W}_{\beta\phantom{\mu}\delta\phantom{\nu}}^{\phantom{\beta}\mu\phantom{\delta}\nu} \mathcal{J}_{\mu \gamma \nu} + \mathcal{W}_{\beta\phantom{\mu}\gamma\phantom{\nu}}^{\phantom{\beta}\mu\phantom{\gamma}\nu} \mathcal{J}_{\mu \delta \nu} + \phantom{}^\star \mathcal{W}_{\beta\phantom{\mu}\delta\phantom{\nu}}^{\phantom{\beta}\mu\phantom{\delta}\nu} \mathcal{J}^{\star}_{\mu \gamma \nu} + \phantom{}^\star \mathcal{W}_{\beta\phantom{\mu}\gamma\phantom{\nu}}^{\phantom{\beta}\mu\phantom{\gamma}\nu} \mathcal{J}^{\star}_{\mu \delta \nu}\right] \mathcal{X}^{\beta} \mathcal{Y}^\gamma \mathcal{Z}^\delta \, , \nonumber 
\end{align}
for any spacetime vectorfields $\mathcal{X},\mathcal{Y},\mathcal{Z}$. The term $K_2^{\mathcal{X}\mathcal{Y}\mathcal{Z}} \left[\mathcal{W}\right]$ vanishes in the case of the homogeneous Bianchi equations. 

Moreover, $Q\left(\mathcal{X},\mathcal{Y},\mathcal{Z},n\right) \geq 0$ holds for any future directed causal vectorfields $\mathcal{X},\mathcal{Y},\mathcal{Z},n$. If they are timelike, then $Q\left(\mathcal{X},\mathcal{Y},\mathcal{Z},n\right)$ in fact controls the sum of squares of all components of $\mathcal{W}$. Note also that both the right hand side and the left hand side of (\ref{bintro})depend only on $\mathcal{W}$ and not derivatives therof. 

The basic strategy to derive estimates for (\ref{BiI}) involves choosing appropriate vectorfields and integrating (\ref{bintro}) over certain spacetime regions. For instance, if the vectorfields $\mathcal{X},\mathcal{Y},\mathcal{Z}$ are Killing and $\mathcal{J}=0$ (the homogeneous Bianchi equations), the right hand side of (\ref{bintro}) vanishes, and we in fact obtain a conservation law relating currents on future spacelike slices to those in the past. 

For our spacetime however, there are no exact Killing vectorfields but only what we are going to call ``ultimately Killing fields". These are vectorfields, whose deformation tensor does not vanish but decays to zero in time. The non-zero deformation tensor requires a careful analysis of the non-vanishing spacetime terms on the right hand side of (\ref{bintro}).\footnote{We remark at this point that it would not make sense to study (\ref{BiI}) on a \emph{fixed} Schwarzschild background: Here the dynamics is trivial due to an algebraic constraint (known as the Buchdahl constraint in the literature \cite{Stewart}) on the Riemann-tensor.}

We also note, schematically at least, the formula for commuting the Bianchi equations with vectorfields:
\begin{align}   \label{commuteintro}                                                                          
D^\alpha \left(\widehat{\mathcal{L}}_\mathcal{X} W\right)_{\alpha \beta \gamma \delta}  = {}^{(\mathcal{X})}\pi \cdot     
DW + D {}^{(\mathcal{X})}\pi \cdot W                                                                             
\end{align}
Here $\widehat{\mathcal{L}}_\mathcal{X}$ is a modified Lie-derivative which gives $\widehat{\mathcal{L}}_\mathcal{X} W$ the algebraic properties of a Weyl-tensor (cf.~section \ref{commute}). Clearly, deriving energies for $\widehat{\mathcal{L}}_T W$ will necessitate an analysis of the non-vanishing term $K_2^{\mathcal{X}\mathcal{Y}\mathcal{Z}} \left[W\right] $.

We are now in a position to explain at a heuristic level, some of the difficulties and challenges one encounters, when one applies the vectorfield techniques outlined in the previous paragraph to the problem we wish to study:
\begin{enumerate}
\item Not all components of $W$ decay for the spacetime under consideration (otherwise we would approach the Minkowski space). This  causes difficulties both at the level of energies associated to (\ref{BiI}) (as they do not decay), and at the level of the error-terms arising from commuting (\ref{BiI}) with the approximate symmetries of the spacetime. As is already apparent from formula (\ref{commuteintro}), in sharp contrast to the stability of Minkowski space, not all error-terms will exhibit a quadratic structure in which both components decay.
\item What is the analogue of the redshift effect \cite{DafRod2} for (\ref{BiI})?
\item How is the obstruction for decay associated with trapped orbits captured for (\ref{BiI})?
\item How does one obtain decay (in the interior, at null infinity, ...) for $W$?
\item What are the appropriate (minimal?) assumptions on the Ricci rotation coefficients?
\end{enumerate}

As perhaps anticipated by the reader, the issues listed above are coupled to one another.

\subsection{Summary of the argument}
We proceed with a summary of the estimates carried out in the paper explaining how the above issues are addressed. In view of the frequently extensive and tedious formulae, we will regularly refer to the bulk of the paper. We hope that by this means, the introduction can also serve as a guide to the paper.

\subsubsection{Ultimately Schwarzschildean spacetimes}
Fix a regular coordinate system $\left(t^\star, r, \theta, \phi\right)$ on the black hole exterior of a Schwarzschild metric $g_M$ with mass $M$. Consider perturbations of $g_M$, such that the resulting metric $g$ on the black hole exterior is what we call \emph{ultimately Schwarzschildean to order $k+1$}  (cf.~Definition \ref{ultS}). At the core of this definition are assumptions at the level of $k$ derivatives of the Ricci-coefficients (i.e.~$k+1$-derivatives of the metric): In particular, both an energy on spacelike slices with null-ears ($\tilde{\Sigma}_{\tau}$, cf.~the figure in section \ref{picsec}) and a spacetime energy for the region enclosed by two such slices  are defined for the Ricci-rotation coefficients $\mathfrak{R}$: These energies measure the approach of the $\mathfrak{R}$ to their Schwarzschild values in the coordinate system $\left(t^\star, r, \theta, \phi\right)$. Roughly speaking, we assume boundedness of both of these energies for $k$ derivatives, $t^{-1}$ decay for $k-1$ derivatives, $t^{-2}$ for $k-2$ derivatives and $t^{-\frac{5}{2}}$ for $k-3$ or less derivatives (cf.~(\ref{iRdec}) and (\ref{Dnorm})). The reason for this hierarchy is related to what we expect to be able to show for the energies involving curvature and will become more clear later.

In addition, we assume the existence of a vectorfield $T$, satisfying $g\left(T,T\right)=0$ on the horizon and approaching the timelike (null on the horizon) Killing field of Schwarzschild in the sense that the deformation tensor associated with $T$ decays. We will refer to the vectorfield $T$ as being an \emph{ultimately Killing field}. See Definition \ref{ultimateKilling}.

\subsubsection{Null-decomposition of the curvature tensor}
We will also pick a null-frame $e_1$, $e_2$, $e_3$, $e_4$ for $g$ ($g\left(e_3,e_4\right)=-2)$), which is assumed to be appropriately close to a previously fixed Schwarzschild frame. We null decompose the curvature tensor with respect to that frame. In the standard notation of the subject, we obtain the components $\alpha, \beta, \rho, \sigma, \underline{\beta}, \underline{\alpha}$ (cf.~section \ref{abrsabbb}). The equations (\ref{BiI}) can then be written as equations for the null-components along null-directions (cf.~(\ref{Bianchi1})-(\ref{Bianchi9})). We note that in Schwarzschild all these components except $\rho=-\frac{2M}{r^3}$ are zero.

\subsubsection{The $T$-energy}        \label{Tintro}      
The first estimate that one may wish to obtain from (\ref{BiI}) is, of course, the analogue of the energy estimate associated with the ultimately Killing field $T$. As $\rho$ itself does not decay, and also in view of the fact that higher derivative estimates need to be derived, we are going to commute the equation sufficiently many times with the vectorfield $T$. However, as is apparent from formula (\ref{commuteintro}), irrespective of how many commutations are performed, there will          
always be a highest order error-term of the form $\rho \cdot \widehat{\mathcal{L}}_T^{k-1} D       
\pi$, which unlike most other terms, does not decay quadratically.\footnote{There are also    
lower order terms which do not decay quadratically, in view of the fact that some            
derivatives of $\rho$ do not decay. We will focus on the highest order term here.} As a      
consequence, in the energy estimate arising from (\ref{bintro}) we will have to control an error-term of the form                                                                       
\begin{align} \label{diffterms}                                                              
\int_{spacetime} \rho \cdot \left(\widehat{\mathcal{L}}^{k-1}_T D \pi \right) \cdot          
\widehat{\mathcal{L}}^k_T W  \, .     
\end{align}                                                                                  
In view of the trapping, we will not be able to control the spacetime-integral of the    
highest order derivative term globally, which means that we will have to put it ``$L^2$      
in space". It follows that to establish boundedness of the error-term at the highest level of derivatives (and hence boundedness of the $k^{th}$ order Weyl energy), we    
would need very strong decay (at least $\frac{1}{t^2}$) of the $L^2$-energy of $k$ derivatives of    
the deformation tensor. However, such a decay has only a chance of being obtained from      
the curvature components via elliptic estimates, if we can show the same decay for $k-1$-derivatives of             
curvature. Unfortunately, with current vectorfield techniques, we can only expect to     
prove $\frac{1}{t}$-decay for the $L^2$-energy of $k-1$ derivatives of curvature.                                
                                                                                             
The resolution of this difficulty relies strongly on the coupled            
character of the problem, i.e.~the fact that the background is not fixed but related to      
the curvature components via the null-structure equations (cf.~section \ref{nseq}). Working out the precise           
contractions of curvature components and components of the deformation tensor in             
(\ref{diffterms})  and inserting the structure equations relating $D\pi$ to the curvature    
components, one realizes that all terms in (\ref{diffterms}) have a special structure. To    
give an example,                                                                           
\begin{align}                                                      
\int_{st} \rho \Big(  \left(\widehat{\mathcal{L}}^{k-1}_T \left[\textrm{$D_3 \mathfrak{f}$ +    
$\alpha$ + l.o.t.}\right]\right) \widehat{\mathcal{L}}^k_T \alpha +                     
\left(\widehat{\mathcal{L}}^{k-1}_T \left[\textrm{$div \mathfrak{f}$ + $\beta$ +                        
l.o.t.}\right]\right) \widehat{\mathcal{L}}^k_T \beta   \Big)                          
\end{align}       
is such a term, where $\mathfrak{f}$ is an expression involving the Ricci rotation coefficients.                                                                           
Remarkably, the derivatives of the deformation tensor have brought in precisely the right           
curvature components from the structure equations, so as to allow an integration by parts in $T$. This lets the            
$T$-derivative fall on $\rho$ and makes this term a cubic error-term with all components     
decaying. The derivative term (which has $D_3f$ and $div f$ respectively) can be integrated by parts as    
well. Clearly, if the derivative falls on $\rho$, we obtain a lower order term. If it falls    
on the highest order derivative term, however, we can use the Bianchi equation               
\begin{align}                                                                                
\alpha_3 = -2\slashed{\mathcal{D}}^\star_2 \beta + \textrm{l.o.t.}                                         
\end{align}                                                                                  
and obtain also a lower order term.\footnote{Note that commutation of derivatives with $\widehat{\mathcal{L}}_T$ derivatives only introduces lower order terms and additional decay.}
At the end of the day, we can show that the worst possible error-terms are effectively    
of the form                                                                                  
\begin{align}                                                                                
\int_{spacetime} \rho \cdot \left(\widehat{\mathcal{L}}^{k-2}_T D \pi \right) \cdot          
\widehat{\mathcal{L}}^k_T W                                                                  
\end{align}     
The gain of a derivative compared to (\ref{diffterms}) is essential to close the estimates. The analysis of this structure is the topic of section \ref{errorterms}.

\subsubsection{The redshift}
As familiar from the wave equation, the vectorfield $T$ does not provide control over all curvature components near the horizon, even if the metric is exactly Schwarzschild (in our case, $T$ may not even be causal everywhere!). For the wave equation it is well-known how to stabilize the $T$ estimate using a redshift vectorfield \cite{DafRod2}. In complete analogy to this case, we construct a timelike vectorfield $N$ (agreeing with $T$ far away from the horizon), which generates boundary-terms controlling all components globally, and a spacetime term $K_1^{NNN}\left[\mathcal{W}\right]$ which has a good sign (and controls all curvature components) near the horizon. The error-terms introduced can be dealt with, once an integrated decay estimate in the interior (away from the horizon) is available. Cf.~Proposition \ref{rsW}.

Since until now we have only commuted with $T$, this technique will merely allow us to obtain non-degenerate control over arbitrary $T$-derivatives of the curvature components. This is sufficient away from the horizon in view of the elliptic estimates (section \ref{elliptic}) available. Close to the horizon we invoke a commutator version of the redshift estimate (developed in the context of the wave equation in \cite{DafRodKerr}), for which one commutes with the redshift vectorfield to obtain estimates for the transversal derivatives on the horizon. For the Bianchi equations, there are two (equivalent) possibilities to adapt this estimate: Either at the level of the tensorial Bianchi equations for the Weyl-tensor commuted with $N$,\footnote{provided one also commuted with $T$ first to eliminate the non-decaying $\rho$ component} or at the level of the null-Bianchi equations. In this paper, we follow the second approach, as it has the advantage that, at the lowest level, we can work with the renormalized null-Bianchi equations for $\rho$ and $\sigma$. Cf.~section \ref{redsection}.

\subsubsection{Decay at infinity} \label{infintro}
Very recently, Dafermos and Rodnianski introduced a novel approach to the decay of the wave equation $\Box \psi =0$ on black hole spacetimes \cite{DafRodnew}. Using a new multiplier, they were able to generate a certain hierarchy of estimates for $r$-weighted energies near infinity provided an integrated decay estimate in the interior is available. A version of the pigeonhole principle finally allowed them to exploit that hierarchy to translate boundedness of $r$-weighted energies into interior decay in $t$ of the natural $L^2$-energy.

This new method of obtaining decay at infinity and exporting it to the interior has an advantage over the traditional method of using the conformal Morawetz multiplier $K=\left(t^2+r^2\right) \partial_t + 4tr \partial_r$ (which is conformally Killing in Minkowski space). Namely, as the latter carries $t$-weights, the error-terms it generates in the interior are typically difficult to control. The multiplier of the new method, on the other hand, carries only $r$-weights and is only applied in a region close to null-infinity. This emphasizes the power of using multipliers locally, as one root of the difficulties with the $K$-vectorfield is precisely that it has to be applied globally.

It turns out that the approach of \cite{DafRodnew} has a very natural generalization to the Bianchi equations. The basic idea is easy to understand and can be explained at a heuristic level. Suppose we use a Minkowskian coordinate system $\left(u,v, \omega \right)$ near infinity, i.e.~such that in particular $\partial_u r \approx -\partial_v r \approx -\frac{1}{2}$ there. In fact, for the purpose of the following heuristic, the reader can in fact consider the Bianchi equations on exact Minkowski space. In this region, the Bianchi equations for the components $\alpha$ and $\beta$ take the form (cf.~section \ref{abrsabbb})
\begin{align}
2\partial_u \alpha - \frac{1}{r} \alpha = -2 \slashed{\mathcal{D}}_2^\star \beta + \textrm{l.o.t.} \textrm{ \ \ \ , \ \ \ } 2\partial_v \beta + \frac{4}{r} \beta = \slashed{div} \alpha + \textrm{l.o.t.}
\end{align}
Multiplying the first by $r^p \alpha$ and using that $\slashed{\mathcal{D}}_2^\star$ is the adjoint of $\slashed{div}$ on the $2$-spheres, we obtain 
\begin{align} \label{igh}
 \partial_u \left(r^p \|\alpha\|^2 \right) + 2\partial_v \left(r^p \|\beta \|^2\right) 
+ \|\alpha\|^2 r^{p-1} \left(-p \cdot  \partial_u r - 1 \right) \nonumber \\ + \|\beta\|^2 r^{p-1} \left(8 - 2 p \partial_v r \right) = \textrm{error} + \textrm{tot.~div on $S^2_{t,r}$} \, .
\end{align}
Upon integration in a characteristic region (the shaded region in the figure below) using the measure $du dv \sin \theta d\theta d\phi$
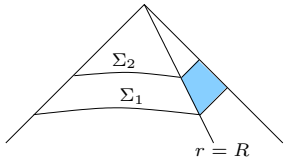
\[
\begin{picture}(0,0)%
\includegraphics{mintro.pstex}%
\end{picture}%
\setlength{\unitlength}{2289sp}%
\begingroup\makeatletter\ifx\SetFigFont\undefined%
\gdef\SetFigFont#1#2#3#4#5{%
  \reset@font\fontsize{#1}{#2pt}%
  \fontfamily{#3}\fontseries{#4}\fontshape{#5}%
  \selectfont}%
\fi\endgroup%
\begin{picture}(3024,1696)(889,-2045)
\put(2950,-2005){\makebox(0,0)[lb]{\smash{{\SetFigFont{7}{8.4}{\rmdefault}{\mddefault}{\updefault}{\color[rgb]{0,0,0}$r=R$}%
}}}}
\put(2134,-1411){\makebox(0,0)[lb]{\smash{{\SetFigFont{7}{8.4}{\rmdefault}{\mddefault}{\updefault}{\color[rgb]{0,0,0}$\Sigma_1$}%
}}}}
\put(2049,-1024){\makebox(0,0)[lb]{\smash{{\SetFigFont{7}{8.4}{\rmdefault}{\mddefault}{\updefault}{\color[rgb]{0,0,0}$\Sigma_2$}%
}}}}
\end{picture}%

\]
the first two terms in (\ref{igh}) will generate positive future boundary terms. The remaining (spacetime) terms on the left hand side will both be positive as long as $2<p<8$. In other words, provided we can control the error-term arising on the timelike boundary of the region (which can be done, provided an integrated decay estimate is available in the interior) and the error-terms on the right hand side (which are absent if the background was exactly Minkowskian and which cause a considerable amount of work in our case), we obtain an estimate for strongly $r$-weighted boundary \emph{and} spacetime terms for the components $\alpha$ and $\beta$. Considering the next pair of Bianchi equations, we will obtain a similar estimate for the pair $\left(\beta,\rho\right)$ with the condition $4<p<6$. In this fashion we can estimate all curvature components with appropriate $r$-weights. From the weighted spacetime-terms we can generate decay in $t$ in the interior using the pigeonhole principle as in \cite{DafRodnew}. Exploiting the hierarchy in the equation one easily obtains the well-known  pointwise decay rates for the null-components in Minkowski space  \cite{ChristKlei2}, for instance.

As mentioned above, a significant amount of work goes into estimating the error-terms on the right hand side. They impose additional non-linear constraints on the admissible $r^p$-weights. In spirit, these estimates are of course similar to the ones in \cite{ChristKlei}, since only the asymptotic region and weights in $r$ are concerned. We emphasize two important differences, however: On the one hand, we are not trying to minimize the number of derivatives. The cubic error-terms are always estimated by putting one term in $L^\infty$ and the others in $L^2$ (unlike $L^4$ estimates in \cite{ChristKlei}). Secondly, by using the null-structure equations directly on the error-terms, we exploit an additional cancellation of the terms which have the worst decay in $r$. This structure goes beyond the null form of the error-terms and simplifies considerably some of the estimates. 
The analysis at infinity is the content of section \ref{decinf1}.

\subsubsection{Integrated decay}
We have seen that both the redshift and the new method of capturing the decay at infinity require an integrated decay estimate in the interior to control their error-terms.\footnote{A boundedness statement can in fact be obtained without the construction of an integrated decay estimate. See section \ref{mtsec}.} How does one obtain such an estimate? 

From the experience with the wave equation, one may try to use a vectorfield of the form $X=f\left(r\right) \partial_r$ in the identity (\ref{bintro}). Carrying out the computation one observes that one cannot obtain a globally (positive, say) spacetime term, which controls all curvature components but that instead the quantities $\rho$ and $\sigma$ will always enter with the opposite sign. More precisely, the main (assuming the metric is exactly Schwarzschild) spacetime term reads
 \begin{align} \label{Xintro}
K_{1}^{XTT}  = \frac{\left(1-\mu\right)^2}{4}  f_{,r} \left[\frac{1}{2} \left(1-\mu\right)^2 |\underline{\alpha}|^2 + \frac{1}{2} \frac{|\alpha|^2}{\left(1-\mu\right)^2} - \frac{1-\mu}{2} \left(\rho^2+\sigma^2\right) \right] \nonumber \\
+f \left(\frac{1}{r} - \frac{3}{2}\frac{\mu}{r}\right)\left[\left(1-\mu\right)^2 |\underline{\beta}|^2 + |\beta|^2 + 2\left(1-\mu\right)\left(\rho^2+\sigma^2\right) \right] \, .
\end{align}
We immediately  recognize the familiar trapping factor of $(r-3M)$ in the second line. Choosing $f$ as a bounded function which changes sign at $r=3M$ we can control all components but $\rho$ and $\sigma$ in a region of bounded $r$ around the degenerating set $r=3M$.\footnote{Note that near infinity and near the horizon we obtain control over all curvature components with the given choice of $f$, since the terms in the second line eventually dominate the ones in the first.} The boundary-terms are in turn controlled by the $T$-energy. We remark that all these arguments need to be stabilized near the horizon by the redshift vectorfield (cf.~section \ref{XestWeyltensor}).

The above observation regarding the role of $\rho$ and $\sigma$ requires us to obtain an integrated decay estimate for $\rho$ and $\sigma$ via different means. The important insight here, which goes back to Price (see \cite{Price} and \cite{Chandrasekhar}), is that the components $\rho$ and $\sigma$ satisfy a wave equation (the so-called Regge-Wheeler equation), which in the case of exact Schwarzschild would be homogeneous and in our case involves inhomogeneous error-terms on the right hand side. Those terms exhibit a quadratic structure of (derivatives of) ``Ricci coefficients $\cdot$ curvature components".

The natural energies for these wave equations do not involve $\rho$ and $\sigma$ itself but rather the rescaled quantities $r^3 \rho$, $r^3\sigma$. On the one hand, this is good news because the quantity $\rho \cdot r^3$ approaches the mass (a constant) for late times and hence we can expect decay for the $L^2$-energy of its derivatives.\footnote{We remark that the $r$-weight which is being introduced in this way corresponds precisely to the maximal $r$-weight that the quantity $\rho$ gets associated via the $r$-weighted multiplier described in section \ref{infintro}}  On the other hand, this complicates the estimates for the error-terms appearing on the right hand side of the equation, as they now get multiplied with the weight of the rescaling factor, $r^3$. In other words, one cannot prove an integrated decay estimate in the interior without understanding the $r$-weighted energies for (all) the curvature components!

Here it turns out that the we can borrow an $\epsilon$ of the strongly $r$-weighted spacetime integrals that we have available at infinity. Coupling the estimate at infinity with the anticipated integrated decay decay estimate in the interior will allow us to absorb the $\epsilon$-contribution and close the estimate. 

Besides controlling the inhomogeneity, an additional problem in the derivation of energy estimates for the renormalized $\rho$ and $\sigma$ components arises from the Regge-Wheeler operator itself. As is well-known for the Regge-Wheeler equation on a fixed Schwarzschild background, an everywhere positive spacetime integral can only be derived if the angular momentum mode of the field satisfies $l \geq 2$. For us, however, the situation is more favorable because we are assuming that the metric approaches Schwarzschild and that we are estimating the curvature components of that particular spacetime. More precisely, the ultimately Schwarzschildean property implies that the zeroth order terms in the anticipated integrated decay estimate are a-priori bounded, while the derivative terms can easily be made non-negative with the help of an appropriate multiplier. In this way the integrated decay assumed in the ultimately Schwarzschildean assumption eliminates the $l=0$ and $l=1$ modes (see section \ref{rssection}).

With these difficulties resolved one finally obtains an integrated decay estimate for the renormalized $\rho$ and $\sigma$ components, which is coupled to the integrated decay estimate for the entire curvature tensor and then, finally, to the $r$-weighted estimates at infinity.

Having understood the basic mechanisms, we recommend the reader at this point to turn to section \ref{mtsec} immediately for a precise statement of the theorems. A brief glance at section \ref{norms} to familiarize with the energies which are being used may be helpful.

\subsection{Some remarks on the Kerr case}
It is natural to enquire which aspects of the mechanisms described stand a chance to (in principle) carry over to the case of a spacetime approaching a Kerr solution of, say, small angular momentum. In this case both the curvature components $\rho$ and $\sigma$ will be non-zero, at least if a null-frame based on the principal null-directions of Kerr is being used (otherwise, other null-components will be non-zero, too). We do not expect any further difficulty entering from the non-decaying $\sigma$-component. The analysis of the intricate error-terms proportional to $\rho$ in the commuted Bianchi equation carried out in section \ref{errorterms} will extend to the dual $\sigma$-component. The fact that $T$ is no-longer (not even asymptotically) timelike near the horizon, on the other hand, can -- at least for the case of small angular momentum -- be dealt with the redshift, cf.~\cite{DafRodKerr}. This leaves the trapping as the main difficulty. In view of the trapping for Kerr itself being well-understood, a generalized approach of Dafermos and Rodnianski \cite{Mihalisnotes} or Andersson-Blue \cite{AndBlue} (using commutation with the second order Carter-operator, i.e.~an ultimately Killing tensor in our language) is highly promising to yield similar results to the ones in the paper. We postpone the analysis of this case to a future paper.
\section{Ultimately Schwarzschildean spacetimes}
\subsection{The Schwarzschild geometry} \label{geoSchw}
In the perhaps most familiar $(t,r, \theta, \phi)$ coordinates, the Schwarzschild metric is given by
\begin{equation} \label{SSc}
 g = -\left(1-\frac{2M}{r}\right)dt^2 + \left(1-\frac{2M}{r}\right)^{-1} dr^2 +
r^2 \left(d\theta^2 + \sin^2 \theta d\phi^2\right) \, ,
\end{equation}
which defines a smooth metric on the manifold $\mathcal{M}_{I}=\left(-\infty, \infty\right) \times \left(2M,\infty\right) \times S^2$. 
Using the the transformations $u=t-r^\star$, $v=t+r^\star$, where $r^\star= r - 2M \log \left(r-2M\right) - M + \log 3M$ is the tortoise coordinate, we may express the metric on $\mathcal{M}_e$ in double null-coordinates $\left(u,v,\theta,\phi\right) \in \left(-\infty, \infty\right) \times \left(-\infty, \infty\right) \times S^2$ as 
\begin{equation}
 g = -\left(1-\mu\right) du dv + r^2  \left(d\theta^2 + \sin^2 \theta d\phi^2\right) \, ,
\end{equation}
where we set $\mu = \frac{2M}{r}$. It is well known that the singularity at $r=2M$ in (\ref{SSc}) is a coordinate singularity and that the metric can be extended through $r=2M$ to a larger manifold. Indeed, let $\chi\left(r\right)$ be an interpolating function which is equal to $1$ for $r\leq 6M$ and equal to $0$ for $r \geq 7M$. Setting
\begin{equation}
 t^\star = t + f\left(r\right) \textrm{ \ \ \ with \ \ $f^\prime\left(r\right) = \chi\left(r\right)\frac{2M}{r-2M}$} \nonumber  \textrm{ \ , \ \ we obtain}
\end{equation}
\begin{equation}
 g = -k_-\left(dt^\star\right)^2 + \frac{4M}{r} \chi dt^\star dr + \left[\chi^2 k_+ + \frac{1}{k_-}\left(1-\chi^2\right)\right] dr^2 
+ r^2 \left(d\theta^2 + \sin^2 \theta d\phi^2\right) \, ,\nonumber
\end{equation}
where we have introduced the shorthand notation $k_\pm = 1 \pm \mu$. Clearly, in the coordinates $(t^\star,r,\theta,\phi)$, $g$ can be defined on the larger manifold $\mathcal{M}_{II} = \left(-\infty, \infty\right) \times \left(0,\infty\right) \times S^2$. We easily see that $r=2M$ is a regular null-hypersurface, which we denote $\mathcal{H}^+$ and call the future event horizon.

Let us fix on $\mathcal{M}_{II} = \left(-\infty, \infty\right) \times \left(0,\infty\right) \times S^2$ a slice of constant $t^\star=\tau_0 \geq 1$ (note that this crosses $\mathcal{H}^+$) and define the submanifold 
\begin{align}
\mathcal{R} = \left[\tau_0, \infty \right) \times \left[2M, \infty\right) \times S^2 \, .
\end{align}
For convenience, let us also define 
\begin{equation}
k_\pm^\chi = 1 \pm \frac{2M}{r} \chi \textrm{ \ \ \ , \ \ \ }  k_{\chi} = g_{rr} = \left[\chi^2 k_+ + \frac{1}{k_-}\left(1-\chi^2\right)\right] = \frac{k_+^\chi k_-^\chi}{k_-} \, .
\end{equation}
\glossary{
name={$k_{\chi}$,$k_{\chi}^+$, $k_\chi^-$}, 
description={functions of $r$}
}

Note that $k_{\chi}$ is  bounded above and below by positive constants depending on $M$ only. 

Next we define various frames on $\mathcal{M}_{II}$. The vectors
\begin{equation}
 \hat{e}_{1,2} = \textrm{frame on $S^2_{t^\star,r}$} \, ,
\end{equation}
\begin{equation}
R = \frac{1}{\sqrt{k_{\chi}}} \partial_r \textrm{ \ \ \ and \ \ \ }  n = \sqrt{k_{\chi}}\partial_{t^\star} - \frac{2M\chi}{r\sqrt{k_{\chi}}} \partial_r 
\end{equation}
satisfy
\begin{equation}
 g\left(R , R\right) = -g\left(n , n \right) = + 1 \textrm{ \ \ \ as well as \ \ \ } g\left(R , n \right) = g\left(T,e_A\right) = g\left(R, e_A\right) = 0 \nonumber \, ,
\end{equation}
and hence constitute an orthonormal frame.
In addition, we may choose a null frame, i.e.~vectors $e_1, e_2, e_3, e_4$ such that
\begin{equation}
g\left(e_3, e_4\right)=-2 \textrm{ \ \ \ and \ \ \ } g\left(e_3,e_3\right) = g\left(e_4,e_4\right) = g\left(e_3, e_A\right) = g\left(e_4,e_A\right) = 0 \, , \nonumber
\end{equation}
as well as
\begin{equation}
 g\left(e_A,e_B\right) = \delta_{AB} \, ,
\end{equation}
for $A,B=1,2$. Such a frame may be defined via $\hat{e}_1, \hat{e}_2$,
\begin{equation} \label{sfchi}
 \hat{e}_3 = \frac{1}{\sqrt{k_{\chi}}} \left(n - R\right) = \partial_{t^\star} - \frac{1+\frac{2M}{r}\chi}{k_{\chi}} \partial_r \, ,
\end{equation}
\begin{equation}
 \hat{e}_4 = \frac{\sqrt{k_{\chi}}}{1} \left(n + R\right) = k_{\chi}\partial_{t^\star} + \left( 1-\frac{2M}{r}\chi\right) \partial_r \, .
\end{equation}

For Schwarzschild, the Ricci coefficients with respect to the null-frame $\left(\hat{e}_1,\hat{e}_2,\hat{e}_3,\hat{e}_4\right)$ are -- using the conventions of \cite{ChristKlei} summarized in  appendix \ref{UF}:
\begin{equation}
 \underline{H}_{AB} = -\frac{1}{r} \left(\frac{1+\frac{2M}{r}\chi}{k_{\chi}}\right) \delta_{AB} = -\frac{1}{r} \left(\frac{1-\frac{2M}{r}}{1-\frac{2M}{r}\chi}\right) \delta_{AB} \, ,
\end{equation}
\begin{equation}
 H_{AB} = \left(1-\frac{2M}{r}\chi\right)\frac{1}{r} \delta_{AB} \, ,
\end{equation}
\begin{equation}
 \underline{\Omega} = +\frac{M}{r^2} \cdot \frac{\left(1-\chi\right) + \chi^\prime r k_-}{\left(1-\frac{2M}{r}\chi\right)^2} \textrm{ \ \ \ , \ \ \ } \Omega = -\frac{M}{r^2} \left(\chi - r \chi^\prime\right) \, ,
\end{equation}
\begin{equation}
 Y_A = \underline{Y}_A = Z_A = \underline{Z}_A =V_A = 0 \, .
\end{equation}

Finally, we collect some useful formulae for further reference:
The timelike Killing field expressed in the chosen null-frame reads
\begin{equation}
 2T = 2\partial_{t^\star} = \left(1-\frac{2M}{r}\chi\right) \hat{e}_3 + \frac{1}{k_{\chi}}\left(1 + \frac{2M}{r}\chi\right)\hat{e}_4 \, .
\end{equation}
The normal of constant $t^\star$ slices $\Sigma_{t^\star}$ is computed to be
\begin{equation}
n = \sqrt{k_{\chi}}\partial_{t^\star} - \frac{2M\chi}{r\sqrt{k_{\chi}}} \partial_r = \frac{1}{2} \sqrt{k_\chi} \hat{e}_3 + \frac{1}{2} \frac{1}{\sqrt{k_\chi}} \hat{e}_4
\end{equation}
and hence the components
\begin{equation}
 n^{t^\star} = \sqrt{k_{\chi}} \textrm{ \ \ \ , \ \ \ } n^r = - \frac{2M\chi}{r\sqrt{k_{\chi}}}
\textrm{ \ \ \ , \ \ \ }
 n_{t^\star} = -\frac{1}{\sqrt{k_{\chi}}} \textrm{ \ \ \ ,  \ \ \ } n_r = 0 \, .
\end{equation}
\subsection{The class of ultimately Schwarzschildean spacetimes} \label{picsec}
Fix a Schwarzschild spacetime of mass $M$, $\left(\mathcal{R},g^M\right)$, equipped with the regular coordinate atlas $(t^\star, r, \theta_i, \phi_i)$ as defined in section \ref{geoSchw}. The Penrose diagram of $\left(\mathcal{R},g_M\right)$ is depicted below, where we use the standard notation $\mathcal{I}^+$ to denote future null-infinity.
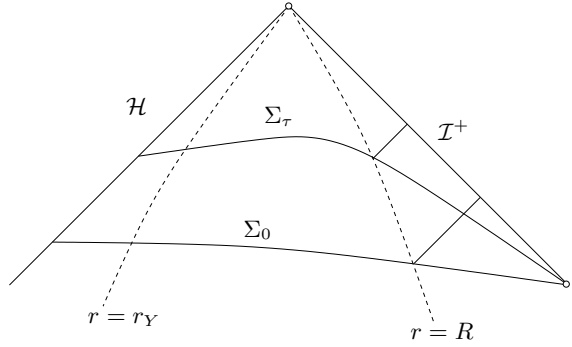
\[
\begin{picture}(0,0)%
\includegraphics{spin2set.pstex}%
\end{picture}%
\setlength{\unitlength}{1776sp}%
\begingroup\makeatletter\ifx\SetFigFont\undefined%
\gdef\SetFigFont#1#2#3#4#5{%
  \reset@font\fontsize{#1}{#2pt}%
  \fontfamily{#3}\fontseries{#4}\fontshape{#5}%
  \selectfont}%
\fi\endgroup%
\begin{picture}(7835,4777)(589,-5089)
\put(6199,-5013){\makebox(0,0)[lb]{\smash{{\SetFigFont{9}{10.8}{\rmdefault}{\mddefault}{\updefault}{\color[rgb]{0,0,0}$r=R$}%
}}}}
\put(4153,-1999){\makebox(0,0)[lb]{\smash{{\SetFigFont{9}{10.8}{\rmdefault}{\mddefault}{\updefault}{\color[rgb]{0,0,0}$\Sigma_{\tau}$}%
}}}}
\put(3872,-3584){\makebox(0,0)[lb]{\smash{{\SetFigFont{9}{10.8}{\rmdefault}{\mddefault}{\updefault}{\color[rgb]{0,0,0}$\Sigma_0$}%
}}}}
\put(6592,-2238){\makebox(0,0)[lb]{\smash{{\SetFigFont{9}{10.8}{\rmdefault}{\mddefault}{\updefault}{\color[rgb]{0,0,0}$\mathcal{I}^+$}%
}}}}
\put(2232,-1874){\makebox(0,0)[lb]{\smash{{\SetFigFont{9}{10.8}{\rmdefault}{\mddefault}{\updefault}{\color[rgb]{0,0,0}$\mathcal{H}$}%
}}}}
\put(1699,-4789){\makebox(0,0)[lb]{\smash{{\SetFigFont{9}{10.8}{\rmdefault}{\mddefault}{\updefault}{\color[rgb]{0,0,0}$r=r_Y$}%
}}}}
\end{picture}%

\]
Let us write
\glossary{
name={$\Sigma_{t^\star}$},
description={slices of constant $t^\star$}
}
$\Sigma_{t^\star}$ for slices of constant $t^\star$. 
Slices with  ``null-ears'' are denoted $\tilde{\Sigma}_{t^\star}$: More precisely, if $S^2_{t^\star, R}$ denotes the sphere of intersection of $\Sigma_{t^\star}$ and a fixed\footnote{The exact $R$ will be fixed later in the paper, with $\frac{1}{R}$ figuring as a potential source of smallness in the argument.} $r=R$ hypersurface, 
\glossary{
name={$R$}, 
description={$r=R$ is a fixed timelike hypersurface close to null-infinity ($R$ is large) }
}
we define
\glossary{
name={$N_{out} \left(S^2_{\tau,r}\right)$}, 
description={outgoing null-hypersurface emanating from $S^2_{\tau,r}$}
}
\begin{equation} 
N_{out} \left(S^2_{\tau,r}\right) := \textrm{future outgoing part of the lightcone emanating from $S^2_{\tau,r}$} \nonumber
\end{equation}
\glossary{name={$\tilde{\Sigma}_{t^\star}$},
description={slices of constant $t^\star$ for $r\leq R$ with null hypersurfaces $N_{out} \left(S^2_{\tau,R}\right)$ attached}
}
\begin{equation}
\tilde{\Sigma}_{t^\star} := \Big(\Sigma_{t^\star} \cup \{r \leq R \} \Big) \cap N_{out} \left(S^2_{t^\star,R}\right) \, .
\end{equation}
The spacetime slab enclosed by $\Sigma_{\tau_1}, \Sigma_{\tau_2}$ and the horizon is denoted by $\mathcal{M}\left({\tau_1},{\tau_2}\right)$, while the spacetime slab enclosed by $\tilde{\Sigma}_{\tau_1}, \tilde{\Sigma}_{\tau_2}$ and the horizon is denoted $\widetilde{\mathcal{M}}\left({\tau_1},{\tau_2}\right)$. 
We also define the asymptotic region $\mathcal{D}\left(\tau_1,\tau_2\right) = \widetilde{\mathcal{M}}\left(\tau_1,\tau_2\right) \cap \{r \geq R \}$. Close to the horizon we also fix a hypersurface $r=r_Y$, whose precise location will be given in Lemma \ref{rscon}.
\glossary{
name={$r_Y$}, 
description={$r=r_Y$ is a fixed timelike hypersurface close to the horizon (see Lemma \ref{rscon}) }
}

We will now fix the differentiable structure of $\mathcal{R}$ and consider small perturbations of the metric on the black hole exterior to the future of $\Sigma_0$. As the differentiable structure is fixed, $S^2_{t^\star,r}$, $\Sigma_{\tau}$, $\tilde{\Sigma}_{\tau}$,  $\mathcal{M}\left(\tau_1,\tau_2\right)$, $\tilde{\mathcal{M}}\left(\tau_1,\tau_2\right), etc.$ will remain well-defined with respect to the perturbed metric, and we will hence use the same notation for them, as long as no confusion arises. We denote the covariant derivative of the perturbed metric $g$ by $D$.
\begin{definition} \label{ultS}
Fix the manifold $\mathcal{R}$ and its differentiable structure.
A metric $g$ on $\mathcal{R}$ is called {\bf ultimately Schwarzschildean (of mass M) to order $\mathbf{k+1}$} ($k\geq 2$) if it has the following properties:
\begin{enumerate}
\item The boundary $\mathcal{H}^+$ is null with respect to $g$. 
 \item The metric $g$ is $C^{k-1}$-close to the Schwarzschild metric in that
\begin{equation}
 |g_{ij} - \left(g^{M}\right)_{ij}| +  |g^{ij} - \left(g^{M}\right)^{ij}| \leq \frac{\epsilon}{r^2}
\end{equation}
\begin{equation}
 |\partial^n_m g_{ij} - \partial^n_m \left(g^{M}\right)_{ij}| + |\partial^n_m g^{ij} - \partial^n_m \left(g^{M}\right)^{ij}| \leq \frac{\epsilon}{r^2}
\end{equation}
holds for $n=1,...,k-1$ and $m \in \{t^\star, r, \theta, \phi\}$. 
\item There exists an admissible null-frame $e_1, ... , e_4$ for $g$, i.e.~a frame satisfying
\begin{enumerate}
\item closesness to the Schwarzschild frame\footnote{In view of the fixed differentiable structure, the $\hat{e}_i$ (defined in terms of coordinate vectorfields in (\ref{sfchi})) remain well-defined with respect to $g$.} $\hat{e}_i$:
\begin{align}
|g\left(e_3, \hat{e}_4\right) + 2| + |g \left(e_4, \hat{e}_3\right) + 2 | + | g\left(e_A, \hat{e}_B\right) - \delta_{AB}| < \epsilon
\end{align}
while all other combinations are $\epsilon$-small in the above sense.
\item The outgoing null-hypersurfaces generated by $e_4$ foliate $\mathcal{R}$. More precisely, there is an optical function $u$ whose level surfaces are the outgoing null-hypersurfaces generated by $e_4$.  This optical function is normalized such that $u=\infty$ is associated with the horizon $\mathcal{H}^+$. Moreover, the $e_1,e_2$ are tangent to the spheres $S^2_{t^\star,u}$ arising from the intersection of the timelike $t^\star=const$ slices and the outgoing null-hypersurfaces $u=const$.
\item the gauge condition $2Y=g\left(D_4 e_A, e_4\right)=0$ 
\item the Ricci coefficients are ultimately Schwarzschildean to order $k$ [see Definition \ref{RRCapproach}] in this frame
\end{enumerate}
\item There exist functions $p: \mathcal{R} \rightarrow \mathbb{R}$ and $q: \mathcal{R} \rightarrow \mathbb{R}$ such that $p$ vanishes on the horizon and that the associated vectorfield $2T=p e_3 + q e_4$ is ultimately Killing to order $k+1$  [see Definition \ref{ultimateKilling}]. Moreover, the deformation tensor satisfies $tr {}^{(T)}\pi=0$
\glossary{
name={$T$}, 
description={$2T=p e_3 + q e_4$ ultimately Killing field}}
\glossary{
name={$p$}, 
description={$3$-component of the ultimately Killing field $T$, vanishing on the horizon, equal to $k_{\chi}^-$ in Schwarzschild}}
\glossary{
name={$q$}, 
description={$4$-component of the ultimately Killing field $T$, equal to $1$ on the horizon, equal to $\frac{k_{\chi}^+}{k_\chi}$ in Schwarzschild}}
and 
\begin{align}
|T\left(p\right) |+ |e_A \left(p\right)| + |T\left(q\right) |+ |e_A \left(q\right)| 
+ | \slashed{D}_3 \left[p - k_{\chi}^- \right] | \nonumber \\ +  | \slashed{D}_4 \left[p - k_{\chi}^- \right] |  +  | \slashed{D}_3 \left[q - \frac{k_{\chi}^+}{k_{\chi}} \right] | + \slashed{D}_4 \left[q - \frac{k_{\chi}^+}{k_{\chi}} \right] |< \epsilon \left(t^\star \right)^{-\frac{5}{4}}
\end{align}
holds in the interior region $r < t^\star$.
\end{enumerate}
\end{definition}
{\bf Remarks. } 
\begin{itemize}
\item Our assumptions are certainly not minimal and leave room for improvement. In particular, the decay properties on the deformation tensor in assumption (4) may be derivable from (3d).
\item The assumption that the frame satisfies $Y=0$ (which corresponds to a partial choice of gauge) will be of importance in the analysis at infinity, cf.~\cite{ChristKlei}. It can easily be achieved by Fermi-propagating $e_1, e_2$ in the $e_4$ direction.

\item The assumption that $T$ can be chosen to be null on the horizon (i.e.~the condition $p |_{\mathcal{H}^+}=0$) is not essential for the results of the paper. However, it allows an elegant proof of uniform boundedness of the highest order curvature norm. We will show later how this assumption can be removed in the context of integrated decay estimates.

\item The assumption that $tr {}^{(T)}\pi=0$ is a technical assumption which allows one to reduce the number of error-terms in the proof at some point. The additional errors could be controlled without further difficulty but would make the paper longer (cf.~section \ref{structureRHS}, Remark \ref{trpremark}, where this is discussed).
\end{itemize}

We will write $\mathbb{P}=\mathbb{P}\left(t^\star,u\right)$ for the projection from the tangent space of $\left(\mathcal{R},g\right)$ to the tangent space of $S^2_{t^\star,u}$. If $D$ denotes the covariant derivative on $\left(\mathcal{R},g\right)$, then we denote by $\slashed{D}_3, \slashed{D}_4$ the projections to $S^2_{t^\star,u}$ of $D_3$ and $D_4$. The covariant derivative intrinsic to $S^2_{t^\star,u}$ will be abbreviated $\slashed{\nabla}$. We will write $\mathcal{D}$ shorthand for any derivative $\mathcal{D} \in \{ \slashed{D}_3, \slashed{D}_4, \slashed{\nabla} \}$.
\subsubsection{The Ricci coefficients} \label{ricosec}
\glossary{name={$\mathfrak{R}$, $\mathfrak{R}^{main}$ }, 
description={collective notation for Ricci-Rotation coefficients, cf.~also the appendix}}
We denote the Ricci coefficients in an admissible null-frame collectively by $\mathfrak{R}$. While $\mathfrak{R}$ contains all components, we generally distinguish
\begin{equation}
\mathfrak{R}^{main}=\{ \widehat{H}, \underline{Y}, Z, \underline{Z}, V, \Omega, \underline{\Omega}, tr H, tr \underline{H} \} \,
\end{equation}
and $\widehat{\underline{H}}$, which has to be treated separately, as it contains information from the radiation field in its $\frac{1}{r}$-term.
Also, $\mathfrak{R}^{(main)}_{SS}$ denotes the corresponding Schwarzschild values.\footnote{That is the values of the components of the Ricci coefficients of $\left(\mathcal{R},g^M\right)$ evaluated with respect to the frame $\hat{e}_i$.} Defining the norm of the null-decomposed Ricci coefficients in the obvious way, i.e.~
\begin{equation}
 || T_{AB} || = \sum_{A,B=0,1} |T\left(e_A, e_B\right)|  \textrm{ \ \ \ and \ \ \ \ }  || U_{A} || = \sum_{A=0,1} |U\left(e_A\right)|,
\end{equation}
we can state
\begin{definition} \label{RRCapproach}
 The Ricci coefficients of $g$ are ultimately Schwarzschildean to order $k$ with respect to the null-frame $\left(e_1, e_2, e_3, e_4\right)$  if the following estimates hold:
\begin{itemize}
\item globally on the black hole exterior
\begin{equation} \label{jo1}
 \Big|tr H - \frac{2}{r} \left(1-\frac{2M}{r}\chi\right) \Big| \ + \ \Big|tr \underline{H} + \frac{2}{r}\left(\frac{1-\frac{2M}{r}}{1-\frac{2M}{r}\chi}\right) \Big| \leq \frac{\epsilon}{r^2}
\end{equation}
\begin{equation} \label{jo2}
 ||\widehat{H}_{AB}|| \leq \frac{\epsilon}{r^2} \textrm{ \ \ \ and \ \ \ } ||\widehat{\underline{H}}_{AB}||  \leq \frac{\epsilon}{r \max\left(1, t^\star -r\right)}
\end{equation}
\begin{equation}  \label{jo3b}
 |\Omega - \Omega_{SS}| + |\underline{\Omega}-\underline{\Omega}_{SS}| \leq \frac{\epsilon}{r^2} 
\end{equation}
\begin{equation}  \label{jo5}
  ||\underline{Y}_A|| + ||Z_A|| +  ||\underline{Z}_A|| \leq \frac{\epsilon}{r^2} 
\end{equation}

\item in the interior region, $r \leq t^\star$, the pointwise estimate
\begin{equation}
\|\mathfrak{R} - \mathfrak{R}_{SS}\| \leq \epsilon \left(t^\star\right)^{-\frac{5}{4}}
\end{equation}
i.e.~(\ref{jo1})-(\ref{jo5}) hold in $r \leq t^\star$ with all right hand sides replaced by $\epsilon \left(t^\star\right)^{-\frac{5}{4}}$.

\item The estimate
\begin{align} \label{iRdec}
 \mathbb{D}^{n} \left[\mathfrak{R}\right] \left(\tau\right) :=  \mathbb{D}^{n} \left[\mathfrak{R}\right] \left(\tau, \infty \right) < \frac{\epsilon}{\tau^{l\left(n\right)}} 
\end{align}
holds for $n=1,...,k$ and 
\begin{equation} \label{Rldef}
 l\left(n\right) = \left\{
\begin{array}{rl}
0 & \text{if } n = k\\
1 & \text{if } n = k-1\\
2 & \text{if } n = k-2 \\
\frac{5}{2} &  \text{if } n < k-2 \, ,
\end{array} \right.
\end{equation}
where
\glossary{
name={$\mathbb{D}^{n} \left[\mathfrak{R}\right] \left(\tau_1,\tau_2\right)$}, 
description={norm on the Ricci-coefficients}}
\begin{align} \label{Dnorm}
\mathbb{D}^{n} \left[\mathfrak{R}\right] \left(\tau_1,\tau_2\right)  = \sup_{\tau \in \left(\tau_1,\tau_2\right)}\overline{\mathbb{E}}^{n} \left[\mathfrak{R}\right] \left( \tilde{\Sigma}_{\tau}\right) +  \overline{\mathbb{E}}^{n} \left[\mathfrak{R}\right] \left( \mathcal{H} \left(\tau_1,\tau_2\right)\right) \nonumber \\  + \overline{\mathbb{I}}^{n} \left[\mathfrak{R}\right] \left(\tilde{\mathcal{M}} \left(\tau_1,\tau_2\right)\right)
\end{align}
with the weighted energies for the Ricci rotation coefficients ($\delta = \frac{1}{100}$)
\begin{align} \label{IR}
\overline{\mathbb{I}}^{n} \left[\mathfrak{R}\right] \left(\tilde{\mathcal{M}} \left(\tau_1,\tau_2\right)\right) =  \sum_{i=0}^n \sum_{k_1+..+k_3= i}  \int_{\tilde{\mathcal{M}} \left(\tau_1,\tau_2\right)} dt^\star dr d\omega \,  r^{2-\delta}  r^{2k_2+2k_3} \nonumber \\  \left(\sum_{j=0}^{l(n)} r^{min\left(2k_1-j+1,0\right)} \tau_1^{j-{l(n)}} \right) \Bigg[ \|\slashed{D}_3^{k_1} \slashed{D}_4^{k_2} \slashed{\nabla}^{k_3} \left(\mathfrak{R}^{main} - \mathfrak{R}^{main}_{SS}\right) \|^2 \nonumber \\ + \frac{1}{r} \|\slashed{D}_3^{k_1} \slashed{D}_4^{k_2} \slashed{\nabla}^{k_3}\widehat{\underline{H}} \|^2 \Bigg]
\end{align}
and\footnote{We write $r^2 dv d\omega$ for the induced (with respect to $e_4$) volume form on $N_{out}\left(S^2_{\tau,R}\right)$.}
\begin{align}
\overline{\mathbb{E}}^{n} \left[\mathfrak{R}\right] \left( \tilde{\Sigma}_{\tau}\right) =  \sum_{i=0}^n\int_{\tilde{\Sigma}_{\tau} \cap \{ r \leq R\}} \|\mathcal{D}^i \left(\mathfrak{R} - \mathfrak{R}_{SS}\right) \|^2+  \sum_{i=0}^n\sum_{k_1+k_2+k_3=i}  \nonumber \\  \int_{N_{out}\left(S^2_{\tau,R}\right)} dv d\omega r^{3-\delta} r^{2k_2+2k_3} \|\slashed{D}_3^{k_1} \slashed{D}_4^{k_2} \slashed{\nabla}^{k_3} \left(\mathfrak{R}^{main} - \mathfrak{R}^{main}_{SS}\right) \|^2 + \nonumber \\  \int_{N_{out}\left(S^2_{\tau,R}\right)} dv d\omega r^{1-\delta} r^{2k_2+2k_3} \|\slashed{D}_3^{k_1} \slashed{D}_4^{k_2} \slashed{\nabla}^{k_3} \widehat{\underline{H}} \|^2 \, 
\end{align}
as well as
\begin{align}
\overline{\mathbb{E}}^{n} \left[\mathfrak{R}\right] \left( \mathcal{H} \left(\tau_1,\tau_2\right)\right) = \sum_{i=0}^n \int_{\mathcal{H} \left(\tau_1,\tau_2\right)} \|\mathcal{D}^i \left(\mathfrak{R} - \mathfrak{R}_{SS}\right) \|^2 \, .
\end{align}
\item In addition, for the boundary term of the highest derivative
\begin{align} \label{awpr}
\sup_{\tau \in \left(\tau_1,\tau_2\right)}\overline{\mathbb{E}}^{k} \left[\mathfrak{R}\right] \left( \tilde{\Sigma}_{\tau}\right) +  \overline{\mathbb{E}}^{k} \left[\mathfrak{R}\right] \left( \mathcal{H} \left(\tau_1,\tau_2\right)\right)  < \epsilon \cdot \left(\tau_1\right)^{-\frac{1}{2}}
\end{align}
and in the asymptotic region
\begin{align} \label{awps}
\overline{\mathbb{I}}^{k} \left[\mathfrak{R}\right] \left(\mathcal{D} \left(\tau_1,\tau_2\right)\right) < \epsilon \cdot \left(\tau_1\right)^{-\frac{1}{2}} \, . 
\end{align} 
\end{itemize}
\end{definition}

{\bf Remarks:} 
\begin{itemize}
\item Clearly, taking a $4$- or an angular derivative improves the $r$-weight in the $L^2$-energy by a power of $2$: this is the factor $r^{2k_2+2k_3}$. 

\item By convention, the $j$ in $\sum_{j=0}^{l(n)}$ runs $0,1,2,\frac{5}{2}$. This sum incorporates that in general, one loses an $r$-weight for every power in $\tau$ gained. However, if a three-derivative is applied ($k_1>0$), decay in $\tau$ is gained allowing one to extract decay in $\tau$ without losing powers of $r$. This is the reason for the factor $min\left(2k_1-j+1,0\right)$. The component $\underline{H}$ is special as it decays less in $r$ near infinity (cf. (\ref{jo2})).

\item The additional assumptions (\ref{awpr}) and (\ref{awps}) are ``justified" by the fact that we are going to show $t^{-1}$ decay for the boundary terms and the \emph{degenerate} spacetime energy norm squared of curvature at this level of derivatives, which ``improves" this assumption. We remark that assumption (\ref{awps}) is technical and can be eliminated with some further work. It simplifies the error-estimates near infinity.

\item the precise reason for the various weights will become much more clear in conjunction with the curvature-energies and our method to prove decay for the latter. Note however at this point that the assumptions on $k$ derivatives of the Ricci-coefficients immediately imply estimates on $k-1$ derivatives of curvature, which are easily inferred from the null-structure equations of section (\ref{nseq}).
\end{itemize}

Before we proceed, we introduce two more energies which will be useful. First an energy on the untilded slices \begin{align} \label{unwR}
\overline{\mathbb{E}}^{n} \left[\mathfrak{R}\right] \left({\Sigma}_{\tau}\right) =  \sum_{i=0}^n\int_{{\Sigma}_{\tau}} \|\mathcal{D}^i \left(\mathfrak{R} - \mathfrak{R}_{SS}\right) \|^2
\end{align}
and second, a $\lambda$-weighted energy on the Ricci-coefficients:
\glossary{
name={$\mathbb{D}^k_\lambda \left[\mathfrak{R}\right] \left(\tau_1, \tau_2\right)$}, 
description={$\lambda$-weighted energy on the Ricci-coefficients}
}
\begin{align} \label{lambdawn}
\mathbb{D}^{n}_\lambda \left[\mathfrak{R}\right] \left(\tau_1,\tau_2\right) = \sup_{\tau} \overline{\mathbb{E}}^{n} \left[\mathfrak{R}\right]\left(\tilde{\Sigma}_{\tau}, \mathcal{H} \right) + B_{\lambda} \cdot {\mathbb{D}}^{n-1} \left[\mathfrak{R}\right]\left(\tau_1,\tau_2\right) 
\nonumber \\ +
 \lambda \cdot \overline{\mathbb{I}}^{n} \left[\mathfrak{R}\right]\left(\tilde{\mathcal{M}}\left(\tau_1,\tau_2\right)\right) + \overline{\mathbb{I}}^{n} \left[\mathfrak{R}\right] \left(\mathcal{D} \left(\tau_1,\tau_2 \right)\right) 
\end{align}
and $\mathbb{D}^{n}_\lambda \left[\mathfrak{R}\right] \left(\tau\right) :=\mathbb{D}^{n}_\lambda \left[\mathfrak{R}\right] \left(\tau, \infty\right)$.
Note that this energy only requires a $\lambda$-small contribution of non-degenerate integrated decay for $n$ derivatives of the Ricci coefficients near the photon sphere, since the last term regards only the asymptotic region. In particular, we see that if $\left(\mathcal{M}, g\right)$ is ultimately Schwarzschildean to order $k+1$, then $\mathbb{D}^{k}_\lambda \left[\mathfrak{R}\right] \left(\tau\right)$ is $\epsilon \cdot \lambda$-small for sufficiently large $\tau$. This definition is again tailored to what we are going to show for the curvature, cf.~the remarks at the end of section \ref{mtsec}. 

\subsubsection{The timelike ultimately Killing field}
We recall the definition of the deformation tensor of a vectorfield $\mathcal{X}$. It is the  covariant two-tensor 
\begin{equation}
\phantom{}^{(\mathcal{X})}\pi \left(\mathcal{Y},\mathcal{Z}\right) = \frac{1}{2} \left(g\left(\mathcal{Z},D_\mathcal{Y} \mathcal{X}\right) + g\left(\mathcal{Y},D_\mathcal{Z} \mathcal{X}\right)\right)
\end{equation}
whose norm is computed in the orthonormal frame:
\begin{equation}
\| \pi \| := \sum_{i,j=1,..4} |\pi\left(e_i, e_j\right)| \, .
\end{equation}
We also introduce the null-decomposition of its traceless part
\begin{align}
\phantom{}^{(\mathcal{X})}\mathbf{i}_{AB} = \phantom{}^{(\mathcal{X})} \widehat{\pi}_{AB}  \textrm{ \ \ \ , \ \ \ } \phantom{}^{(\mathcal{X})}\mathbf{j} = \phantom{}^{(\mathcal{X})} \widehat{\pi}_{34} \, , \nonumber \\ 
\phantom{}^{(\mathcal{X})}\mathbf{m}_{A} = \phantom{}^{(\mathcal{X})} \widehat{\pi}_{4A}  \textrm{ \ \ \  , \ \ \ } \phantom{}^{(\mathcal{X})}\mathbf{\underline{m}}_A = \phantom{}^{(\mathcal{X})} \widehat{\pi}_{3B}  \, , \nonumber \\
\phantom{}^{(\mathcal{X})}\mathbf{n} = \phantom{}^{(\mathcal{X})} \widehat{\pi}_{44}  \textrm{ \ \ \ , \ \ \ } \phantom{}^{(\mathcal{X})}\mathbf{\underline{n}} = \phantom{}^{(\mathcal{X})} \widehat{\pi}_{33} \, .
\end{align}

\begin{definition} \label{ultimateKilling}
The vectorfield $T$ is called {\bf ultimately Killing to order $\mathbf{k+1}$} if its deformation tensor satisfies
\begin{itemize}
\item on the entire black hole exterior the pointwise estimate
\begin{equation}
 \|\phantom{}^{(T)}\pi \textrm{ \ \ except ${}^{(T)}i_{AB}$}\| \leq \frac{\epsilon}{r^2}  \textrm{\ \ \  as well as \ \ \  } \|{}^{(T)}i_{AB}\| \leq \frac{\epsilon}{r \cdot \max\left(1, t^\star-r\right)} \nonumber
\end{equation}
 \item in the interior region, $r \leq t^\star$, the pointwise estimate
\begin{equation} \label{decaydeft}
\|\phantom{}^{(T)}\pi\| \leq \epsilon \left(t^\star\right)^{-\frac{5}{4}}
\end{equation}
\item For any $\tau_0 \leq \tau_1 < \tau_2$
\begin{align}
\sum_{k_1+k_2+k_3=i} \int_{\tilde{\mathcal{M}}\left(\tau_1, \tau_2\right)} dt^\star dr d\omega \, r^{2-\delta} \left(\sum_{j=0}^{l(n)} r^{min\left(2k_1-j+2,1\right)} \tau_1^{j-{l(n)}} \right) r^{2k_2+2k_3} \Big( \nonumber \\ 
\|\slashed{D}_3^{k_1} \slashed{D}_4^{k_2} \slashed{\nabla}^{k_3}  \ \phantom{}^{(T)}\left(\mathbf{j},\mathbf{m},\underline{\mathbf{m}},\mathbf{n},\underline{\mathbf{n}}\right) \|^2 
+ \frac{1}{r}\|\slashed{D}_3^{k_1} \slashed{D}_4^{k_2} \slashed{\nabla}^{k_3}  \ \phantom{}^{(T)}\mathbf{i} \|^2 \Big)
< \frac{\epsilon}{\tau_1^{l(n)}} \nonumber
\end{align}
where $i=1,...,k$ and $l\left(n\right)$ as in (\ref{Rldef}).
\end{itemize}
\end{definition}
\section{The Bianchi equations and the Bel-Robinson identity}
In this section we collect the main energy identities and commutation formulae that we are going to need in the paper. We use the conventions of \cite{ChristKlei} and refer the reader to this reference for detailed derivations.
 
\subsection{The energy identity}
As we are going to work with higher order energies arising from commuting the Bianchi equation with various vectorfields, it is convenient to study directly the more general inhomogeneous Bianchi equations
\begin{equation} \label{h2}
 D^{\alpha} \mathcal{W}_{\alpha \beta \gamma \delta} = \mathcal{J}_{\beta \gamma \delta} \, ,
\end{equation}
for any tensor $\mathcal{W}$ which has the algebraic properties of a Weyl-tensor. Defining the left and right duals 
\begin{align}
{}^\star  \mathcal{W}_{\alpha \beta \gamma \delta} = \frac{1}{2} \in_{\alpha \beta \mu \nu}  \mathcal{W}^{\mu \nu}_{\phantom{\mu \nu}\gamma \delta} \textrm{ \ \ \ \ and \ \ \ \ }  \mathcal{W}^\star_{\phantom{\star} \alpha \beta \gamma \delta} = \frac{1}{2}  \mathcal{W}_{\alpha \beta}^{\phantom{\alpha \beta}\mu \nu} \in_{\mu \nu \gamma \delta} \nonumber \\   \mathcal{J}^\star_{\phantom{\star} \beta \gamma \delta} = \frac{1}{2}  \mathcal{J}_{\beta \mu \nu} \in^{\mu \nu}_{\phantom{\mu \nu}\gamma \delta} 
\end{align}
with $\in_{\alpha \beta \gamma \delta}$ denoting the volume form of $g$, we can also write (\ref{h2}) as
\begin{align} 
D^{\alpha} \phantom{}^\star \mathcal{W}_{\alpha \beta \gamma \delta} = \mathcal{J}^\star_{\beta \gamma \delta} \, .
\end{align}
We define the symmetric, traceless Bel-Robinson tensor
\begin{equation}
 Q\left[\mathcal{W}\right]_{\alpha \beta \gamma \delta} = \mathcal{W}_{\alpha \rho \gamma \sigma} \mathcal{W}_{\beta\phantom{\rho}\delta\phantom{\sigma}}^{\phantom{\beta}\rho\phantom{\delta}\sigma} + \phantom{}^\star \mathcal{W}_{\alpha \rho \gamma \sigma} \phantom{}^\star \mathcal{W}_{\beta\phantom{\rho}\delta\phantom{\sigma}}^{\phantom{\beta}\rho\phantom{\delta}\sigma} \, .
\end{equation}
It satisfies the identity (\ref{bintro}), from which we derive 
\begin{align} \label{mainid}
 \int_{\Sigma_{\tau_2}} Q\left[\mathcal{W}\right]\left(\mathcal{X},\mathcal{Y},\mathcal{Z},n_{\Sigma} \right) d\mu_{\Sigma} + \int_{\mathcal{H}(\tau_1,\tau_2)} Q\left[\mathcal{W}\right]\left(\mathcal{X},\mathcal{Y},\mathcal{Z},n_{\mathcal{H}} \right) d\mu_{\mathcal{H}} \nonumber \\ + \int_{\mathcal{M}\left(\tau_1,\tau_2\right)}  \left(K_1^{\mathcal{X}\mathcal{Y}\mathcal{Z}}\left[\mathcal{W}\right] + K_2^{\mathcal{X}\mathcal{Y}\mathcal{Z}}\left[\mathcal{W}\right]\right) d\mu   =  \int_{\Sigma_{\tau_1}} Q\left[\mathcal{W}\right]\left(\mathcal{X},\mathcal{Y},\mathcal{Z},n_{\Sigma}\right) d\mu_{\Sigma} \, ,
\end{align}
for vectorfields $\mathcal{\mathcal{X}},\mathcal{\mathcal{Y}},\mathcal{\mathcal{Z}}$, where
\glossary{
name={$K_1^{\mathcal{X}\mathcal{Y}\mathcal{Z}}\left[\mathcal{W}\right]$, $ K_2^{\mathcal{X}\mathcal{Y}\mathcal{Z}}\left[\mathcal{W}\right] $ }, 
description={spacetime terms in the main-energy identity (\ref{mainid})}}
\begin{equation}
 K_1^{\mathcal{X}\mathcal{Y}\mathcal{Z}}\left[\mathcal{W}\right] = Q^{\alpha \beta \gamma \delta} \left(\phantom{}^{(\mathcal{X})}\pi_{\alpha \beta} \mathcal{Y}_\gamma \mathcal{Z}_{\delta} + \phantom{}^{(\mathcal{Y})}\pi_{\alpha \beta} \mathcal{Z}_\gamma \mathcal{X}_{\delta} + \phantom{}^{(\mathcal{Z})}\pi_{\alpha \beta} \mathcal{X}_\gamma \mathcal{Y}_{\delta} \right)
\end{equation}
\begin{equation} \label{K2}
 K_2^{\mathcal{X}\mathcal{Y}\mathcal{Z}} \left[\mathcal{W}\right] = \left(div Q\right)_{\beta \gamma \delta}\mathcal{X}^\beta \mathcal{Y}^\gamma \mathcal{Z}^\delta
\end{equation}
where
\begin{equation} \label{divQ}
\left(div Q\right)_{\beta \gamma \delta} = \mathcal{W}_{\beta\phantom{\mu}\delta\phantom{\nu}}^{\phantom{\beta}\mu\phantom{\delta}\nu} \mathcal{J}_{\mu \gamma \nu} + \mathcal{W}_{\beta\phantom{\mu}\gamma\phantom{\nu}}^{\phantom{\beta}\mu\phantom{\gamma}\nu} \mathcal{J}_{\mu \delta \nu} + \phantom{}^\star \mathcal{W}_{\beta\phantom{\mu}\delta\phantom{\nu}}^{\phantom{\beta}\mu\phantom{\delta}\nu} \mathcal{J}^{\star}_{\mu \gamma \nu} + \phantom{}^\star \mathcal{W}_{\beta\phantom{\mu}\gamma\phantom{\nu}}^{\phantom{\beta}\mu\phantom{\gamma}\nu} \mathcal{J}^{\star}_{\mu \delta \nu} \, .
\end{equation}
An analogous identity holds for the tilded slices $\tilde{\Sigma}$ and $\tilde{\mathcal{M}}$. In that case, there is an additional boundary term at null-infinity.

Regarding the boundary term, we note \cite{ChristKlei} that $Q\left[\mathcal{W}\right] \left(\mathcal{X},\mathcal{Y},\mathcal{Z},n\right)$ is non-negative for any set of causal vectorfields $\mathcal{X},\mathcal{Y},\mathcal{Z},n$. It will sometimes be convenient to evaluate this boundary-term for the Schwarzschildean normal $\tilde{n}_{\Sigma}=\frac{1}{2} k_{\chi} e_3 + \frac{1}{2}  \frac{1}{k_{\chi}} e_4 $, which is $C^1$-close to $n_{\Sigma}$ by the ultimately Schwarzschildean assumption and use the estimate 
\begin{align} \label{bndyssrel}
|Q\left[\mathcal{W}\right] \left(\mathcal{X},\mathcal{Y},\mathcal{Z},n_{\Sigma} \right) - Q \left[\mathcal{W}\right]\left(\mathcal{X},\mathcal{Y},\mathcal{Z},\tilde{n}_{\Sigma}\right)| \leq \frac{\epsilon}{r^2} Q\left[\mathcal{W}\right] \left(N,N,N,N\right) \, ,
\end{align}
where $N$ is a timelike vectorfield with $b \leq g\left(N,N\right)\leq B$ for constants $b,B$ depending on $M$ only.

For future applications, we finally collect the null-components of the deformation tensor of a vectorfield of the form
\begin{align}
\mathcal{X} = \mathfrak{B} e_3 + \mathfrak{D} e_4 \textrm{ \ \ \ \ \ with $\mathfrak{B}$ and $\mathfrak{D}$ \ $C^1$-bounded functions} \, .
\end{align}
They read (writing the short-hand $\pi$ for ${}^{(\mathcal{X})}\pi$)
\begin{equation} \label{nucos}
 \pi^{34} = \frac{1}{4} \left[-e_3\left(\mathfrak{B}\right) - e_4 \left(\mathfrak{D}\right) + 2\mathfrak{B} \underline{\Omega} + 2\mathfrak{D} \Omega \right] \, ,
\end{equation}
\begin{equation}
 \pi^{33} = \frac{1}{4} \left[-2e_4\left(\mathfrak{B}\right) - 4\mathfrak{B} \Omega\right] \textrm{ \ \ \ , \ \ \ }  \pi^{44} = \frac{1}{4} \left[-2e_3\left(\mathfrak{D}\right) - 4\mathfrak{D} \underline{\Omega}\right] \, ,
\end{equation}
\begin{equation}
 \pi^{AB} = \mathfrak{B} \underline{H}^{AB} + \mathfrak{D} H^{AB} \, ,
\end{equation}
\begin{equation}
 \pi^{4A} = -\frac{1}{2} \left[-e_A\left(\mathfrak{D}\right) +\mathfrak{B} \underline{Y}^A + \mathfrak{D} Z^A + \mathfrak{D} V^A \right] \, ,
\end{equation}
\begin{equation} \label{nucoe}
 \pi^{3A} = -\frac{1}{2} \left[-e_A\left(\mathfrak{B}\right) -\mathfrak{B} V^A + \mathfrak{B} \underline{Z}^A + \mathfrak{D} Y^A \right] \, .
\end{equation}

\subsection{The (tensorial) Bianchi equations under commutation} \label{commute}
We now state the commutation formulae that arise from commuting the Bianchi equations with vectorfields.  For this we need the notion of a modified Lie-derivative (cf.~\cite{ChristKlei}), which ensures that $\widehat{\mathcal{L}}_X W_{\alpha \beta \gamma \delta}$ is again a Weyl field.
\begin{definition}
The modified Lie derivative of a Weyl-tensor $W$ is defined as
\glossary{
name={$\widehat{\mathcal{L}}_X$},
description={modified Lie-derivative to make $\widehat{\mathcal{L}}_X W$ traceless}
}
\begin{align}
\widehat{\mathcal{L}}_X W := \mathcal{L}_X W - \frac{1}{2} {}^{(X)}[W] + \frac{3}{8} tr {}^{(X)}\pi \, W
\end{align}
with
\begin{align}
{}^{(X)}[W]_{\alpha \beta \gamma \delta} = \pi^{\mu}_{\phantom{\mu} \alpha}W_{\mu \beta \gamma \delta} + \pi^{\mu}_{\phantom{\mu} \beta}W_{\alpha \mu \gamma \delta} + \pi^{\mu}_{\phantom{\mu} \gamma}W_{\alpha \beta \mu \delta} + \pi^{\mu}_{\phantom{\mu} \delta}W_{\alpha \beta \gamma \mu} \, .
\end{align}
\end{definition}
\begin{proposition}
Let
\begin{equation}
 D^\alpha W_{\alpha \beta \gamma \delta} = j_{\beta \gamma \delta}
\end{equation}
and $X$ an arbitrary vectorfield. Then
\begin{align} \label{Weylcommute1}
D^\alpha\left(\widehat{\mathcal{L}}_X W_{\alpha \beta \gamma \delta} \right) = \widehat{\mathcal{L}}_X j_{\beta \gamma \delta} + \frac{1}{2} \widehat{\pi}^{\mu \nu} D_{\nu} W_{\mu \beta \gamma \delta} + \frac{1}{2} D^\alpha \widehat{\pi}_{\alpha \lambda} W^\lambda_{\phantom{\lambda}\beta \gamma \delta} \nonumber \\ + \frac{1}{2} \Bigg( \left(D_{\beta} \widehat{\pi}_{\alpha \lambda} - D_{\lambda} \widehat{\pi}_{\alpha \beta} \right) W^{\alpha \lambda}_{\phantom{\alpha \lambda} \gamma \delta} + \left(D_{\gamma} \widehat{\pi}_{\alpha \lambda} - D_{\lambda} \widehat{\pi}_{\alpha \gamma} \right) W^{\alpha \phantom{\beta} \lambda}_{\phantom{\alpha}\beta \phantom{\lambda} \delta} \nonumber \\ + \left(D_{\delta} \widehat{\pi}_{\alpha \lambda} - D_{\lambda} \widehat{\pi}_{\alpha \delta} \right) W^{\alpha \phantom{\beta \gamma} \lambda}_{\phantom{\alpha}\beta \gamma \phantom{\lambda}}\Bigg) = \widehat{\mathcal{L}}_X j_{\beta \gamma \delta} + J_{\beta \gamma \delta}\left(X,W\right) \, .
\end{align}
\glossary{
name={$J \left(X,W\right)$}, 
description={error-term arising in the context of commutation with a vectorfield $X$}
}
Here the last equality is to be understood as a definition of $J_{\beta \gamma \delta}\left(X,W\right)$.
\end{proposition}
\begin{proof}
This is Proposition 7.1.2 of \cite{ChristKlei}. 
\end{proof}
\begin{corollary}
\begin{equation} \label{nXWeylcommute}
D^\alpha \left(\widehat{\mathcal{L}}^n_T \widehat{\mathcal{L}}_T W\right)_{\alpha \beta \gamma \delta} = \widehat{\mathcal{L}}^n_T J_{\beta \gamma \delta} \left(T,W\right) + \sum_{k=0}^{n-1} \widehat{\mathcal{L}}^k_T J_{\beta \gamma \delta} \left(T, \widehat{\mathcal{L}}^{n-1-k}_T \widehat{\mathcal{L}}_T W\right) \, .
\end{equation}
\end{corollary}
%
%
%
%

\subsection{Bianchi equations for the null-curvature components} \label{abrsabbb}
For $\mathcal{W}$ a Weyl tensor we define its null-components
\begin{equation}
 \underline{\alpha}_{\mu \nu} = \mathbb{P}^\rho_\mu \mathbb{P}^\sigma_\nu \mathcal{W}_{\rho \gamma \sigma \delta} \left(e_3\right)^\gamma \left(e_3\right)^\delta \, ,
\end{equation}
\begin{equation}
 \alpha_{\mu \nu} =\Pi^\rho_\mu \mathbb{P}^\sigma_\nu \mathcal{W}_{\rho \gamma \sigma \delta} \left(e_4\right)^\gamma \left(e_4\right)^\delta  \, ,
\end{equation}
\begin{equation}
  \underline{\beta}_{\mu} = \frac{1}{2} \mathbb{P}^\rho_\mu \mathcal{W}_{\rho \gamma \sigma \delta} \left(e_3\right)^\sigma \left(e_3\right)^\gamma \left(e_4\right)^\delta  \, ,
\end{equation}
\begin{equation}
  \beta_{\mu} = \frac{1}{2} \mathbb{P}^\rho_\mu \mathcal{W}_{\rho \gamma \sigma \delta} \left(e_4\right)^\sigma \left(e_3\right)^\gamma \left(e_4\right)^\delta  \, ,
\end{equation}
\begin{equation}
  \rho = \frac{1}{4} \mathcal{W}_{\rho \gamma \sigma \delta} \left(e_3\right)^\rho \left(e_4\right)^\sigma \left(e_3\right)^\gamma \left(e_4\right)^\delta  \, ,
\end{equation}
\begin{equation}
  \sigma = \frac{1}{4} \epsilon^{\rho \gamma} \mathcal{W}_{\rho \gamma \sigma \delta}  \left(e_3\right)^\sigma \left(e_4\right)^\delta  \, ,
\end{equation}
with $\mathbb{P}^\rho_\mu$ the projection onto the $2$-spheres $S^2_{t^\star,u}$. A computation reveals

\begin{lemma}
For Schwarzschild and its null-frame $\hat{e}_0, \hat{e}_1, \hat{e}_2, \hat{e}_3$ all curvature components are zero except
\begin{equation}
\rho = -\frac{2M}{r^3} =: \overline{\rho} \, .
\end{equation}
\end{lemma}

We will consider certain derivatives of these null-components in the null directions. The following definitions are useful for a (differentiated) null-component $u$ of signature $s$ and rank $k$ (cf.~\cite{ChristKlei}):
\begin{align} \label{subnot}
u_3 &= \slashed{D}_3 u + \frac{3-s\left(u\right)+k\left(u\right)}{2} tr \underline{H} u \, , \nonumber \\
u_4 &= \slashed{D}_4 u + \frac{3+s\left(u\right)+k\left(u\right)}{2} tr H u \, .
\end{align}
Here, the signature is defined recursively as follows. For a $p$-tupel ($p\geq 0$) $n_p$ of $3$'s and $4$'s (in any order) 
\begin{align}
s\left(\alpha_{n_p} \right) &= 2 + \left(\textrm{number of $4$'s - number of $3$'s in $n_p$}\right)  \, ,\nonumber \\
s\left(\beta_{n_p} \right) &= 1 + \left(\textrm{number of $4$'s - number of $3$'s in $n_p$}\right)  \, ,
\nonumber \\
s\left(\rho_{n_p}, \sigma_{n_p} \right) &= 0 + \left(\textrm{number of $4$'s - number of $3$'s in $n_p$}\right)  \, ,
\nonumber \\
s\left(\underline{\beta}_{n_p} \right) &= -1 + \left(\textrm{number of $4$'s - number of $3$'s in $n_p$}\right)  \, ,\nonumber \\ 
s\left(\underline{\alpha}_{n_p} \right) &= -2 + \left(\textrm{number of $4$'s - number of $3$'s in $n_p$}\right) \nonumber \, .
\end{align}
The rank $k$ is simply the length of the tuple, e.g.~$k\left(\alpha_{n_p}\right) = p$.
Note that $s\left(\alpha_{4n_p}\right)= s\left(\alpha_{n_p}\right)+1$ and $s\left(\alpha_{3n_p}\right)= s\left(\alpha_{n_p}\right)-1$. We also define the short-hand notation
\begin{equation}
\vartheta^\pm \left(u_{n_p}\right) = \frac{3 \pm s\left(u_{n_p}\right) + p}{2} \, .
\end{equation}
Following \cite{ChristKlei}, we use the following notation for operators on the spheres $S^2_{t^\star,u}$:
\begin{itemize}
\item $\slashed{\mathcal{D}}_1$ takes any $S^2$-tangent $1$-form $\xi$ into the pair of functions $\left(\slashed{div} \xi, \slashed{curl} \xi\right)$
\item $\slashed{\mathcal{D}}_1^\star$, the $L^2$-adjoint of $\slashed{\mathcal{D}}_1$, takes any pair of scalars $\rho, \sigma$ into the $S^2$-tangent $1$-form $-\slashed{\nabla}_A \rho + \epsilon_{AB} \slashed{\nabla}^B \sigma$.
\item $\slashed{\mathcal{D}}_2$ takes any $2$-covariant symmetric traceless tensor $\xi$ into the $S^2$-tangent $1$-form $\slashed{div} \xi$.
\item $\slashed{\mathcal{D}}_2^\star$, the $L^2$ adjoint of $\slashed{\mathcal{D}}_2$ takes any $S^2$-tangent $1$-form $\xi$ into the $2$-form $-\frac{1}{2} \left(\slashed{\nabla}_B \xi_A + \slashed{\nabla}_A \xi_B - \left(\slashed{div} \xi\right) \slashed{g}_{AB}\right)$.
\end{itemize}
Finally, we decompose the Bianchi equations as equations for the null-curvature components (see Proposition \ref{GNDK} for the definition of $\Theta \left(J\right), \underline{\Theta}\left(J\right)$):
\begin{align} \label{Bianchi1}
\alpha_3 &= \slashed{D}_3 \alpha + \frac{1}{2}tr \underline{H} \alpha =  - 2\slashed{\mathcal{D}}_2^\star \beta+  E_3\left(\alpha\right) \\
E_3\left(\alpha\right)
&=  4\underline{\Omega} \alpha - 3 \left( \widehat{H} \rho + {}^\star \widehat{H} \sigma\right) + \left(V+4Z\right)  \widehat{\otimes} \beta  - 2 \Theta\left(J\right) \nonumber 
\end{align}
\begin{align} \label{Bianchi2}
\beta_4 &= \slashed{D}_4 \beta + 2 tr H \beta =  \slashed{div} \alpha + E_4\left(\beta\right)  \\
E_4\left(\beta\right)&=  -2\Omega \beta + \left(2V + \underline{Z} \right) \cdot \alpha + 3Y \rho +  3\left( {}^\star Y\sigma\right) -J_{4A4} \nonumber
\end{align}
\begin{align} \label{Bianchi3}
\underline{\alpha}_4 &= \slashed{D}_4 \underline{\alpha}+ \frac{1}{2}tr H \underline{\alpha} =+ 2\slashed{\mathcal{D}}_2^\star \underline{\beta}  +E_4\left(\underline{\alpha}\right)  \\ \nonumber
E_4\left(\underline{\alpha}\right) &= 4\Omega \underline{\alpha}  -3 \widehat{\underline{H}} \rho - 3 \left( - {}^\star \widehat{\underline{H}} \sigma\right) + \left(V-4\underline{Z}\right) \widehat{\otimes} \underline{\beta} - 2\underline{\Theta}\left(J\right) \nonumber
\end{align}
\begin{align} \label{Bianchi4}
\underline{\beta}_3 &= \slashed{D}_3 \underline{\beta} + 2 tr \underline{H} \underline{\beta}  = - \slashed{div} \underline{\alpha}+  E_3\left(\underline{\beta}\right)  \\
E_3\left(\underline{\beta}\right) &=  -2\underline{\Omega}\underline{\beta} + \left(2V - Z \right) \cdot \underline{\alpha}  -3\underline{Y} \rho  + 3\left({}^\star \underline{Y}\sigma\right) + J_{3A3} \nonumber 
\end{align}
\begin{align} \label{Bianchi5}
\beta_3 &= \slashed{D}_3 \beta + tr \underline{H} \beta = + \slashed{\mathcal{D}}_1^\star \left(-\rho, \sigma\right) + E_3\left(\beta\right)   \\ 
E_3\left(\beta\right) &= 2\underline{\Omega} \beta  + 2 \widehat{H} \cdot \underline{\beta} + \underline{Y} \cdot \alpha + 3Z \rho + 3 \left( {}^\star Z \sigma\right) \nonumber + J_{3A4}
\end{align}
\begin{align} \label{Bianchi10}
\underline{\beta}_4 &= \slashed{D}_4 \underline{\beta} + tr H \underline{\beta}=  + \slashed{\mathcal{D}}_1^\star \left(\rho, \sigma\right) +  E_4\left(\beta\right) \\
 E_4\left(\underline{\beta}\right) &=  + 2 \Omega \underline{\beta}  + 2 \widehat{\underline{H}} \cdot \beta - Y \cdot \underline{\alpha} - 3 \left(\underline{Z} \rho - {}^\star \underline{Z} \sigma\right) - J_{4A3} \nonumber
\end{align}
\begin{align} \label{Bianchi6}
\rho_4 &= \slashed{D}_4 \rho + \frac{3}{2} tr H \rho =\slashed{div} \beta + E_4\left(\rho\right)  \\
E_4\left(\rho\right) &= - \frac{1}{2} \widehat{\underline{H}} \cdot \alpha + V \cdot \beta + 2 \left(\underline{Z} \cdot \beta - Y \cdot \underline{\beta}\right) - \frac{1}{2}J_{434} \nonumber
\end{align}
\begin{align} \label{Bianchi7}
\sigma_4 &= \slashed{D}_4 \sigma +\frac{3}{2} tr  {H} \sigma =  -\slashed{curl} \beta + E_4\left(\sigma\right) \\
E_4\left(\sigma\right)&= + \frac{1}{2} \widehat{\underline{H}} \cdot {}^\star \alpha - V \cdot {}^\star \beta - 2 \left(\underline{Z} \cdot {}^\star \beta - Y \cdot {}^\star\underline{\beta}\right) - \frac{1}{2}J^\star_{434}\nonumber
\end{align}
\begin{align} \label{Bianchi8}
\rho_3 &= \slashed{D}_3 \rho+\frac{3}{2} tr \underline{H} \rho =  -\slashed{div} \underline{\beta} + E_3\left(\rho\right)  \\
E_3\left(\rho\right) &= - \frac{1}{2}\widehat{H} \cdot \underline{\alpha} + V \cdot \underline{\beta} + 2 \left(\underline{Y} \cdot \beta - Z \cdot \underline{\beta}\right) + \frac{1}{2}J_{334} \nonumber
\end{align}
\begin{align} \label{Bianchi9}
\sigma_3 &= \slashed{D}_3 \sigma + \frac{3}{2}tr \underline{H} \sigma =  - \slashed{curl} \underline{\beta} + E_3\left(\sigma\right) \\
E_3\left(\sigma\right)  &= - \frac{1}{2}\widehat{H} \cdot {}^\star \underline{\alpha} + V \cdot {}^\star \underline{\beta} + 2 \left(\underline{Y} \cdot {}^\star \beta + Z \cdot {}^\star  \underline{\beta}\right) + \frac{1}{2} J^\star_{334}
\nonumber 
\end{align}
Note that at the lowest order we are going to renormalize $\tilde{\rho}=r^3 \rho + 2M$ and $\tilde{\sigma} = r^3 \sigma$ or alternatively $\hat{\rho}=\frac{\tilde{\rho}}{r^3}$, for which we expect decay. Hence we also collect the equations
\begin{align} \label{renormrho}
\hat{\rho}_3 = \left[\rho + \frac{2M}{r^3}\right]_3 = -\slashed{div} \underline{\beta} + E_3 \left(\rho\right) - \frac{6M}{r^3}\left(\frac{\slashed{D}_3 r}{r} + \frac{1}{2} tr \underline{H}\right) = - \slashed{div} \underline{\beta} + \hat{E}_3 \left(\rho\right) \, ,
\nonumber \\ 
\hat{\rho}_4 = \left[\rho + \frac{2M}{r^3}\right]_4 = \slashed{div} \beta + E_4 \left(\rho\right) - \frac{6M}{r^3}\left(\frac{\slashed{D}_4 r}{r} + \frac{1}{2} tr H\right) =  \slashed{div} \beta + \hat{E}_4 \left(\rho\right) \, .
\end{align}

\subsection{Commuting the null-Bianchi equations}
Besides the important commutation formula of section \ref{commute} we are also going to commute the Bianchi equations for the null-curvature components directly. 

From the first order Bianchi equations of the previous section we derive the higher order Bianchi equations inductively using the following Lemma, which we quote from \cite{ChristKlei} (Lemma 7.3.3). 
\begin{lemma} \label{commutelemma}
Let $U_{A_1... A_k}$ be an $S^2$-tangent $k$-covariant tensor on $\left(\mathcal{R},g\right)$. Then
\begin{align}
\slashed{D}_4 \slashed{\nabla}_B U_{A_1... A_k} - \slashed{\nabla}_B  \slashed{D}_4 U_{A_1... A_k} + H_{BC} \slashed{\nabla}_C U_{A_1... A_k} = F_{4 B A_1 ... A_k} \, ,
\end{align}
\begin{align}
\slashed{D}_3 \slashed{\nabla}_B U_{A_1... A_k} - \slashed{\nabla}_B  \slashed{D}_3 U_{A_1... A_k} + \underline{H}_{BC} \slashed{\nabla}_C U_{A_1... A_k} = F_{3 B A_1... A_k} \, ,
\end{align}
\begin{align}
\slashed{D}_3 \slashed{D}_4 U_{A_1... A_k} - \slashed{D}_4 \slashed{D}_3 U_{A_1... A_k} = F_{3 4 A_1 ... A_k} \, ,
\end{align}
where
\begin{align}
F_{3 B A_1 ... A_k} = \underline{Y}_B\slashed{D}_4 U_{A_1... A_k} + \left(Z_B -V_B\right) \slashed{D}_3 U_{A_1... A_k}  + \sum_{i=1}^k \Big(H_{A_iB} \underline{Y}_C \nonumber \\- H_{BC} \underline{Y}_{A_i} + \underline{H}_{A_iB}{Z}_C -  \underline{H}_{BC}Z_{A_i} - \in_{A_i C} {}^\star \underline{\beta}\left(R\right)_B\Big) U_{A_1 ... A_{i-1} C A_{i+1} A_k} \, , \nonumber
\end{align}
\begin{align}
F_{4 B A_1 ... A_k} ={Y}_B\slashed{D}_3 U_{A_1... A_k} + \left(\underline{Z}_B +V_B\right) \slashed{D}_4 U_{A_1... A_k}  + \sum_{i=1}^k \Big(\underline{H}_{A_iB} Y_C \nonumber \\- \underline{H}_{BC} Y_{A_i} + H_{A_iB}\underline{Z}_C - H_{BC} \underline{Z}_{A_i} + \in_{A_i C} {}^\star \beta\left(R\right)_B\Big) U_{A_1 ... A_{i-1} C A_{i+1} A_k} \, , \nonumber 
\end{align}
\begin{align}
F_{3 4 A_1 ... A_k} = -2\Omega \slashed{D}_3 U_{A_1... A_k} + 2 \underline{\Omega} \slashed{D}_4 U_{A_1... A_k} + \left(Z_B - \underline{Z}_B\right) \slashed{\nabla}_B U_{A_1... A_k} + \nonumber \\ 2 \sum_{i=1}^k \Big(Y_{A_i} \underline{Y}_C - Y_C \underline{Y}_{A_i} + \underline{Z}_{A_i} Z_C - Z_{A_i} \underline{Z}_C + \epsilon_{A_iC} \sigma \left(R\right)\Big) U_{A_1 ... A_{i-1} C A_{i+1} A_k} \, . \nonumber
\end{align}
\end{lemma}
\begin{remark} There is a typo in \cite{ChristKlei} regarding the sign of the first term of the expression for $F_{34A_1 ... A_k}$.
\end{remark}

We recall the definition (\ref{subnot}) and declare that for $\Omega_i$ a basis of angular momentum vectorfields in Schwarzschild, the signature of $\slashed{\mathcal{L}}_{\Omega_i}^m u_{n_p}$ is defined to be identical to the signature of $u_{n_p}$. 
\glossary{
name={$E^{4n_l}_{3}\left(\alpha\right)$},
description={inhomogeneities in the (commuted) Bianchi equations}
}
\begin{proposition}
Let $n_l$ be an $l$-tuple of $3$'s and $4$'s, $l \geq 0$. The Bianchi equations imply the following equations for higher derivatives:
For $\boxed{\alpha}$
\begin{align} \label{a3nk}
\alpha_{3n_l} & = -2 \slashed{\mathcal{D}}^\star_2 \beta_{n_l} + E_3^{n_l} \left(\alpha\right)
\end{align}
with the error defined recursively as
\begin{align}
E_3^{n_p3} \left(\alpha\right) &= \slashed{D}_3 E^{n_p}_3 \left(\alpha\right) + \vartheta^-\left(\alpha_{3n_p} \right) tr  \underline{H} E^{n_l}_3 \left(\alpha\right) \nonumber \\  &- \vartheta^-\left(\beta_{n_l}\right) \left(\slashed{\nabla} tr \underline{H} \right) \widehat{\otimes} \beta_{n_p} - \widehat{\underline{H}} \slashed{div} \beta_{n_p}  + {}^\star \widehat{\underline{H}} \slashed{curl} \beta_{n_p} + 2\widehat{s} \left(F_3\left(\beta_{n_p}\right) \right) \nonumber \\
E_3^{n_p4} \left(\alpha\right) &= \slashed{D}_4 E^{n_p}_3 \left(\alpha\right) + \vartheta^+\left(\alpha_{3n_p}\right) tr H E^{n_p}_3 \left(\alpha\right) \nonumber \\  &- \vartheta^+\left(\beta_{n_p}\right) \left(\slashed{\nabla} tr H \right) \widehat{\otimes} \beta_{n_p} - \widehat{H} \slashed{div} \beta_{n_p}  + {}^\star \widehat{H} \slashed{curl} \beta_{n_p} + 2\widehat{s} \left(F_4\left(\beta_{n_p}\right) \right) \nonumber
\end{align}
and also
\begin{align}
\slashed{\mathcal{L}}^m_{\Omega_i} \alpha_{3n_l} & = -2 \slashed{\mathcal{D}}^\star_2 \left(\slashed{\mathcal{L}}^m_{\Omega_i} \beta_{n_l}\right) + E_3^{n_l\Omega_i^m} \left(\alpha\right)  \\ 
E_3^{n_l\Omega_i^m} \left(\alpha\right) &= \slashed{\mathcal{L}}_{\Omega_i} \left( E_3^{n_l\Omega_i^{m-1}} \left(\alpha\right)\right) + 2 \left[\slashed{\mathcal{D}}^\star_2, \slashed{\mathcal{L}}_{\Omega_i}\right]\slashed{\mathcal{L}}^{m-1}_{\Omega_i} \beta_{n_l} \nonumber
\end{align}
For $\boxed{\beta}$
\begin{align} 
\beta_{3n_l} &= \slashed{\nabla} \rho_{n_l} + {}^\star \slashed{\nabla} \sigma_{n_l} + E_3^{n_l} \left(\beta\right) \label{b3nk} \\ 
\beta_{4n_l} &= \slashed{div} \alpha_{n_l} + E_4^{n_l} \left(\beta\right) \label{b4nk}
\end{align}
\begin{align} 
 E_3^{n_p3} \left(\beta\right) &= \slashed{D}_3 E_3^{n_p} \left(\beta\right) +  \vartheta^-\left(\beta_{3n_p}\right) tr  \underline{H} E^{n_p}_3  \left(\beta\right) + F_3\left(\rho_{n_p}\right)+ {}^\star F_3\left(\sigma_{n_p}\right) \nonumber \\  &- \vartheta^-\left(\rho_{n_p}\right) \rho_{n_p} \slashed{\nabla} tr \underline{H}  - \vartheta^-\left(\rho_{n_p}\right) \sigma_{n_p} {}^\star \slashed{\nabla} tr \underline{H} - \underline{\widehat{H}}\cdot \slashed{\nabla} \rho_{n_p} - {}^\star \underline{\widehat{H}}\cdot \slashed{\nabla} \sigma_{n_p} \nonumber \\
E_3^{n_p4} \left(\beta\right) &= \slashed{D}_4 E_3^{n_p} \left(\beta\right) + \vartheta^+\left(\beta_{3n_p}\right) tr  H E^{n_p}_3  \left(\beta\right) + F_4\left(\rho_{n_p}\right)+ {}^\star F_4\left(\sigma_{n_p}\right) \nonumber \\  &- \vartheta^+\left(\rho_{n_p}\right) \rho_{n_p} \slashed{\nabla} tr H   - \vartheta^+\left(\sigma_{n_p}\right) \sigma_{n_p} {}^\star \slashed{\nabla} tr H - \widehat{H}\cdot \slashed{\nabla} \rho_{n_p} - {}^\star {\widehat{H}}\cdot \slashed{\nabla} \sigma_{n_p} \nonumber
\\ 
E_4^{n_p3} \left(\beta\right) &=  \slashed{D}_3 E_4^{n_p} \left(\beta\right) + \vartheta^-\left(\beta_{4n_p}\right) tr  H E^{n_p}_4  \left(\beta\right)  \nonumber \\
&-\vartheta^-\left(\alpha_{n_p}\right) \slashed{\nabla} tr \underline{H} \cdot \alpha_{n_p} - \widehat{\underline{H}} \cdot \slashed{\nabla} \alpha_{n_p} + tr F_3 \left(\alpha_{n_p}\right) \nonumber
 \\
E_4^{n_p4} \left(\beta\right) &= \slashed{D}_4 E_4^{n_p} \left(\beta\right) + \vartheta^+\left(\beta_{4n_p}\right) tr  H E^{n_p}_4  \left(\beta\right) \nonumber \\
&-\vartheta^+\left(\alpha_{n_p}\right) \slashed{\nabla} tr H \cdot \alpha_{n_p} - \widehat{H} \cdot \slashed{\nabla} \alpha_{n_p} + tr F_4 \left(\alpha_{n_p}\right)
\end{align}
and in addition
\begin{align}
\slashed{\mathcal{L}}^m_{\Omega_i} \beta_{3n_l} &= \slashed{\nabla} \slashed{\mathcal{L}}^m_{\Omega_i}\rho_{n_l} + {}^\star \slashed{\nabla} \slashed{\mathcal{L}}^m_{\Omega_i} \sigma_{n_l} + E_3^{n_l\Omega_i^m} \left(\beta\right)  \nonumber \\ 
\slashed{\mathcal{L}}^m_{\Omega_i} \beta_{4n_l} &= \slashed{div} \slashed{\mathcal{L}}^m_{\Omega_i} \alpha_{n_l} + E_4^{n_l\Omega_i^m} \left(\beta\right)
\end{align}
\begin{align}
E_3^{n_l\Omega_i^m} \left(\beta\right) &= \slashed{\mathcal{L}}_{\Omega_i} \left( E_3^{n_l\Omega_i^{m-1}} \left(\beta\right)\right) - \left[\slashed{\mathcal{D}}^\star_1, \slashed{\mathcal{L}}_{\Omega_i}\right]\slashed{\mathcal{L}}^{m-1}_{\Omega_i} \left(\rho,\sigma\right)_{n_l} \nonumber \\
E_4^{n_l\Omega_i^m} \left(\beta\right) &= \slashed{\mathcal{L}}_{\Omega_i} \left( E_4^{n_l\Omega_i^{m-1}} \left(\beta\right)\right) - \left[\slashed{div}, \slashed{\mathcal{L}}_{\Omega_i}\right]\slashed{\mathcal{L}}^{m-1}_{\Omega_i} \alpha_{n_l} 
\end{align}
For $\boxed{\rho}$
\begin{align}
\rho_{3n_l} &= -\slashed{div} \underline{\beta}_{n_l} + E_3^{n_l} \left(\rho\right) \label{r3nk}  \\
\rho_{4n_l} &= \slashed{div} \beta_{n_l} + E_4^{n_l} \left(\rho\right) \label{r4nk}
\end{align}
\begin{align} 
 E_3^{n_p3} \left(\rho\right) &= \slashed{D}_3 E_3^{n_p} \left(\rho\right) +  \vartheta^-\left(\rho_{3n_p}\right) tr  \underline{H} E^{n_p}_3  \left(\rho\right) - tr F_3\left(\underline{\beta}_{n_p}\right) \nonumber \\  &+\vartheta^- \left(\underline{\beta}_{n_p}\right) \underline{\beta}_{n_p} \cdot \slashed{\nabla} tr \underline{H} + \underline{\widehat{H}} \cdot \slashed{\nabla} \underline{\beta}_{n_p} \nonumber \\
  E_3^{n_p4} \left(\rho\right) &=  \slashed{D}_4 E_3^{n_p} \left(\rho\right) +  \vartheta^+\left(\rho_{3n_p}\right) tr  \underline{H} E^{n_p}_3  \left(\rho\right) - tr F_4\left(\underline{\beta}_{n_p}\right) \nonumber \\  &+ \vartheta^+\left(\underline{\beta}_{n_p}\right) \underline{\beta}_{n_p} \cdot \slashed{\nabla} tr H + \widehat{H} \cdot \slashed{\nabla} \underline{\beta}_{n_p}
\nonumber \\
 E_4^{n_p3} \left(\rho\right) &= \slashed{D}_3 E_4^{n_p} \left(\rho\right) +  \vartheta^-\left(\rho_{4n_p}\right) tr  \underline{H} E^{n_p}_4  \left(\rho\right) + tr F_3\left({\beta}_{n_p}\right) \nonumber \\  &-\vartheta^-\left(\beta_{n_p}\right) {\beta}_{n_p} \cdot \slashed{\nabla} tr \underline{H} - \underline{\widehat{H}} \cdot \slashed{\nabla} {\beta}_{n_p}
 \nonumber \\
 E_4^{n_p4} \left(\rho\right) &= \slashed{D}_4 E_4^{n_4} \left(\rho\right) + \vartheta^+\left(\rho_{4n_p}\right) tr  H E^{n_p}_4  \left(\rho\right) - tr F_4\left({\beta}_{n_p}\right) \nonumber \\  &-\vartheta^+\left(\beta_{n_p}\right) {\beta}_{n_p} \cdot \slashed{\nabla} tr H - \widehat{H} \cdot \slashed{\nabla} {\beta}_{n_p}
\end{align}
and in addition
\begin{align}
\slashed{\mathcal{L}}^m_{\Omega_i} \rho_{3n_l} &= -\slashed{div} \slashed{\mathcal{L}}^m_{\Omega_i}\underline{\beta}_{n_l} + E_3^{n_l\Omega_i^m} \left(\rho\right)  \nonumber \\ 
\slashed{\mathcal{L}}^m_{\Omega_i} \rho_{4n_l} &= \slashed{div} \slashed{\mathcal{L}}^m_{\Omega_i} \beta_{n_l} + E_4^{n_l\Omega_i^m} \left(\rho\right)
\end{align}
\begin{align}
E_3^{n_l\Omega_i^m} \left(\rho\right) &= \slashed{\mathcal{L}}_{\Omega_i} \left( E_3^{n_l\Omega_i^{m-1}} \left(\rho\right)\right) - \left[-\slashed{div}, \slashed{\mathcal{L}}_{\Omega_i}\right]\slashed{\mathcal{L}}^{m-1}_{\Omega_i} \underline{\beta}_{n_l} \nonumber \\
E_4^{n_l\Omega_i^m} \left(\rho\right) &= \slashed{\mathcal{L}}_{\Omega_i} \left( E_4^{n_l\Omega_i^{m-1}} \left(\rho\right)\right) - \left[\slashed{div}, \slashed{\mathcal{L}}_{\Omega_i}\right]\slashed{\mathcal{L}}^{m-1}_{\Omega_i} \beta_{n_l} 
\end{align}
For $\boxed{\sigma}$
\begin{align}
\sigma_{3n_l} &= -\slashed{curl} \underline{\beta}_{n_l} + E_3^{n_l} \left(\sigma\right) \label{s3nk} \\ \sigma_{4n_l} &= -\slashed{curl} \beta_{n_l} + E_4^{n_l} \left(\sigma\right) \label{s4nk}
\end{align}
\begin{align}
 E_3^{n_p3} \left(\sigma\right) &= \slashed{D}_3 E_3^{n_p} \left(\sigma\right) + \vartheta^-\left(\sigma_{3n_p}\right) tr  \underline{H} E^{n_p}_3  \left(\sigma\right) - tr F_3\left({}^\star\underline{\beta}_{n_p}\right) \nonumber \\  &-\vartheta^-\left(\underline{\beta}_{n_p}\right) {}^\star \underline{\beta}_{n_p} \cdot \slashed{\nabla} tr \underline{H} - \underline{\widehat{H}} \cdot \slashed{\nabla} {}^\star \underline{\beta}_{n_p}
  \nonumber \\
 E_3^{n_p4} \left(\sigma\right) &= \slashed{D}_4 E_3^{n_p} \left(\sigma\right) + \vartheta^+\left(\sigma_{3n_p}\right) tr  {H} E^{n_p}_3  \left(\sigma\right) - tr F_4\left({}^\star\underline{\beta}_{n_p}\right) \nonumber \\  &-\vartheta^+\left(\underline{\beta}_{n_p}\right) {}^\star \underline{\beta}_{n_p} \cdot \slashed{\nabla} tr H - {\widehat{H}} \cdot \slashed{\nabla} {}^\star \underline{\beta}_{n_p}
 \nonumber \\
 E_4^{n_p3} \left(\sigma\right) &= \slashed{D}_3 E_4^{n_p} \left(\sigma\right) + \vartheta^-\left(\sigma_{4n_p}\right) tr  \underline{H} E^{n_p}_4  \left(\sigma\right) - tr F_3\left({}^\star{\beta}_{n_p}\right) \nonumber \\  &+\vartheta^+\left(\beta_{n_p}\right) {}^\star{\beta}_{n_p} \cdot \slashed{\nabla} tr \underline{H} +\underline{\widehat{H}} \cdot \slashed{\nabla} {}^\star {\beta}_{n_p}
 \nonumber \\
 E_4^{n_p4} \left(\sigma\right) &= \slashed{D}_4 E_4^{n_4} \left(\sigma\right) +  \vartheta^+\left(\sigma_{4n_p}\right) tr  H E^{n_p}_4  \left(\sigma\right) - tr F_4\left({}^\star {\beta}_{n_p}\right) \nonumber \\  &+\vartheta^+\left(\beta_{n_p}\right) {}^\star {\beta}_{n_p} \cdot \slashed{\nabla} tr H + \widehat{H} \cdot \slashed{\nabla}{}^\star {\beta}_{n_p}
\end{align}
\begin{align}
\slashed{\mathcal{L}}^m_{\Omega_i} \sigma_{3n_l} &= -\slashed{curl} \slashed{\mathcal{L}}^m_{\Omega_i}\underline{\beta}_{n_l} + E_3^{n_l\Omega_i^m} \left(\sigma\right)  \nonumber \\ 
\slashed{\mathcal{L}}^m_{\Omega_i} \sigma_{4n_l} &= -\slashed{curl} \slashed{\mathcal{L}}^m_{\Omega_i} \beta_{n_l} + E_4^{n_l\Omega_i^m} \left(\sigma\right)
\end{align}
\begin{align}
E_3^{n_l\Omega_i^m} \left(\sigma\right) &= \slashed{\mathcal{L}}_{\Omega_i} \left( E_3^{n_l\Omega_i^{m-1}} \left(\sigma\right)\right) - \left[-\slashed{curl}, \slashed{\mathcal{L}}_{\Omega_i}\right]\slashed{\mathcal{L}}^{m-1}_{\Omega_i} \underline{\beta}_{n_l} \nonumber \\
E_4^{n_l\Omega_i^m} \left(\sigma\right) &= \slashed{\mathcal{L}}_{\Omega_i} \left( E_4^{n_l\Omega_i^{m-1}} \left(\sigma\right)\right) - \left[-\slashed{curl}, \slashed{\mathcal{L}}_{\Omega_i}\right]\slashed{\mathcal{L}}^{m-1}_{\Omega_i} \beta_{n_l} 
\end{align}
For $\boxed{\underline{\beta}}$
\begin{align}
\underline{\beta}_{3n_l} &= -\slashed{div} \underline{\alpha}_{n_l} + E_3^{n_l} \left(\underline{\beta}\right) \label{bb3nk}  \\
\underline{\beta}_{4n_l} &= -\slashed{\nabla} \rho_{n_l} - {}^\star \slashed{\nabla} \sigma_{n_l} + E_4^{n_l} \left(\underline{\beta}\right)  \label{bb4nk}
\end{align}
\begin{align}
 E_3^{n_p3} \left(\underline{\beta}\right) &= \slashed{D}_3 E_3^{n_p} \left(\underline{\beta}\right) +  \vartheta^-\left(\underline{\beta}_{3n_p}\right) tr  \underline{H} E^{n_p}_3  \left(\underline{\beta}\right) - tr F_3\left(\underline{\alpha}_{n_p}\right) \nonumber \\  &+\vartheta^-\left(\underline{\alpha}_{n_p}\right) \slashed{\nabla} tr \underline{H} \cdot \underline{\alpha}_{n_p} + \widehat{\underline{H}} \cdot \slashed{\nabla} \underline{\alpha}_{n_p} \nonumber \\
E_3^{n_p4} \left(\underline{\beta}\right) &= \slashed{D}_4 E_3^{n_p} \left(\underline{\beta}\right) +  \vartheta^+\left(\underline{\beta}_{3n_p}\right) tr {H} E^{n_p}_3  \left(\underline{\beta}\right) - tr F_4\left(\underline{\alpha}_{n_p}\right) \nonumber \\  &+\vartheta^+\left(\underline{\alpha}_{n_p}\right) \slashed{\nabla} tr {H} \cdot \underline{\alpha}_{n_p} + \widehat{{H}} \cdot \slashed{\nabla} \underline{\alpha}_{n_p} \nonumber \\
E_4^{n_p3} \left(\underline{\beta}\right) &= \slashed{D}_3 E_4^{n_p} \left(\underline{\beta}\right) +  \vartheta^-\left(\underline{\beta}_{4n_p}\right)tr  H E^{n_p}_4  \left(\underline{\beta}\right) -F_3\left(\rho_{n_p}\right)- {}^\star F_3\left(\sigma_{n_p}\right) \nonumber \\  &+ \vartheta^-\left(\rho_{n_p}\right) \rho_{n_p} \slashed{\nabla} tr \underline{H}   -  \vartheta^-\left(\sigma_{n_p}\right)\sigma_{n_p} {}^\star \slashed{\nabla} tr \underline{H} - \widehat{\underline{H}}\cdot \slashed{\nabla} \rho_{n_p} - {}^\star {\widehat{\underline{H}}}\cdot \slashed{\nabla} \sigma_{n_p} \nonumber
\\
E_4^{n_p4} \left(\underline{\beta}\right) &= \slashed{D}_4 E_4^{n_p} \left(\underline{\beta}\right) +   \vartheta^+\left(\underline{\beta}_{4n_p}\right)  tr  H E^{n_p}_4  \left(\underline{\beta}\right) + \vartheta^+\left(\rho_{n_p}\right) \rho_{n_p} \slashed{\nabla} tr H  \nonumber \\
& - \vartheta^+\left(\sigma_{n_p}\right) \sigma_{n_p} {}^\star \slashed{\nabla} tr H - \widehat{H}\cdot \slashed{\nabla} \rho_{n_p} - {}^\star {\widehat{H}}\cdot \slashed{\nabla} \sigma_{n_p} \nonumber 
\end{align}
\begin{align}
\slashed{\mathcal{L}}^m_{\Omega_i} \underline{\beta}_{3n_l} &= -\slashed{div} \slashed{\mathcal{L}}^m_{\Omega_i} \underline{\alpha}_{n_l} + E_3^{n_l\Omega_i^m} \left(\underline{\beta}\right)   \\
\slashed{\mathcal{L}}^m_{\Omega_i} \underline{\beta}_{4n_l} &= -\slashed{\nabla} \slashed{\mathcal{L}}^m_{\Omega_i} \rho_{n_l} - {}^\star \slashed{\nabla} \slashed{\mathcal{L}}^m_{\Omega_i} \sigma_{n_l} + E_4^{n_l\Omega_i^m} \left(\underline{\beta}\right)  
\end{align}
\begin{align}
E_3^{n_l\Omega_i^m} \left(\underline{\beta}\right)  &= \slashed{\mathcal{L}}^m_{\Omega_i}\left( E_3^{n_l\Omega_i^{m-1}} \left(\underline{\beta}\right)\right)  + \left[\slashed{div}, \slashed{\mathcal{L}}_{\Omega_i}\right] \slashed{\mathcal{L}}^{m-1}_{\Omega_i} \underline{\alpha}_{n_l} \\
E_4^{n_l\Omega_i^m} \left(\underline{\beta}\right)  &=\slashed{\mathcal{L}}^m_{\Omega_i} \left(E_4^{n_l\Omega_i^{m-1}} \left(\underline{\beta}\right)  \right) + \left[\slashed{\mathcal{D}}_1^\star, \slashed{\mathcal{L}}_{\Omega_i}\right]\slashed{\mathcal{L}}^{m-1}_{\Omega_i} \left(\rho,\sigma\right)_{n_l}
\end{align}
For $\boxed{\underline{\alpha}}$
\begin{align} \label{ab4nk}
\underline{\alpha}_{4n_l} &= 2 \slashed{\mathcal{D}}^\star_2 \underline{\beta}_{n_l} + E_4^{n_l} \left(\underline{\alpha}\right) 
\end{align}
\begin{align}
E_4^{3n_l} \left(\underline{\alpha}\right) &= \slashed{D}_3 E^{n_l}_4 \left(\underline{\alpha}\right) +  \vartheta^-\left(\underline{\alpha}_{4n_l}\right) tr H E^{n_l}_4 \left(\underline{\alpha}\right) \nonumber \\  &+ \vartheta^-\left(\underline{\beta}_{n_l}\right) \left(\slashed{\nabla} tr \underline{H} \right) \widehat{\otimes} \underline{\beta}_{n_l} + \widehat{\underline{H}} \slashed{div} \underline{\beta}_{n_l}  - {}^\star \widehat{\underline{H}} \slashed{curl} \underline{\beta}_{n_l} - 2\widehat{s} \left(F_3\left(\underline{\beta}_{n_l}\right) \right) \nonumber
\\
E_4^{4n_l} \left(\underline{\alpha}\right) &= \slashed{D}_4 E^{n_4}_4 \left(\underline{\alpha}\right) +  \vartheta^+\left(\underline{\alpha}_{3n_l}\right) tr  \underline{H} E^{n_l}_4 \left(\underline{\alpha}\right) \nonumber \\  &+\vartheta^+\left(\underline{\beta}_{n_l}\right) \left(\slashed{\nabla} tr {H} \right) \widehat{\otimes} \underline{\beta}_{n_l} + \widehat{{H}} \slashed{div} \underline{\beta}_{n_l}  - {}^\star \widehat{{H}} \slashed{curl} \underline{\beta}_{n_l} + 2\widehat{s} \left(F_4\left(\underline{\beta}_{n_l}\right) \right) \nonumber
\end{align}
\begin{align}
\slashed{\mathcal{L}}^m_{\Omega_i} \underline{\alpha}_{4n_l} & = +2 \slashed{\mathcal{D}}^\star_2 \left(\slashed{\mathcal{L}}^m_{\Omega_i} \underline{\beta}_{n_l}\right) + E_4^{n_l\Omega_i^m} \left(\underline{\alpha}\right)  \\ 
E_4^{n_l\Omega_i^m} \left(\underline{\alpha}\right) &= \slashed{\mathcal{L}}_{\Omega_i} \left( E_4^{n_l\Omega_i^{m-1}} \left(\underline{\alpha}\right)\right) - 2 \left[\slashed{\mathcal{D}}^\star_2, \slashed{\mathcal{L}}_{\Omega_i}\right]\slashed{\mathcal{L}}^{m-1}_{\Omega_i} \underline{\beta}_{n_l} \nonumber
\end{align}
\end{proposition}
\subsection{More commutation formulae}
The commuted Bianchi equations in the last section are not yet of the form most useful for the multiplier estimates. The point is that we would like to write the left hand side of (\ref{a3nk}) as $\left(\alpha_{n_l}\right)_3$, so that the equation becomes an honest inhomogeneous Bianchi equation for the components $\alpha_{n_l}$ and $\beta_{n_l}$. Of course, this simply means to pushing the $3$ trough the $n_l$.

Let $u$ be any curvature component, $n_l$ a tuple of $3$'s and $4$'s and $a,b \in \{3,4\}$. We define the commutator
\begin{equation}
\mathcal{C}_{ab}\left[u_{n_l}\right] = u_{abn_l} - u_{an_lb} \, .
\end{equation}
Note that we can write both
\begin{align} \nonumber
u_{3n_l4} = u_{34n_l} - \mathcal{C}_{34}\left[u_{n_l}\right] \textrm {\ \ and \ \ }
u_{3n_l4} = u_{43n_l} - \mathcal{C}_{34}\left[u_{n_l}\right] + \left[u_{34}-u_{43}\right]_{n_l} \, .
\end{align}
Let us agree on the following convention: If the Bianchi equation for $u_{3n_l4}$ is derived using $u_{34n_l}$ of the previous section, the resulting error-term will be denoted with a tilde. On the other hand, if the equation for $u_{3n_l4}$ is derived using the Bianchi equation for $u_{43n_l}$, the error-term will acquire a hat. For the Bianchi equations for $u_{3n_l3}$ and $u_{4n_l4}$ there is no ambiguity, as they are derived using $u_{33n_l}$ and $u_{44n_l}$ respectively. In this case, the error-term is also denoted with a tilde. For instance, for $\underline{\beta}$:
\begin{align} \label{bb33}
\underline{\beta}_{3n_l3} =  -\slashed{div} \underline{\alpha}_{3n_l} + \tilde{E}_3^{3n_l} \left(\underline{\beta}\right) \textrm{ \ \ \ , \ \ \ \ } 
 \tilde{E}_3^{3n_l} \left(\underline{\beta}\right) = E_3^{3n_l}\left(\underline{\beta}\right) - \mathcal{C}_{33} \left[\underline{\beta}_{n_l}\right] \, ,
\end{align}
\begin{align} \label{bb34}
\underline{\beta}_{3n_l4} =  -\slashed{div} \underline{\alpha}_{4n_l} + \tilde{E}_3^{4n_l} \left(\underline{\beta}\right)  \textrm{ \ \ , \ \ \ } 
 \tilde{E}_3^{4n_l} \left(\underline{\beta}\right) = E_3^{4n_l}\left(\underline{\beta}\right) - \mathcal{C}_{34} \left[\underline{\beta}_{n_l}\right] \, , 
\end{align}
\begin{align} \label{bb43}
\underline{\beta}_{4n_l3} =  -\slashed{\mathcal{D}}_1^\star \left(\rho, \sigma\right)_{3n_l} + \tilde{E}_4^{3n_l} \left(\underline{\beta}\right)  \textrm{ \  , \ } 
 \tilde{E}_4^{3n_l} \left(\underline{\beta}\right) = E_4^{3n_l}\left(\underline{\beta}\right) - \mathcal{C}_{43} \left[\underline{\beta}_{n_l}\right] \, ,
\end{align}
\begin{align} \label{bb44}
\underline{\beta}_{4n_l4} = \slashed{\mathcal{D}}_1^\star\left( \rho, \sigma\right)_{4n_l} + \tilde{E}_4^{4n_l} \left(\underline{\beta}\right) \textrm{ \ \ \ , \ \ \ } 
 \tilde{E}_4^{4n_l} \left(\underline{\beta}\right) = E_4^{4n_l}\left(\underline{\beta}\right) - \mathcal{C}_{44} \left[\underline{\beta}_{n_l}\right] \, .
\end{align}
On the other hand, we could write
\begin{align} \label{bb342}
\underline{\beta}_{3n_l4} &=  \slashed{\mathcal{D}}_1^\star \left(\rho,\sigma\right)_{3n_l} + \hat{E}_4^{3n_l} \left(\underline{\beta}\right) \nonumber \\
 \hat{E}_4^{3n_l} \left(\underline{\beta}\right) &= E_4^{3n_l}\left(\underline{\beta}\right) + \left[\underline{\beta}_{34}-\underline{\beta}_{43}\right]_{n_l} - \mathcal{C}_{34} \left[\underline{\beta}_{n_l}\right]
\end{align}
\begin{align} \label{bb432}
\underline{\beta}_{4n_l3} &=  -\slashed{div} \underline{\alpha}_{4n_l} + \hat{E}_3^{4n_l} \left(\underline{\beta}\right) \nonumber \\
 \hat{E}_3^{4n_l} \left(\underline{\beta}\right) &= E_3^{4n_l}\left(\underline{\beta}\right) + \left[\underline{\beta}_{43}-\underline{\beta}_{34}\right]_{n_l} - \mathcal{C}_{43} \left[\underline{\beta}_{n_l}\right]
\end{align}
The formulae for the other components are easily derived. We also collect the following formulae
\begin{align} \label{a43}
\alpha_{43n_l}  &= \alpha_{34n_l} + \Big[\alpha \left[\vartheta^-\left(\alpha\right) \slashed{D}_4 tr \underline{H} -\vartheta^+\left(\alpha\right) \slashed{D}_3 tr H\right]  + F_{34} \left(\alpha\right)\Big]_{n_l} 
\end{align}
\begin{align} \label{b4334}
\beta_{43n_l} &=  \beta_{34n_l} +\Big[ \beta \left[\vartheta^+\left(\beta\right) \slashed{D}_3 tr H -\vartheta^-\left(\beta\right) \slashed{D}_4 tr \underline{H} \right]  + F_{34} \left(\beta\right)\Big]_{n_l}
\end{align}
\begin{align} \label{r3443}
\rho_{43n_l} = \rho_{34n_l} + \Big[\rho  \left[\vartheta^+\left(\rho\right) \slashed{D}_3 tr H - \vartheta^-\left(\rho\right) \slashed{D}_4 tr \underline{H} \right]  + F_{34} \left(\rho\right)\Big]_{n_l}
\end{align}
\begin{align} \label{s3443}
\sigma_{43n_l} &=  \sigma_{34n_l} + \Big[\sigma  \left[\vartheta^+\left(\sigma\right) \slashed{D}_3 tr H - \vartheta^-\left(\sigma\right) \slashed{D}_4 tr \underline{H} \right]  + F_{34} \left(\sigma\right)\Big]_{n_l}
\end{align}
\begin{align} \label{bb3443}
\underline{\beta}_{34n_l} &= \underline{\beta}_{43n_l} + \Big[\underline{\beta}\left[-\vartheta^+\left(\underline{\beta}\right)\slashed{D}_3 tr H + \vartheta^-\left(\underline{\beta}\right) \slashed{D}_4 tr \underline{H} \right]  + F_{43} \left(\underline{\beta}\right)\Big]_{n_l} \, .
\end{align}

\subsection{Commutation with angular momentum}
We can commute the Bianchi equations for the quantities $u_{3n_l4}$ etc. derived in the previous section with the $\Omega_i$. In this way we obtain, for $m\geq 0$ (using the convention that $\hat{E}_{\cdot}^{\cdot n_l\Omega_i^0}\left(u\right)=\hat{E}_{\cdot}^{\cdot n_l}\left(u\right)$ and  $\tilde{E}_{\cdot}^{\cdot n_l\Omega_i^0}\left(u\right)=\tilde{E}_{\cdot}^{\cdot n_l}\left(u\right)$) the formulae
\begin{align}
\left(\slashed{\mathcal{L}}_{\Omega_i}^m \alpha_{4n_l}\right)_3 = -2\slashed{\mathcal{D}}_2^\star \left(\slashed{\mathcal{L}}_{\Omega_i}^m \beta_{4n_l}\right) + \hat{E}_3^{4n_l\Omega_i^m}\left(\alpha\right)
\end{align}
\begin{align}
\hat{E}_3^{4n_l\Omega_i^m}\left(\alpha\right) = \slashed{\mathcal{L}}_{\Omega_i} \left(\hat{E}_3^{4n_l\Omega_i^{m-1}}\left(\alpha\right)\right) + 2 \left[\slashed{\mathcal{D}}_2^\star, \slashed{\mathcal{L}}_{\Omega_i}\right] \left(\slashed{\mathcal{L}}_{\Omega_i}^{m-1} \beta_{4n_l}\right) \nonumber \\ 
-  \left[\slashed{D}_3, \slashed{\mathcal{L}}_{\Omega_i}\right] \left(\slashed{\mathcal{L}}_{\Omega_i}^{m-1} \alpha_{4n_l}\right) 
-\vartheta^-\left(\alpha_{4n_l}\right) \left(\Omega_i tr \underline{H}\right) \slashed{\mathcal{L}}_{\Omega_i}^{m-1} \alpha_{4n_l}
\end{align}
For $\beta$:
\begin{align}
\left(\slashed{\mathcal{L}}_{\Omega_i}^m \beta_{4n_l}\right)_4 = \slashed{div} \left(\slashed{\mathcal{L}}_{\Omega_i}^m \alpha_{4n_l}\right) + \tilde{E}_4^{4n_l\Omega_i^m} \left(\beta\right) \nonumber \\
\tilde{E}_4^{4n_l\Omega_i^m} \left(\beta\right) = \slashed{\mathcal{L}}_{\Omega_i} \left(\tilde{E}_4^{4n_l\Omega_i^{m-1}} \left(\beta\right)\right) + \left[\slashed{div}, \slashed{\mathcal{L}}_{\Omega_i}\right] \left(\slashed{\mathcal{L}}_{\Omega_i}^{m-1} \alpha_{4n_l}\right) \nonumber \\  - \left[\slashed{D}_4, \slashed{\mathcal{L}}_{\Omega_i}\right] \left(\slashed{\mathcal{L}}_{\Omega_i}^{m-1} \beta_{4n_l}\right) 
-\vartheta^+\left(\beta_{4n_l}\right) \left(\Omega_i tr H\right) \slashed{\mathcal{L}}_{\Omega_i}^{m-1} \beta_{4n_l}
\end{align}
\begin{align}
\left(\slashed{\mathcal{L}}_{\Omega_i}^m \beta_{4n_l} \right)_3 &=  \slashed{\mathcal{D}}_1^\star \left(-\slashed{\mathcal{L}}_{\Omega_i}^m\rho_{4n_l}, \slashed{\mathcal{L}}_{\Omega_i}^m\sigma_{4n_l}\right) + \tilde{E}_3^{4n_l\Omega_i^m} \left(\beta\right) 
\end{align}
\begin{align}
\tilde{E}_3^{4n_l\Omega_i^m}\left(\beta\right) = \slashed{\mathcal{L}}_{\Omega_i} \left(\tilde{E}_3^{4n_l \Omega_i^{m-1}}\left(\beta\right)\right) + 2 \left[\slashed{\mathcal{D}}_1^\star, \slashed{\mathcal{L}}_{\Omega_i}\right]  \left(-\slashed{\mathcal{L}}_{\Omega_i}^{m-1}\rho_{4n_l}, \slashed{\mathcal{L}}_{\Omega_i}^{m-1}\sigma_{4n_l}\right) \nonumber \\ 
-  \left[\slashed{D}_3, \slashed{\mathcal{L}}_{\Omega_i}\right] \left(\slashed{\mathcal{L}}_{\Omega_i}^{m-1} \beta_{4n_l}\right) 
-\vartheta^-\left(\beta_{4n_l}\right) \left(\Omega_i tr \underline{H}\right) \slashed{\mathcal{L}}_{\Omega_i}^{m-1} \beta_{4n_l} \nonumber
\end{align}
For $\rho$ and $\sigma$
\begin{align}
\left(\slashed{\mathcal{L}}_{\Omega_i}^m \rho_{4n_l}\right)_4 = \slashed{div} \left(\slashed{\mathcal{L}}_{\Omega_i}^m \beta_{4n_l}\right) + \tilde{E}_4^{4n_l\Omega_i^m} \left(\rho\right) \nonumber \\
\tilde{E}_4^{4n_l\Omega_i^m} \left(\rho\right) =\slashed{\mathcal{L}}_{\Omega_i} \left(\tilde{E}_4^{4n_l \Omega_i^{m-1}} \left(\rho\right)\right) + \left[\slashed{div}, \slashed{\mathcal{L}}_{\Omega_i}\right] \left(\slashed{\mathcal{L}}_{\Omega_i}^{m-1} \beta_{4n_l}\right) \nonumber \\  - \left[\slashed{D}_4, \slashed{\mathcal{L}}_{\Omega_i}\right] \left(\slashed{\mathcal{L}}_{\Omega_i}^{m-1} \rho_{n_l}\right) 
-\vartheta^+\left(\rho_{4n_l}\right) \left(\Omega_i tr H\right) \slashed{\mathcal{L}}_{\Omega_i}^{m-1} \rho_{4n_l}
\end{align}
\begin{align}
\left(\slashed{\mathcal{L}}_{\Omega_i}^m \sigma_{4n_l}\right)_4 = -\slashed{curl} \left(\slashed{\mathcal{L}}_{\Omega_i}^m \beta_{4n_l}\right) + \tilde{E}_4^{4n_l\Omega_i^m} \left(\sigma\right) \nonumber \\
\tilde{E}_4^{4n_l\Omega_i^m} \left(\sigma\right) = \slashed{\mathcal{L}}_{\Omega_i} \left(\tilde{E}_4^{4n_l \Omega_i^{m-1}} \left(\sigma\right)\right) + \left[-\slashed{curl}, \slashed{\mathcal{L}}_{\Omega_i}\right] \left(\slashed{\mathcal{L}}_{\Omega_i}^{m-1} \beta_{4n_l}\right) \nonumber \\  -  \left[\slashed{D}_4, \slashed{\mathcal{L}}_{\Omega_i}\right] \left(\slashed{\mathcal{L}}_{\Omega_i}^{m-1} \sigma_{n_l}\right) 
-\vartheta^+\left(\sigma_{4n_l}\right) \left(\Omega_i tr H\right) \slashed{\mathcal{L}}_{\Omega_i}^{m-1} \sigma_{4n_l}
\end{align}
\begin{align}
\left(\slashed{\mathcal{L}}_{\Omega_i}^m \rho_{4n_l} \right)_3 &=  -\slashed{div} \left(\slashed{\mathcal{L}}_{\Omega_i}^m \underline{\beta}_{4n_l}\right) + \hat{E}_3^{4n_l\Omega_i^m} \left(\rho\right) 
\end{align}
\begin{align}
\hat{E}_3^{4n_l\Omega_i^m}\left(\rho\right) = \slashed{\mathcal{L}}_{\Omega_i} \left(\hat{E}_3^{4n_l \Omega_i^{m-1}}\left(\rho\right)\right) + \left[-\slashed{div}, \slashed{\mathcal{L}}_{\Omega_i}\right] \left( \slashed{\mathcal{L}}_{\Omega_i}^{m-1}\underline{\beta}_{4n_l}\right) \nonumber \\ 
- \left[\slashed{D}_3, \slashed{\mathcal{L}}_{\Omega_i}\right] \left(\slashed{\mathcal{L}}_{\Omega_i}^{m-1} \rho_{4n_l}\right) 
-\vartheta^-\left(\rho_{4n_l}\right) \left(\Omega_i tr \underline{H}\right) \slashed{\mathcal{L}}_{\Omega_i}^{m-1} \rho_{4n_l}
\end{align}
\begin{align}
\left(\slashed{\mathcal{L}}_{\Omega_i}^m \sigma_{4n_l} \right)_3 &=  -\slashed{curl} \left(\slashed{\mathcal{L}}_{\Omega_i}^m\underline{\beta}_{4n_l} \right) + \hat{E}_3^{4n_l\Omega_i^m} \left(\sigma\right) 
\end{align}
\begin{align}
\hat{E}_3^{4n_l\Omega_i^m}\left(\sigma\right) = \slashed{\mathcal{L}}_{\Omega_i} \left(\hat{E}_3^{4n_l \Omega_i^{m-1}}\left(\sigma\right)\right) + \left[-\slashed{curl}, \slashed{\mathcal{L}}_{\Omega_i}\right] \left( \slashed{\mathcal{L}}_{\Omega_i}^{m-1}\underline{\beta}_{4n_l}\right) \nonumber \\ 
- \left[\slashed{D}_3, \Omega_i\right] \left(\slashed{\mathcal{L}}_{\Omega_i}^{m-1} \sigma_{4n_l}\right) 
-\vartheta^-\left(\sigma_{4n_l}\right) \left(\Omega_i tr \underline{H}\right) \slashed{\mathcal{L}}_{\Omega_i}^{m-1} \sigma_{4n_l}
\end{align}
For $\underline{\beta}$
\begin{align}
\left(\slashed{\mathcal{L}}_{\Omega_i}^m \underline{\beta}_{4n_l}\right)_4 = \slashed{\mathcal{D}}_1^\star\left( \slashed{\mathcal{L}}_{\Omega_i}^m \rho_{4n_l}, \slashed{\mathcal{L}}_{\Omega_i}^m \sigma_{4n_l}\right)  + \tilde{E}_4^{4n_l\Omega_i^m} \left(\underline{\beta}\right) \nonumber \\
\tilde{E}_4^{4n_l\Omega_i^m} \left(\beta\right) = \slashed{\mathcal{L}}_{\Omega_i} \left(\tilde{E}_4^{4n_l\slashed{\mathcal{L}}_{\Omega_i}^{m-1}} \left(\underline{\beta}\right)\right) + \left[\slashed{div}, \slashed{\mathcal{L}}_{\Omega_i}\right] \left(\slashed{\mathcal{L}}_{\Omega_i}^{m-1} \left(\rho_{4n_l}, \sigma_{4n_l}\right) \right) \nonumber \\  - \left[\slashed{D}_4, \slashed{\mathcal{L}}_{\Omega_i}\right] \left(\slashed{\mathcal{L}}_{\Omega_i}^{m-1} \underline{\beta}_{4n_l}\right) 
-\vartheta^+\left(\underline{\beta}_{4n_l}\right) \left(\Omega_i tr H\right) \slashed{\mathcal{L}}_{\Omega_i}^{m-1} \underline{\beta}_{4n_l}
\end{align}
\begin{align}
\left(\slashed{\mathcal{L}}_{\Omega_i}^m \underline{\beta}_{4n_l} \right)_3 &=  -\slashed{div} \left(\slashed{\mathcal{L}}_{\Omega_i}^m \underline{\alpha}_{4n_l}\right) + \hat{E}_3^{4n_l\Omega_i^m} \left(\underline{\beta}\right) 
\end{align}
\begin{align}
\hat{E}_3^{4n_l\Omega_i^m}\left(\underline{\beta}\right) = \slashed{\mathcal{L}}_{\Omega_i} \left(\hat{E}_3^{4n_l \Omega_i^{m-1}}\left(\underline{\beta}\right)\right) + \left[-\slashed{div}, \slashed{\mathcal{L}}_{\Omega_i}\right] \left( \slashed{\mathcal{L}}_{\Omega_i}^{m-1}\underline{\alpha}_{4n_l}\right) \nonumber \\ 
- \left[\slashed{D}_3, \slashed{\mathcal{L}}_{\Omega_i}\right] \left(\slashed{\mathcal{L}}_{\Omega_i}^{m-1} \underline{\beta}_{4n_l}\right) 
-\vartheta^-\left(\underline{\beta}_{4n_l}\right) \left(\Omega_i tr \underline{H}\right) \slashed{\mathcal{L}}_{\Omega_i}^{m-1} \underline{\beta}_{4n_l}
\end{align}
\begin{align}
\left(\slashed{\mathcal{L}}_{\Omega_i}^m \underline{\beta}_{3n_l} \right)_4 &= - \slashed{div} \left( \slashed{\mathcal{L}}_{\Omega_i}^m \underline{\alpha}_{4n_l}\right) + \tilde{E}_3^{4n_l\Omega_i^m} \left(\underline{\beta}\right) 
\end{align}
\begin{align}
\tilde{E}_3^{4n_l\Omega_i^m}\left(\underline{\beta}\right) = \slashed{\mathcal{L}}_{\Omega_i} \left(\tilde{E}_3^{4n_l{\Omega_i}^{m-1}}\left(\underline{\beta}\right)\right) + \left[-\slashed{div}, \slashed{\mathcal{L}}_{\Omega_i}\right] \left( \slashed{\mathcal{L}}_{\Omega_i}^{m-1}\underline{\alpha}_{4n_l}\right) \nonumber \\ 
- \left[\slashed{D}_4, \slashed{\mathcal{L}}_{\Omega_i}\right] \left(\slashed{\mathcal{L}}_{\Omega_i}^{m-1} \underline{\beta}_{3n_l}\right) 
-\vartheta^+\left(\underline{\beta}_{3n_l}\right) \left(\slashed{\mathcal{L}}_{\Omega_i} tr {H}\right) \slashed{\mathcal{L}}_{\Omega_i}^{m-1} \underline{\beta}_{4n_l}
\end{align}
\begin{align}
\left(\slashed{\mathcal{L}}_{\Omega_i}^m \underline{\beta}_{3n_l} \right)_3 &= - \slashed{div} \left( \slashed{\mathcal{L}}_{\Omega_i}^m \underline{\alpha}_{3n_l}\right) + \tilde{E}_3^{3n_l\Omega_i^m} \left(\underline{\beta}\right) 
\end{align}
\begin{align}
\tilde{E}_3^{3n_l\Omega_i^m}\left(\underline{\beta}\right) = \slashed{\mathcal{L}}_{\Omega_i} \left(\tilde{E}_3^{3n_l{\Omega_i}^{m-1}}\left(\underline{\beta}\right)\right) + \left[-\slashed{div}, \slashed{\mathcal{L}}_{\Omega_i}\right] \left( \slashed{\mathcal{L}}_{\Omega_i}^{m-1}\underline{\alpha}_{3n_l}\right) \nonumber \\ 
- \left[\slashed{D}_3, \slashed{\mathcal{L}}_{\Omega_i}\right] \left(\slashed{\mathcal{L}}_{\Omega_i}^{m-1} \underline{\beta}_{3n_l}\right) 
-\vartheta^-\left(\underline{\beta}_{3n_l}\right) \left(\slashed{\mathcal{L}}_{\Omega_i} tr \underline{H}\right) \slashed{\mathcal{L}}_{\Omega_i}^{m-1} \underline{\beta}_{3n_l}
\end{align}
For $\underline{\alpha}$
\begin{align}
\left(\slashed{\mathcal{L}}_{\Omega_i}^m \underline{\alpha}_{4n_l}\right)_4 = 2\slashed{\mathcal{D}}_2^\star \left(\slashed{\mathcal{L}}_{\Omega_i}^m \underline{\beta}_{4n_l} \right) + \tilde{E}_4^{4n_l\Omega_i^m} \left(\underline{\alpha}\right) \nonumber \\
\tilde{E}_4^{4n_l\Omega_i^m} \left(\alpha\right) = \slashed{\mathcal{L}}_{\Omega_i} \left(\tilde{E}_4^{4n_l \Omega_i^{m-1}} \left(\underline{\alpha}\right)\right) + 2 \left[\slashed{\mathcal{D}}_2^\star, \slashed{\mathcal{L}}_{\Omega_i}\right] \left(\slashed{\mathcal{L}}_{\Omega_i}^{m-1} \underline{\beta}_{4n_l} \right) \nonumber \\  - \left[\slashed{D}_4, \slashed{\mathcal{L}}_{\Omega_i}\right] \left(\slashed{\mathcal{L}}_{\Omega_i}^{m-1} \underline{\alpha}_{4n_l}\right) 
-\vartheta^+\left(\underline{\alpha}_{4n_l}\right) \left(\Omega_i tr H\right) \slashed{\mathcal{L}}_{\Omega_i}^{m-1} \underline{\alpha}_{4n_l}
\end{align}
\begin{align}
\left(\slashed{\mathcal{L}}_{\Omega_i}^m \underline{\alpha}_{4n_l} \right)_3 &= 2\slashed{\mathcal{D}}_2^\star \left(\slashed{\mathcal{L}}_{\Omega_i}^m \underline{\beta}_{3n_l} \right) + \tilde{E}_4^{3n_l\Omega_i^m} \left(\underline{\alpha}\right) 
\end{align}
\begin{align}
 \tilde{E}_4^{3n_l\Omega_i^m} \left(\underline{\alpha}\right)  =\slashed{\mathcal{L}}_{\Omega_i} \left( \tilde{E}_4^{3n_l{\Omega_i}^{m-1}} \left(\underline{\alpha}\right) \right) + 2\left[\slashed{\mathcal{D}}_2^\star, \Omega_i\right] \left( \slashed{\mathcal{L}}_{\Omega_i}^{m-1}\underline{\beta}_{3n_l}\right) \nonumber \\ 
- \left[\slashed{D}_3,\slashed{\mathcal{L}}_{\Omega_i}\right] \left(\slashed{\mathcal{L}}_{\Omega_i}^{m-1} \underline{\alpha}_{4n_l}\right) 
-\vartheta^-\left(\underline{\alpha}_{4n_l}\right) \left(\Omega_i tr \underline{H}\right) \slashed{\mathcal{L}}_{\Omega_i}^{m-1} \underline{\alpha}_{4n_l}
\end{align}
\begin{align}
\left(\slashed{\mathcal{L}}_{\Omega_i}^m \underline{\alpha}_{3n_l} \right)_4 &= 2\slashed{\mathcal{D}}_2^\star \left(\slashed{\mathcal{L}}_{\Omega_i}^m \underline{\beta}_{3n_l} \right) + \hat{E}_4^{3n_l\Omega_i^m} \left(\underline{\alpha}\right) 
\end{align}
\begin{align}
 \hat{E}_4^{3n_l\Omega_i^m} \left(\underline{\alpha}\right)  = \slashed{\mathcal{L}}_{\Omega_i} \left( \hat{E}_4^{3n_l\slashed{\mathcal{L}}_{\Omega_i}^{m-1}} \left(\underline{\alpha}\right) \right) + 2\left[\slashed{\mathcal{D}}_2^\star, \slashed{\mathcal{L}}_{\Omega_i}\right] \left( \slashed{\mathcal{L}}_{\Omega_i}^{m-1}\underline{\beta}_{3n_l}\right) \nonumber \\ 
- \left[\slashed{D}_4, \slashed{\mathcal{L}}_{\Omega_i}\right] \left(\left(\slashed{\mathcal{L}}_{\Omega_i}\right)^{m-1} \underline{\alpha}_{3n_l}\right) 
-\vartheta^+\left(\underline{\alpha}_{3n_l}\right) \left(\slashed{\mathcal{L}}_{\Omega_i} tr {H}\right) \slashed{\mathcal{L}}_{\Omega_i}^{m-1} \underline{\alpha}_{3n_l}
\end{align}
\subsection{Commutation properties for null decomposition and Lie-derivatives}
We also need to relate estimates for the null components of the Weyl-field
$\widehat{\slashed{\mathcal{L}}}_{T} W$ to ``slashed" derivatives of the null-curvature components. This relation is easily inferred from Proposition 7.3.1 of \cite{ChristKlei}:
\begin{lemma} \label{LTTL}
Let $u$ be any null-component of a Weyl-tensor $W$. Then
\begin{align}
u \left( \widehat{\mathcal{L}}_T W\right) \equiv \widehat{\slashed{\mathcal{L}}}_T u  \nonumber
\end{align}
where $\equiv$ denotes equality up to lower order terms of the form
\begin{equation}
\left(\alpha, \underline{\alpha}, \beta, \underline{\beta}, \rho, \sigma \right) \cdot \left(\textrm{decaying RRC, ${}^{(T)}\pi$} \right) \, . \nonumber
\end{equation}
Moreover, if $tr {}^{(T)}\pi = 0$ (as is assumed in the context of the ultimately Schwarzschildean assumption), the $\rho$-term is absent and the lower order terms are quadratically decaying.
\end{lemma}
\subsection{Null components of the Bel-Robinson tensor} \label{ncbelrob}
We collect the components of the Bel-Robinson tensor in the null-frame:
\begin{equation}
 Q_{3333} = 2|\underline{\alpha}|^2 \textrm{ \ \ \ \ \  , \ \ \ \ \ }  Q_{4444} = 2|\alpha|^2 \, , \nonumber 
\end{equation}
\begin{equation}
 Q_{3334} = 4|\underline{\beta}|^2 \textrm{ \ \ \ \ \ \ , \ \ \ \ } Q_{3444} = 4|{\beta}|^2 \, , \nonumber
\end{equation}
\begin{equation}
 Q_{3344} = 4\left(\rho^2 + \sigma^2\right) \, , \nonumber
\end{equation}
\begin{equation}
 Q_{A444} = 4\alpha_{AB} \beta_B \textrm{ \ \ \ \ \ , \  \ \ \ \ } Q_{A333} = -4\underline{\alpha}_{AB} \underline{\beta}_B \, , \nonumber
\end{equation}
\begin{equation}
Q_{A344} = 4\rho \beta_A - 4 \sigma^\star \beta_A  \textrm{ \ \ \ \ \ , \  \ \ \ \ } Q_{A433} = -4\rho \underline{\beta}_A - 4\sigma^\star \underline{\beta}_A \, , \nonumber
\end{equation}
\begin{equation}
Q_{AB44} = 2 \delta_{AB} |\beta|^2 + 2\rho \alpha_{AB} - 2 \sigma^\star \alpha_{AB} \, , \nonumber
\end{equation}
\begin{equation}
Q_{AB33} = 2 \delta_{AB} |\underline{\beta}|^2 + 2\rho \underline{\alpha}_{AB} + 2 \sigma^\star \underline{\alpha}_{AB} \, , \nonumber
\end{equation}
\begin{equation}
Q_{AB34} = 2 \left(\delta_{AB} \beta \cdot \underline{\beta} - \beta_A \underline{\beta}_B - \beta_B \underline{\beta}_A \right) + 2 \delta_{AB} \left(\rho^2 + \sigma^2\right)  \, . \nonumber
\end{equation}

\subsection{The null-structure equations} \label{nseq}
We collect the null structure equations (Proposition 7.4.1 of \cite{ChristKlei}).
\begin{align} \label{Hb3}
\slashed{D}_3 \widehat{\underline{H}}= + tr \underline{H} \widehat{\underline{H}} = -2 \slashed{\mathcal{D}}_2^\star \underline{Y} - 2\underline{\Omega} \widehat{\underline{H}}+ \left(\left(Z+\underline{Z} - 2V\right) \widehat{\otimes} \underline{Y}\right) - \underline{\alpha}
\end{align}
\begin{align} \label{trHb3}
\slashed{D}_3 \left(tr \underline{H}\right) + \frac{1}{2}\left(tr \underline{H}\right)^2 = 2 \slashed{div} \underline{Y} - 2 \underline{\Omega}tr \underline{H} + 2 \underline{Y} \cdot \left(Z+\underline{Z} - 2V \right) - \widehat{\underline{H}} \cdot  \widehat{\underline{H}}
\end{align}
\begin{align} \label{curlYb}
\slashed{curl} \underline{Y} = \underline{Y} \wedge  \left(Z+\underline{Z} - 2V\right)
\end{align}
\begin{align} \label{Hb4}
\slashed{D}_4 \widehat{\underline{H}} + tr {H} \widehat{\underline{H}} = -2 \slashed{\mathcal{D}}_2^\star \underline{Z} + 2 \Omega \widehat{\underline{H}} - \frac{1}{2} tr \underline{H} \widehat{H} + \left(Y \widehat{\otimes} \underline{Y}\right) + \left(\underline{Z} \widehat{\otimes} \underline{Z}\right)
\end{align}
\begin{align} \label{trHb4}
\slashed{D}_4 \left(tr \underline{H}\right) + \frac{1}{2} tr H \left(tr \underline{H}\right) \nonumber \\ = 2 \slashed{div} \underline{Z} + 2 \Omega tr \underline{H} - \widehat{H} \cdot  \widehat{\underline{H}} + 2\left(Y \cdot \underline{Y} + \underline{Z} \cdot \underline{Z}\right) + 2\rho
\end{align}
\begin{align} \label{curlZb}
\slashed{curl} \underline{Z} = \frac{1}{2} \widehat{H} \wedge  \widehat{\underline{H}} - Y \wedge\underline{Y} - \sigma
\end{align}
\begin{align} \label{divHb}
\left(\slashed{div} \underline{H}\right)_A - V_B \underline{H}_{AB} = \slashed{\nabla}_A tr \underline{H} - V_A tr \underline{H} + \underline{\beta}_A
\end{align}
\begin{align} \label{divH}
\left(\slashed{div} H\right)_A + V_B {H}_{AB} = \slashed{\nabla}_A tr {H} + V_A tr {H} + {\beta}_A
\end{align}
\begin{align} \label{Gauss}
K = -\frac{1}{4} tr H tr \underline{H} + \frac{1}{2} \widehat{H} \cdot  \widehat{\underline{H}} - \rho
\end{align}
\begin{align} \label{V3}
\slashed{D}_3 V_A = -2 \slashed{\nabla}_A \underline{\Omega} - \underline{H}_{AB} \left(V_B + Z_B\right) + 2\underline{\Omega} \left(V_A - Z_A\right) + H_{AB} \underline{Y}_B + 2\Omega \underline{Y}_A - \underline{\beta}_A
\end{align}
\begin{align} \label{V4}
\slashed{D}_4 V_A = 2 \slashed{\nabla}_A {\Omega} - {H}_{AB} \left(-V_B + \underline{Z}_B\right) + 2{\Omega} \left(V_A + \underline{Z}_A\right) - \underline{H}_{AB} {Y}_B - 2\underline{\Omega} Y_A - \beta_A
\end{align}
\begin{align} \label{H3}
\slashed{D}_3 \widehat{{H}} + \frac{1}{2} tr \underline{H} \widehat{{H}} = -2 \slashed{\mathcal{D}}_2^\star Z + 2\underline{\Omega} \widehat{{H}}  - \frac{1}{2} tr {H} \underline{\widehat{H}}+ \left( \underline{Y} \widehat{\otimes} Y\right) + \left({Z} \widehat{\otimes} {Z}\right)
\end{align}
\begin{align} \label{trH3}
\slashed{D}_3 \left(tr {H}\right) + \frac{1}{2}\left(tr \underline{H}\right) tr H = \nonumber \\ 2 \slashed{div} \underline{Z} + 2 \underline{\Omega} \left(tr {H}\right) - \widehat{H} \cdot  \widehat{\underline{H}} + 2\left(Y \cdot \underline{Y} + Z \cdot {Z}\right) + 2\rho 
\end{align}
\begin{align} \label{H4}
\slashed{D}_4 \widehat{{H}} + tr {H} \widehat{{H}} = -2 \slashed{\mathcal{D}}_2^\star Y - 2 \Omega \widehat{{H}} + \left(\left(Z+\underline{Z} + 2V\right) \widehat{\otimes} {Y}\right) - {\alpha}
\end{align}
\begin{align} \label{trH4}
\slashed{D}_4 \left(tr {H}\right) + \frac{1}{2} \left(tr {H}\right)^2 = 2 \slashed{div} Y - 2 \Omega tr {H} + 2 {Y} \cdot \left(Z+\underline{Z} + 2V \right) - \widehat{{H}} \cdot  \widehat{H}
\end{align}
\begin{align} \label{Yb4}
\slashed{D}_4 \underline{Y}_A - \slashed{D}_3 \underline{Z}_A = 4 \Omega \underline{Y}_A + \underline{H}_{AB} \left(\underline{Z}_B - Z_B\right) - \underline{\beta}_A
\end{align}
\begin{align} \label{Y3}
\slashed{D}_3 {Y}_A - \slashed{D}_4 {Z}_A = 4 \underline{\Omega} {Y}_A + {H}_{AB} \left(Z_B -\underline{Z}_B\right) + {\beta}_A
\end{align}
\begin{align} \label{O43}
\slashed{D}_4 \underline{\Omega} + \slashed{D}_3 \Omega = Y \cdot \underline{Y} + V \cdot \left(Z - \underline{Z}\right) - Z \cdot \underline{Z} +4 \Omega \underline{\Omega} + \rho
\end{align}

\section{The energies} \label{norms}
We define the following energies for $W$:
\begin{align}
\|W\|^2 := \|\alpha\|^2 + \|\underline{\alpha}\|^2 + \|\beta\|^2 + \|\underline{\beta}\|^2 + |\rho|^2 + |\sigma |^2 \, .
\end{align}
\begin{align}
\|\tilde{W}\|^2 := \|\alpha\|^2 + \|\underline{\alpha}\|^2 + \|\beta\|^2 + \|\underline{\beta}\|^2 + |\rho -\overline{\rho}|^2 + |\sigma |^2 \, .
\end{align}
Recall $\mathcal{D} = \{\slashed{D}_3, \slashed{D}_4, \slashed{\nabla}\}$.
We define the higher order norms
\begin{align} \label{highernorm}
\|\mathcal{D}^k\tilde{W}\|^2 :=  \|\mathcal{D}^k\alpha\|^2 + \|\mathcal{D}^k\underline{\alpha}\|^2 + \|\mathcal{D}^k\beta\|^2 + \|\mathcal{D}^n\underline{\beta}\|^2 \nonumber \\ + |\mathcal{D}^k \left(\rho -\overline{\rho}\right)|^2 + |\mathcal{D}^k\left(\sigma \right)|^2 \, ,
\end{align}
where $\mathcal{D}^k$ denotes any $k$-tupel of operators from $\mathcal{D}$. We avoid ambiguities regarding the ordering of derivatives by understanding the definition as including all permutations of the $k$ derivatives applied.

\subsection{The master energies} \label{master}
We define the zeroth order energies
\begin{align}
\mathbb{E}^{0}_{(0,p_2,p_3,p_4,p_5,0,0)} \left[W\right] \left( \tilde{\Sigma}_{\tau}\right) = \int_{\tilde{\Sigma}_{\tau} \cap \{ r \leq R\}} \| \tilde{W} \|^2 \nonumber \\ + \int_{N_{out}\left(S^2_{\tau,R}\right)} dv d\omega \ \Big[r^{p_2} \| \alpha\|^2 + r^{p_3}\|  \beta\|^2  + r^{p_4}  \| \widehat{\rho}, \sigma\|^2  + r^{p_5} \| \underline{\beta}\|^2 \Big] \, ,
\end{align}
\begin{align}
\mathbb{E}^{0}_{(0,0, p_3,p_4,p_5,p_6)} \left[W\right] \left(\mathcal{I}_{\tau_1}^{\tau_2}\right) =\nonumber \\ = \int_{\mathcal{I}_{\tau_1}^{\tau_2}} du d\omega \  \Big[ r^{p_3}\|  \beta\|^2 + r^{p_4}  \| \widehat{\rho}, \sigma\|^2 + r^{p_5} \| \underline{\beta}\|^2 + r^{p_6} \| \underline{\alpha}\|^2 \Big]
\end{align}
and weighted spacetime energy
\begin{align}
\mathbb{I}^{0}_{(0,p_2,p_3,p_4,p_5,p_6,0)}  \left[W\right] \left(\tilde{\mathcal{M}} \left(\tau_1,\tau_2\right)\right) =  \int_{\tilde{\mathcal{M}} \left(\tau_1,\tau_2\right)} dt^\star dr d\omega \ \Bigg[  r^{p_2} \| \alpha\|^2 \nonumber \\+ r^{p_3}\|  \beta\|^2 + r^{p_4}  \| \widehat{\rho}, \sigma \|^2 + r^{p_5} \| \underline{\beta}\|^2  + r^{p_6} \| \underline{\alpha}\|^2 \Bigg] \, .
\end{align}
These energies should be understood as capturing the characteristic decay in $r$ of the curvature components familiar from the stability of Minkowski space. Which tuples of $p$ are admissible will be discussed in section \ref{decmatrix}.

Recalling the notation (\ref{subnot}) we define the following $n^{th}$-order weighted energies on the slices $\tilde{\Sigma}_{\tau}$ ($n\geq 1$):
\begin{align} \label{ennorm}
\mathbb{E}^{n}_{(p_1,p_2,p_3,p_4,p_5,p_6,0)} \left[W\right] \left( \tilde{\Sigma}_{\tau}\right) = \int_{\tilde{\Sigma}_{\tau} \cap \{ r \leq R\}} \|\mathcal{D}^n \tilde{W} \|^2 \nonumber \\ +
 \sum_{k_1+k_2+k_3=n-1}  \int_{N_{out}\left(S^2_{\tau,R}\right)} dv d\omega \nonumber \\  r^{2k_2+2k_3} \Bigg[ r^{p_1} \| \slashed{D}_3^{k_1} \slashed{D}_4^{k_2} \slashed{\nabla}^{k_3} \alpha_4\|^2 + r^{p_2} \| \slashed{D}_3^{k_1} \slashed{D}_4^{k_2} \slashed{\nabla}^{k_3} \left(\slashed{\nabla}\alpha, \beta_4\right)\|^2  \nonumber \\
+ r^{p_3} \|\left(\alpha_3, \slashed{\nabla} \beta, \rho_4,\sigma_4\right)\|^2  + r^{p_4} \| \slashed{D}_3^{k_1} \slashed{D}_4^{k_2} \slashed{\nabla}^{k_3} \left(\beta_3,\underline{\beta}_4, \slashed{\nabla}\rho, \slashed{\nabla}\sigma\right)\|^2  \nonumber \\ + r^{p_5} \|\slashed{D}_3^{k_1} \slashed{D}_4^{k_2} \slashed{\nabla}^{k_3} \left(\rho_3,\sigma_3,\slashed{\nabla}\underline{\beta}, \underline{\alpha}_4\right)\|^2 +  r^{p_6} \|\slashed{D}_3^{k_1} \slashed{D}_4^{k_2} \slashed{\nabla}^{k_3} \left(\underline{\beta}_3, \slashed{\nabla}\underline{\alpha}\right)\|^2\Bigg] \, ,
\end{align}
\begin{align}
\mathbb{E}^{n}_{(0, p_1,p_2,p_3,p_4,p_5,p_6)} \left[W\right] \left( \mathcal{I}_{\tau_1}^{\tau_2}\right) =  \sum_{k_1+k_2+k_3=n-1} \int_{\mathcal{I}_{\tau_1}^{\tau_2}} du d\omega \, r^{2k_2+2k_3} \Bigg[ \nonumber \\ r^{p_2} \| \slashed{D}_3^{k_1} \slashed{D}_4^{k_2} \slashed{\nabla}^{k_3}  \left(\slashed{\nabla}\alpha, \beta_4\right)\|^2  + r^{p_3} \| \slashed{D}_3^{k_1} \slashed{D}_4^{k_2} \slashed{\nabla}^{k_3}  \left(\alpha_3, \slashed{\nabla} \beta, \rho_4,\sigma_4\right)\|^2 \nonumber \\ + r^{p_4} \| \slashed{D}_3^{k_1} \slashed{D}_4^{k_2} \slashed{\nabla}^{k_3}  \left(\beta_3, \underline{\beta}_4, \slashed{\nabla}\rho, \slashed{\nabla}\sigma\right)\|^2 \nonumber \\ + r^{p_5} \|\slashed{D}_3^{k_1} \slashed{D}_4^{k_2} \slashed{\nabla}^{k_3} \left(\rho_3,\sigma_3,\slashed{\nabla}\underline{\beta}, \underline{\alpha}_4\right)\|^2  \nonumber \\ + r^{p_6} \|\slashed{D}_3^{k_1} \slashed{D}_4^{k_2} \slashed{\nabla}^{k_3}  \left(\underline{\beta}_3, \slashed{\nabla}\underline{\alpha}\right)\|^2 +  r^{p_7} \|\slashed{D}_3^k \slashed{D}_4^l \slashed{\nabla}^m \left(\underline{\alpha}_3\right)\|^2\Bigg]
\end{align}
and the $n^{th}$ order weighted spacetime energies
\begin{align}
\mathbb{I}^{n,(non)deg}_{(p_1,p_2,p_3,p_4,p_5,p_6, p_7)} \left[W\right] \left(\tilde{\mathcal{M}} \left(\tau_1,\tau_2\right)\right) =   \sum_{k_1+k_2+k_3=n-1} \int_{\tilde{\mathcal{M}} \left(\tau_1,\tau_2\right)} dt^\star dr d\omega  \nonumber \\
 w_{(non)deg}\left(r\right) \cdot r^{2k_2+2k_3} \Bigg[ r^{p_1} \| \slashed{D}_3^{k_1} \slashed{D}_4^{k_2} \slashed{\nabla}^{k_3}  \left(\alpha_4\right)\|^2 \nonumber \\ + r^{p_2} \| \slashed{D}_3^{k_1} \slashed{D}_4^{k_2} \slashed{\nabla}^{k_3}  \left(\slashed{\nabla}\alpha, \beta_4\right)\|^2  + r^{p_3} \| \slashed{D}_3^{k_1} \slashed{D}_4^{k_2} \slashed{\nabla}^{k_3}  \left(\alpha_3, \slashed{\nabla} \beta, \rho_4,\sigma_4\right)\|^2 \nonumber \\ + r^{p_4} \| \slashed{D}_3^{k_1} \slashed{D}_4^{k_2} \slashed{\nabla}^{k_3}  \left(\beta_3, \underline{\beta}_4, \slashed{\nabla}\rho, \slashed{\nabla}\sigma\right)\|^2 \nonumber \\ + r^{p_5} \| \slashed{D}_3^{k_1} \slashed{D}_4^{k_2} \slashed{\nabla}^{k_3}  \left( \rho_3,\sigma_3, \slashed{\nabla}\underline{\beta}, \underline{\alpha}_4\right)\|^2  \nonumber \\ + r^{p_6} \|\slashed{D}_3^{k_1} \slashed{D}_4^{k_2} \slashed{\nabla}^{k_3}  \left(\underline{\beta}_3,\slashed{\nabla}\underline{\alpha}\right)\|^2 +  r^{p_7} \|\slashed{D}_3^{k_1} \slashed{D}_4^{k_2} \slashed{\nabla}^{k_3}  \left(\underline{\alpha}_3\right)\|^2\Bigg] \, ,
\end{align}
with the weight $w_{nondeg}\left(r\right)=1$ in the non-degenerate case and $w_{deg} = \frac{\left(r-3M\right)^2}{r^2}$ in the degenerate case.
\begin{remark}
It is apparent from these energies that in the asymptotic region differentiating by $\slashed{D}_4$ or $\slashed{\nabla}$ gains two powers of $r$. To avoid ambiguities in the ordering of the derivatives, the definition is understood as including all possible permutations of the derivatives applied. Such permutations only differ by lower order terms, as we shall see in detail later.
\end{remark}
\begin{remark}
Note that because we are dealing with the characteristic energies, $\alpha_4$ is missing on $\mathcal{I}$, while $\underline{\alpha}_3$ is missing on $N_{out}\left(S^2_{\tau,R}\right)$.
\end{remark}
In addition, we define the summed energies
\glossary{
name={$\overline{\mathbb{E}}^{n}_P \left[W\right] \left( \tilde{\Sigma}_{\tau}\right)$}, 
description={weighted curvature energy, weight encoded in decay matrix $P$}
}
\glossary{
name={$\overline{\mathbb{I}}^{n}_{\tilde{P}} \left[W\right] \left(\mathcal{U} \right)$}, 
description={weighted spacetime curvature energy, weight encoded in decay matrix $\tilde{P}$}
}
\begin{align} \label{summednorms}
\overline{\mathbb{E}}^{n}_P \left[W\right] \left( \tilde{\Sigma}_{\tau}\right) &= \sum_{i=0}^n \mathbb{E}^{i}_{(p^i_1,p^i_2,p^i_3,p^i_4,p^i_5,p_6^i,0)} \left[W\right] \left( \tilde{\Sigma}_{\tau}\right)
 \\
\overline{\mathbb{I}}^{n, (non)deg}_{P} \left[W\right] \left(\tilde{\mathcal{M}} \left(\tau_1,\tau_2\right)\right) &= \sum_{i=0}^n \mathbb{I}^{i, (non)deg}_{(p^i_1,p^i_2,p^i_3,p^i_4,p^i_5,p^i_6,p^i_7)} \left[W\right] \left(\tilde{\mathcal{M}} \left(\tau_1,\tau_2\right)\right) \nonumber 
\end{align}
where $P=(p^i_1,p^i_2,p^i_3,p^i_4,p^i_5,p^i_6, p_7^i)$ denotes an $\left(n+1\right) \times 7$ matrix encoding the decay at each order. These matrices are defined in the following section. For convenience, we also define the unweighted energies
\begin{align} \label{unweightedEn}
\overline{\mathbb{E}}^{n}  \left[W\right] \left(\tilde{\Sigma}_{\tau} \right)= \sum_{i=0}^n \int_{\tilde{\Sigma}_{\tau}} \|\mathcal{D}^i \tilde{W} \|^2   \textrm{ \ \ \  EXCEPT \ \  $\|\slashed{D}_3^n \underline{\alpha}\|^2$ on $N_{out}$}
\end{align}
\begin{align} \label{unweightedIn}
\overline{\mathbb{I}}^{n,(non)deg} \left[W\right] \left(\tilde{\mathcal{M}} \left(\tau_1,\tau_2\right)\right)= \sum_{i=0}^n \int_{\tilde{\mathcal{M}} \left(\tau_1,\tau_2\right)}  \frac{1}{r^{1+\delta}} w_{(non)deg} \|\mathcal{D}^i \tilde{W} \|^2 
\end{align}
Finally, for the horizon we define
\begin{align}
{\mathbb{E}}^{n} \left[W\right] \left( \mathcal{H}\left(\tau_1,\tau_2\right)\right) &= \ \int_{\mathcal{H}\left(\tau_1,\tau_2\right)} dt^\star d\omega \| \mathcal{D}^n \tilde{W} \|^2 \textrm{ \ \ \  EXCEPT \ \  $\|\slashed{D}_3^n \underline{\alpha}\|^2$}
\end{align}
\begin{align}
\overline{\mathbb{E}}^{n} \left[W\right] \left( \mathcal{H}\left(\tau_1,\tau_2\right)\right) &= \sum_{i=0}^n \int_{\mathcal{H}\left(\tau_1,\tau_2\right)} dt^\star d\omega \| \mathcal{D}^i \tilde{W} \|^2 \textrm{ \ \ \  EXCEPT \ \  $\|\slashed{D}_3^n \underline{\alpha}\|^2$}
\end{align}

\subsection{The decay matrices} \label{decmatrix}
We need some definitions:
\begin{definition}
We call a tuple $\left(0,p_2,p_3,p_4,p_5,0,0\right)$ \underline{boundary $0$-admissible} if all $p_i \geq 2$ and in addition
\begin{align}
2<p_2<7-\delta \textrm{ \ \ \ , \ \ \ } p_3 \leq \min\left(6,p_2\right) \textrm{ \ \ \ , \ \ \ } p_4 \leq \min\left(4,p_3\right)\textrm{ \ \ \ , \ \ \ } p_5\leq 2 \, . \nonumber
\end{align}
To a boundary $0$-admissible tuple $p$ we associate the \underline{bulk $0$-admissible} tuple
\begin{align}
\tilde{p} = \left(0,p_2-1,p_2-1,p_3-1,p_4-1,p_5-1,0\right)  
\end{align}
in case that $p_3<6$, $p_4<4$ and $p_2<2$. In case that $p_3=6$ we replace $p_3-1=5$ by $5-\delta$. In case that $p_4=4$ we replace $p_4-1=3$ by $3-\delta$ and in case $p_5=2$ we replace $p_5-1=1$ by $1-\delta$. 
\end{definition}
\begin{definition}
We call a tuple $\left(p_1,p_2,p_3,p_4,p_5,p_6,0\right)$ \underline{boundary admissible} if 
\begin{align}
2&<p_1 \leq 10-\delta \textrm{ \  ,\ \ } p_2 \leq \min\left(9-\delta,p_1\right) \textrm{  \ , \ } p_3\leq \min\left(8-\delta,p_2\right) \nonumber \\ p_4 &\leq \min\left(6,p_3\right) \textrm{ \  , \ } p_5\leq \min\left(4,p_4\right) \textrm{ \  ,\ \ } p_6\leq2 \, , \nonumber
\end{align}
or if it is equal to the tuple $\left(10-\delta,9-\delta,8,6,4,2,0\right)$.
To a  boundary admissible tuple we associate the \underline{bulk-admissible} tuple 
\begin{align}
\tilde{p}=\left(p_1-1,p_1-1,p_2-1,p_3-1,p_4-1,p_5-1,p_6-1\right) \nonumber
\end{align}
in case that $p_3<8$, $p_4<6$, $p_5<4$ and $p_6<2$ hold. In case that $p_3=8$ we replace $p_3-1=7$ by $7-\delta$. In case that $p_4=6$ we replace $p_4-1=5$ by $5-\delta$. In case that $p_5=4$ we replace $p_5-1=3$ by $3-\delta$ and in case $p_6=2$ we replace $p_6-1=1$ by $1-\delta$.

Finally, an $\left(n+1 \right)\times 7$ matrix $P$ is called boundary admissible, if its first line is a boundary $0$-admissible tuple, and the remaining $n$ lines consist of identical boundary admissible tuples. A boundary admissible matrix $P$ has an associated bulk-admissible matrix $\tilde{P}$: The $n^{th}$ line of $\tilde{P}$ is defined to be the bulk admissible tuple associated with the $n^{th}$ line of $P$.
\end{definition}

We define the following $\left(n+1\right) \times 7$ boundary admissible decay-matrices:
\glossary{
name={$P_i$}, 
description={boundary-admissible decay matrix}
}
\glossary{
name={$\tilde{P}_i$}, 
description={bulk-associated decay matrix}
}
\begin{equation} \label{P0}
\mathbf{P}_0 = \left( \begin{array}{ccccccc}
0 & 7-\delta & 6 & 4 & 2 & 0 & 0\\ 
10-\delta & 9-\delta & 8 & 6 & 4 & 2 & 0 \\
\vdots & \vdots & \vdots & \vdots & \vdots & \vdots & \vdots
\end{array} \right)
 \end{equation}
 \begin{equation} \label{P1}
\mathbf{P}_1 = \left( \begin{array}{ccccccc}
0 & 6-\delta & 6-\delta & 4 & 2 & 0 & 0\\ 
9-\delta & 9-\delta & 8-\delta & 6 & 4 & 2 & 0 \\
\vdots & \vdots & \vdots & \vdots & \vdots & \vdots & \vdots
\end{array} \right)
 \end{equation}
\begin{equation} \label{P2}
\mathbf{P}_2 = \left( \begin{array}{ccccccc}
0 & 5-\delta & 5-\delta & 4 & 2 & 0 & 0\\ 
8-\delta & 8-\delta & 8-\delta & 6 & 4 & 2 & 0 \\
\vdots & \vdots & \vdots & \vdots & \vdots & \vdots & \vdots
\end{array} \right)
 \end{equation}
 \begin{equation} \label{P3}
\mathbf{P}_3 = \left( \begin{array}{ccccccc}
0 & 4-\delta & 4-\delta & 4-\delta & 2 & 0 & 0\\ 
7-\delta & 7-\delta & 7-\delta & 6 & 4 & 2 & 0 \\
\vdots & \vdots & \vdots & \vdots & \vdots & \vdots & \vdots
\end{array} \right) \, .
 \end{equation}
 We will also need the auxiliary decay matrix
 \begin{equation} \label{Prho}
 \mathbf{P}_{\rho} = \left( \begin{array}{ccccccc}
0 & 4 & 4 & 4 & 2 & 0 & 0\\ 
6 & 6 & 6 & 6 & 4 & 2 & 0 \\
\vdots & \vdots & \vdots & \vdots & \vdots & \vdots & \vdots
\end{array} \right) \, .
 \end{equation}
Note that the natural $L^2$-based energy for the curvature components would correspond to a matrix $P$ in (\ref{summednorms}) with all non-zero entries equal to $2$ and for which the weight $r^{2k_2+2k_3}$ is dropped in (\ref{ennorm}).

We also define an ordering relation on these matrices. We will say $P\leq Q$ if all entries of $P$ are smaller or equal than that of $Q$. 

\begin{remark}
Note that the energy induced by a boundary admissible decay matrix with second line $\left(p_1,...,p_6,0\right)$ controls in particular the lowest order energy induced by a boundary 0-admissible tuple with first line $(0,p_2-2, ..., p_6-2,0)$ in view of the Poincare inequality $\int_{S^2} \left(f - \overline{f}\right)^2 r^2 dA_{S^2} \leq c r^2 \int_{S^2} |\slashed{\nabla} f|^2  r^2 dA_{S^2}$. Cf.~Lemma \ref{pcc}.
\end{remark}

\begin{remark}
Applying Sobolev inequalities one can obtain the familiar characteristic pointwise decay in $r$ of the curvature components from the boundedness of the $\overline{\mathbb{E}}^2_{P_0}\left[W\right]$ and $\overline{\mathbb{I}}^2_{P_0}\left[W\right]$ energies.
\end{remark}

\subsection{Auxiliary norms}
Recall that we fixed $\delta=\frac{1}{100}$.
For $\mathcal{W}$ a Weyl tensor and $\mathcal{U}$ a spacetime region
\begin{align} \label{degXall}
I^{(non)deg} \left[\mathcal{W}\right]\left( \mathcal{U}\right) = \int_{ \mathcal{U}} dt^\star dr r^2 d\omega \,  w_{(non)deg} \left(r\right) \Bigg(\frac{1}{r^{1+\delta}}| \alpha|^2  \nonumber \\ + \frac{1}{r^{1+\delta}} |\underline{\alpha}|^2 + \frac{r^2}{r^{1+\delta}} |\beta|^2 + \frac{r^2}{r^{1+\delta}} |\underline{\beta}|^2  + \frac{r^2}{r}\left(|\rho|^2 +  | \sigma|^2 \right) \Bigg) 
\end{align}
\begin{align} 
E \left[\mathcal{W}\right] \left(\tilde{\Sigma}_{\tau} \right) = \int_{N_{out}\left(\tau,R\right)} r^2 dv  d\omega \,   \left(\| \alpha\|^2 + \|\beta\|^2 +  \|\underline{\beta}\|^2  + |\rho|^2 +  | \sigma|^2 \right) \nonumber \\
\int_{ \Sigma_{\tau} \cap \{r\leq R\}} r^2 dr  d\omega \,   \left(\| \alpha\|^2 + \|\underline{\alpha}\|^2 + \|\beta\|^2 +  \|\underline{\beta}\|^2  + |\rho|^2 +  | \sigma|^2 \right)
\end{align}
\begin{align} 
E \left[\mathcal{W}\right] \left( \mathcal{H}\left(\tau_1,\tau_2\right) \right) = \int_{ \mathcal{H}\left(\tau_1,\tau_2\right)} dt^\star r^2  d\omega \,   \left(\| \alpha\|^2 + \|\beta\|^2 +  \|\underline{\beta}\|^2  + |\rho|^2 +  | \sigma|^2 \right) \nonumber
\end{align}
Note that these energies are only expected to decay if $\mathcal{W}$ arose from commutation with an approximate Killing field.

In the context of the spin-reduction, it is convenient to introduce also certain weighted energies:
\begin{align} 
E_{spin1} \left[\mathcal{W}\right] \left( \tilde{\Sigma}_{\tau} \right) = \int_{N_{out}\left(\tau,R\right)} r^2 dv  d\omega \,   \left(\| \alpha\|^2 +  \|\underline{\beta}\|^2 + r^2 \left(\|\beta\|^2   + |\rho|^2 +  | \sigma|^2 \right)\right) \nonumber \\
\int_{ \Sigma_{\tau} \cap \{r\leq R\}} r^2 dr  d\omega \,   \left(\| \alpha\|^2 + \|\underline{\alpha}\|^2 +   \|\underline{\beta}\|^2  +  \|\beta\|^2 +  |\rho|^2 +  | \sigma|^2 \right) \nonumber
\end{align}
which gives the components $\beta$, $\rho$, $\sigma$ an additional weight.\footnote{The reason that $\underline{\beta}$ does not also have it is that $\underline{\beta}$ will not appear weighted on the characteristic hypersurface $N_{out}$. Cf.~the spin reduction performed in section \ref{spired}.} This additional weight has already been incorporated into the spacetime energy (\ref{degXall}).
For the spin $0$-curvature components $\rho$ and $\sigma$, more precisely the renormalized quantities $\phi_n=\tilde{\rho}\left(\mathcal{L}_T^n W\right)=r^3 \rho\left(\mathcal{L}_T^n W\right) - 2M \delta^n_0$, $\psi_n=\tilde{\sigma}\left(\mathcal{L}_T^n W\right)=r^3\sigma\left(\mathcal{L}_T^n W\right) + 2M \delta^n_0$ we define the energy
\begin{align} \label{strhoe}
E \left[\mathcal{D}\phi_n\right]\left( \tilde{\Sigma}_{\tau}\right) =  \int_{\tilde{\Sigma}_{\tau} \cap \{ r \leq R\}} dr d\omega |\mathcal{D}\phi_n |^2 + \int_{N_{out}\left(S^2_{\tau,R}\right)} dv d\omega \left[|\slashed{D}_4\phi_n |^2 + |\slashed{\nabla} \phi_n |^2 \right] 
\end{align}
and
\begin{align} \label{strho}
I_{deg}  \left[\mathcal{D}\phi_n \right] \left({\mathcal{U}}\right) =\int_{\mathcal{U}} dt^\star dr d\omega \Bigg[ \frac{\left(r-3M\right)^2}{r^{3+\delta}} \left(\left(T \phi_n \right)^2 + \left(\slashed{\nabla}\phi_n \right)^2 \right) \nonumber \\ + \chi_{\{r \leq r_0\}} \frac{1}{r^{1+\delta}} \left(\partial_r \phi_n\right)^2 + \frac{1}{r^{1+\delta}} \left(T^\perp \phi_n \right)^2 + \frac{\phi_n^2}{r^{3+\delta}}  \Bigg] 
\end{align}
where $T^\perp = -p e_3 + q e_4$ 
\glossary{
name={$T^\perp$}, 
description={$2T^\perp = -p e_3 + q e_4$, vectorfield (non-degenerate derivative in integrated decay estimate for $\rho$ and $\sigma$)}
}
is perpendicular to $T=pe_3 + q e_4$. Note that there is globally non-degenerate control of the $T^\perp$-derivative only.  The analogous definitions are made for $\psi_n$ (replace $\phi_n$ by $\psi_n$ everywhere in the above). 

\subsection{Constants and Conventions}
Recall that $\delta = \frac{1}{100}$ has already been fixed. Unless  noted otherwise, $B$ ($b$) will denote a large (small) constant depending only on the mass $M$. If the constant depends on other values, it will typically be indexed by them, e.g.~$B_h$ denotes a constant which also depends on $h$. The quantity $\epsilon$ is the ultimate source of smallness and arises from the closeness assumptions in the definition of ultimately Schwarzschildean. We use the usual algebra of constants $B \epsilon = \epsilon$, $B \cdot B = B$ etc.
\section{The main theorems} \label{mtsec}
Recall that the definition of an ultimately Schwarzschildean spacetime (Definition \ref{ultS}) depends on the choice of an $\epsilon>0$. Recall also the energies of the Ricci-coefficients defined in section \ref{ricosec}.
We have the following statement of uniform boundedness:
\begin{theorem} \label{theo1}
[Boundedness] There is an $\epsilon>0$ such that the following statement is true: If $\left(\mathcal{R}, g\right)$ is a spacetime which is ultimately Schwarzschildean to order $k+1$ (for some $k>7$) in the sense of Definition \ref{ultS}, then the estimate
\begin{align} \label{theo1est}
 \sup_{\tau>\tau_0} \overline{\mathbb{E}}^{k} \left[W\right] \left(\Sigma_{\tau} \right) \leq B \left( \mathbb{C}^k \left[\mathfrak{R}\right] \left(\tau_0\right)+ \overline{\mathbb{E}}^{k} \left[W\right] \left(\Sigma_{0} \right) \right) \, 
\end{align}
holds, where
\begin{align} \label{Dnorm}
\mathbb{C}^{k} \left[\mathfrak{R}\right] \left(\tau_1,\tau_2\right)  = \sup_{\tau \in \left(\tau_1,\tau_2\right)}\overline{\mathbb{E}}^{k} \left[\mathfrak{R}\right] \left( {\Sigma}_{\tau}\right) +  \overline{\mathbb{E}}^{k} \left[\mathfrak{R}\right] \left( \mathcal{H} \left(\tau_1,\tau_2\right)\right)  + \overline{\mathbb{I}}^{k} \left[\mathfrak{R}\right] \left({\mathcal{M}} \left(\tau_1,\tau_2\right)\right) \nonumber
\end{align}
and $\mathbb{C}^{k} \left[\mathfrak{R}\right] \left(\tau_0\right) = \mathbb{C}^{k} \left[\mathfrak{R}\right] \left(\tau_0,\infty\right)$, which is bounded by the ultimately Schwarzschildean property.
\end{theorem}
We have formulated the statement in terms of the untilded $\mathcal{M}$,$\Sigma$ as it is easier to prove in this form. Of course, the estimate (\ref{theo1est}) then also holds for the tilded slices $\tilde{\Sigma}$ and the $\mathbb{D}^{n} \left[\mathfrak{R}\right]$-energy.

In other words, the energy involving $k$-derivatives of curvature is uniformly bounded. Note that the assumptions regard only $k+1$ derivatives of the metric, while the assertion provides bounds for $k$ derivatives of curvature, i.e.~$k+2$ derivatives of the metric. To prove Theorem \ref{theo1} it will be sufficient to exploit multiplier estimates based on the ultimately Killing vectorfield $T$ and a version of the redshift vectorfield. \newline

To formulate any type of decay statement one has to be aware of two important issues, which can be heuristically understood as follows. 
\begin{enumerate}
\item As discussed in the introduction, the components $\rho$ and $\sigma$ (and Lie-T-derivatives thereof) have to be renormalized and treated separately, if one wants to prove integrated decay estimates. These renormalized components, $\phi$, $\psi$, may be thought of as satisfying an inhomogeneous Regge-Wheeler equation on a Schwarzschild background of the type
\begin{equation} 
\mathcal{P}_{Regge-Wheeler} \left(\phi\right) = r^3 D \left( W \cdot \mathfrak{R} \right) \label{RWhe} \, ,
\end{equation}
the $r^3$-weight arising from the renormalization. Clearly, to prove energy estimates for $\phi$ one has to understand $r$-weighted curvature energies to estimate the inhomogeneity. This means that one cannot understand local integrated decay without at the same time understanding the characteristic decay in $r$ of the null-curvature components!

\item It is well known that for the homogeneous Regge-Wheeler equation on Schwarzschild only the angular modes $l\geq 2$ of $\phi$ decay, while the $l=0$ and $l=1$ modes describe perturbations to nearby Kerr solutions. Consequently, any decay statement will need to eliminate these modes.
\end{enumerate}

The first problem will be resolved by coupling the local integrated decay estimate -- which will require an $\epsilon$ of an $r$-weighted spacetime integral of curvature at infinity by the above heuristics -- with the $r$-weighted multiplier estimates near infinity, as described in the introduction. 

The second problem will be resolved by the ultimately Schwarzschildean assumption, which will allow us to a-priori estimate the troublesome lowest order terms in an integrated decay estimate (these lowest order terms are the manifestation of the $l=0$ and $l=1$ modes). We emphasize that in the case of the Regge-Wheeler equation on a fixed Schwarzschild background, these terms can only be controlled from the angular term by applying a Poincar\'e inequality which excludes the $l=0$ and $l=1$ modes. Here the situation is more favorable as we are assuming approach to Schwarzschild at a certain order, which takes care of the zeroth order terms.
\begin{theorem}  \label{maintheointdec1}
[Integrated weighted energy decay] With the assumptions of Theorem \ref{theo1} we have for any boundary-admissible decay matrix with $P>P_{\rho}$ and for any $n=0,...,k-1$ the estimate
\begin{align} \label{mestI}
\overline{\mathbb{E}}^{n+1}_P \left[W\right] \left(\tilde{\Sigma}_ {\tau_2}\right) + \overline{\mathbb{I}}^{n+1,deg}_{\tilde{P}} \left[W\right] \left(\tilde{\mathcal{M}}\left(\tau_1,\tau_2\right)\right)  \leq
B \cdot \overline{\mathbb{E}}^{n+1}_{P} \left[W\right] \left(\tilde{\Sigma}_{\tau_1}\right) \nonumber \\ +B \cdot \mathbb{D}^{n+1} \left[\mathfrak{R}\right] \left(\tau_1,\tau_2\right)  \, 
\end{align}
and, for any $\lambda>0$,
\begin{align} \label{mestI2}
\overline{\mathbb{E}}^{n}_P \left[W\right] \left(\tilde{\Sigma}_ {\tau_2}\right) + \overline{\mathbb{I}}^{n,nondeg}_{\tilde{P}} \left[W\right] \left(\tilde{\mathcal{M}}\left(\tau_1,\tau_2\right)\right)  \leq
B_{\lambda} \cdot \overline{\mathbb{E}}^{n+1}_{P} \left[W\right] \left(\tilde{\Sigma}_{\tau_1}\right)  \nonumber \\
+B \cdot \mathbb{D}^{n+1}_{\lambda} \left[\mathfrak{R}\right] \left(\tau_1,\tau_2\right)  \, ,
\end{align}
where $\mathbb{D}^k_\lambda \left[\mathfrak{R}\right]\left(\tau_1,\tau_2\right)$ was defined in (\ref{lambdawn}). 
Here the constants $B$ also depend on lower order energies of curvature and Ricci-coefficients, which are bounded by the ultimately Schwarzschildean assumption.
\end{theorem}
\begin{remark} \label{decremark}
Applying the estimate (\ref{mestI}) for $n=k-1$ from the data up to any slice in the future with decay matrix $P_0$ provides a stronger (than Theorem \ref{theo1}) statement of boundedness, which moreover, as we will see, does not hinge on $T$ being null on the horizon. It also yields boundedness of a degenerate (at the photon sphere) $r$-weighted spacetime integral (integrated decay). 
\end{remark}
\begin{remark}
In fact, one can write $B_{\lambda} \cdot \overline{\mathbb{E}}^{n}_{P} \left[W\right] \left(\tilde{\Sigma}_{\tau_1}\right)  + B \cdot \overline{\mathbb{E}}^{n+1}_{P} \left[W\right] \left(\tilde{\Sigma}_{\tau_1}\right)$ for the first term on the right hand side of (\ref{mestI2}). 
\end{remark}
Note that the non-degenerate estimate (\ref{mestI2}) needs only the $\mathbb{D}_\lambda^{n+1} \left[\mathfrak{R}\right]$-energy on the right hand side, thus requiring only a small amount of the non-degenerate integrated decay estimate for $n$-derivatives of the Ricci-coefficients, cf.~(\ref{lambdawn}). \newline

Iterating Theorem \ref{maintheointdec1} one can obtain
\begin{theorem} \label{maintheorem}
[Decay] With the assumptions as in Theorem \ref{maintheointdec1}, we also have, for $n=1, ... ,k-1$ and any $\lambda>0$ the estimates ($\tau \geq \tau_0$)
\begin{align} \label{prop2}
 \overline{\mathbb{E}}^{n}_{P_l} \left[W\right] \left(\tilde{\Sigma}_{\tau}\right) + \overline{\mathbb{E}}^{n} \left[W\right] \left( \mathcal{H}\left(\tau,\infty\right)\right) \nonumber 
 \\ \leq  \frac{B_{\lambda} }{\tau^{l}} \cdot \overline{\mathbb{E}}^{k}_{P_0} \left[W\right] \left(\tilde{\Sigma}_{0}\right) + B \sum_{i=1}^l \frac{1}{\tau^i}  \mathbb{D}^{n+i}_\lambda \left[\mathfrak{R}\right]\left(\tau\right)
\end{align}
\begin{align} \label{prop1}
\overline{\mathbb{I}}^{n,deg}_{\tilde{P}_l} \left[W\right] \left(\tilde{\mathcal{M}} \left(\tau,\infty\right)\right)  \nonumber \\
\leq  \frac{B_{\lambda} }{\tau^{l}} \cdot \overline{\mathbb{E}}^{k}_{P_0} \left[W\right] \left(\tilde{\Sigma}_{0}\right)  + B \sum_{i=1}^l \frac{1}{\tau^i}  \mathbb{D}^{n+i}_\lambda \left[\mathfrak{R}\right]\left(\tau\right) + B \cdot \mathbb{D}^n \left[\mathfrak{R}\right] \left(\tau\right)
\end{align}
\begin{align} \label{prop1b}
 \overline{\mathbb{I}}^{n-1,nondeg}_{\tilde{P}_l} \left[W\right]\left(\tilde{\mathcal{M}} \left(\tau, \infty\right)\right) \nonumber \\ \leq \frac{B_{\lambda} }{\tau^{l}} \cdot \overline{\mathbb{E}}^{k}_{P_0} \left[W\right] \left(\tilde{\Sigma}_{0}\right)  + B \sum_{i=0}^l \frac{1}{\tau^i}  \mathbb{D}^{n+i}_\lambda \left[\mathfrak{R}\right]\left(\tau\right) 
\end{align}
with
\begin{equation} 
l = l\left(n\right) = \left\{
\begin{array}{rl}
1 & \text{if } n= k-1\\
2 & \text{if } n = k-2 \\
3 & \text{if } n < k-2 \\
\end{array} \right.
\end{equation}
The constants $B$, $B_{\lambda}$ depend only on the mass $M$ and lower order energies of the Ricci coefficients (i.e.~$ \mathbb{D}^k \left[\mathfrak{R}\right] \left(\tau_0\right)$) and curvature which are $\epsilon$-small by the ultimately Schwarzschildean assumption.
\end{theorem}

In view of the powerful techniques of \cite{ChristKlei}, where estimates for the deformation tensor and the Ricci-coefficients are obtained from appropriate estimates on the Weyl tensor, it is reasonable to state
\begin{conjecture}
The spacetime of Theorem \ref{maintheorem} is weakly ultimately Schwarzschildean to order $k+2$.
\end{conjecture}
Here by \emph{weakly ultimately Schwarzschildean} we understand that the spacetime is ultimately Schwarzschildean in the sense of Definition \ref{ultS}, except that the highest integrated decay bound on the Ricci components in (\ref{iRdec}) degenerates at the photon sphere. This is of course a consequence of the fact that we can only prove a degenerate spacetime estimate for $k$ derivatives of curvature. \newline

We close the section with a brief discussion of the relevance of these results to the fully non-linear problem, i.e.~the problem of \emph{improving} the assumptions made on the Ricci coefficients. For this it is essential to prove in addition that the spacetime under consideration is ultimately Schwarzschildean to order $k-1$ \emph{with all constants and decay rates improved}. As remarked after the definition (\ref{lambdawn}), choosing $\tau=\tau_f$ sufficiently large, the energy $\mathbb{D}^k_\lambda \left[\mathfrak{R}\right]\left(\tau_f\right)$ is $\epsilon \cdot \lambda$-small. Hence the estimate (\ref{mestI2}) implies that $\overline{\mathbb{I}}^{k-1,nondeg}_{\tilde{P}_0} \left[W\right]\left(\tilde{\mathcal{M}} \left(\tau_f, \infty\right)\right) \leq \epsilon \cdot \lambda $, as the first term on the right hand side of (\ref{mestI2}) is arbitrarily small up to time $\tau_f$ by Cauchy stability. Moreover, the estimate (\ref{prop1b}) tells us that $\overline{\mathbb{I}}^{k-1,deg}_{\tilde{P}_1} \left[W\right]\left(\tilde{\mathcal{M}} \left(\tau, \infty\right)\right) \leq \epsilon \tau^{-1}$. At least  in principle -- i.e.~modulo the estimates for the Ricci-coefficients from curvature -- this is sufficient to improve the assumptions at that order, in particular assumption (\ref{awps}) and the assumption $\overline{\mathbb{I}}^{k} \left[\mathfrak{R}\right] \left(\tilde{\mathcal{M}} \left(\tau,\infty\right)\right) \leq \epsilon$ on the Ricci-coefficients.

There is one problem, however: We will not be able to improve the decay rate for the lowest order energy. In fact, due to the presence of the  $\mathbb{D}^n_\lambda \left[\mathfrak{R}\right]$ term in (\ref{prop1}), the decay \emph{rate} of the curvature can never be better than what we are assuming on the the lowest order energy of the Ricci-coefficients. The existence of this term originates from an error-term in the energy estimates, which is not cubic. As discussed in the introduction, we will show that this term is lower order in differentiability (allowing one to prove the Theorems above) but we will not be able to eliminate its non-cubic nature entirely. 

As a consequence, the results of the paper do not quite allow one to move immediately to the fully non-linear problem: There is still a missing ingredient regarding the nature of this error-term. This problem can in fact be studied in the context of a linearization of the equations. This is work in progress by the author and collaborators \cite{DafHoRodip}. 
\section{The $T$-energy}
In this section, we derive the energy of the $n$-commuted Weyl field $\widehat{\mathcal{L}}^n_{T} W$ as arising from the ultimately Killing field $T$. In Schwarzschild, this energy involves all components of curvature with good signs but with a degeneration of some components at the horizon. To stabilize the estimate under perturbations, we will have to invoke the redshift. Those estimates are carried out in section \ref{anondege}. In the following section, we simply allow for an $\epsilon$-small negative contribution near the horizon in our estimates (cf.~Lemma \ref{Tboundaryterm}). 

The uncommuted Weyl field ($n=0$ in the expression above) requires special attention. This is because the natural Bel-Robinson energy for $W$ will involve the term $\int |\rho|^2$, which does not decay. To remedy this we derive a renormalized energy directly from the (renormalized) Bianchi equations in section \ref{normalizedenergy}.

\subsection{The boundary term}
Let $\mathcal{W}$ denote any Weyl-field\footnote{\label{apl}For prospective applications, the reader can think of $\mathcal{W}=\widehat{\mathcal{L}}^i_{T} W$ for $i\geq 1$.} with null-components $\alpha, \beta, \rho, \sigma, \underline{\beta}, \underline{\alpha}$. 

We choose $\mathcal{X}=\mathcal{Y}=\mathcal{Z}=T$ in (\ref{bintro}). A computation using the ``Schwarzschildean" normal 
$\tilde{n}=\frac{1}{2}\sqrt{k_\chi} e_3 + \frac{1}{2} \frac{1}{\sqrt{k_\chi}}e_4$ (which is $C^1$ close to $n_{\Sigma}$ by the ultimately Schwarzschildean property) shows
\begin{align} \label{Tenergy}
J_{\mu}^{TTT} \left[\mathcal{W}\right]  \tilde{n}^\mu_{\Sigma} := Q\left(T,T,T, \tilde{n}^\mu_{\Sigma} \right) = \frac{1}{16} \Bigg\{ 2|\underline{\alpha}|^2 \left[ p^3 \sqrt{k_+} \right] + 2|\alpha|^2 q^3\frac{1}{\sqrt{k_+}} \nonumber \\
+ 4|\underline{\beta}|^2 \left(p^2 \sqrt{k_\chi} + p^3  \frac{1}{\sqrt{k_\chi}}\right) 
+ 4|\beta|^2 \left( p q^2  \frac{1}{\sqrt{k_\chi}} + q^3  \sqrt{k_\chi} \right) \nonumber \\
+ 4\left(\rho^2+\sigma^2\right) \left( p^2 q  \frac{1}{\sqrt{k_\chi}} + p q^2  \sqrt{k_\chi}\right) \, . \nonumber
\end{align}
Note that in the Schwarzschild case $p$ vanishes on the horizon while $q$ equals $1$, leading to a degeneration of all components except $\alpha$ and $\beta$. Since $|p-k_{\chi}^-|\leq \epsilon$ and $|q-\frac{k_\chi^+}{k_{\chi}}| \leq \epsilon$ by the ultimately Schwarzschildean assumption we conclude:
\begin{lemma} \label{Tboundaryterm}
\begin{align}
 \int_{\Sigma_{\tau}} Q \left(T,T,T,n_{\Sigma}\right) d\mu_{\Sigma_{\tau}} \geq b \int_{\Sigma_{\tau} \cap \{r \geq r_Y-\frac{3}{4} \left(r_Y-2M\right)\}} |\mathcal{W}|^2 r^2 dr d\omega \nonumber \\
-\epsilon \int_{\Sigma_{\tau} \cap \{r \leq r_Y\}} |\mathcal{W}|^2 r^2 dr d\omega
\end{align}
and the analogous identity for the slices $\tilde{\Sigma}_{\tau}$.\footnote{Note that the $\underline{\alpha}$ component will not appear on $N_{out}\left(S^2_{\tau,R}\right)$.}
\end{lemma}
The spacetime term $K_1^{TTT}\left[\mathcal{W}\right]$ is discussed in section \ref{K1TTT}, while we postpone the analysis of the much more intricate error term $K_2^{TTT}\left[\mathcal{W}\right]$ to section \ref{errorterms}.
\subsection{The error $K_1^{TTT}\left[\mathcal{W}\right]$} \label{K1TTT}
We have to understand 
\begin{equation} \label{contract}
 K_1^{TTT}\left[\mathcal{W}\right] = 3 Q^{\alpha \beta \gamma \delta} \left[ \mathcal{W}\right] \left( \phantom{}^{(T)}\pi_{\alpha \beta} \right) T_\gamma T_{\delta} \, 
\end{equation}
for a Weyl-tensor $\mathcal{W}$ (cf.~footnote \ref{apl}).
\begin{proposition} \label{K1est}
The error term $K_1^{TTT}\left[\mathcal{W}\right]$ satisfies
\begin{align} \label{wts}
\int_{\tilde{\mathcal{M}}\left(\tau_1, \tau_2\right)} K_1^{TTT}\left[\mathcal{W}\right] \leq \epsilon \left(\tau_1\right)^{-\frac{1}{4}} \  \sup_{\tau \in \left(\tau_1,\tau_2\right)} E \left[\mathcal{W}\right]\left( \tilde{\Sigma}_{\tau}\right) \nonumber \\+ \frac{\epsilon}{\sqrt{\tau_1}}\int_{\tilde{\mathcal{M}}\left(\tau_1, \tau_2\right)\cap \{r\geq R\}} r^{1-\delta} \|\underline{\alpha}\|^2 dt^\star dr d\omega \, .
\end{align}
Moreover, if $\tilde{\mathcal{M}}\left(\tau_1, \tau_2\right)$ is replaced by $\mathcal{M}\left(\tau_1, \tau_2\right)$ (and $\tilde{\Sigma}_\tau$ by $\Sigma_\tau$), the last term can be dropped.
\end{proposition}
\begin{proof}
In the region ${\tilde{\mathcal{M}}\left(\tau_1, \tau_2\right)\cap \{r\leq R\}}$ and for $\mathcal{M}\left(\tau_1,\tau_2\right)$ the estimate is easily obtained without the last term in view of the uniform $\left(t^\star\right)^{-\frac{5}{4}}$-decay of all components of the deformation tensor, cf.~(\ref{decaydeft}). In the region $r\geq R$ we have have to take into account that $\underline{\alpha}$ does not appear on the characteristic slices, which gives rise to the last term. In particular, using the formulae of section (\ref{ncbelrob}) and the null-components of the deformation tensor we can estimate the most difficult terms which arise from the contraction (\ref{contract}).
\begin{align}
\int \rho\left(\mathcal{W}\right) \underline{\alpha} \left(\mathcal{W}\right)\cdot {}^{(T)}\mathbf{i} \, \sqrt{g} dt^\star dr d\omega \leq \nonumber \\ 
\epsilon \cdot \int du \frac{1}{u^\frac{3}{2}} \int dv d\omega \sqrt{g} \, \rho\left(\mathcal{W}\right)^2 + \epsilon \int \|\underline{\alpha}\left(\mathcal{W}\right)\|^2 \left(t^\star\right)^{-\frac{1}{2}} \frac{\sqrt{g}}{r^2}dt^\star dr d\omega \nonumber \\ \leq \frac{\epsilon}{\sqrt{\tau_1}} \sup_{\tau \in \left(\tau_1,\tau_2\right)} E \left(\mathcal{W},  \tilde{\Sigma}_{\tau}\right) +  \frac{\epsilon}{\sqrt{\tau_1}}\int_{\tilde{\mathcal{M}}\left(\tau_1, \tau_2\right)\cap \{r\geq R\}} r^{1-\delta} \|\underline{\alpha}\|^2 dt^\star dr d\omega
\end{align}
and using that the component $n$ of the deformation tensor decays like $\frac{1}{u^\frac{1}{2} r^\frac{3}{3}}$,
\begin{align}
\int \| \underline{\alpha} \left(\mathcal{W}\right)\|^2 \, | {}^{(T)}\mathbf{n}| \, \sqrt{g} dt^\star dr d\omega \leq  \frac{\epsilon}{\sqrt{\tau_1}}\int_{\tilde{\mathcal{M}}\left(\tau_1, \tau_2\right)\cap \{r\geq R\}} r^{1-\delta} \|\underline{\alpha}\|^2 dt^\star dr d\omega \nonumber \, .
\end{align}
\end{proof}
\subsection{A non-degenerate energy} \label{anondege}
Note that in Schwarzschild the Ricci coefficients take the following values close to the horizon:
\begin{equation}
\frac{1}{4}tr \underline{H} \approx \Omega \approx -\frac{M}{r^2} \approx -\frac{1}{4M} \textrm{ \ \ \ \ \ and \ \ \ \ } tr H \approx \underline{\Omega} \approx 0
\end{equation}
Define
\begin{equation} \label{gammdefo}
 \gamma \left(r\right) = \xi\left(r\right) \left(1 + \frac{1}{c_{red}}k_-\right) \textrm{ \ \ \ and  \ \ \ } \eta \left(r\right) = \xi \left(r\right) \frac{1}{c_{red}}{k_-} \, ,
\end{equation}
where $\xi$ an interpolating function equal to $1$ in $r\leq r_Y$ and equal to zero for $r \geq r_Y+ \frac{r_Y-2M}{2}$ and $c_{red}=\frac{1}{20k}$ a small redshift constant ($k$ being the number of derivatives in the ultimately Schwarzschildean assumption). We will see the reason for that later.)
For $r \leq r_Y$ sufficiently close to the horizon we have in Schwarzschild
\begin{align}
-\slashed{D}_3 \gamma &= -\left[\partial_t - \partial_r \right] \gamma = \frac{1}{c_{red}} \frac{2M}{r^2} > \frac{1}{4M c_{red}} \, ,  \nonumber \\ 
|\slashed{D}_4 \gamma| &= |\left[k_+\partial_t +k_- \partial_r \right] \gamma| < \frac{1}{10M}  \, .\nonumber
\end{align}
These inequalities are stable under perturbation as we are working in a regular coordinate system. We summarize this as 
\begin{lemma} \label{rscon}
Choose $c_{red}=\frac{1}{20k}$. We can find $r_Y>2M$ such that in the region $r \leq r_Y$ the following inequalities hold
\begin{align}
|\slashed{D}_4 \gamma| + |\gamma \underline{\Omega}| + | tr H | \leq \frac{1}{20M}
\textrm{ \ \ \ , \ \ \ } \frac{1}{2} < - M \cdot tr \underline{H}  < 2 \, , \nonumber \\
-\slashed{D}_3 \gamma \geq \frac{1}{4M c_{red}} \textrm{ \ \ \ , \ \ \ } \frac{1}{2} < - 4M \Omega < 2  \textrm{ \ \ \ , \ \ \ } 0 \leq \frac{1-\mu}{c_{red}} \leq \frac{1}{10} \nonumber \, .
\end{align}
\end{lemma}

Defining now the vectorfield
\glossary{
name={$Y$}, 
description={redshift vectorfield (and, unfortunately, a null component of the Ricci-rotation-coefficients, cf.~the appendix. Since we work in a gauge where $Y=0$, no confusion can arise.)}
}
\begin{equation} \label{Ydef}
Y = \gamma e_3 + \eta e_4 \, ,
\end{equation}
we observe using the formulae (\ref{nucos})-(\ref{nucoe}) in conjunction with the bounds of Lemma \ref{rscon} that in $r\leq r_Y$ the bounds
\begin{align} \label{defsum}
{}^{(Y)}\pi^{33} \geq \frac{1}{8M} \textrm{ \ , } {}^{(Y)}\pi^{34} \geq \frac{1}{18Mc_{red}}\textrm{ \ , } {}^{(Y)}\pi^{44} \geq \frac{1}{18Mc_{red}} \, , \nonumber \\
 |{}^{(Y)}\pi^{AB} \delta_{AB}| \leq \frac{2}{M}  \textrm{ \ \ , \ \ }  \| {}^{(Y)}\pi^{A3}\| + \| {}^{(Y)}\pi^{A4}\| \leq \epsilon 
\end{align}
hold. Finally, we set
\glossary{
name={$N$}, 
description={$N=T+Y$, everywhere timelike vectorfield, coincides with $T$ away from the horizon}
}
\begin{equation}
 N = T + Y \, ,
\end{equation}
which is future directed timelike everywhere, and state
\begin{proposition} \label{rsW}
The $N$-boundary-term satisfies
\begin{align} \label{fiid1}
 Q\left[\mathcal{W}\right] \left(N,N,N, n^\mu_{\Sigma}\right) \geq b \Bigg(|\underline{\alpha}|^2 + | \alpha|^2 + |\underline{\beta}|^2 + |\beta|^2 + \left(\rho^2 + \sigma^2 \right) \Bigg) \, 
\end{align}
everywhere on the black hole exterior. The quantity 
\begin{equation}
K_1^{NNN} \left[\mathcal{W}\right] = 3Q\left[\mathcal{W}\right] \left( \phantom{}^{(T+Y)}\pi, T+Y, Y+Y\right)  
\end{equation}
satisfies
\begin{align}
K_1^{NNN} \left[\mathcal{W}\right] &\geq b \cdot  Q\left[\mathcal{W}\right]\left(N,N,N, n^\mu_{\Sigma}\right)    \textrm{ \ \ \ \ for $r \leq r_Y$}\nonumber
\end{align}  
\begin{equation}
 |K_1^{NNN}\left[\mathcal{W}\right]| \leq B \cdot Q\left[\mathcal{W}\right]\left(T,T,T, n_{\Sigma}\right) \textrm{ \ \ \ \ for $r_Y \leq r \leq r_Y + \frac{r_Y-2M}{2}$} \nonumber
\end{equation}
\begin{equation}
K_1^{NNN}\left[\mathcal{W}\right] = K_1^{TTT}\left[\mathcal{W}\right] \textrm{ \ \ \ \ for $r\geq r_Y+ \frac{r_Y-2M}{2}$}\nonumber
\end{equation}
\end{proposition}
\begin{proof}
Since $N=\left(\frac{p}{2} + \gamma\right) e_3 + \left(\frac{q}{2} + \eta \right) e_4 = x e_3 + y e_4$ is timelike, the bound (\ref{fiid1}) is immediate. 
For the bulk term, only the statement in $r\leq r_Y$ is not obvious. For this we need to explicitly compute the contraction $Q_{abcd} \left[\mathcal{W}\right] \phantom{}^{(Y)}\pi^{ab} N^c N^d$.
\begin{align}
Q_{abcd} \left[\mathcal{W}\right] \phantom{}^{(Y)}\pi^{ab} N^c N^d = Q_{3333} \phantom{}^{(Y)}\pi^{33} x^2 + Q_{3334} \left(2\phantom{}^{(Y)}\pi^{34} x^2 + 2xy \phantom{}^{(Y)}\pi^{33}\right) \nonumber \\
+ Q_{3434} \left(\phantom{}^{(Y)}\pi^{33} y^2 + 4\phantom{}^{(Y)}\pi^{34} xy + \phantom{}^{(Y)}\pi^{44} x^2\right) + Q_{3444} \left(2\phantom{}^{(Y)}\pi^{34} y^2 + 2\phantom{}^{(Y)}\pi^{44} xy \right) \nonumber \\
+ Q_{4444} \phantom{}^{(Y)}\pi^{44} y^2 
+ \underline{ \phantom{}^{(Y)}\pi^{AB} \left(Q_{AB33} x^2 + 2 Q_{AB34} xy + Q_{AB44} y^2\right)} \nonumber \\
Q_{3Acd} \left[\mathcal{W}\right] \phantom{}^{(Y)} \pi^{3A} N^c N^c + Q_{4Acd}\left[\mathcal{W}\right] \phantom{}^{(Y)} \pi^{4A} N^c N^c \, . \nonumber
\end{align}
Note that the terms in the last line are $\epsilon$-small in view of (\ref{defsum}). As $x$ and $y$ are uniformly bounded above and below and in view of the properties of ${}^{(Y)}\pi$ summarized in (\ref{defsum}), we find (using the formulae for the components of the Bel-Robinson tensor collected in section \ref{ncbelrob}) that the wrong-signed underlined term can be absorbed by the other, positive, terms.  This yields the desired bound in $r\leq r_Y$.
\end{proof}
We remark that the above proposition should be thought of as the analogue of the redshift for the wave equation, cf.~\cite{DafRod2, Mihalisnotes}.
\subsection{A (decaying) $T$-energy at the lowest order} \label{normalizedenergy}
As mentioned at the beginning of this section, the $T$-energy at the lowest order involves the $L^2$ energy of $\rho$ itself, which does not decay. We will now derive a renormalized $T$-energy at the lowest order, using the renormalized null Bianchi equations directly:
\begin{proposition}
At the lowest order, we have the estimate
\begin{align}
\int_{\Sigma_{\tau_2} \cap \{ r \geq r_Y - \frac{3}{4} \left(r_Y-2M\right) \}} \|\tilde{W}\|^2 \, r^2 dr \, d\omega \leq \int_{\Sigma_{\tau_1} } \|\tilde{W}\|^2 \, r^2 dr \, d\omega  \nonumber \\ + \epsilon \sup_{\tau_1,\tau_2} \int_{\Sigma_{\tau} \cap \{r \geq r_Y\}} |\tilde{W}|^2 r^2 dr d\omega
+ B  \|r^{1+\delta} \left(\mathfrak{R}-\mathfrak{R}_{SS}\right)^2 \|_{L^\infty}  \int_{\mathcal{M}\left(\tau_1,\tau_2\right)} \frac{1}{r^{1+\delta}} \|\tilde{W}\|^2 \nonumber \\ + B \|\rho^2\|_{L^\infty} \int_{\mathcal{M}\left(\tau_1,\tau_2\right)} \|\mathfrak{R}-\mathfrak{R}_{SS}\|^2 \, , \nonumber 
\end{align}
where we recall that $\tilde{W}$ contains the renormalized component $\hat{\rho} = \rho + \frac{2M}{r^3}$. The same identity holds for the slices $\tilde{\Sigma}_{\tau}$ and regions $\tilde{\mathcal{M}}\left(\tau_1,\tau_2\right)$.
\end{proposition}
\begin{proof}
As in the previous section we will use the vectorfield $\tilde{T}$ for computations and then argue by stability of the estimate in view of the $C^1$ closeness of $\tilde{T}$ to $T$. Write the Bianchi equations as
\begin{align} \label{TB1}
\tilde{T}\left(\|\beta\|^2\right) + \left[k_-^\chi tr \underline{H} + 2\frac{k_+^\chi}{k_\chi} tr H\right] \|\beta\|^2 = k_-^\chi\left( \slashed{\mathcal{D}}^\star_1 \left(-\hat{\rho},\sigma\right) + E_3 \left(\beta\right) \right) \beta \nonumber \\ + \frac{k_+^\chi}{k_\chi}\left[\slashed{div} \alpha + E_4\left(\beta\right) \right] \beta
\end{align}
\begin{align} \label{TB2}
\tilde{T}\left(\hat{\rho}^2\right) + \frac{3}{2}\hat{\rho}^2 \left[k_-^\chi tr \underline{H} + \frac{k_+^\chi}{k_\chi} tr H\right] = -k_-^\chi \left[\slashed{div} \underline{\beta} + \hat{E}_3\left(\rho\right) \right] \hat{\rho} \nonumber \\ + \frac{k_+^\chi}{k_\chi} \left[\slashed{div} \beta + \hat{E}_4\left(\rho\right) \right] \hat{\rho} \end{align}
\begin{align} \label{TB3}
\tilde{T}\left(\sigma^2\right) + \frac{3}{2}{\sigma}^2 \left[k_-^\chi tr \underline{H} + \frac{k_+^\chi}{k_\chi} tr H\right] = - k_-^\chi \left[ \slashed{curl} \underline{\beta} + \hat{E}_3\left(\sigma\right) \right]\hat{\sigma} \nonumber \\ +  \frac{k_+^\chi}{k_\chi} \left[\slashed{curl} \beta + \hat{E}_4\left(\sigma\right)\right] \sigma  
\end{align}
\begin{align} \label{TB4}
\tilde{T}\left(\|\underline{\beta}\|^2\right) + \left[2 k_-^\chi tr \underline{H} + \frac{k_+^\chi}{k_\chi} tr H\right] \|\underline{\beta}\|^2  =  +k_-^\chi \left[-\slashed{div} \underline{\alpha} + E_3\left(\underline{\beta}\right) \right] \underline{\beta} \nonumber \\  + \frac{k_+^\chi}{k_\chi} \left( \slashed{\mathcal{D}}^\star_1 \left(\hat{\rho},\sigma\right) + E_4 \left(\underline{\beta}\right) \right) \underline{\beta}
\end{align}
Note that the terms proportional to $\hat{\rho}^2$ and $\sigma^2$ already decay in view of the assumptions on the Ricci coefficients. 

We now multiply the four equations by appropriate weights (depending only on $r$), add them up and integrate over the spacetime region $\mathcal{M}\left(\tau_1,\tau_2\right)$. The weights can be read of from expression (\ref{Tenergy}) and ensure that after integration by parts all spacetime terms which appear decay: Multiplying (\ref{TB1}) by $w_1=r^2 \left(\frac{k_+^\chi}{k_\chi}\right)^2$, (\ref{TB2}) and (\ref{TB3}) by $w_2 = r^2\left(\frac{k_+^\chi}{k_\chi}\right)k_-^\chi$ and finally (\ref{TB4}) by $w_3 = r^2\left(k_-^\chi\right)^2$ and integrating with respect to $dt^\star \, dr \, d\omega$ yields the Proposition.

As an example we look at the $\alpha$ spacetime term arising from equation (\ref{TB1}) after integration by parts:
\begin{align}
\int dt^\star dr d\omega \Bigg( \frac{1}{2} \partial_t \left( r^2 \left(\frac{k_+^\chi}{k_\chi}\right)^2\| \alpha\|^2\right) -  \frac{1}{2} \partial_r \left(r^2 \left(\frac{k_+^\chi}{k_\chi}\right)^3 \| \alpha\|^2\right) + \nonumber \\ \left[\frac{1}{2} \partial_r \left(\left( \frac{k_+^\chi}{k_{\chi}}\right)^3 \right) + \left(\frac{1}{2} tr \underline{H} - 4\underline{\Omega}\right)\left( \frac{k_+^\chi}{k_{\chi}}\right)^3 \right] \| \alpha\|^2 - \left[E_3\left(\alpha\right) - 4\underline{\Omega}\alpha\right] \left( \frac{k_+^\chi}{k_{\chi}}\right)^3 \alpha \Bigg) \nonumber
\end{align}
The first term in the second line is zero in Schwarzschild and decays in the perturbed case, in view of the assumptions on the Ricci-coefficients. The second term is easily estimated by the terms appearing on the right hand side of the Proposition. 
\end{proof}
\section{The error $K_2^{\mathcal{X}\mathcal{Y}\mathcal{Z}}\left[\widehat{\mathcal{L}}^n_T W\right]$} \label{errorterms}
In order for the identity (\ref{mainid}) to be useful for us, we need to understand the terms $K_1^{\mathcal{X}\mathcal{Y}\mathcal{Z}}\left[\mathcal{W}\right]$ and $K_2^{\mathcal{X}\mathcal{Y}\mathcal{Z}}\left[\mathcal{W}\right]$ on the right hand side. The two terms are of a different nature. If we commute with the ultimately Killing field $T$, $K_2^{\mathcal{X}\mathcal{Y}\mathcal{Z}}\left[\widehat{\mathcal{L}}^n_T W\right]$ is the error arising from the commutation and hence expected to be small for any $\mathcal{X},\mathcal{Y},\mathcal{Z}$ with bounded coefficients. The term $K_1^{\mathcal{X}\mathcal{Y}\mathcal{Z}}\left[\widehat{\mathcal{L}}^n_T W\right]$, on the other hand, is only expected to be small if $\mathcal{X},\mathcal{Y},\mathcal{Z}$ are ultimately Killing, cf.~Proposition \ref{K1est}. This entire section is devoted to controlling the error $K_2^{\mathcal{X}\mathcal{Y}\mathcal{Z}}\left[\widehat{\mathcal{L}}^n_T W\right]$ as arising from commutation with $T$. Remark: As this part of the argument does not depend on the assumption that $tr {}^{(T)}\pi=0$ (nor the gauge $Y=0$), we will derive all the formulae including this term.

\subsection{General null decomposition of $J\left(T,W\right)$}
Let $W$ be any Weyl field which itself satisfies the inhomogeneous Bianchi equation $
D^\alpha \left( W\right)_{\alpha \beta \gamma \delta} =  \bar{J}_{\beta \gamma \delta}$. We consider the commuted Bianchi equation (cf.~(\ref{Weylcommute1}))
\begin{equation} 
D^\alpha \left(\widehat{\mathcal{L}}_T W\right)_{\alpha \beta \gamma \delta} =  J_{\beta \gamma \delta} \left(T,W\right) +\widehat{\mathcal{L}}_T \bar{J}_{\beta \gamma \delta} \, \, .
\end{equation} 
Assuming that $\bar{J}$ also arose from commutation with such a vectorfield, it becomes clear that we have to understand the structure of the term $J\left(T,W\right)$ and Lie-derivatives thereof. We define, following \cite{ChristKlei},
\begin{align}
\phantom{}^{(T)}\mathfrak{p}_\gamma &= D^{\alpha} \phantom{}^{(T)}\widehat{\pi}_{\alpha \gamma} \, , \\
 \phantom{}^{(T)}\mathfrak{q}_{\alpha \beta \gamma} &= D_{\beta} \phantom{}^{(T)}\widehat{\pi}_{\gamma \alpha} - D_{\gamma} \phantom{}^{(T)}\widehat{\pi}_{\beta \alpha} - \frac{1}{3} \left(\phantom{}^{(T)}\mathfrak{p}_\gamma g_{\alpha \beta} - \phantom{}^{(T)}\mathfrak{p}_\beta g_{\alpha \gamma} \right) \, ,
\end{align}
and decompose $J\left(T,W\right)$ as
\begin{equation}
 J\left(T,W\right) = \frac{1}{2} \left(J^1\left(T,W\right) + J^2\left(T,W\right)\right) \, ,
\end{equation}
where
\begin{equation} \label{J1}
 J^1\left(T,W\right)_{\beta \gamma \delta} = \phantom{}^{(T)}\widehat{\pi}^{\mu \nu} D_{\nu} W_{\mu \beta \gamma \delta} \, ,
\end{equation} 
\begin{align}  \label{J2}
J^2\left(T,W\right)_{\beta \gamma \delta} =  \phantom{}^{(T)}\mathfrak{p}_{\lambda} W^{\lambda}_{\phantom{\lambda} \beta \gamma \delta}  + \phantom{}^{(T)}\mathfrak{q}_{\alpha \beta \lambda} W^{\alpha \lambda}_{\phantom{\alpha \lambda} \gamma \delta} \nonumber \\ + \phantom{}^{(T)}\mathfrak{q}_{\alpha \gamma \lambda} W^{\alpha \phantom{\beta} \lambda}_{\phantom{\alpha} \beta \phantom{\lambda}\delta} + \phantom{}^{(T)}\mathfrak{q}_{\alpha \delta \lambda} W^{\alpha \phantom{\beta \gamma} \lambda}_{\phantom{\alpha} \beta \gamma}  \, .
\end{align}
Note that $\mathfrak{q}_{\beta \gamma \delta}$ is a Weyl current, i.e.~it satisfies 
\begin{equation}
 \mathfrak{q}_{\beta \gamma \delta} = - \mathfrak{q}_{\beta \delta \gamma} \textrm{ \ \ \ and \ \ \ } 0 = g^{\beta \gamma} \mathfrak{q}_{\beta \gamma \delta} = -\frac{1}{2} \mathfrak{q}_{34 \delta} - \frac{1}{2} \mathfrak{q}_{43\delta} + \delta^{AB} \mathfrak{q}_{AB \delta} \, .
\end{equation}
We recall that the error-term $K_2^{\mathcal{Y}\mathcal{Z}\mathcal{U}} \left[\mathcal{W}\right]$  in the energy identity for the Bel-Robinson tensor arises from $J^1\left(T,W\right) + J^2\left(T,W\right)$ via formula (\ref{divQ}). Hence we have to estimate the integral of 
\begin{equation}
 D\left(T,W\right)\left(\mathcal{Y},\mathcal{Z},\mathcal{U}\right)  =  K_2^{\mathcal{Y},\mathcal{Z},\mathcal{U}}\left[\widehat{\mathcal{L}}_T W\right] 
\end{equation}
for appropriate vectorfields $\mathcal{Y},\mathcal{Z},\mathcal{U}$. We proceed with a null decomposition of this term:

\begin{proposition} \label{GNDK}
We have 
\begin{equation}
 D(T,W)_{333} = 4 \underline{\alpha} \left(\widehat{\mathcal{L}}_T W\right) \cdot \underline{\Theta}\left(T,W\right) + 8 \underline{\beta} \left(\widehat{\mathcal{L}}_T W\right) \cdot \underline{\Xi} \left(T,W\right)
\end{equation}
\begin{align}
 D(T,W)_{334} = 8 \rho \left(\widehat{\mathcal{L}}_T W\right) \cdot \underline{\Lambda}\left(T,W\right) - 8 \sigma \left(\widehat{\mathcal{L}}_T W\right) \cdot \underline{K}\left(T,W\right) \nonumber \\  -  8 \underline{\beta} \left(\widehat{\mathcal{L}}_T W\right) \cdot \underline{I} \left(T,W\right)
\end{align}
\begin{align}
 D(T,W)_{443} = 8 \rho \left(\widehat{\mathcal{L}}_T W\right) \cdot {\Lambda}\left(T,W\right) + 8 \sigma \left(\widehat{\mathcal{L}}_T W\right) K\left(T,W\right) \nonumber \\  +  8 {\beta} \left(\widehat{\mathcal{L}}_T W\right) I \left(T,W\right)
\end{align}
\begin{equation}
 D(T,W)_{444} = 4 \alpha \left(\widehat{\mathcal{L}}_T W\right) \cdot \Theta\left(T,W\right) - 8 \beta \left(\widehat{\mathcal{L}}_T W\right) \cdot \Xi \left(T,W\right)
\end{equation}
where we have used the null decomposition of $J\left(T,W\right)$:
\begin{align}
\Lambda(J) &= \frac{1}{4} J_{434} \textrm{ \ \ \ \ \ \ \ \ \ \ \ \ \ \ \ \ \ \ \ \ \ \ \ \ \ } \underline{\Lambda}\left(J\right) = \frac{1}{4} J_{343} \nonumber \\
K(J) &= \frac{1}{4} \epsilon^{AB} J_{4AB} \textrm{ \ \ \ \ \ \ \ \ \ \ \ \ \ \ \ \ \ \ \ } \underline{K}\left(J\right) =  \frac{1}{4} \epsilon^{AB} J_{3AB} \nonumber \\
\Xi(J)_A &= \frac{1}{2} J_{44A} \textrm{ \ \ \ \ \ \ \ \ \ \ \ \ \ \ \ \ \ \ \ \ \ \ \ \ \ }\underline{\Xi}\left(J\right) = \frac{1}{2} J_{33A} \nonumber \\
I(J)_A &= \frac{1}{2} J_{34A} \textrm{ \ \ \ \ \ \ \ \ \ \ \ \ \ \ \ \ \ \ \ \ \ \ \ \ \ }\underline{I}\left(J\right) = \frac{1}{2} J_{43A} \nonumber \\
\Theta_{AB} &= \frac{1}{2} \left( J_{A4B} + J_{B4A} - \left(\delta^{CD}J_{C4D}\right) \delta_{AB} \right)\nonumber \\
\underline{\Theta}_{AB} &= \frac{1}{2} \left( J_{A3B} + J_{B3A} - \left(\delta^{CD}J_{C3D}\right) \delta_{AB}\right)
\end{align}
\end{proposition}
\begin{proof}
See 8.1.7c in \cite{ChristKlei}.
\end{proof} 
Clearly, from the point of view of decay in the interior, the worst terms of $J\left(T,W\right)$ are the terms in $J^2\left(T,W\right)$ which are proportional to the non-decaying component $\rho$.\footnote{Note that $J^1\left(T,W\right)$ only has derivatives of $\rho$ and is hence ``a derivative better" than the terms we mentioned.} We collect the null decomposition of both terms in the following Lemma.
\begin{lemma} \label{J1I}
We have the following formulae for the null-decomposition of $J^1\left(T,W\right)$
\begin{equation}
\Lambda \left(J^1\right) \equiv \rho \left[ -\frac{3}{2} tr H  \ \phantom{}^{(T)}\widehat{\pi}^{34} - \frac{3}{2} tr \underline{H} \ \phantom{}^{(T)}\widehat{\pi}^{33} - \frac{3}{8} tr H \left(\delta^{AB} \  \phantom{}^{(T)}\widehat{\pi}_{AB}\right)\right]
\end{equation}
\begin{equation}
\Xi\left(J^1\right) \equiv \rho \left[\frac{3}{4} tr H \  \phantom{}^{(T)}\widehat{\pi}^{3A}\right]
\end{equation}
\begin{equation}
I \left(J^1\right) \equiv \rho \left[ -\frac{3}{4} tr \underline{H} \  \phantom{}^{(T)}\widehat{\pi}^{3A} - \frac{3}{2} tr H \ \phantom{}^{(T)}\widehat{\pi}^{4A} \right]
\end{equation}
\begin{equation}
K\left(J^1\right) \equiv 0 \equiv \Theta\left(J^1\right)
\end{equation}
where $\equiv$ denotes equality up to terms of the form $ \phantom{}^{(T)}\pi \cdot  \slashed{D} \left(W \textrm{ \ but not $\rho$}\right)$, i.e.~decaying quadratically as $T$ is ultimately Killing.
\end{lemma}
\begin{proof}
Direct computation. For instance, for the component $\Lambda$ we find
\begin{eqnarray}
4\Lambda = \widehat{\pi}^{34} D_4 W_{3434} + \widehat{\pi}^{33} D_4 W_{3434} + \widehat{\pi}^{A4} D_4 W_{A434} \nonumber \\+\widehat{\pi}^{3A} D_A W_{3434} + \widehat{\pi}^{3A} D_3 W_{A434} +\widehat{\pi}^{AB} D_A W_{B434} \nonumber \\ \equiv \widehat{\pi}^{4A} \widetilde{\bar{J}}_{A3434} + 4\widehat{\pi}^{34} D_4 \rho +  4\widehat{\pi}^{44} D_3 \rho - 6 Y_A  \widehat{\pi}^{3A} \rho - 6 {}^\star Y_A  \widehat{\pi}^{3A} \rho \nonumber \\ + 4 \widehat{\pi}^{3A} \slashed{\nabla}_A \rho +  \widehat{\pi}^{AB} \left(\slashed{\nabla}_A \beta_B + \slashed{\nabla}_B \beta_A -3H_{AB} \rho - 3{}^\star \widehat{H}_{AB} \sigma\right) 
\end{eqnarray}
The other components are computed similarly, using the formulae in the appendix.
\end{proof}

\begin{lemma} \label{J2l}
We have the following formulae for the null-decomposition of $J^2\left(T,W\right)$
\begin{equation}
 \Lambda \left(J^2\right) \equiv \frac{3}{4}\rho\left[-\frac{1}{2}D_4 \left(tr \pi\right) + D^A \psi_{4A} + 2p \rho \right] \nonumber
\end{equation}
\begin{equation}
 I\left(J^2\right)_A \equiv \frac{3}{4} \rho \left[2q \beta_A + D_3 \psi_{A4}  + \frac{1}{2} D_A tr \pi\right] \nonumber 
 \end{equation}
\begin{equation}
 \Xi\left(J^2\right)_A \equiv -\frac{3}{4} \rho \left[2p \beta_A + D_4 \psi_{A4} \right]
\end{equation}
\begin{equation}
 \Theta_{AB}\left(J^2\right) \equiv \frac{3}{4}\rho\left[ D_A \psi_{4B} + D_B \psi_{4A} - \delta_{AB} \left(D^A \psi_{4A}\right) -2q \alpha_{AB} \right]
\end{equation}
\begin{equation}
 K\left(J^2\right) \equiv -\frac{3}{4}\rho \cdot \epsilon_{AB} \left[ D_A \psi_{4B} - p\sigma \epsilon_{AB} \right]
\end{equation}
and
\begin{align}
\underline{\Lambda} \left(J^2\right) = \frac{3}{4}\rho\left[-\frac{1}{2}D_3 \left(tr \pi\right) + D^A \psi_{3A} + 2q \rho \right]
\end{align}
\begin{equation}
\underline{I}\left(J^2\right)_A \equiv \frac{3}{4} \rho \left[-2p \underline{\beta}_A + D_4 \psi_{A3}  + \frac{1}{2} D_A tr \pi\right] \nonumber 
 \end{equation}
\begin{equation}
\underline{\Xi}\left(J^2\right)_A \equiv -\frac{3}{4} \rho \left[-2q \underline{\beta}_A + D_4 \psi_{A3} \right]
\end{equation}
\begin{equation}
\underline{\Theta}_{AB}\left(J^2\right) \equiv \frac{3}{4}\rho\left[ D_A \psi_{3B} + D_B \psi_{3A} - \delta_{AB} \left(D^A \psi_{3A}\right) -2p \underline{\alpha}_{AB} \right]
\end{equation}
\begin{equation}
\underline{K} \left(J^2\right) \equiv -\frac{3}{4}\rho \cdot \epsilon_{AB} \left[ D_A \psi_{3B} + q\sigma \epsilon_{AB} \right]
\end{equation}
where $\equiv$ denotes equality up to terms which do not involve the curvature component $\rho$ and are hence of the form ``$\left(W \textrm{ \ but not \ } \rho \right)\cdot D \phantom{}^{(T)}\pi$", i.e.~decaying quadratically.
\end{lemma}

\begin{proof}
Collecting the terms proportional to $\rho$ for $J^2\left(T,W\right)$, we find\footnote{Note the typo regarding the $\Theta$-equation in \cite{ChristKlei}, where the factor of $\frac{1}{2}$ is missing.}
\begin{equation}
 \Lambda \left(J\right) = \frac{1}{4} J_{434} \equiv \rho \left(-\frac{1}{2} \mathfrak{p}_4 + \frac{3}{4} \delta_{BA} \mathfrak{q}_{A4B}\right) 
\end{equation}
\begin{equation}
 K\left(J\right) = \frac{1}{4} \epsilon^{AB} J_{4AB} \equiv -\frac{3}{4}\rho \mathfrak{q}_{A4B} \epsilon_{AB}
\end{equation}
\begin{equation}
 \Xi \left(J\right)_A = \frac{1}{2}J_{44A} \equiv -\frac{3}{4}\rho \mathfrak{q}_{4A4}
\end{equation}
\begin{equation}
 I\left(J\right)_A = \frac{1}{2} J_{34A} \equiv \frac{1}{2}\left(\rho \mathfrak{p}_A - \frac{3}{2} \mathfrak{q}_{34A} \rho\right)
\end{equation}
\begin{align}
 \Theta\left(J\right)_{AB} = \frac{1}{2} \left(J_{A4B} + J_{B4A} - \left(\delta^{CD} J_{C4D}\right) \delta_{AB}\right) \nonumber \\ \equiv \frac{3}{4} \rho \left(\mathfrak{q}_{A4B} + \mathfrak{q}_{B4A} - \left(\mathfrak{q}_{C4D} \delta^{CD}\right)\delta_{AB}\right) \, .
\end{align}
with the analogous formulae for the bared quantities obtained form interchanging $3$'s and $4$'s in these formulae.
We define the antisymmetric tensor
 \begin{equation}
 \psi_{a b} = \frac{1}{2} \left(D_a T_b - D_b T_a \right)
 \end{equation}
 and compute, using the definition of the deformation tensor and commuting covariant derivatives
\begin{align}
 \mathfrak{q}_{A4B} &= D_4 \widehat{\pi}_{BA} - D_B \widehat{\pi}_{4A} + \frac{1}{3}\mathfrak{p}_4 \delta_{AB} 
 \nonumber \\
 &= \frac{1}{2} \left(D_4 D_B T_A + D_4 D_A T_B - D_B D_4 T_A - D_B D_A T_4   \right) - \frac{1}{4} D_4 tr \pi \delta_{AB} + \frac{1}{3}\mathfrak{p}_4 \delta_{AB} 
  \nonumber \\
 &= \frac{1}{2} \left(R_{4BA3}T^3 + R_{4BA4}T^4 + R_{4AB(a)}T^{(a)} + D_A D_4 T_B - R_{BA43}T^3 - D_A D_B T_4   \right) \nonumber \\ &- \frac{1}{4} D_4 tr \pi \delta_{AB} + \frac{1}{3}\mathfrak{p}_4 \delta_{AB} 
 \nonumber \\
 &= -\frac{1}{4} \delta_{AB} D_4 \left(tr \pi\right) + D_A \psi_{4B} + \frac{1}{2} \left(-2q \alpha_{AB} + 2p \rho \delta_{AB} - 2p\sigma \epsilon_{AB} \right)   + \frac{1}{3}\mathfrak{p}_4 \delta_{AB} \nonumber
 \end{align}
 \begin{align}
 \mathfrak{\mathfrak{q}}_{4A4} &= D_A \widehat{\pi}_{44} - D_4 \widehat{\pi}_{A4} 
 \nonumber \\
 &= \frac{1}{2}\left[2 D_A D_4 T_4 - D_4 D_A T_4 - D_4 D_4 T_A\right]
  \nonumber \\
 &= \frac{1}{2}\left[D_4 D_A T_4 - D_4 D_4 T_A + 2 R_{A443}T^3\right] = 2p \beta_A + D_4 \psi_{A4} 
 \nonumber
 \end{align}
 \begin{align}
 \mathfrak{q}_{34A} &= D_4 \widehat{\pi}_{A3} - D_A \widehat{\pi}_{43} - \frac{1}{3} \mathfrak{p}_A g_{34} = -2q \beta_A - D_3 \psi_{A4} -\frac{1}{2} D_A tr \pi + \frac{2}{3}\mathfrak{p}_A\nonumber \, .
 \end{align}
as well as
\begin{align}
\mathfrak{q}_{A3B} = -\frac{1}{4}\delta_{AB} D_3 \left(tr \pi\right) + D_3 \psi_{3B} + \frac{1}{2} \left(-2p \underline{\alpha}_{AB} + 2q \rho \delta_{AB} + 2q \sigma \epsilon_{AB}\right) + \frac{1}{3} \mathfrak{p}_3 \delta_{AB} \nonumber
\end{align}
\begin{align}
 \mathfrak{\mathfrak{q}}_{3A3} = -2q \underline{\beta}_A + D_3 \psi_{A3} 
 \nonumber
 \end{align}
 \begin{align}
 \mathfrak{q}_{43A}  = 2p \underline{\beta}_A - D_4 \psi_{A3} -\frac{1}{2} D_A tr \pi + \frac{2}{3}\mathfrak{p}_A\nonumber \, .
 \end{align}
 from which the formulae of the lemma are easily derived.
\end{proof}

\subsection{The structure of $J\left(T,W\right)$}
We combine Lemmas \ref{J1I} and \ref{J2l} expressing the derivatives in terms of the slashed derivatives and the null component $\mathcal{P}_A := \psi_{A4}$ and $\mathcal{Q}_A := \psi_{A3}$.
\begin{lemma} \label{ncLT}
Let $\{\pi, \textrm{dec. RRC}\}$ denote any term which is of the form  ``constant times a component of the deformation tensor" or ``constant times a decaying Ricci rotation coefficient".
We have the following formulae for the null-decomposition of $J^2\left(T,W\right)$
\begin{equation}
 \Lambda \left(J\right) \equiv \frac{3}{4}\rho\left[-\frac{1}{2}\slashed{D}_4 \left(tr \pi\right) - \slashed{\nabla}^A \mathcal{P}_A +2 p \left(\rho + \frac{2M}{r^3}\right) + \{\pi, \textrm{dec. RRC}\} \right]
\end{equation}
\begin{equation}
 I\left(J\right)_A \equiv \frac{3}{4} \rho \left[ \frac{1}{2} \slashed{\nabla}_A tr \pi + \slashed{D}_3 \mathcal{P}_A + 2q \beta_A + \{\pi, \textrm{dec. RRC}\}  \right] 
 \end{equation}
\begin{equation}
 \Xi\left(J\right)_A \equiv -\frac{3}{4} \rho \left[\slashed{D}_4\mathcal{P}_A + 2p \beta_A + \{\pi, \textrm{dec. RRC}\} \right]
\end{equation}
\begin{equation}
 \Theta_{AB}\left(J\right) \equiv \frac{3}{4}\rho\left[ 2\slashed{\mathcal{D}}_2^\star \mathcal{P}_A  -2q \alpha_{AB} + \{\pi, \textrm{dec. RRC}\}  \right]
\end{equation}
\begin{equation}
 K\left(J\right) \equiv -\frac{3}{4}\rho \cdot \epsilon_{AB} \left[ -\slashed{\nabla}_A \mathcal{P}_B  - p\sigma \epsilon_{AB} + \{\pi, \textrm{dec. RRC}\} \right]
\end{equation}
and
\begin{equation}
\underline{\Lambda} \left(J\right) \equiv \frac{3}{4}\rho\left[-\frac{1}{2}\slashed{D}_3 \left(tr \pi\right) - \slashed{\nabla}^A \mathcal{Q}_A +2 q \left(\rho + \frac{2M}{r^3}\right) + \{\pi, \textrm{dec. RRC}\} \right]
\end{equation}
\begin{equation}
\underline{I}\left(J\right)_A \equiv \frac{3}{4} \rho \left[ \frac{1}{2} \slashed{\nabla}_A tr \pi + \slashed{D}_4 \mathcal{Q}_A - 2p \underline{\beta}_A + \{\pi, \textrm{dec. RRC}\}  \right] 
 \end{equation}
\begin{equation}
\underline{\Xi}\left(J\right)_A \equiv -\frac{3}{4} \rho \left[\slashed{D}_3\mathcal{Q}_A - 2q \underline{\beta}_A + \{\pi, \textrm{dec. RRC}\} \right]
\end{equation}
\begin{equation}
\underline{\Theta}_{AB}\left(J\right) \equiv \frac{3}{4}\rho\left[ 2\slashed{\mathcal{D}}_2^\star \mathcal{Q}_A  -2p \underline{\alpha}_{AB} + \{\pi, \textrm{dec. RRC}\}  \right]
\end{equation}
\begin{equation}
\underline{K}\left(J\right) \equiv -\frac{3}{4}\rho \cdot \epsilon_{AB} \left[ -\slashed{\nabla}_A \mathcal{Q}_B  + q\sigma \epsilon_{AB} + \{\pi, \textrm{dec. RRC}\} \right]
\end{equation}
\end{lemma}
\begin{proof}
Note that
\begin{align}
D^A \psi_{A4} &= \slashed{\nabla}^A \psi_{A4} - \frac{1}{2} tr H \psi_{34} + V_A \psi_{A4}   \nonumber \\ &= \slashed{\nabla}^A \psi_{A4} + \frac{1}{2} tr H \left(\slashed{D}_3 p - \slashed{D}_4 q -2\underline{\Omega} p + 2\Omega q\right)\nonumber \\
&= \slashed{\nabla}^A \psi_{A4} + \frac{1}{2} tr H \left(2\slashed{D}_3 p - 4 \underline{\Omega} p + \pi_{34} \right)
\end{align}
and
\begin{equation}
D_3 \psi_{A4} = \slashed{D}_3 \psi_{A4} - Z_A \psi_{34} -  2Z_B \psi_{AB} - 2\underline{\Omega} \psi_{A4}
\end{equation}
\begin{equation}
D_4 \psi_{A4} = \slashed{D}_4 \psi_{A4} - Y_A \psi_{34} -  2Y_B \psi_{AB} + 2\Omega \psi_{A4}
\end{equation}
\begin{align}
D_A \psi_{4B} + D_B \psi_{4A} - \delta_{AB} \left(D^A \psi_{4A}\right) \equiv - 2\slashed{\mathcal{D}}_2^\star \psi_{4A} - \widehat{H}_{AB} \psi_{34} - 2 \widehat{H_{CB}\psi_{AC}} \nonumber \\
\equiv 2\slashed{\mathcal{D}}_2^\star \psi_{A4} - \widehat{H}_{AB} \psi_{34} 
\end{align}
and similarly for exchanging $3$ and $4$ indices. Hence for $\Lambda$ we find
\begin{align}
\Lambda \left(J^1\right) \equiv \frac{3}{4} \rho \left[-\frac{1}{2} tr H \pi_{34} -\frac{1}{4} tr H tr \pi - \frac{1}{2} tr \underline{H} \pi_{44} - \frac{1}{2} tr H \delta^{AB} \pi_{AB} + \frac{1}{4} tr H tr \pi \right] \nonumber
\end{align}
\begin{align}
\Lambda \left(J^2\right) \equiv \frac{3}{4} \rho \left[ -\frac{1}{2}\slashed{D}_4 \left(tr \pi\right) - \slashed{\nabla}^A \mathcal{P}_A +2 p \rho - \frac{1}{2} tr H \left(2\slashed{\nabla}_3 p - 4 \underline{\Omega} p + \pi_{34}\right) \right] \nonumber
\end{align}
which upon adding produces
\begin{align}
\Lambda \left(J\right) \equiv \frac{3}{4} \rho \Big[ -\frac{1}{2}\slashed{D}_4 \left(tr \pi\right) - \slashed{\nabla}^A \mathcal{P}_A +2 p \rho - \frac{1}{2} tr H \left(2\slashed{\nabla}_3 p - 4 \underline{\Omega} p\right) \nonumber \\ -tr \chi \pi_{34} - \frac{1}{2} tr \underline{H} \pi_{44} - \frac{1}{2} tr H \delta^{AB} \pi_{AB} \Big] \nonumber
\end{align}
Since up to decaying Ricci coefficients $2 p \rho - \frac{1}{2} tr H \left(2\slashed{\nabla}_3 p - 4 \underline{\Omega} p\right) = 2p \left(\rho + \frac{2M}{r^3}\right)$ the result follows. For later purposes we also collect the explicit form of $I$. It is
\begin{align}
\underline{I}\left(J\right)_A \equiv \frac{3}{4} \rho \left[ \frac{1}{2} \slashed{\nabla}_A tr \pi + \slashed{D}_3 \mathcal{P}_A + 2q {\beta}_A - Z_A \psi_{34} - 2\underline{\Omega} \mathcal{P}_A  - tr \underline{H} \widehat{\pi}^{3A} - 2 tr H \widehat{\pi}^{4A} \right] \, . \nonumber
\end{align} 
\end{proof}
\subsection{$K_2^{\mathcal{X}\mathcal{Y}\mathcal{Z}} \left[\widehat{\mathcal{L}}_T W\right]$}
Using the previous Lemma we can now look at the spacetime integrands $D(T,W)$. 

\begin{proposition} \label{K2T}
The terms $D(T,W)$ arising from $J^1 \left(T,W\right) + J^2\left(T,W\right)$ satisfy
\begin{align}
D\left(T,W\right)_{444} \equiv 6 \slashed{D}_4 \left(\rho_0 \beta  \left(\widehat{\mathcal{L}}_T W\right) \cdot \mathcal{P} \right)+ 6\slashed{div} \left(\rho_0\alpha \left(\widehat{\mathcal{L}}_T W\right) \cdot \mathcal{P} \right) \nonumber \\ - 3 \mathcal{L}_T \Big[q \rho |\alpha\left(W\right)|^2 - 2 p \rho |\beta\left(W\right)|^2\Big]
\end{align}
\begin{align}
D\left(T,W\right)_{443} \equiv  6 \slashed{D}_3 \left( \rho_0 \mathcal{P} \cdot {\beta} \left(\widehat{\mathcal{L}}_T W\right)  \right) - 3 \slashed{D}_4 \left( \rho_0 \cdot tr {}^{(T)}\pi \rho \left(\widehat{\mathcal{L}}_T W\right)  \right) \nonumber \\ - 6 \slashed{\mathcal{D}}_1 \left(- \mathcal{P}_A \ \rho _ 0 \  \rho \left(\widehat{\mathcal{L}}_T W\right),  \mathcal{P}_A \ \rho _ 0 \  \sigma \left(\widehat{\mathcal{L}}_T W\right)\right) \nonumber \\
+ 6 \mathcal{L}_T \Big(p\rho \left(\rho+ \frac{2M}{r^3}\right)^2 + p\rho_0 \sigma^2 + q \rho_0 |\beta|^2 \Big)
\end{align}
\begin{align}
D\left(T,W\right)_{333} \equiv -6 \slashed{D}_3 \left(\rho_0 \underline{\beta}  \left(\widehat{\mathcal{L}}_T W\right) \cdot \mathcal{Q} \right)+ 6\slashed{div} \left(\rho_0 \underline{\alpha} \left(\widehat{\mathcal{L}}_T W\right) \cdot \mathcal{P} \right) \nonumber \\ - 3 \mathcal{L}_T \Big[p \rho |\underline{\alpha}\left(W\right)|^2 - 2 q \rho |\underline{\beta}\left(W\right)|^2\Big]
\end{align}
\begin{align}
D\left(T,W\right)_{334} \equiv  -6 \slashed{D}_4 \left( \rho_0 \mathcal{Q} \cdot \underline{\beta} \left(\widehat{\mathcal{L}}_T W\right)  \right) - 3 \slashed{D}_3 \left( \rho_0 \cdot tr {}^{(T)}\pi \rho \left(\widehat{\mathcal{L}}_T W\right)  \right) \nonumber \\ + 6 \slashed{\mathcal{D}}_1 \left(- \mathcal{Q}_A \ \rho _ 0 \  \rho \left(\widehat{\mathcal{L}}_T W\right),  -\mathcal{Q}_A \ \rho _ 0 \  \sigma \left(\widehat{\mathcal{L}}_T W\right)\right) \nonumber \\
+ 6 \mathcal{L}_T \Big(q\rho \left(\rho+ \frac{2M}{r^3}\right)^2 + q\rho_0 \sigma^2 + p \rho_0 |\underline{\beta}|^2 \Big)
\end{align}
where $\equiv$ denotes equality up to terms of the following form
\begin{align} \label{ermod1}
\left(\widehat{\mathcal{L}}_T W  \right) \cdot \{\rho_0 \pi \ , \   \rho_ 0 \textrm{dec. RRC}  \}  
\end{align}
and lower order terms of the form
\begin{align} \label{ermod2}
\left( W \cdot \tilde{W}  \right) \cdot \{ \pi \ , \  \textrm{dec. RRC}  \}  
\end{align}
\end{proposition}
\begin{proof}
First note that all terms arising from $J^1 \left(T,W\right)$ are of the form (\ref{ermod1}), (\ref{ermod2}) or even lower order.
Denote by $\rho_0= \rho\left(W\right)$ the non-decaying $\rho$ component of the uncommuted Weyl tensor we compute, modulo cubic terms in which all three components decay and terms of the form $\left(\widehat{\mathcal{L}}_T W\right) \cdot \{\pi, \textrm{dec. RRC}\}  \cdot \rho_0$:
\begin{align}
D\left(T,W\right)_{444} \equiv 4 \alpha \left(\widehat{\mathcal{L}}_T W\right) \cdot \frac{3}{2}\rho_0 \left(\slashed{\mathcal{D}}_2^\star \mathcal{P}_A - q \alpha_{AB}\right) \nonumber \\ - 8 \beta \left(\widehat{\mathcal{L}}_T W\right) \cdot \left(-\frac{3}{4} \rho_0 \left(2p \beta_A + \slashed{D}_4 \mathcal{P}_A \right)\right) \nonumber
\end{align}
Using Lemma \ref{LTTL} we can write
\begin{align}
D\left(T,W\right)_{444} \equiv 6 q \, \rho_0 \Bigg[ -\frac{1}{2}  \widehat{\mathcal{L}}_T \left(|\alpha|^2\right)  \Bigg] 
+ 12 p \, \rho_0 \Bigg[ \frac{1}{2} \widehat{\mathcal{L}}_T \left(|\beta|^2\right)  \Bigg]  \nonumber \\ 
6\rho_0 \left[\slashed{div} \alpha \left(\widehat{\mathcal{L}}_T W\right) - \slashed{D}_4 \beta  \left(\widehat{\mathcal{L}}_T W\right)\right] \mathcal{P}_A - \frac{6}{\rho_0} \left(\slashed{D}_4 \rho_0\right) \beta  \left(\widehat{\mathcal{L}}_T W\right) \cdot\mathcal{P}_A \nonumber \\ + 6 \slashed{D}_4 \left(\rho_0 \beta  \left(\widehat{\mathcal{L}}_T W\right) \cdot \mathcal{P}_A \right) + 6\slashed{div} \left(\rho_0\alpha \left(\widehat{\mathcal{L}}_T W\right) \cdot \mathcal{P} \right) - 6  \alpha \left(\widehat{\mathcal{L}}_T W\right) \slashed{\nabla} \rho_0 \mathcal{P} \nonumber
\end{align}
and hence
\begin{align}
D\left(T,W\right)_{444} \equiv  6\rho_0 \left[\left(2 tr H +  2\Omega + \frac{3}{2} tr H\right) \beta  \left(\widehat{\mathcal{L}}_T W\right)  - \alpha \left(\widehat{\mathcal{L}}_T W\right) \slashed{\nabla} \rho_0 - J_{4A4} \right] \mathcal{P}_A \nonumber \\ + 6 \slashed{D}_4 \left(\rho_0 \beta  \left(\widehat{\mathcal{L}}_T W\right) \cdot \mathcal{P} \right)+ 6\slashed{div} \left(\rho_0\alpha \left(\widehat{\mathcal{L}}_T W\right) \cdot \mathcal{P} \right) - 3 \mathcal{L}_T \left[q \rho_0 \alpha^2 - 2 p \rho \beta^2\right] \nonumber
\end{align}
Note that the terms in the last line are pure derivatives. Here we have used in particular the Bianchi equation (\ref{Bianchi2}) and the fact that the adjoint of $\slashed{\mathcal{D}}_2^\star$ is $\slashed{div}$. 

Similarly, again modulo cubic terms in which all three components decay and terms of the form $\left(\widehat{\mathcal{L}}_T W\right) \cdot  \{\pi, \textrm{dec. RRC}\} \cdot \rho_0$:
\begin{align}
 D\left(T,W\right)_{443} \equiv 8 \rho \left(\widehat{\mathcal{L}}_T W\right) \cdot \frac{3}{4} \rho_0 \left[2p \left(\rho + \frac{2M}{r^3}\right) - \slashed{div} \mathcal{P}_A - \frac{1}{2}  \slashed{D}_4 \left(tr \pi\right)\right] \nonumber \\ + 8 \sigma \left(\widehat{\mathcal{L}}_T W\right) \left(-\frac{3}{4} \rho_0 \left(- 2p \sigma - \slashed{curl} \ \mathcal{P}_A \right) \right) \nonumber \\  +  8 {\beta} \left(\widehat{\mathcal{L}}_T W\right) \cdot \left(\frac{3}{4} \rho_0 \left(2q \beta_A +  \slashed{D}_3 \mathcal{P}_A+ \frac{1}{2} \slashed{\nabla}_A tr \pi\right)\right) \nonumber
 \end{align}
 \begin{align}
 \equiv 12p \rho_0 \Bigg[\frac{1}{2} \widehat{\mathcal{L}}_T \left(\rho + \frac{2M}{r^3}\right)^2+  \left(\rho + \frac{2M}{r^3}\right)  \left(-\widehat{\mathcal{L}}_T \left(\frac{2M}{r^3}\right) - \frac{1}{8} tr {}^{(T)}\pi \rho \right) \Bigg] \nonumber \\ 
 +12p \rho_0 \Bigg[\frac{1}{2} \widehat{\mathcal{L}}_T \left(\sigma^2\right) - \frac{1}{8} tr {}^{(T)}\pi \sigma^2 \Bigg]
 +12 q \rho_0 \Bigg[ \frac{1}{2} \widehat{\mathcal{L}}_T \left(\beta^2\right)  \Bigg] \nonumber \\
 + 6 \rho_0 \mathcal{P}_A \Bigg[\slashed{\mathcal{D}}_1^\star \left(- \rho \left(\widehat{\mathcal{L}}_T W\right),  \sigma \left(\widehat{\mathcal{L}}_T W\right)\right) -\slashed{D}_3 {\beta} \left(\widehat{\mathcal{L}}_T W\right) \Bigg] \nonumber \\
 + 3 \rho_0 tr {}^{(T)}\pi \Bigg[\slashed{D}_4 \rho \left(\widehat{\mathcal{L}}_T W\right) - \slashed{div}  {\beta} \left(\widehat{\mathcal{L}}_T W\right) \Bigg] \nonumber 
 \\ 
 + 6 \slashed{D}_3 \left( \rho_0 \mathcal{P} \cdot {\beta} \left(\widehat{\mathcal{L}}_T W\right)  \right) - 3 \slashed{D}_4 \left( \rho_0 \cdot tr {}^{(T)}\pi \rho \left(\widehat{\mathcal{L}}_T W\right)  \right) \nonumber \\ - 6 \slashed{\mathcal{D}}_1 \left(- \mathcal{P}_A \ \rho _ 0 \  \rho \left(\widehat{\mathcal{L}}_T W\right),  + \mathcal{P}_A \ \rho _ 0 \  \sigma \left(\widehat{\mathcal{L}}_T W\right)\right) \, .
\end{align}
Using the Bianchi equations (\ref{Bianchi4}) and (\ref{Bianchi6}) yields the statement in the lemma. 
 The computation for $D\left(T, W\right)_{333}$ and $D\left(T, W\right)_{334}$ proceeds analogously.
 \end{proof}
Note that since these terms are going to be integrated, we have essentially gained a derivative: The worst error-terms involving $\rho_0$ are not of the form $\rho_0 \cdot u\left(\widehat{\mathcal{L}}_T W\right) \cdot D\pi$ but only $\rho_0 \cdot u\left(\widehat{\mathcal{L}}_T W\right) \cdot \pi$.
It is not hard to see that the structure we revealed at the first commutation with $T$, survives to higher orders.
\subsection{$K_2^{\mathcal{X}\mathcal{Y}\mathcal{Z}}\left[\widehat{\mathcal{L}}^{n+1}_T W\right]$}
Consider the $n+1$ times $T$-commuted Bianchi equation:
\begin{equation} 
D^\alpha \left(\widehat{\mathcal{L}}_T^{n+1} W\right)_{\alpha \beta \gamma \delta} = \widehat{\mathcal{L}}^n_T J_{\beta \gamma \delta} \left(T,W\right) + \sum_{i=0}^{n-1} \widehat{\mathcal{L}}^i_T J_{\beta \gamma \delta} \left(T, \widehat{\mathcal{L}}^{n-i}_T W\right) 
= J^{hard} + J^{easy} \nonumber \, .
\end{equation}
We use the index ``easy" for the second (summed) term, in view of the fact that a $T$-derivative has already fallen on $W$ in this expression. Hence this term already decays quadratically. On the other hand, the most difficult term in $J^{hard}$ is the one where all derivatives fall on the deformation tensor, in view of the non-decaying component $\rho$ of $W$.  
We set
\begin{equation}
K_2^{\mathcal{X}\mathcal{Y}\mathcal{Z}}\left[\widehat{\mathcal{L}}^{n+1}_T W\right] = \phantom{}^{easy}K_2^{\mathcal{X}\mathcal{Y}\mathcal{Z}}\left[\widehat{\mathcal{L}}^{n+1}_T W\right] + \phantom{}^{hard}K_2^{\mathcal{X}\mathcal{Y}\mathcal{Z}}\left[\widehat{\mathcal{L}}^{n+1}_T W \right] \, ,
\end{equation}
with the terms arising from the different terms on the right hand side, $J^{easy}$ and $J^{hard}$ via formula (\ref{divQ}). In view of the quadratically decaying structure of $J^{easy}$ we easily see
\begin{proposition} \label{easyerror}
Let $\left(\mathcal{R},g\right)$ be ultimately Schwarzschildean to order $k+1$ ($k>7$). For $1 \leq n \leq k-1$ and any $\lambda>0$,
\begin{align} \label{eeid}
 \Big| \int_{\tilde{\mathcal{M}}\left(\tau_1,\tau_2\right)} \sum_{a,b,c \in \{3,4\}} \phantom{}^{easy}K_2^{e_a e_b e_c}\left[\widehat{\mathcal{L}}^{n+1}_T W\right]\Big| \nonumber \\
 \leq  \left(\lambda + \frac{B_{\lambda} \epsilon}{\sqrt{\tau_1}}  \right) \cdot \sup_{\tau} \, \overline{\mathbb{E}}^{n+1} \left[W\right] \left(\tilde{\Sigma}_{\tau}\right) 
  + B_{\lambda} \epsilon \left(\tau_1\right)^{-\frac{1}{4}} \cdot \sup_{\tau}  \overline{\mathbb{E}}^{n} \left[\mathfrak{R}\right] \left(\tilde{\Sigma}_{\tau}\right) \nonumber \\
  +  \frac{\epsilon}{\sqrt{\tau_1}} \int_{\tilde{\mathcal{M}}\cap\{r\geq R\}} dt^\star dr d\omega r^{\frac{1}{2}} \left( \|\mathcal{D}^{n+1} \underline{\alpha}\|^2 \right) \, .
\end{align} 
Moreover, as in Proposition \ref{K1est}, if $\tilde{\mathcal{M}}\left(\tau_1, \tau_2\right)$ is replaced by $\mathcal{M}\left(\tau_1, \tau_2\right)$ (and $\tilde{\Sigma}_\tau$ by $\Sigma_\tau$), the last term can be dropped. For $n=0$, ${}^{easy}K_2^{e_a e_b e_c}\left[\widehat{\mathcal{L}}_T W\right] = 0$.
\end{proposition}
\begin{proof}
Observe that schematically 
\begin{equation} 
J_{easy} \equiv \sum_{s=0}^{n-1} \left[\left(D\widehat{\mathcal{L}}^s_T \pi\right) \left( \widehat{\mathcal{L}}^{n-s}_T W \right) + \left(\widehat{\mathcal{L}}^{s}_T \pi\right) \left( D \widehat{\mathcal{L}}^{n-s}_T W \right) \right] \, ,
\end{equation}
where the $\equiv$ ignores all lower order terms which arise from the commutation of $D$ with the Lie derivative\footnote{these terms will not only be of lower order but also introduce additional decay} and all combinatorial (integer) factors for these terms. We then estimate, first in $r\leq R$:
\begin{align}
\Big| \int_{\tilde{\mathcal{M}}\left(\tau_1,\tau_2\right) \cap \{r \leq R\}} \sum_{a,b,c \in \{3,4\}} \phantom{}^{easy}K_2^{e_a e_b e_c}\left[\widehat{\mathcal{L}}^{n+1}_T W\right]\Big| \leq \int dt^\star  \|\widehat{\mathcal{L}}^{n+1}_T W  \|_{L^2} \Bigg[ \nonumber \\
 \sum_{s=0}^{\lfloor \frac{n-1}{2} \rfloor} \left(\|D\widehat{\mathcal{L}}^s_T \pi\|_{L^\infty} \|\widehat{\mathcal{L}}^{n-s}_T W  \|_{L^2} + \|\widehat{\mathcal{L}}^s_T \pi\|_{L^\infty} \|D\widehat{\mathcal{L}}^{n-s}_T W  \|_{L^2} \right)  \nonumber \\ + \sum_{\lfloor \frac{n-1}{2} \rfloor + 1}^{n-1} \| D\widehat{\mathcal{L}}^s_T \pi\|_{L^2} \|  \cdot \widehat{\mathcal{L}}^{n-s}_T W  \|_{L^\infty} +  \| \widehat{\mathcal{L}}^s_T \pi\|_{L^2} \|  \cdot D\widehat{\mathcal{L}}^{n-s}_T W  \|_{L^\infty}  \Bigg] \nonumber  \, ,
 \end{align}
 to which we apply Cauchy's inequality and Sobolev embedding:
 \begin{align}
 \leq \lambda \cdot \sup_{\tau} \, \overline{\mathbb{E}}^{n+1} \left[W\right] \left(\tilde{\Sigma}_{\tau}\right)  \nonumber \\ 
 + B_{\lambda}\sum_{i=0}^{\lfloor \frac{n-1}{2} \rfloor} \left(\tau_2-\tau_1\right)^2 \cdot \sup_{\tau} \, \mathbb{E}^{i+3} \left[\mathfrak{R}\right] \left(\tilde{\Sigma}_{\tau}\right) \cdot \sup_{\tau} \, \mathbb{E}^{n+1-i} \left[W\right] \left(\tilde{\Sigma}_{\tau}\right) \nonumber \\
  + B_{\lambda}\sum_{i=\lfloor \frac{n-1}{2} \rfloor+1}^{n-1} \left(\tau_2-\tau_1\right)^2 \cdot \sup_{\tau} \, \mathbb{E}^{n-i+3} \left[W\right] \left(\tilde{\Sigma}_{\tau}\right) \cdot \sup_{\tau} \, \mathbb{E}^{i+1} \left[\mathfrak{R}\right] \left(\tilde{\Sigma}_{\tau}\right) \nonumber \, .
 \end{align}
One easily checks that for $n=k-1=7$ (the worst case) the ultimately Schwarzschildean assumption ensures that one can absorb the $\left(\tau_2-\tau_1\right)^2$-term using the decay of one of the $sup$-terms.

We can do the same estimate in $r\geq R$, the only thing we have to be careful about is that not all null-components appear on characteristic slices. This is easily accounted for by adding the spacetime term appearing on the right hand side of (\ref{eeid}). Note in this context that potentially divergent terms like
\begin{align}
\int dt^\star dr d\omega \, r^2 \, \underline{\alpha} \left(\widehat{\mathcal{L}}^{n+1}_T W\right) \cdot \left(\slashed{D}_3 \widehat{\mathcal{L}}^{n-1}_T  \left({}^{T} \mathbf{i}\right) \right)\cdot \widehat{\mathcal{L}}_T \underline{\alpha}
\end{align}
(i.e.~where all terms of the cubic expression only decay like $\frac{1}{r}$)
cannot appear in $K_2 \left[\mathcal{W}\right]$ because of the null-structure of the error-terms (cf.~the signature considerations in \cite{ChristKlei}). At least one component in the cubic error-term has improved decay in $r$.
\end{proof}

For $\phantom{}^{hard}K_2^{\mathcal{X}\mathcal{Y}\mathcal{Z}}$ we have
\begin{proposition} \label{harderror}
Let $\left(\mathcal{R},g\right)$ be ultimately Schwarzschildean to order $k+1$ ($k>7$) and $\mathcal{X}=p_\mathcal{X} e_3 + q_\mathcal{X} e_4$, $\mathcal{Y}=p_\mathcal{Y} e_3 + q_\mathcal{Y} e_4$, $\mathcal{Z}=p_\mathcal{Z} e_3 + q_\mathcal{Z} e_4$ be vectorfields for bounded spacetime functions $p_i$ and $q_i$ satisfying $|T\left(p_i\right)| + |T\left(q_i\right)|\leq \epsilon \cdot \tau^{-\frac{5}{4}}$ in the interior $r < t^\star$. For $0 \leq n \leq k-1$
and any positive $\lambda_2$ we have
\begin{align} \label{harderroresta}
\Big| \int_{\tilde{\mathcal{M}}}  {}^{hard}K_2^{\mathcal{X}\mathcal{Y}\mathcal{Z}}\left[\widehat{\mathcal{L}}^{n+1}_T W\right]\Big| \leq  \left(\lambda_2 + \epsilon \left(\tau_1\right)^{-\frac{1}{4}}\right) \sup_{\tau \in \left(\tau_1, \tau_2\right)}  \overline{\mathbb{E}}^{n+1} \left[W\right]\left(\tilde{\Sigma}_{\tau}\right)  \nonumber \\
+  \overline{\mathbb{I}}^{n+1,nondeg} \left[\mathfrak{R}\right] \left(\tilde{\mathcal{M}}\left(\tau_1,\tau_2\right)\right) + B_{\lambda_2} \cdot \overline{\mathbb{E}}^{n} \left[\mathfrak{R}\right]\left(\mathcal{H}, \tilde{\Sigma}_{\tau_2}, \tilde{\Sigma}_{\tau_1}\right) \nonumber \\ 
+ B  \int_{\tilde{\mathcal{M}}\left(\tau_1,\tau_2\right),\Sigma_{\tau_1},\Sigma_{\tau_2}, \mathcal{H}} \frac{1}{r^2} \left(  \sum_{i=1}^n \| \widehat{\mathcal{L}}^{i}_T W \|^2 + \|\tilde{W}\|^2\right)  \nonumber \\ + \frac{\epsilon}{\sqrt{\tau_1}} \int_{\tilde{\mathcal{M}}\cap\{r\geq R\}} dt^\star dr d\omega r^{\frac{1}{2}} \left( \|\mathcal{D}^{n+1} \underline{\alpha}\|^2 \right) + \lambda_2  \sup_{\mathcal{H},i} |p_i| \int_{\mathcal{H}} \|\widehat{\mathcal{L}}^{n+1}_T W \|^2 \, .
\end{align}
As in Proposition \ref{easyerror}, if $\tilde{\mathcal{M}}\left(\tau_1, \tau_2\right)$ is replaced by $\mathcal{M}\left(\tau_1, \tau_2\right)$ (and $\tilde{\Sigma}_\tau$ by $\Sigma_\tau$), the penultimate term can be dropped.
\end{proposition}

\begin{proof}
Observe that the null-decomposition of $J=J\left(T,W\right)$ almost commutes with $\widehat{\mathcal{L}}_T$:
\begin{equation}
\Lambda\left(\widehat{\mathcal{L}}_T J\right) = \widehat{\mathcal{L}}_T \left(\Lambda\left(J\right)\right) + \left(\textrm{decaying RRC}\right) \cdot \left(\Lambda\left(J\right), .... , \underline{\theta} \left(J\right)\right) \, ,
\end{equation}
where the error is of lower oder and introduces additional decay. Similarly for the other null-components. Now the $\widehat{\mathcal{L}}^n_T$-derivative of each null-component of $J$  is of one of the following forms
\begin{enumerate}
\item[(1)] $\widehat{\mathcal{L}}^{n}_T \left[\pi D\left(W \textrm{ \ but not $\rho$}\right) + D\pi \cdot \left(W \textrm{ \ but not $\rho$}\right)\right]$ or $\widehat{\mathcal{L}}^{n-1}_T \left[\pi \widehat{\mathcal{L}}_T D\rho + D\pi \cdot \widehat{\mathcal{L}}_T\rho \right]$, i.e.~a product of two (or more, arising from commutation) decaying components
\item[(2)] $\widehat{\mathcal{L}}^{n}_T \pi \ D\rho$ 
\item[(3)] $\widehat{\mathcal{L}}^{n}_T D\pi \ \rho$ (these terms were essentially computed in Lemma \ref{ncLT})
\end{enumerate}
In the expression for $K_2$ arising from these terms, terms of type (1) can be estimated as in ${}^{easy}K_2$.\footnote{Note again that, in view of signature considerations, terms involving twice the curvature component $\underline{\alpha}$ in conjunction with the (weakly $r$-decaying) null-component ${}^{(\mathcal{X})}\mathbf{i}$ cannot appear.} We only mention the term involving the highest derivative of the weakly decaying component ${}^{T}\mathbf{i}$. In view of signature considerations this term is of the form\footnote{Note that this term is easily estimated if $\tilde{\mathcal{M}}$ is replaced by $\mathcal{M}$, in view of the uniform decay of the deformation tensor (\ref{decaydeft}). The problem here is that $\underline{\alpha}$ does not appear on the characteristic hypersurfaces $N_{out}\left(S^2_{\tau,R}\right)$.}
\begin{align}
\int_{r\geq R} dt^\star dr d\omega \, r^2 \, \underline{\alpha} \left(\widehat{\mathcal{L}}^{n+1}_T W\right) \cdot \left(\slashed{D}_3 \widehat{\mathcal{L}}^{n}_T  \left({}^{(T} \mathbf{i}\right) \right)\cdot \left(W \textrm{ \ \ but not \ \ }  \underline{\alpha}, \rho \right) \nonumber \\
\| r^2 \left(W \textrm{ \ \ but not \ \ }  \underline{\alpha}, \rho \right)  \|_{L^\infty} \Big[ \int_{\tilde{\mathcal{M}}\cap\{r\geq R\}} dt^\star dr d\omega r^{\frac{1}{2}} \left( \|\mathcal{D}^{n+1} \underline{\alpha}\|^2 + \|\mathcal{D}^{n+1} \pi \|^2 \right)  \Big] \nonumber \\
\frac{\epsilon}{\tau_1} \Big[ \overline{\mathbb{I}}^{n+1,nondeg} \left[\mathfrak{R}\right] \left(\tilde{\mathcal{M}}\left(\tau_1,\tau_2\right)\right) + \int_{\tilde{\mathcal{M}}\cap\{r\geq R\}} dt^\star dr d\omega r^{\frac{1}{2}}  \|\mathcal{D}^{n+1} \underline{\alpha}\|^2  \Big] \, .
 \nonumber
\end{align}

For the terms of type (2), we first note that using the Bianchi equations for $\rho$, they can be transformed into terms of type (1) and terms of the form $\widehat{\mathcal{L}}^{n}_T \pi \ \left(tr H, tr \underline{H}\right) \rho$. Realizing that the decay in $r$ at infinity is not an issue for the latter term in view of $\rho$ decaying like $\frac{1}{r^3}$, we integrate it by parts in $T$ to obtain, modulo lower oder boundary terms which are already present on the right hand side of (\ref{harderroresta}):
\begin{align} \label{tyst}
\int_{\tilde{\mathcal{M}}\left(\tau_1,\tau_2\right)} \left(\widehat{\mathcal{L}}^{n}_T \pi \ \rho_0 \right)\left(\widehat{\mathcal{L}}^{n+1}_T W\right) \sim \int_{\tilde{\mathcal{M}}\left(\tau_1,\tau_2\right)} \left(\widehat{\mathcal{L}}^{n+1}_T \pi \ \rho_0 \right)\left(\widehat{\mathcal{L}}^{n}_T W\right)
\end{align}
for $n \geq 1$, while the right hand side equal to  $\int_{\tilde{\mathcal{M}}\left(\tau_1,\tau_2\right)} \left(\widehat{\mathcal{L}}_T \pi \ \rho_0 \right)\tilde{W}$, in case that $n=0$.
An application of Cauchy's inequality yields the terms found on the right hand side of Proposition \ref{harderror}.

Finally, for the terms of type (3), we can redo the computation of Proposition \ref{K2T} and transform this term into terms of type (1) and (2):
\begin{lemma}
We have the following generalization of Proposition \ref{K2T}
\begin{align}
D\left(T,\widehat{\mathcal{L}}^{n}_T W\right)_{444} \equiv 6 \slashed{D}_4 \left(\rho_0 \beta  \left(\widehat{\mathcal{L}}^{n+1}_T W\right) \cdot \widehat{\mathcal{L}}^{n}_T \mathcal{P} \right)+ 6\slashed{div} \left(\rho_0\alpha \left(\widehat{\mathcal{L}}^{n+1}_T W\right) \cdot \widehat{\mathcal{L}}^{n}_T \mathcal{P} \right) \nonumber \\ - 3 \mathcal{L}_T \Big[q \rho |\alpha\left(\widehat{\mathcal{L}}^{n}_T W\right)|^2 - 2 p \rho |\beta\left(\widehat{\mathcal{L}}^{n}_T W\right)|^2\Big] \nonumber
\end{align}
\begin{align}
D\left(T,W\right)_{443} \equiv  6 \slashed{D}_3 \left( \rho_0 \widehat{\mathcal{L}}^n_T \mathcal{P} \cdot {\beta} \left(\widehat{\mathcal{L}}^{n+1}_T W\right)  \right) - 3 \slashed{D}_4 \left( \rho_0 \cdot \widehat{\mathcal{L}}^{n}_T tr {}^{(T)}\pi \, \rho \left(\widehat{\mathcal{L}}_T W\right)  \right) \nonumber \\ - 6 \slashed{\mathcal{D}}_1 \left(- \widehat{\mathcal{L}}^n_T \mathcal{P}_A \ \rho _ 0 \  \rho \left(\widehat{\mathcal{L}}^{n+1}_T W\right),  \widehat{\mathcal{L}}^n_T\mathcal{P}_A \ \rho _ 0 \  \sigma \left(\widehat{\mathcal{L}}^{n+1}_T W\right)\right) \nonumber \\
+ 6 \mathcal{L}_T \Big(p\rho \left[\widehat{\mathcal{L}}^{n}_T \left(\rho+ \frac{2M}{r^3}\right)\right]^2 + p\rho_0 \sigma\left(\widehat{\mathcal{L}}^n_T W\right)^2 + q \rho_0 \|\widehat{\mathcal{L}}^n_T\beta\|^2 \Big) \nonumber
\end{align}
and similarly for the other two components. Here $\equiv$ denotes equality up to lower order terms which are of the form arising from (1) or (2).
\end{lemma}
\begin{proof}
We show the statement only for the first component, as the others are treated completely analogously. 
\begin{align}
D\left(T,\widehat{\mathcal{L}}^{n}_T W\right)_{444} = 4 \alpha \left(\widehat{\mathcal{L}}^{n+1}_T W\right) \cdot \Theta \left(T, \mathcal{L}^{n}_T W\right) - 8  \beta \left(\widehat{\mathcal{L}}^{n+1}_T W\right) \cdot \Xi \left(T, \mathcal{L}^{n}_T W\right) \nonumber \\
\equiv 4 \alpha \left(\widehat{\mathcal{L}}^{n+1}_T W\right) \cdot \widehat{\mathcal{L}}^{n}_T \Theta \left(T, W\right)
- 8  \beta \left(\widehat{\mathcal{L}}^{n+1}_T W\right) \cdot \widehat{\mathcal{L}}^{n}_T \Xi \left(T,  W\right) \nonumber
\end{align}
Now we can use Lemma \ref{ncLT}, where we computed the part of the null-components $\Theta$ and $\Xi$, which is proportional to $\rho$. Since commuting with the Lie-T-derivative only introduces lower order terms we obtain
\begin{align}
D\left(T,\widehat{\mathcal{L}}^{n}_T W\right)_{444} \equiv 4 \alpha \left(\widehat{\mathcal{L}}^{n+1}_T W\right) \cdot \frac{3}{2}\rho_0 \left(\slashed{\mathcal{D}}_2^\star \widehat{\mathcal{L}}^{n}_T \mathcal{P}_A - q \alpha \left(\widehat{\mathcal{L}}^{n}_T W\right) \right) \nonumber \\ 
- 8 \beta \left(\widehat{\mathcal{L}}^{n+1}_T W\right) \cdot \left(-\frac{3}{4} \rho_0 \left(2p \beta \left(\widehat{\mathcal{L}}^{n}_T W\right) + \slashed{D}_4 \widehat{\mathcal{L}}^{n}_T \mathcal{P}_A \right) \right) \, ,
\end{align}
from which point on we can follow the computations of Proposition \ref{K2T} to produce the result.
\end{proof}
We now integrate the main terms collected in the previous Lemma. Since they arise from
\begin{align}
D\left(T,\widehat{\mathcal{L}}^{n}_T W\right) \left(p_\mathcal{X} e_3 + q_\mathcal{X} e_4, p_\mathcal{Y} e_3 + q_\mathcal{Y} e_4, p_\mathcal{Z} e_3 + q_\mathcal{Z} e_4\right) \, ,
\end{align}
we observe that for vectorfields satisfying $p_i=0$ ($i=\mathcal{X}, \mathcal{Y}, \mathcal{Z}$) on the horizon, the boundary term on the horizon will vanish. In general, the strength of the horizon term will be of the size of the $p_i$'s on the horizon. Hence, for the terms appearing in the Lemma 
\begin{align}
\int_{\tilde{\mathcal{M}}\left(\tau_1,\tau_2\right)} D\left(T,\widehat{\mathcal{L}}^{n}_T W\right) \left(p_1 e_3 + q_1 e_4, p_2 e_3 + q_2 e_4, p_3 e_3 + q_3 e_4\right) \nonumber \\
\leq B \cdot \sup_{\mathcal{H},i} |p_i| \int_{\mathcal{H}} \left( \lambda_2 \|\widehat{\mathcal{L}}^{n+1}_T W \|^2 + \frac{1}{\lambda_2} \|\mathcal{D}^n \pi\|^2 \right) \nonumber \\
+ \lambda_2 \, \sup_{\tau \in \left(\tau_1, \tau_2\right)} \ \overline{\mathbb{E}}^{n+1} \left[W\right]\left(\Sigma_{\tau}\right) + B_{\lambda_2} \sup_{\tau \in \left(\tau_1, \tau_2\right)} \overline{\mathbb{E}}^{n} \left[\mathfrak{R}\right]\left(\Sigma_{\tau}\right) \nonumber \\
+ \textrm{spacetime terms of the form (\ref{tyst})}
\end{align}
\end{proof}

\subsection{Summary of the error estimates}
Combining Propositions \ref{easyerror} and \ref{harderror} we can summarize our estimate of the errorterm $K_2$. We note in passing that the curvature terms in Proposition \ref{harderror} can be controlled by the energies of the Ricci-coefficients one order higher, simply by the definition of curvature in terms of derivatives of the Ricci-coefficients (cf.~the null structure equations in section \ref{nseq}). Hence:

\begin{proposition} \label{K12summary}
Let $\left(\mathcal{R},g\right)$ be ultimately Schwarzschildean to order $k+1$ ($k>7$) and $\mathcal{X}=p_\mathcal{X} e_3 + q_\mathcal{X} e_4$, $\mathcal{Y}=p_\mathcal{Y} e_3 + q_\mathcal{Y} e_4$, $\mathcal{Z}=p_\mathcal{Z} e_3 + q_\mathcal{Z} e_4$ be vectorfields for bounded spacetime functions $p_i$ and $q_i$ satisfying $|T\left(p_i\right)| + |T\left(q_i\right)|\leq \epsilon \cdot \tau^{-\frac{5}{4}}$. For any positive $\lambda$ and $0\leq n \leq k-1$ we have
\begin{align} \label{harderrorest}
\Big| \int_{\tilde{\mathcal{M}}\left(\tau_1,\tau_2\right)}  K_{2}^{\mathcal{X}\mathcal{Y}\mathcal{Z}}\left[\widehat{\mathcal{L}}^{n+1}_T W\right]\Big| \leq  \left(\lambda +B_{\lambda} \epsilon \left(\tau_1\right)^{-\frac{1}{4}}\right) \cdot \sup_{\tau \in \left(\tau_1, \tau_2\right)} \ \overline{\mathbb{E}}^{n+1} \left[W\right]\left(\tilde{\Sigma}_{\tau}\right)  \nonumber \\
+ B \cdot \mathbb{D}^{n+1} \left[\mathfrak{R}\right] \left(\tau_1,\tau_2\right) + B_{\lambda} \cdot \mathbb{D}^{n} \left[\mathfrak{R}\right] \left(\tau_1,\tau_2\right) \nonumber \\ 
 + \epsilon \left(\tau_1\right)^{-\frac{1}{4}} \cdot \overline{\mathbb{I}}^{n+1,r\geq R} \left[W\right] \left(\tilde{\mathcal{M}}\left(\tau_1,\tau_2\right)\right)+ \lambda  \sup_{\mathcal{H},i} |p_i| \int_{\mathcal{H}} \|\widehat{\mathcal{L}}^{n+1}_T W \|^2 
\end{align}
For $\mathcal{X}=\mathcal{Y}=\mathcal{Z}=T$ the last term vanishes and no highest order curvature-term appears on the horizon.

Finally, if $\tilde{\mathcal{M}}\left(\tau_1, \tau_2\right)$ is replaced by $\mathcal{M}\left(\tau_1, \tau_2\right)$, $\tilde{\Sigma}_\tau$ by $\Sigma_\tau$ and the $\mathbb{D}$-energies by the $\mathbb{C}$-energies, the estimate holds without the penultimate term on the right hand side.
\end{proposition}
\section{The redshift (from the null-Bianchi equations)} \label{redsection}
\subsection{The main result}
The main result of this section is 
\begin{proposition} [Estimates close to the horizon] \label{redshiftcomplete}
We have for $0\leq n\leq k-1$:
\begin{align} \label{mrst}
\overline{\mathbb{E}}^{n+1}_{r \leq r_Y} \left[W\right] \left(\Sigma_{\tau_2}\right) + \overline{\mathbb{E}}^{n+1} \left[W\right] \left( \mathcal{H}\left(\tau_1,\tau_2\right)\right) + \overline{\mathbb{I}}^{n+1, deg}_{r \leq r_Y} \left[W\right] \left(\mathcal{M}\left(\tau_1,\tau_2\right)\right)
\leq  \nonumber \\ 2 \cdot \overline{\mathbb{E}}^{n+1}_{r \leq r_Y} \left[W\right] \left(\Sigma_{\tau_1}\right)
  + B \cdot \overline{\mathbb{I}}^{n+1, deg}_{|r-r_Y| \leq \frac{r_Y-2M}{2}} \left[W\right] \left(\mathcal{M}\left(\tau_1,\tau_2\right)\right) + Err^{n+1}_{hoz}\left[\mathfrak{R}\right] \left(\tau_1,\tau_2\right)
\end{align}
with the error 
\begin{align} \label{hzee}
Err^{n+1}_{hoz}\left[\mathfrak{R}\right] \left(\tau_1,\tau_2\right) =  \epsilon \cdot \overline{\mathbb{I}}^{n+1, deg}_{r \leq r_Y+\frac{M}{2}} \left[\mathfrak{R}\right] \left(\mathcal{M}\left(\tau_1,\tau_2\right)\right) \nonumber \\  + B \cdot \overline{\mathbb{I}}^{n, deg}_{r \leq r_Y+\frac{M}{2}} \left[\mathfrak{R}\right] \left(\mathcal{M}\left(\tau_1,\tau_2\right)\right) + B \cdot \overline{\mathbb{E}}^{n}_{r \leq r_Y+\frac{M}{2}} \left[\mathfrak{R}\right] \left(\Sigma_{\tau_1},\Sigma_{\tau_2}, \mathcal{H} \right)
\end{align}
where $\epsilon$ arises from a pointwise estimate for $\tilde{W}$ and the constants $B$ depend only $M$ and, in (\ref{mrst}), the $\overline{\mathbb{I}}^{max\{2,n\}, deg}_{r \leq r_Y+\frac{M}{2}} \left[\mathfrak{R}\right] \left(\mathcal{M}\left(\tau_1,\tau_2\right)\right)$-energy.
\end{proposition}
\begin{proof}
This will follow from the sequence of estimates proven in the remainder of this section.
\end{proof}
Note that we are estimating $n+1$ derivatives of curvature, requiring only an $\epsilon$ of $n+1$-derivatives of the Ricci rotation coefficients.
\subsection{Estimates for the uncommuted equations}
Recall that we defined
\begin{equation}
\gamma = \xi \left(r\right) \left(1 +\frac{1}{c_{red}}\left(1-\mu\right)\right)
\end{equation}
in (\ref{gammdefo}).
Multiply (\ref{Bianchi3}) by $\gamma \alpha$ and integrate the resulting equation over the spacetime region $\mathcal{Y}=\mathcal{M}\left(\tau_1,\tau_2\right) \cap \{ r \leq r_Y\}$. Integrate the angular $\underline{\beta}$ term by parts (producing $div \underline{\alpha}$) yields upon  inserting (\ref{Bianchi4}) the identity
\begin{align} 
\int_{\mathcal{Y}} \sqrt{g} dt^\star dr d\omega \Bigg( \frac{1}{2} \slashed{D}_4 \left(\gamma \|\underline{\alpha}\|^2\right) 
+ \left[-\frac{1}{2} \slashed{D}_4 \gamma - \gamma\left(4\Omega - \frac{1}{2} tr H\right)\right] \|\underline{\alpha}\|^2   + \nonumber \\  \slashed{D}_3 \left(\gamma \|\underline{\beta}\|^2\right) + \left[-\slashed{D}_3 \gamma + \gamma\left(4\underline{\Omega} + 4 tr \underline{H}\right)\right] \|\underline{\beta}\|^2 \Big)]
=  \int_{\mathcal{Y}} \sqrt{g} dt^\star dr d\omega \, e_{\gamma}\left[\underline{\alpha},\underline{\beta}\right]  \nonumber
\end{align}
with the error 
\begin{equation}
e_{\gamma}\left[\underline{\alpha},\underline{\beta}\right] = \gamma \left[E_4\left(\underline{\alpha}\right) - 4\Omega\underline{\alpha}\right] \underline{\alpha} + 2\gamma \left[E_3\left(\underline{\beta}\right)+2\underline{\Omega}\underline{\beta} \right]\underline{\beta}  \, ,\nonumber
\end{equation}
for which we ignore the entirely harmless cubic term $2 \underline{\alpha} \underline{\beta} \slashed{\nabla} \left(\gamma \frac{\sqrt{g}}{\sqrt{g_{S^2}}} \right)$ arising from the integration by parts (this term can always be estimated as $\epsilon \left(\|\alpha\|^2 + \|\beta\|^2\right)$ using the ultimately Schwarzschildean assumption). Inserting the relations
\begin{align}
\slashed{D}_3 \left(\gamma \underline{\beta}^2\right) &=  D_3 \left(\gamma \underline{\beta}^2\right) = D_a \left(\left(e_3\right)^a \gamma \underline{\beta}^2\right) - \gamma \underline{\beta}^2 \left[-2\underline{\Omega} + tr \underline{H}\right] \nonumber \\
\slashed{D}_4 \left(\gamma \underline{\alpha}^2\right) &=  D_4 \left(\gamma \underline{\alpha}^2\right) = D_a \left(\left(e_4\right)^a \gamma \underline{\alpha}^2\right) - \gamma \underline{\alpha}^2 \left[-2\Omega + tr H\right]
\end{align}
back into the above identity we can write
\begin{align} 
\int_{\mathcal{Y}} \sqrt{g} dt^\star dr d\omega \Bigg( \frac{1}{2} {D}_a \left(\left(e_4\right)^a\gamma \|\underline{\alpha}\|^2\right) 
+ \left[-\frac{1}{2} \slashed{D}_4 \gamma - \gamma\left(3\Omega \right)\right] \|\underline{\alpha}\|^2   + \nonumber \\  {D}_a \left(\left(e_3\right)^a \gamma \|\underline{\beta}\|^2\right) + \left[-\slashed{D}_3 \gamma + \gamma\left(6\underline{\Omega} + 2 tr \underline{H}\right)\right] \|\underline{\beta}\|^2 \Big)]
=  \int_{\mathcal{Y}} \sqrt{g} dt^\star dr d\omega \, e_{\gamma}\left[\underline{\alpha},\underline{\beta}\right]  \nonumber \, .
\end{align}
By Lemma \ref{rscon}, in the region under consideration both
\begin{align}
-\frac{1}{2} \slashed{D}_4 \gamma - 3 \gamma \Omega \geq b
\textrm{ \ \ \ \ and \ \ \ \ }
-\slashed{D}_3 \gamma + \gamma\left(6\underline{\Omega} + 2 tr \underline{H}\right) \geq b
\end{align}
hold and  we obtain the estimate\footnote{We have not written the measure explicitly here as no confusion can arise from weights in $r$ in this region.}
\begin{align}
\int_{\Sigma_{\tau_2} \cap \{r \leq r_Y\}}\left(\|\underline{\alpha}\|^2 + \|\underline{\beta}\|^2\right)  + \int_{\mathcal{Y}}  \left(\|\underline{\alpha}\|^2 + \|\underline{\beta}\|^2\right) +\int_{\mathcal{H}\left(\tau_1,\tau_2\right)}  \|\underline{\beta}\|^2 \nonumber \\ \leq B \int_{\Sigma_{\tau_1} \cap \{r \leq r_Y\}}\left(\|\underline{\alpha}\|^2 + \|\underline{\beta}\|^2\right) +B \int_{\mathcal{M} \cap \{|r-r_Y|\leq \frac{r_Y-2M}{2}\}} \left(\|\underline{\alpha}\|^2 + \|\underline{\beta}\|^2\right) \nonumber \\
+ B \Big| \int_{\mathcal{Y}} e_{\gamma}\left[\underline{\alpha},\underline{\beta}\right] \Big|
\end{align}

The same procedure applied to the Bianchi equations (\ref{Bianchi10}), (\ref{Bianchi8}) and (\ref{Bianchi9}) produces the identity
\begin{align}
\int_{\mathcal{Y}} \sqrt{g} dt^\star dr d\omega\Bigg( \frac{1}{2}  \slashed{D}_4 \left(\gamma \underline{\beta}^2\right) 
+ \left[-\frac{1}{2} \slashed{D}_4 \gamma + \gamma\left( 2tr H - 2\Omega \right)\right] \underline{\beta}^2  \nonumber \\ + \frac{1}{2} \slashed{D}_3 \left(\gamma \hat{\rho}^2 + \gamma \hat{\sigma}^2 \right) + \left[-\frac{1}{2} \slashed{D}_3 \gamma +\frac{3}{2} \gamma tr \underline{H} \hat{\rho} \right]  \left( \hat{\rho}^2 + \hat{\sigma}^2 \right) \Bigg) = \nonumber \\ 
\int_{\mathcal{Y}} \sqrt{g} dt^\star dr d\omega \, e_{\gamma}\left[\underline{\beta},\left(\hat{\rho},\hat{\sigma}\right)\right] 
\end{align}
with the error
\begin{equation}
e_{\gamma}\left[\underline{\beta},\left(\hat{\rho},\hat{\sigma}\right)\right] = \gamma \hat{E}_3\left(\rho\right) \hat{\rho} + \gamma \hat{E}_3\left(\sigma\right) \hat{\sigma} + \gamma \left[E_4 \left(\underline{\beta}\right) -2\Omega \underline{\beta} \right] \underline{\beta}
\end{equation}
(again we ignore a harmless cubic error-term of the form $\underline{\beta} \left(\hat{\rho}, \hat{\sigma}\right) \slashed{\nabla} \left(\gamma \frac{\sqrt{g}}{\sqrt{g_{S^2}}} \right)$, which is readily estimated). Upon integration over $\mathcal{Y}$ this leads to the estimate
\begin{align} \label{che1}
\int_{\Sigma_{\tau_2} \cap \{r \leq r_Y\}}\left(\|\underline{\beta}\|^2 + \|\hat{\rho},\hat{\sigma}\|^2\right)  + \int_{\mathcal{Y}}\left(\|\underline{\beta}\|^2 + \|\hat{\rho},\hat{\sigma}\|^2\right) +\int_{\mathcal{H}\left(\tau_1,\tau_2\right)}  \|\hat{\rho},\hat{\sigma}\|^2 \nonumber \\ \leq B \int_{\Sigma_{\tau_1} \cap \{r \leq r_Y\}}\left(\|\underline{\beta}\|^2 + \|\hat{\rho},\hat{\sigma}\|^2\right) \nonumber \\+B \int_{\mathcal{M} \cap \{|r-r_Y|\leq \frac{r_Y-2M}{2}\}} \left(\|\underline{\beta}\|^2 + \|\hat{\rho},\hat{\sigma}\|^2\right) 
+ B \Big| \int_{\mathcal{Y}} e_{\gamma}\left[\underline{\beta},\left(\hat{\rho},\hat{\sigma}\right)\right] \Big| \, .
\end{align}
where we have taken into account that both inequalities
\begin{align}
-\frac{1}{2} \slashed{D}_4 \gamma + \gamma \left(-\Omega + \frac{3}{2}tr H\right) \geq b
\textrm{ \ \ and \ \ }
-\slashed{D}_3 \gamma + \gamma\left(+\underline{\Omega} + \frac{3}{2} tr \underline{H}\right) \geq b \nonumber 
\end{align}
hold in $r \leq r_Y$

The other two sets of Bianchi equations are a little more subtle. We can not achieve a good sign for both of the spacetime terms at the same time. For instance, using the renormalized Bianchi equation (\ref{renormrho}) and (\ref{Bianchi7}) we derive
\begin{align}
\int_{\mathcal{Y}} \sqrt{g} dt^\star dr d\omega \Bigg( \frac{1}{2} \slashed{D}_4 \left(\gamma \hat{\rho}^2 + \gamma \hat{\sigma}^2 \right) 
+ \left[-\frac{1}{2} \slashed{D}_4 \gamma + \frac{3}{2} \gamma \, tr H \hat{\rho} \right]  \left( \hat{\rho}^2 + \hat{\sigma}^2 \right)  \nonumber \\ + \frac{1}{2}  \slashed{D}_3 \left(\gamma \beta^2\right) + \left[-\frac{1}{2} \slashed{D}_3 \gamma - \gamma\left(2\underline{\Omega} -  tr \underline{H} \right)\right] \beta^2 \Bigg)
= \int_{\mathcal{Y}} \sqrt{g} dt^\star dr d\omega \, e_{\gamma}\left[\left(\hat{\rho},\hat{\sigma}\right),\beta\right]  \nonumber
\end{align}
with the error
\begin{equation}
e_{\gamma}\left[\left(\hat{\rho},\hat{\sigma}\right),\beta\right] =  \gamma \hat{E}_4\left(\rho\right) \hat{\rho} + \gamma \hat{E}_4\left(\sigma\right) \hat{\sigma} + \gamma \left[E_3\left(\beta\right) - 2 \underline{\Omega}\beta\right] \beta
\end{equation}
After integration, the $\beta$ spacetime term will admit a good sign, while the $(\hat{\rho},\hat{\sigma})$-term will have the wrong sign as
\begin{align}
-\frac{1}{2} \slashed{D}_4 \gamma + \gamma \left(tr H + \Omega \right)
\end{align} 
is negative close to the horizon.
However, since $(\hat{\rho},\hat{\sigma})$ is already controlled by the previous estimate, we can add a multiple of the estimate (\ref{che1}) to also control the spacetime term of $\beta$. Similarly, for the final set of equations:
\begin{align}
\int_{\mathcal{Y}} \sqrt{g} dt^\star dr d\omega \Bigg( \frac{1}{4} \slashed{D}_3 \left(\gamma \alpha^2\right) 
+ \left[-\frac{1}{4} \slashed{D}_3 \gamma + \gamma\left( \frac{1}{4} tr \underline{H} \alpha - 2 \underline{\Omega} \alpha \right)\right]\alpha^2  \nonumber \\  \frac{1}{2} \slashed{D}_4 \left(\gamma \beta^2\right) + \left[-\frac{1}{2}\slashed{D}_4 \gamma + \gamma\left(2tr H + 2\Omega\right)\right] \beta^2 \Bigg)
= \int_{\mathcal{Y}} \sqrt{g} dt^\star dr d\omega \, e_{\gamma}\left[{\beta},\alpha\right] 
\end{align}
with error
\begin{equation}
e_{\gamma}\left[{\beta},\alpha\right] = \frac{1}{2} \gamma \left[E_3\left(\alpha\right) - 4\underline{\Omega} \alpha\right] + \gamma \left[E_4\left(\beta\right) + 2\Omega \beta \right] \beta
\end{equation}
Hence while the $\alpha$-spacetime term is positive, the $\beta$-spacetime term has the wrong sign (which will only get strenghtened after the contribution from the integration by parts). However, the $\beta$-term has been controlled by the previous step and we can add a multiple of the latter to control this term.
Adding the estimates up we can summarize our findings in the following
\begin{proposition} \label{redshift1}
We have the estimate
\begin{align}
\int_{\Sigma_{\tau_2} \cap \{r \leq r_Y\}} \|\tilde{W}\|^2 + \int_{\mathcal{Y}} \|\tilde{W}\|^2 \leq B \int_{\Sigma_{\tau_1} \cap \{r \leq r_Y\}} \|\tilde{W}\|^2 \nonumber \\+B \int_{\mathcal{M} \cap \{|r-r_Y|\leq \frac{r_Y-2M}{2}\}} \|\tilde{W}\|^2 
 + B \Big| \int_{\mathcal{Y}} e_{\gamma}\left[\underline{\alpha},\underline{\beta}\right] \Big| \nonumber \\+ B \Big| \int_{\mathcal{Y}} e_{\gamma}\left[\underline{\beta},\left(\hat{\rho},\hat{\sigma}\right)\right] \Big|  + B \Big| \int_{\mathcal{Y}} e_{\gamma}\left[\left(\hat{\rho},\hat{\sigma}\right), \beta\right] \Big| + B \Big| \int_{\mathcal{Y}} e_{\gamma}\left[\beta, \alpha\right] \Big|   \, .
\end{align}
\end{proposition}
For the error-terms we have
\begin{lemma}
\begin{align}
 \Big| \int_{\mathcal{Y}} e_{\gamma}\left[\underline{\alpha},\underline{\beta}\right] \Big|  + \Big| \int_{\mathcal{Y}} e_{\gamma}\left[\underline{\beta},\left(\hat{\rho},\hat{\sigma}\right)\right] \Big|  + \Big| \int_{\mathcal{Y}} e_{\gamma}\left[\left(\hat{\rho},\hat{\sigma}\right), \beta\right] \Big| + \Big| \int_{\mathcal{Y}} e_{\gamma}\left[\beta, \alpha\right] \Big| \nonumber \\
 \leq \|\left(\mathfrak{R}-\mathfrak{R}_{SS}\right)^2 \|_{L^\infty} \cdot \int_\mathcal{Y} \|\tilde{W}\|^2 + B \|\rho^2\|_{L^\infty} \int_\mathcal{Y}\|\mathfrak{R}-\mathfrak{R}_{SS}\|^2 \, .
\end{align}
\end{lemma}
\begin{proof}
This follows from inspecting the error-terms individually. Most terms are of the form ``decaying Ricci rotation coefficient" $\cdot$ ``decaying curvature component" and are estimated by the first term on the right hand side of the lemma. The only other term arises when the curvature component is $\rho$. This is accounted for by the last term.
\end{proof}
\subsection{The higher derivative redshift estimate}
The estimate of Proposition \ref{redshift1} can be derived for all higher derivatives as well, this being essentially the statement of Proposition \ref{redshiftcomplete}. We will now prove the latter.

\subsubsection*{Estimating $\underline{\alpha}_{3n_l}$ and $\underline{\beta}_{3n_l}$}
From the equations for $\underline{\beta}_{33n_l}$ and $\underline{\alpha}_{43n_l}$, derived from (\ref{bb3nk}) and (\ref{ab4nk}) we obtain, modulo terms which will vanish after integration and a cubic term involving the angular derivative of the measure
\begin{align} \label{hdrs1}
\frac{1}{2} \slashed{D}_4 \| \gamma \underline{\alpha}_{3n_l}\|^2 + \slashed{D}_3  \|\gamma \underline{\beta}_{3n_l}\|^2
+ 2 \|\underline{\beta}_{3n_l}\|^2  \Big[-\frac{ \slashed{D}_3}{2} \gamma + \vartheta^-\left(\underline{\beta}_{3n_l}\right) tr \left(\underline{H}\right) r\Big] \nonumber \\ 
+ \|\underline{\alpha}_{3n_l}\|^2 \Big[-\frac{1}{2}\slashed{D}_4 \gamma + \vartheta^+\left(\underline{\alpha}_{3n_l}\right) \gamma tr H   \Big]  = e_\gamma \left[\underline{\alpha}_{3n_l}, \underline{\beta}_{3n_l}\right]  
\end{align}
where
\begin{align}
e_\gamma\left[\underline{\alpha}_{3n_l}, \underline{\beta}_{3n_l}\right] =  \gamma \hat{E}_4^{3n_l}\left(\underline{\alpha}\right)\underline{\alpha}_{3n_l} + 2\gamma E_3^{3n_l}\left(\underline{\beta}\right) \underline{\beta}_{3n_l}  -2\gamma \mathcal{C}_{33}\left[\underline{\beta}_{n_l}\right]  \underline{\beta}_{3n_l}\nonumber 
\end{align}
For the latter expression we prove
\begin{lemma} \label{lemrs1}
For any $\lambda>0$,
\begin{align}
\int_{\mathcal{Y}} e_\gamma\left[\underline{\alpha}_{3n_l}, \underline{\beta}_{3n_l}\right] \leq \int_{\mathcal{Y}} \left[6+2\left(\textrm{number of $3$'s in $n_l$}\right)\right]\gamma \, \Omega \| \underline{\alpha}_{3n_l}\|^2 \nonumber \\ + \int_{\mathcal{Y}} \left[-2\left(\textrm{number of $4$'s in $n_l$}\right)\right]\gamma \, \Omega \| \underline{\beta}_{3n_l}\|^2 +
 \lambda \int_{\mathcal{Y}} \left(\|\underline{\alpha}_{3n_l} \|^2 +\|\underline{\beta}_{3n_l} \|^2 \right) \nonumber \\ + \lambda \int_{\mathcal{H}, \Sigma_{\tau_1}, \Sigma_{\tau_2}}  \|\underline{\beta}_{3n_l} \|^2
+B \|\left(\mathfrak{R}-\mathfrak{R}_{SS}\right)^2\|_{L^\infty} \cdot \overline{\mathbb{I}}^{l+1, deg}_{r \leq r_Y} \left[W\right] \left(\mathcal{M}\left(\tau_1,\tau_2\right)\right) \nonumber  \\
 + \lambda \cdot \overline{\mathbb{I}}^{l+1, deg}_{|r - r_Y| \leq \frac{r_Y-2M}{r}} \left[W\right] \left(\mathcal{M}\left(\tau_1,\tau_2\right)\right) \nonumber \nonumber \\ 
+ B \| \tilde{W} \|^2_{L^\infty} \cdot \overline{\mathbb{I}}^{l+1, deg}_{r \leq r_Y+\frac{M}{2}} \left[\mathfrak{R}\right] \left(\mathcal{M}\left(\tau_1,\tau_2\right)\right) 
\nonumber \\
+ B_\lambda \cdot \left[\overline{\mathbb{I}}^{l, deg}_{r \leq r_Y+\frac{M}{2}} \left[\mathfrak{R}, W\right] \left(\mathcal{M}\left(\tau_1,\tau_2\right)\right) + \overline{\mathbb{E}}^{l, deg}_{r \leq r_Y} \left[\mathfrak{R}, W\right] \left(\Sigma_{\tau_1}, \Sigma_{\tau_2}, \mathcal{H}\right)\right] \, .\nonumber 
\end{align}
\end{lemma}
\begin{remark}
The first term on the right hand side is the redshift term, which has a good sign when brought to the left hand side.\footnote{As is familiar from the wave equation, its strength improves, the more $3$-derivatives are taken.} The contribution for $\underline{\beta}_{3n_l}$ has the wrong sign. However, since we have the large term $-\slashed{D}_3 \gamma$ already available on the left hand side of (\ref{hdrs1}), we will absorb the negative contribution. The terms in lines two and three will eventually be absorbed by the left hand side (once we added the estimates for all the components). The remaining terms involving Ricci-coefficients are contained in the expression for (\ref{hzee}).
The remaining lower order curvature terms will eventually be controlled by reiterating the redshift estimate, while the term in the fourth line is present in (\ref{mrst}). \end{remark}
\begin{proof}
Recall that $\hat{E}_4^{3n_l}\left(\underline{\alpha}\right)=E_4^{3n_l}\left(\underline{\alpha}\right) -\mathcal{C}_{34}\left[\underline{\alpha}_{n_l}\right] - \left[\underline{\alpha}_{43} -\underline{\alpha}_{34}\right]_{n_l}$.
Starting with the commutation term
\begin{equation}
-\mathcal{C}_{34}\left[\underline{\alpha}_{n_l}\right]  \underline{\alpha}_{3n_l} = -\left(\underline{\alpha}_{34n_l} - \underline{\alpha}_{3n_l4}\right)\underline{\alpha}_{3n_l} \, \nonumber
\end{equation}
we see that we need to push through the $4$-derivative. Every time it hits a $3$, Lemma \ref{commutelemma} applies and introduces a lower oder term. Since terms which contain $l$ derivatives are lower order (and accounted for by the term in the penultimate line of Lemma \ref{lemrs1}), we only have to focus on the highest order term. Furthermore, as both $\underline{\Omega}$ and $Z, \underline{Z}$ decay, we only need to follow the term which is proportional to $\Omega$ in the expression $F_{34A_1 ... A_k}$ of Lemma \ref{commutelemma}. It is easy to see that this term will be 
\begin{equation}
-\mathcal{C}_{34}\left[\underline{\alpha}_{n_l}\right]  \underline{\alpha}_{3n_l} = -\left(\alpha_{34n_l} - \alpha_{3n_l4}\right)\underline{\alpha}_{3n_l} \sim \left(\textrm{number of $3$'s in $n_l$}\right) 2\Omega \underline{\alpha}_{3n_l}  \underline{\alpha}_{3n_l} \nonumber
\end{equation}
The $C_{33}\left[\underline{\beta}\right]$-term is dealt with analogously and gives a negative contribution.

Inspecting the term $E_3^{3n_l}\left(\underline{\beta}\right)$ and how it arises inductively from derivatives of $E_3 \left(\underline{\beta}\right)$ we observe that it consists of the following type of terms (note that $\underline{\Omega}$ decays)
\begin{itemize}
\item products of curvature components with Ricci coefficients with both components decaying
\item derivatives falling on the expression $-3\underline{Y} \rho$ in $E_3\left(\underline{\beta}\right)$, which are problematic in view of the fact that $\rho$ does not decay. 
\end{itemize}
The terms of the first type are unproblematic if sufficiently many derivatives are taken and are easily seen to be estimated by the terms of the right hand side in the lemma. For the terms of the second type we make the following observation: The worst term is clearly the one where all derivatives fall on $\underline{Y}$. This is because if a derivative falls on $\rho$ it can be replaced by $\rho$ and terms which decay using the Bianchi equation, hence gaining one derivative. It follows that except for the term where all derivatives fall on $\underline{Y}$, we can estimate $\rho$ pointwise and $k$-derivatives of $\underline{Y}$ using the spacetime energy for the Ricci coefficients appearing on the right hand side of the lemma. We still need to estimate the term
\begin{align}
\int_{\mathcal{Y}} \gamma \, \rho \left[ \slashed{D}^{n_l} \slashed{D}_3 \underline{Y} \right] \underline{\beta}_{3n_l}
\end{align}
For this we will be able to exploit a cancellation with the other $\underline{\alpha}$-terms. We first note from equation (\ref{a43}) and Lemma \ref{commutelemma} that
\begin{align} \label{hcm}
- \left[\underline{\alpha}_{43} -\underline{\alpha}_{34}\right]_{n_l} = + 2\Omega \underline{\alpha}_{3n_l} - 2\underline{\Omega}\underline{\alpha}_{4n_l} + \textrm{terms prop.~to $\underline{\alpha}_{n_l}$}
\end{align}
Furthermore, inspecting how $E_4^{3n_l}\left(\underline{\alpha}\right)$ arises inductively from derivatives of $E_4 \left(\underline{\alpha}\right)$, we observe that it consists of terms of the following type 
\begin{enumerate}
\item products of curvature components with Ricci coefficients with both components decaying
\item derivatives falling on the expression $4 \Omega \underline{\alpha}$ in $E_4\left(\underline{\alpha}\right)$, which in view of the fact that $\Omega$ itself does not decay are problematic. 
\item derivatives falling on the expression $-3\widehat{\underline{H}} \rho$ in $E_4\left(\underline{\alpha}\right)$, which in view of the fact that $\rho$ does not decay are problematic. 
\end{enumerate}
Again, the terms of the first type are easily estimated. For the {\bf terms of the second type} we observe

\emph{Derivatives all falling on the same factor.} If all derivatives fall on $\alpha$ we find a contribution of $4\Omega \alpha_{3n_l}$, which adds to the $2\Omega \underline{\alpha}_{3n_l}$ in (\ref{hcm}) to produce the factor of $6$ in the Lemma.\footnote{In rewriting these terms like this one picks up lower order terms, which involve $\Omega$ multiplying $l$-derivatives of $\alpha$. These are accounted for in the lower order energies on the right hand side of the Lemma.}  If all $l+1$ derivatives fall on $\Omega$ we can (since the $l+1$-derivatives have at least one $3$ derivative in it) repeatedly use the commutation formula
\begin{equation}
\slashed{D}_3 \slashed{D}_4 \Omega = \slashed{D}_4 \slashed{D}_3 \Omega - 2\Omega \slashed{D}_3 \Omega + 2\underline{\Omega} \slashed{D}_4 \Omega + \left(Z_B - \underline{Z}_B\right) \slashed{\nabla}_B \Omega
\end{equation}
and the structure equation (\ref{O43}) to create terms of the first type and a term where $l$ derivatives fall on $\rho$ (recall that $\underline{\Omega}$ decays). For the latter we can again use the Bianchi equation for $\rho_3$ or $\rho_4$ to gain a derivative or transform it into a cubic term.

\emph{Derivatives splitting up}. If the derivatives split, we can estimate one of the terms pointwise and the other in $L^2$. Writing $\Omega = \Omega + \frac{M}{r^2} - \frac{M}{r^2}$ one sees that
\begin{align}
\int_{\mathcal{Y}} \sum_{s=1}^{l} \left[\slashed{D}^{s} \Omega \cdot \slashed{D}^{{l-s+1}} \underline{\alpha} \right] \underline{\alpha}_{3n_l} \nonumber \\
\leq  \lambda \int_{\mathcal{Y}} \|\underline{\alpha}_{3n_l}\|^2 + \frac{B}{\lambda} \Big[\overline{\mathbb{I}}^{l,deg}\left[W\right] \left(\mathcal{Y}\right)+  \sup_{\tau} \overline{\mathbb{E}}^{l}_{r\leq r_Y} \left[W\right] \left({\Sigma}_{\tau}\right) \Big] \nonumber \, ,
\end{align}
where $B$ also depends on the $\mathbb{D}^l \left[\mathfrak{R}\right]$-energy which is bounded by the ultimately Schwarzschildean assumption.

For the {\bf terms of the third type} we argue similarly as for the $\underline{\beta}$-term, which leaves us to control the term where all derivatives fall on $\widehat{\underline{H}}$. Combining it with the $\underline{\beta}$ term we need to estimate
\begin{align}
I = \int_{\mathcal{Y}} \gamma \, \rho \Bigg( \left[ 2  \slashed{D}_3 \slashed{D}^{n_l} \underline{Y} \right] \underline{\beta}_{3n_l} + \left[\slashed{D}^{n_l} \slashed{D}_3 \widehat{\underline{H}}\right] \underline{\alpha}_{3n_l} \Bigg)
\end{align}
Note that we can choose the ordering of derivatives as commuting only introduces lower order terms which have already been estimated. The idea is to integrate the first term by parts moving a three derivative on $\underline{\beta}_{3n_l}$. Modulo a lower order term we will obtain $\underline{\beta}_{3n_l3}$ for which the Bianchi equation yields $-\slashed{div} \alpha_{3n_l}$. Integrating again by parts to move the angular derivative back to the $\underline{Y}$ term and commuting it through, we obtain cancellation of the highest order terms after plugging in the structure equation (\ref{Hb3}). This leaves, ignoring the boundary terms for the moment the highest order term
\begin{equation}
\int_{\mathcal{Y}} \gamma \, \rho \underline{\alpha}_{n_l} \underline{\alpha}_{3n_l}
\end{equation}
which after another integration by parts is seen to be of lower order. It is easy to check that in this process only boundary and lower order terms appear, which are collected on the right hand side of the Lemma.
\end{proof}
With the lemma at hand we can integrate the identity (\ref{hdrs1}) in $\mathcal{Y}$ and obtain favorable spacetime- and boundary terms if the first term on the right hand side of the Lemma \ref{lemrs1} is taken into account. 

\subsubsection*{Estimating $\underline{\beta}_{3n_l}$ and $\left(\rho,\sigma\right)_{3n_l}$}
Using the equations for $\rho_{33n_l}$ and $\sigma_{33n_l}$ (derived from  (\ref{r3nk}) and (\ref{s3nk})) and for $\underline{\beta}_{34n_l}$ (derived from (\ref{bb3nk})) yields, modulo a total divergence on $S^2_{t,r}$)  
\begin{align} \label{hdrs2}
\frac{1}{2} \slashed{D}_3 \left[\gamma \left( \rho_{3n_l}^2 +\sigma_{3n_l}^2\right)\right] + \frac{1}{2} \slashed{D}_4 \left[\gamma \|\underline{\beta}_{3n_l}\|^2 \right] \nonumber \\ 
+  \left( \rho_{3n_l}^2 +\sigma_{3n_l}^2\right)  \Big[-\frac{1}{2}\slashed{D}_3 \gamma + \vartheta^-\left(\rho_{3n_l}\right) \gamma \, tr \underline{H}\Big] \nonumber \\     
+ \left[-\frac{1}{2}\slashed{D}_4 \gamma + 2\vartheta^+\left(\underline{\beta}_{3n_l}\right) \gamma \,  tr H\right] \|\underline{\beta}_{3n_l}\|^2 = e_\gamma \left[\left(\rho,\sigma\right)_{3n_l}, \underline{\beta}_{3n_l}\right]  \end{align}
\begin{align}
e_\gamma \left[\left(\rho,\sigma\right)_{3n_l}, \underline{\beta}_{3n_l}\right] = E_3^{3n_l} \left(\rho\right) \cdot \rho_{3n_l} + E_3^{3n_l} \left(\sigma\right) \cdot \sigma_{3n_l} \nonumber \\ \underline{\beta}_{3n_l} \left(E_4^{3n_l}\left(\underline{\beta}\right) + \Big[\underline{\beta}\left[2\slashed{D}_4 tr \underline{H} - \slashed{D}_3 tr H \right] + F_{43}\left(\underline{\beta}\right)\Big]_{n_l}\right) \nonumber \\ -\gamma \mathcal{C}_{34}\left[\underline{\beta}_{n_l}\right]  \underline{\beta}_{3n_l} -\gamma \mathcal{C}_{33}\left[\left(\rho,\sigma\right)_{n_l}\right]  \left(\rho,\sigma\right)_{3n_l}
\end{align}
We have the analogue of Lemma \ref{lemrs1}:
\begin{lemma} \label{lemrs2}
\begin{align}
\int_{\mathcal{Y}} e_\gamma \left[\left(\rho,\sigma\right)_{3n_l}, \underline{\beta}_{3n_l}\right] \leq \int_{\mathcal{Y}} \gamma \Omega \left[4 + 2 \left(\textrm{number of $3$'s in $n_l$} \right)\right]  \| \underline{\beta}_{3n_l}\|^2 \nonumber \\ -2 \int_{\mathcal{Y}} \gamma \Omega \left(\textrm{number of $4$'s in $n_l$} \right) \| \left(\rho,\sigma\right)_{3n_l}\|^2
 + \lambda \int_{\mathcal{Y}}  \left( \|\underline{\beta}_{3n_l} \|^2 + \|\left(\rho,\sigma\right)_{3n_l} \|^2  \right) \nonumber \\ + \lambda \int_{\mathcal{H}, \Sigma_{\tau_1}, \Sigma_{\tau_2}} \left( \|\underline{\beta}_{3n_l} \|^2 + \|\left(\rho,\sigma\right)_{3n_l} \|^2 \right) \nonumber \\ + \textrm{error-terms as in Lemma \ref{lemrs1}} \nonumber
\end{align}
\end{lemma}
\begin{proof}
The proof is completely analogous to the previous one. Hence we will only reveal the structure regarding the term which is proportional to $\rho$. The most difficult term is 
\begin{align}
I = \int_{\mathcal{Y}} \rho \gamma \left[ -3 \slashed{D}^{n_l} \slashed{D}_3 \underline{Z} + \frac{3}{2}  \slashed{\nabla} \slashed{D}^{n_l} tr \underline{H} \right] \underline{\beta}_{3n_l}
\end{align}
First use the structure equation (\ref{Yb4}), which modulo terms of lower order produces a familiar $\underline{\beta}_{n_l}\underline{\beta}_{3n_l}$-term which can be integrated by parts as in the previous lemma and a term which is like $- \slashed{D}^{n_l} \slashed{D}_4 \underline{Y} \underline{\beta}_{3n_l}$. Moving the $4$-derivative onto the $\underline{\beta}_{3n_l}$-term, inserting the Bianchi equations $\underline{\beta}_{34n_l} = - \slashed{\nabla} \rho_{3n_l} - {}^\star \slashed{\nabla} \sigma_{3n_l}+ \textrm{l.o.t.}$ and moving the angular derive back onto the $Y$-term yields a term $+3\slashed{D}^{n_l} \slashed{div} Y \cdot \rho_{3n_l}$ (note that the $\sigma_{3n_l}$-term is of lower order in view of the structure equation for $\slashed{curl}Y$). For the $\slashed{\nabla} \slashed{D}^{n_l} tr \underline{H} \underline{\beta}_{3n_l}$-term we similarly move the angular derivative onto $\underline{\beta}_{3n_l}$, insert the Bianchi equation $\rho_{3n_l3} = - \slashed{div} \underline{\beta}_{3n_l} + \textrm{l.o.t.}$ and move one derivative from $\rho_{3n_l3}$ back to the $tr \underline{H}$-term. This leaves the highest order term $-\frac{3}{2} \slashed{D}^{n_l} \slashed{D}_3 tr \underline{H} \rho_{3n_l}$. The sum of the two highest order terms is however of lower order by virtue of the structure equation (\ref{trHb3}).
Again due to the good main-term in the first line of Lemma \ref{lemrs2} we obtain non-negative boundary and spacetime terms upon integration of (\ref{hdrs2}) in $\mathcal{Y}$.
\end{proof}

\subsubsection*{Estimating $\beta_{3n_l}$}
Using (\ref{b3nk}) in conjunction with (\ref{r4nk}) and (\ref{s4nk}) we obtain
\begin{align} \label{hdrs3}
\frac{1}{2} \slashed{D}_3  \left[\gamma \|\beta_{3n_l}\|^2 \right]  + \frac{1}{2} \slashed{D}_4  \left[\gamma \left( \rho_{3n_l}^2 +\sigma_{3n_l}^2\right)\right] \nonumber \\ 
+   \Big[-\frac{1}{2}\slashed{D}_3 \gamma + \vartheta^-\left(\beta_{3n_l}\right) \gamma \, tr \underline{H}\Big] \|\beta_{3n_l}\|^2 \nonumber \\     
+ \left[-\frac{1}{2}\slashed{D}_4 \gamma + 2\vartheta^+\left(\rho_{3n_l}\right) \gamma \,  tr H\right]\left( \rho_{3n_l}^2 +\sigma_{3n_l}^2\right)  = e_\gamma \left[\beta_{3n_l},\left(\rho,\sigma\right)_{3n_l}\right]  \, ,
\end{align}
\begin{align}
e_\gamma \left[\beta_{3n_l},\left(\rho,\sigma\right)_{3n_l}\right] = -\gamma \mathcal{C}_{34}\left[\left(\rho,\sigma\right)_{n_l}\right]  \left(\rho,\sigma\right)_{3n_l} -\gamma \mathcal{C}_{33}\left[\beta_{n_l}\right]  \beta_{3n_l} \nonumber \\ 
+E_3^{3n_l} \left(\beta\right) \cdot \beta_{3n_l} + \underline{\rho}_{3n_l} \left(E_4^{3n_l}\left(\rho\right) + \Big[\rho \left[\frac{3}{2} \slashed{D}_4 tr \underline{H} - \frac{3}{2}\slashed{D}_3 tr H \right] + F_{43}\left(\rho\right)\Big]_{n_l}\right) \nonumber \\
+\underline{\sigma}_{3n_l} \left(E_4^{3n_l}\left(\sigma\right) + \Big[\sigma \left[\frac{3}{2} \slashed{D}_4 tr \underline{H} - \frac{3}{2}\slashed{D}_3 tr H \right] + F_{43}\left(\sigma\right)\Big]_{n_l}\right) \nonumber \, .
\end{align}
As before we have
\begin{lemma} \label{lemrs3}
\begin{align}
e_\gamma \left[\beta_{3n_l},\left(\rho,\sigma\right)_{3n_l}\right] \leq \int_{\mathcal{Y}} \gamma \Omega \left[2+2 \left(\textrm{number of $3$'s in $n_l$} \right)\right]  \| \left(\rho,\sigma\right)_{3n_l}\|^2 \nonumber \\ -2 \int_{\mathcal{Y}} \gamma \Omega \left(\textrm{number of $4$'s in $n_l$} \right) \| \beta_{3n_l}\|^2
 + \lambda \int_{\mathcal{Y}}  \left( \|{\beta}_{3n_l} \|^2 + \|\left(\rho,\sigma\right)_{3n_l} \|^2  \right) \nonumber \\ + \lambda \int_{\mathcal{H}, \Sigma_{\tau_1}, \Sigma_{\tau_2}} \left( \|{\beta}_{3n_l} \|^2 + \|\left(\rho,\sigma\right)_{3n_l} \|^2 \right) \nonumber \\ + \textrm{error-terms as in Lemma \ref{lemrs1}} \nonumber \, .
\end{align}
\end{lemma}
\begin{proof}
We only present the control of the $\rho$-integral. The crucial term is, after using the structure equations to establish that
\begin{align}
\rho \left[\frac{3}{2} \slashed{D}_4 tr \underline{H} - \frac{3}{2}\slashed{D}_3 tr H \right] + F_{43}\left(\rho\right) = \nonumber \\ \rho\left[3 \slashed{div} \left(\underline{Z}-Z\right) + 2\Omega \rho_3 -2\underline{\Omega}\rho_4\right] + \textrm{quadr.~dec.~terms} \, ,
\end{align}
\begin{align}
I = \int_{\mathcal{Y}} \gamma \rho \Big[ 3 \slashed{D}^{n_l} \slashed{div} \left(\underline{Z}-Z\right) \rho_{3n_l} + \left(-\frac{3}{2} \slashed{D}^{n_l}\slashed{\nabla} tr \underline{H} + 3 \slashed{D}^{n_l} \slashed{D}_3 Z \right) \beta_{3n_l} \Big] \, .
\end{align}
Starting with the last term we move the $3$-derivative onto $\beta_{3n_l}$ to obtain $\slashed{\nabla}\rho_{3n_l} +{}^\star \sigma_{3n_l} + \textrm{l.o.t}$ using Bianchi. Moving the $\slashed{\nabla}$ derivative back to the left, the $\sigma$-term is seen to be of lower order in view of the structure equation for $\slashed{curl} Z$. The $\rho_{3n_l}$ term which arises cancels to highest order with the $Z$-part of the first term. Hence we need to establish a cancellation between the $\underline{Z}$-part of the first term and the third term. This is achieved by moving the $\slashed{\nabla}$ derivative onto $\beta_{3n_l}$, using the Bianchi equation $\rho_{43n_l}= \slashed{div} \beta_{3n_l} + \textrm{l.o.t}$, and moving back the $4$-derivative onto $tr \underline{H}$. The structure equation (\ref{trHb4}) finally establishes cancellation to the highest order.
\end{proof}

\subsubsection*{Estimating $\alpha_{3n_l}$}
Using (\ref{a3nk}) and (\ref{b4nk})
\begin{align} \label{hdrs4}
\frac{1}{2} \slashed{D}_3  \left[\gamma \|\alpha_{3n_l}\|^2 \right]  +  \slashed{D}_4  \left[\gamma \|\beta_{3n_l}\|^2\right] \nonumber \\ 
+   \Big[-\frac{1}{2}\slashed{D}_3 \gamma + \vartheta^-\left(\alpha_{3n_l}\right) \gamma \, tr \underline{H}\Big] \|\alpha_{3n_l}\|^2 \nonumber \\     
+ 2\left[-\frac{1}{2}\slashed{D}_4 \gamma + 2\vartheta^+\left(\beta_{3n_l}\right) \gamma \,  tr H\right]\|\beta_{3n_l}^2\|  = e_\gamma \left[\alpha_{3n_l},\beta_{3n_l}\right]  
\end{align}
\begin{align}
e_\gamma \left[\alpha_{3n_l},\beta_{3n_l}\right] = -2\gamma \mathcal{C}_{34}\left[\beta_{n_l}\right]  \beta_{3n_l} -\gamma \mathcal{C}_{33}\left[\alpha_{n_l}\right]  \alpha_{3n_l}
+\gamma E_3^{3n_l} \left(\alpha\right) \cdot \alpha_{3n_l}  \nonumber \\ + 2\gamma \beta_{3n_l} \left(E_4^{3n_l}\left(\beta\right) + \Big[\beta \left[\slashed{D}_4 tr \underline{H} - 2\slashed{D}_3 tr H \right] + F_{43}\left(\beta\right)\Big]_{n_l}\right) \nonumber 
\end{align}
As before we have
\begin{lemma} 
\begin{align}
e_\gamma \left[\alpha_{3n_l},\beta_{3n_l}\right] \leq \int_{\mathcal{Y}} 2 \gamma \Omega \left[-2+2+2 \left(\textrm{number of $3$'s in $n_l$} \right)\right]  \| \beta_{3n_l}\|^2 \nonumber \\ -2 \int_{\mathcal{Y}} \gamma \Omega \left(\textrm{number of $4$'s in $n_l$} \right) \| \alpha_{3n_l}\|^2
 + \lambda \int_{\mathcal{Y}}  \left( \|{\alpha}_{3n_l} \|^2 + \|\beta_{3n_l} \|^2  \right) \nonumber \\ + \lambda \int_{\mathcal{H}, \Sigma_{\tau_1}, \Sigma_{\tau_2}} \left( \|{\alpha}_{3n_l} \|^2 + \|\beta_{3n_l} \|^2\right) \nonumber \\ + \textrm{error-terms as in Lemma \ref{lemrs1}} \nonumber
\end{align}
\end{lemma}
Note that the tern $2\gamma \beta_{3n_l} \left(E_4^{3n_l}\left(\beta\right)\right)$ now makes a negative contribution because of the minus sign in $E_4\left(\beta\right) = - 2\Omega \beta + ...$, while both $2\gamma \beta_{3n_l} \slashed{D}^{n_l} F_{43}\left(\beta\right)$ and  $-2\gamma \mathcal{C}_{34}\left[\beta_{n_l}\right]  \beta_{3n_l}$ still contribute a positive $\Omega$-term.
The critical $\rho$-term
\begin{align}
\int_{\mathcal{Y}} \rho \, \gamma \Big[2 \cdot 3 \slashed{D}^{n_l} \slashed{D}_3 Y \cdot \beta_{3n_l} -3 \slashed{D}^{n_l} \slashed{D}_3 \widehat{H} \alpha_{3n_l} \Big]
\end{align}
is seen to be of lower order by integrating by first using the structure equation
\begin{equation}
\slashed{D}_3 Y - \slashed{D}_4 Z = \beta + \textrm{l.o.t.} \, ,
\end{equation}
and then integrating by parts twice (moving $\slashed{D}_4$ onto $\beta_{3n_l}$ and returning the angular derivative arising from the Bianchi equations back to $Z$) and finally using the structure equation
\begin{equation}
\slashed{D}_3 \widehat{H} = - 2\slashed{\mathcal{D}}^\star_2 Z + \textrm{l.o.t.} 
\end{equation}

The previous for redshift estimates will allow us to control all $3n_l$-derivatives. However, to control all derivatives via the Bianchi equation we need control over at least two $4n_l$-derivatives as well:

\subsubsection*{Estimating $\beta_{4n_l}$}
\begin{align} \label{hdrs5}
\frac{1}{2} \slashed{D}_3  \left[\gamma \|\beta_{4n_l}\|^2 \right]  +  \frac{1}{4} \slashed{D}_4  \left[\gamma \|\alpha_{3n_l}\|^2\right] 
+   \Big[-\frac{1}{2}\slashed{D}_3 \gamma + \vartheta^-\left(\beta_{4n_l}\right) \gamma \, tr \underline{H}\Big] \|\beta_{4n_l}\|^2 \nonumber \\     
+ \frac{1}{2} \left[-\frac{1}{2}\slashed{D}_4 \gamma + 2\vartheta^+\left(\alpha_{3n_l}\right) \gamma \,  tr H\right]\|\alpha_{3n_l}\|^2  = e_\gamma \left[\beta_{4n_l},\alpha_{3n_l}\right]  
\end{align}
\begin{align}
e_\gamma \left[\beta_{4n_l},\alpha_{3n_l}\right]  = -\gamma \mathcal{C}_{43}\left[\beta_{n_l}\right]  \beta_{4n_l} -\frac{1}{2}\gamma \mathcal{C}_{34}\left[\alpha_{n_l}\right]  \alpha_{3n_l}
\nonumber \\ + \frac{1}{2}\gamma E_3^{4n_l} \left(\alpha\right) \cdot \alpha_{3n_l} + \gamma \beta_{4n_l} \cdot E_4^{3n_l}\left(\beta\right)  \nonumber 
\end{align}
For the obligatory Lemma regarding the error-term we can be more naive, since we already have a large (proportional to $\frac{1}{c_{red}}$) spacetime term containing all three-derivatives available:
\begin{lemma} 
\begin{align} \label{esf}
e_\gamma \left[\beta_{4n_l},\alpha_{3n_l}\right]  \leq \left(2+2l\right) \|2\gamma \Omega\|_{L^\infty} \int_{\mathcal{Y}} \left(\|\beta_{3n_l}\|^2 + \|\beta_{4n_l}\|^2+ \| \alpha_{3n_l}\|^2\right) \nonumber \\
 + \lambda \int_{\mathcal{H}, \Sigma_{\tau_1}, \Sigma_{\tau_2}} \left( \|{\alpha}_{3n_l} \|^2 + \|\beta_{4n_l} \|^2\right) + \textrm{error-terms as in Lemma \ref{lemrs1}} \
\end{align}
\end{lemma}
\begin{proof}
Everything goes through as previously, except that the terms
\begin{align}
-\gamma \mathcal{C}_{43}\left[\beta_{n_l}\right]  \beta_{4n_l} + \gamma \beta_{4n_l} \cdot E_4^{3n_l}\left(\beta\right)
\end{align}
will introduce a highest order mixed term of the form 
\begin{align}
\gamma\left[-2 -2 \cdot \left(\textrm{number of $4$'s in $n_l$}\right)\right]\Omega \beta_{n_l3} \beta_{n_l4}
\end{align}
For this term we use Cauchy's inequality.  We leave the considerations for the term proportional to $\rho$ to the reader.
\end{proof}
\begin{remark}
The highest order terms on the right hand side of (\ref{esf}) will be absorbed by terms on the left once we add the estimates for all the quantities, since we control these derivatives already with a largeness factor of $\frac{1}{c_{red}}$ from the previous steps.
\end{remark}

\subsubsection*{Estimating $\alpha_{4n_l}$}
\begin{align} \label{hdrs6}
\frac{1}{2} \slashed{D}_3  \left[\gamma \|\alpha_{4n_l}\|^2 \right]  +  \slashed{D}_4  \left[\gamma \|\beta_{4n_l}\|^2\right] 
+   \Big[-\frac{1}{2}\slashed{D}_3 \gamma + \vartheta^-\left(\alpha_{4n_l}\right) \gamma \, tr \underline{H}\Big] \|\alpha_{4n_l}\|^2 \nonumber \\     
+ 2 \left[-\frac{1}{2}\slashed{D}_4 \gamma + 2\vartheta^+\left(\beta_{4n_l}\right) \gamma \,  tr H\right]\|\beta_{4n_l}^2\|  = e_\gamma \left[\alpha_{4n_l},\beta_{4n_l}\right]  
\end{align}
\begin{align}
e_\gamma \left[\alpha_{4n_l},\beta_{4n_l}\right]  = -2\gamma \mathcal{C}_{43}\left[\alpha_{n_l}\right]  \alpha_{4n_l} -\gamma \mathcal{C}_{44}\left[\beta_{n_l}\right]  \beta_{4n_l}
 + \gamma E_3^{4n_l} \left(\alpha\right) \cdot \alpha_{4n_l} \nonumber \\+2\gamma \beta_{4n_l} \cdot E_4^{4n_l}\left(\beta\right) -  \gamma \alpha_{4n_l}\Big[\alpha \left[\frac{1}{2}\slashed{D}_4 tr \underline{H} - \frac{5}{2} \slashed{D}_3 tr H \right] + F_{43}\left(\alpha\right)\Big]_{n_l} \nonumber 
\end{align}
which goes with
\begin{lemma}
\begin{align}
e_\gamma \left[\beta_{4n_l},\alpha_{3n_l}\right]  \leq \left(2+2l\right) \|2\gamma \Omega\|_{L^\infty} \int_{\mathcal{Y}} \left(\|\beta_{4n_l}\|^2 + \|\beta_{4n_l}\|^2+ \| \alpha_{4n_l}\|^2\right) \nonumber \\
 + \lambda \int_{\mathcal{H}, \Sigma_{\tau_1}, \Sigma_{\tau_2}} \left( \|{\alpha}_{4n_l} \|^2 + \|\beta_{4n_l} \|^2\right) + \textrm{error-terms as in Lemma \ref{lemrs1}} \nonumber
\end{align}
\end{lemma}
Once more, we remark that the main terms will be absorbed by the left hand side, once the derivatives for all quantities are added.
Using the Bianchi equations we now prove that the control over the derivatives we have considered is sufficient to control all derivatives:
\begin{lemma} \label{allfrom3}
We have the following estimates on the spheres $S^2_{t^\star,u}$ in the region $r<5M$:
\begin{align}
\sum_{l=0}^{K+1} \sum_{n_l} \int_{S^2} \|\mathcal{D}^{n_l} \tilde{W}\|^2 \leq B \sum_{l=0}^K \sum_{n_l} \int_{S^2} \Big[\|\underline{\alpha}_{3n_l}\|^2 + \|\underline{\beta}_{3n_l}\|^2 \nonumber \\ + \|\left(\rho,\sigma\right)_{3n_l}\|^2 + \|\beta_{3n_l}\|^2  + \|{\alpha}_{3n_l}\|^2 +  \|\beta_{4n_l}\|^2  + \|{\alpha}_{4n_l}\|^2 \Big] \nonumber \\ + B \Big[\|\underline{\alpha}\|^2 + \|\underline{\beta}\|^2 + \|\left(\hat{\rho},\sigma\right)|^2 + \|\beta\|^2  + \|{\alpha}\|^2 \Big] \nonumber \\ +  B \sum_{l=0}^{K} \sum_{n_l} \int_{S^2} \|\mathcal{D}^{n_l}\left( \mathfrak{R}-\mathfrak{R}_{SS}\right)\|^2
\end{align}
\end{lemma}
\begin{proof}
Use the commuted Bianchi equations to estimate $\|\slashed{\nabla} \underline{\alpha}_{n_l}\|^2$ from $\|\underline{\beta}_{n_l3}\|^2$, $\|\slashed{\nabla}\underline{\beta}_{n_l}\|^2$ and $\|\underline{\alpha}_{4n_l}\|^2$ from $\|\left(\rho,\sigma\right)_{3n_l}\|^2$, etc.~cf. Lemma \ref{allfrom4}. Finally, use the fact that $u_{34} = \slashed{\nabla}^2 u + \textrm{l.o.t.}$ holds for the curvature components to estimate the remaining angular derivatives from the $34$-derivatives. Note that $r$-weights are irrelevant in this region.
\end{proof}

We can now complete the Proof of Proposition \ref{redshiftcomplete}. For fixed $l$ add up and integrate the equations (\ref{hdrs1}), (\ref{hdrs2}), (\ref{hdrs3}), (\ref{hdrs4}), (\ref{hdrs5}), (\ref{hdrs6}) using the Lemmata for the error-terms. Do this for all $l$ and sum over all permutations $n_l$ to establishe Proposition \ref{redshiftcomplete}, except for additional lower order terms ($l$-derivatives of curvature) appearing on the right hand side (cf.~the last two lines of Lemma \ref{lemrs1}).  Iterate the estimate for these terms, until one finally needs an estimate for zero derivatives, at which point we insert the estimate of Proposition \ref{redshift1}. 
\section{Elliptic Estimates in the interior region} \label{elliptic}
In the interior region, we can estimate all derivatives provided we control $T$-derivatives:
\begin{proposition} \label{ellipticint}
Let $\left(\mathcal{R},g\right)$ be ultimately Schwarzschildean to order $k+1$ with $k >7$.
For $n\geq 1$ we have
\begin{align} 
\int_{\Sigma \cap \{r \geq r_Y\}} |\mathcal{D}^n \tilde{W}|^2 r^2 dr d\omega \leq B \sum_{i=1}^n \int_{\Sigma} | \widehat{\mathcal{L}}_T^i W|^2 r^2 dr d\omega + B \int_{\Sigma} |\tilde{W}|^2 r^2 dr d\omega \nonumber \\ + B \cdot \overline{\mathbb{E}}^{n-1}\left[\mathfrak{R}\right] \left(\Sigma_{\tau}\right)  \nonumber \, .
\end{align}
Moreover,
\begin{align}
\overline{\mathbb{I}}^{n, deg}_{r_Y - \frac{r_Y-2M}{2} \leq r \leq R+M} \left[W\right] \left(\tilde{\mathcal{M}}\left(\tau_1,\tau_2\right)\right) \leq B \cdot \mathbb{D}^{n-1} \left[\mathfrak{R}\right] \left(\tau_1, \tau_2\right)  + \nonumber \\ B  \int_{\tilde{\mathcal{M}}\left(\tau_1,\tau_2\right)}\frac{1}{r^2} \Big\{ \sum_{i=1}^n \left(1-\frac{3M}{r} \delta_i^n \right)^2 | \widehat{\mathcal{L}}_T^i W|^2 + |\tilde{W}|^2 \Big\} r^2 dt^\star dr d\omega  \nonumber
\end{align}
and
\begin{align}
\overline{\mathbb{I}}^{n, nondeg}_{r_Y - \frac{r_Y-2M}{2} \leq r \leq R+M} \left[W\right] \left(\tilde{\mathcal{M}}\left(\tau_1,\tau_2\right)\right) \nonumber \\ \leq B  \int_{\tilde{\mathcal{M}}\left(\tau_1,\tau_2\right)}\frac{1}{r^2} \Big\{ \sum_{i=1}^n | \widehat{\mathcal{L}}_T^i W|^2 + |\tilde{W}|^2 \Big\} r^2 dt^\star dr d\omega  + B \cdot \mathbb{D}^{n-1} \left[\mathfrak{R}\right] \left(\tau_1, \tau_2\right) \nonumber
\end{align}
Here the constants $B$ depend on the mass $M$ and lower order energies of the Ricci-coefficients and Weyl-curvature, which are bounded (or decaying) by the ultimately Schwarzschildean assumption.
\end{proposition}

\begin{proof}
(Sketch.) Write the Bianchi equation as a $div$-$curl$-system for the electric and magnetic part of the Weyl-tensor, the latter defined with respect to the timelike normal to $\Sigma$, $n_{\Sigma}$. Apply standard estimates from \cite{ChristKlei}. The (quadratic) inhomogeneity is easily estimated by lower order energies of the Ricci-coefficients and Weyl-curvature (using Sobolev embedding).
\end{proof}
\section{The region near infinity} \label{decinf1}
In section \ref{redsection} we revealed a remarkable hierarchy in the Bianchi equations, which allowed us to obtain estimates for all derivatives of curvature. This hierarchy is also present near infinity and somewhat even more remarkable. This is because close to the horizon all curvature components decay at the same rate in $\tau$ (there is merely a difference in the strength of the redshift factor), while at infinity the curvature components and their derivatives each have a characteristic decay in $r$, which the estimates need to reveal.  It turns out that using appropriate $r^p$-weighted multipliers (for some positive $p$) at the level of the null-Bianchi equations is sufficient to establish the correct asymptotics. In principle, just as near the horizon, it suffices to commute the Bianchi equations with $3-$ and $4-$derivatives (and obtain angular derivatives from the wave equations satisfied by the components). However, for the \emph{optimal} $r$-weights (as exhibited in the energies of section \ref{norms}), this commutation seems unavoidable.  Since every set of Bianchi equations will admit a different $r$-weighted estimates, depending on how many derivatives have been taken, we have introduced the concept of boundary admissible tuples and matrices in section \ref{decmatrix}, which should be understood as encoding the $r^p$ weighted decay for the individual components at each order.

\subsection{The main result}
\begin{theorem} \label{maintheoinf}
For any boundary admissible matrix $P$ we have for $0\leq n \leq k$
\begin{align}
\overline{\mathbb{E}}^n_P \left[W\right] \left(N_{out} \left(\tau_2,R\right)\right) + \overline{\mathbb{I}}^n_{\tilde{P}} \left[W\right] \left(\mathcal{D}\left(\tau_1,\tau_2\right)\right)  \leq
\overline{\mathbb{E}}^n_{P} \left[W\right] \left(N_{out} \left(\tau_1,R\right)\right) \nonumber \\ + B \cdot \overline{\mathbb{I}}^{n,deg}_{R-M<r<R+M}\left[W\right]\left(\tilde{\mathcal{M}}\left({\tau_1,\tau_2}\right)\right) \nonumber \\ + B \left[
\sup_{\tau \in \left(\tau_1,\tau_2\right)}\overline{\mathbb{E}}^{n} \left[\mathfrak{R}\right] \left(N_{out}\left(\tau,R\right)\right)  + \overline{\mathbb{I}}^{n} \left[\mathfrak{R}\right] \left(\mathcal{D} \left(\tau_1,\tau_2\right)\right) \right] \, .
\end{align}
The constant $B$ in the second line depends on $M$ and on lower order curvature energies which are bounded by the ultimately Schwarzschildean assumption.
\end{theorem}

\subsection{Estimates for the uncommuted equations} \label{infuncom}
Multiplying (\ref{Bianchi1}) by $r^{p_2-2} \left(1-\mu\right)^q  \alpha$ for positive integers $p_2,q$ produces
\begin{align}
\frac{1}{2}\slashed{D}_3 \left[r^{p_2-2}  \left(1-\mu\right)^q \| \alpha \|^2 \right] \nonumber \\ +
 \|\alpha\|^2 \left[ - \frac{p_2-2}{2} \slashed{D}_3 r - 4 \underline{\Omega} r + \frac{1}{2} tr \left(\underline{H}\right) r -\frac{1}{2} q r \frac{\slashed{D}_3 \left(1-\mu\right)}{1-\mu} \right]  \left(1-\mu\right)^q r^{p_2-3}  \nonumber \\ =
-2r^{p_2-2} \left(1-\mu\right)^q \slashed{\mathcal{D}}^\star_2 \beta \cdot \alpha + r^{p_2-2} \left(1-\mu\right)^q \left(E_3 \left(\alpha\right) - 4 \underline{\Omega} \alpha\right) \nonumber
\end{align}
Using that $div$ is the adjoint of $\slashed{\mathcal{D}}^\star_2$ and inserting the Bianchi equation (\ref{Bianchi2}), we can write this as
\begin{align}
\frac{1}{2}D_a \left(\left(e_3\right)^a r^{p_2-2} \left(1-\mu\right)^q \| \alpha \|^2 \right) + D_a  \left[\left(e_4\right)^a r^{p_2-2}\left(1-\mu\right)^q \|\beta\|^2 \right] \nonumber \\
+ \|\alpha\|^2r^{p_2-3}  \left(1-\mu\right)^q \left[ - \frac{p_2-2}{2}  \slashed{D}_3 r - 3 \underline{\Omega} r  -\frac{1}{2} r q \frac{\slashed{D}_3 \left(1-\mu\right)}{1-\mu} \right]  = \nonumber \\
+ \|\beta\|^2 r^{p_2-3}  \left(1-\mu\right)^q \left[3r \, tr H + 6\Omega r - \left(p_2-2\right)  \slashed{D}_4 r  - q r \frac{\slashed{D}_4 \left(1-\mu\right)}{1-\mu} \right]  \nonumber  \\
+ r^{p_2-2} \left(1-\mu\right)^q \Big( \left(E_3\left(\alpha\right) - 4\underline{\Omega}\alpha\right) \cdot \alpha + 2  \left(E_4\left(\beta\right) + 2\Omega \beta \right) \cdot \beta + 2 \left[\slashed{div} \left(\alpha \beta\right)\right]\Big) \nonumber 
\end{align}
We integrate over the region $\mathcal{D}^{\tau_2}_{\tau_1}$ taking into account the following observations:
 \begin{itemize}
\item the terms in the last line are error-terms. For the last term we integrate the angular derivative by parts to produce a cubic error-term in view of the fact that both $r$ and the measure approach something spherically symmetric. 
\item The first square bracket is equal to $\left[\left(p_2-2\right) + \mathcal{O}\left(\frac{1}{r}\right) \right]$
\item The second square bracket is equal to $\left[\left(8-p_2\right) + \mathcal{O}\left(\frac{1}{r}\right) \right]$. However, choosing $q$ to be a sufficiently large negative number ($q=-4$ is good enough), we can ensure that the $\mathcal{O}\left(\frac{1}{r}\right)$-term has a positive sign.
\end{itemize}
Using that $1-\mu \approx 1$ in this region and that the measure brings in another power of $2$ in $r$, it follows that for sufficiently large $R$ we have the estimate
\begin{align} \label{st1}
 \int_{N_{out}\left(S^2_{\tau_2,R}\right)} dv \, d\omega_2 \, r^{p_2} \|\alpha\|^2 + \int_{\mathcal{I}^{\tau_2-R}_{\tau_1-R}} du \, d\omega_2 \, r^{p_2} \|\beta\|^2 \nonumber \\
+ \int_{\mathcal{D}^{\tau_2}_{\tau_1}} \|\alpha\|^2 \, r^{{p_2}-1} \Big[\left(p_2-2\right) + \mathcal{O}\left(\frac{1}{r}\right) \Big] dt^\star \, dr \, d\omega \nonumber \\ 
+ \int_{\mathcal{D}^{\tau_2}_{\tau_1}} \|\beta\|^2 \, r^{{p_2}-1} \Big[\left(8-p_2\right) + \mu \Big] dt^\star \, dr \, d\omega \leq B \Big| \int_{\mathcal{D}^{\tau_2}_{\tau_1}}  \, e_{p_2}\left[\alpha,\beta\right]  \Big| \nonumber \\ + B \cdot \mathbb{I}^{0,deg}_{R-M<r<R+M}\left[W\right]\left(\mathcal{M}\left({\tau_1,\tau_2}\right)\right) + 
2 \int_{N_{out}\left(S^2_{\tau_1,R}\right)} dv \, d\omega_2 \, r^{p_2} \|\alpha\|^2 \, ,
\end{align}
with the error
\begin{align}
e_{p_2}\left[\alpha,\beta\right] =  r^{p_2-2} \left(1-\mu\right)^q \Big( \left(E_3\left(\alpha\right) - 4\underline{\Omega} \alpha \right) \cdot \alpha + 2  \left(E_4\left(\beta\right) + 2\Omega \beta \right) \cdot \beta + 2 \left[\slashed{div} \left(\alpha \beta\right)\right]\Big) \nonumber \, .
\end{align} 
It follows that for $2<p_2\leq8$ we obtain good terms on the left hand side, provided $R$ is chosen sufficiently large. However, there is a non-linear restriction from the error-terms in the penultimate line as we will see in the next section. 

Let us define the renormalized quantity $\widehat{\rho} = \rho + \frac{2M}{r^3}$.
Multiplying (\ref{Bianchi5}) by $r^{p_3} \left(1-\mu\right)^q \beta$ and integrating by parts as before to insert (\ref{Bianchi6}) and (\ref{Bianchi7}), yields for sufficiently large $R$ the estimate
\begin{align} \label{st3}
 \int_{N_{out}\left(S^2_{\tau_2,R}\right)} dv \, d\omega_2 \, r^{p_3} |\beta|^2 +  \int_{\mathcal{I}^{\tau_2-R}_{\tau_1-R}} du \, d\omega_2 \, r^{p_3} \left(\widehat{\rho}^2 +\sigma^2\right) \nonumber \\
 \int_{\mathcal{D}^{\tau_2}_{\tau_1}} dt^\star dr d\omega \, |\beta|^2 \, r^{{p_3}-1} \Big[\left(p_3 -4\right) + \mathcal{O}\left(\frac{1}{r}\right)  \Big] + \nonumber \\     
\int_{\mathcal{D}^{\tau_2}_{\tau_1}} dt^\star dr d\omega \,  \left(\widehat{\rho}^2 +\sigma^2\right) r^{p_3-1} \left[ \left(6- p_3\right) + \mu  \right]   \leq B \Big| \int_{\mathcal{D}^{\tau_2}_{\tau_1}}  e_{p_3}\left[\beta, \left(\rho,\sigma\right)\right] \Big| \nonumber \\ + B \cdot \mathbb{I}^{0,deg}_{R-M<r<R+M}\left[W\right]\left(\mathcal{M}\left({\tau_1,\tau_2}\right)\right) 
+4 \int_{N_{out}\left(S^2_{\tau_1,R}\right)} dv \, d\omega_2 \, r^{p_3} |\beta|^2
\end{align}
with the error
\begin{align}
e_{p_3}\left[\beta, \left(\rho,\sigma\right)\right] =  r^{p_3-2} \left(1-\mu\right)^q \left[ \left(E_3\left(\beta \right) - 2\underline{\Omega}\beta\right)  \cdot \beta +  \hat{E}_4\left(\rho\right)  \widehat{\rho} + E_4\left(\sigma\right) \sigma + \slashed{div} \left[ \left(-\widehat{\rho}, \sigma\right) \beta \right] \right] \nonumber
\end{align}
hence (\ref{st3}) produces a good estimate for $4<{p_3}\leq 6$. However, note that $p_3 \leq4$ is also admissible as long as $p_3<p_2$, as we can estimate the wrong signed $\beta$-term by adding a multiple of the previous estimate (\ref{st1}).

We now consider the underlined quantities. 
We start with equations (\ref{Bianchi8}), (\ref{Bianchi9}) and (\ref{Bianchi10}), from which we derive
\begin{align} \label{st4}
 \int_{N_{out}\left(S^2_{\tau_2,R}\right)} dv \, d\omega_2 \, r^{p_4} \left(\widehat{\rho}^2+{\sigma}^2\right) +  \int_{\mathcal{I}^{\tau_2-R}_{\tau_1-R}} du \, d\omega_2 \, r^{p_4} |\underline{\beta}|^2 + \nonumber \\
 \int_{\mathcal{D}^{\tau_2}_{\tau_1}}  dt^\star dr d\omega \,  \left(\widehat{\rho}^2+{\sigma}^2\right) \, r^{{p_4}-1} \Big[\left(p_4 -6\right) + \mathcal{O}\left(\frac{1}{r}\right)  \Big] 
\nonumber \\ + \int_{\mathcal{D}^{\tau_2}_{\tau_1}} dt^\star dr d\omega \,  |\underline{\beta}|^2 \,  r^{p_4-1} \left[ \left(4- p_4\right) + \mu  \right]   \leq \Big| \int_{\mathcal{D}^{\tau_2}_{\tau_1}} e_{p_4}\left[\left(\rho,\sigma\right), \underline{\beta} \right] \Big| + \nonumber \\  B \cdot \mathbb{I}^{0,deg}_{R-M<r<R+M}\left[W\right]\left(\mathcal{M}\left({\tau_1,\tau_2}\right)\right) + 4 \int_{N_{out}\left(S^2_{\tau_2,R}\right)} dv \, d\omega_2 \, r^{p_4} \left(\tilde{\rho}^2+\tilde{\sigma}^2\right) 
\end{align}
with the error
\begin{align}
e_{p_4}\left[\left(\rho,\sigma\right), \underline{\beta} \right] = r^{p_4-2} \left(1-\mu\right)^q \Big[ \left(E_4\left(\underline{\beta}\right)  - 2\Omega \underline{\beta}\right) \cdot \underline{\beta} + \hat{E}_3\left(\rho\right) \widehat{\rho} + E_4\left(\sigma\right) \sigma  -\slashed{div} \left[\beta \cdot \left(\widehat{\rho}, \sigma\right)\right] \Big] \nonumber
\end{align}
Obviously, the first spacetime term demands ${p_4}>6$ while the $\underline{\beta}$-term requires ${p_4}\leq 4$ in order to be positive. However, since we already control the spacetime integral of $\left(\widehat{\rho}^2+{\sigma}^2\right)$ from the previous estimate, the restriction is ${p_4} \leq \min\left(p_3,4\right)$ only.

Finally, we multiply (\ref{Bianchi4}) by $r^{p_5} \underline{\beta} \left(1-\mu\right)^a$ and integrate in the region $\mathcal{D}^{\tau_2}_{\tau_1}$ to arrive at the estimate
\begin{align} \label{st2}
 \int_{N_{out}\left(S^2_{\tau_2,R}\right)} dv \, d\omega_2 \, r^{p_5} |\underline{\beta}|^2 +  \int_{\mathcal{I}^{\tau_2-R}_{\tau_1-R}} du \, d\omega_2 \, r^{p_5} |\underline{\alpha}|^2 \nonumber \\
+ \int_{\mathcal{D}^{\tau_2}_{\tau_1}} dt^\star dr d\omega \, |\underline{\beta}|^2 \, r^{{p_5}-1}  \Big[\left(p_5 -8\right) + \mathcal{O}\left(\frac{1}{r}\right)  \Big]  \nonumber \\ 
+ \int_{\mathcal{D}^{\tau_2}_{\tau_1}}dt^\star dr d\omega \,  |\underline{\alpha}|^2 \, r^{{p_5}-1} \Big[\left(2-p_5\right) + \mu \Big]\leq \Big|\int_{\mathcal{D}^{\tau_2}_{\tau_1}} e_{p_5}\left[\underline{\beta}, \underline{\alpha}\right] \Big| \nonumber \\  + B \cdot \mathbb{I}^{0,deg}_{R-M<r<R+M}\left[W\right]\left(\mathcal{M}\left({\tau_1,\tau_2}\right)\right) 
+ 4 \int_{N_{out}\left(S^2_{\tau_1,R}\right)} dv \, d\omega_2 \, r^{p_5} |\underline{\beta}|^2
\end{align}
with the error
\begin{align}
e_{p_5}\left[\underline{\beta}, \underline{\alpha}\right] = \frac{1}{2} r^{p_5-2} \left(1-\mu\right)^q \Big( \left(E_3 \left(\underline{\beta}  \right)+ 2\underline{\Omega} \underline{\beta}\right) \cdot \underline{\beta} + \frac{1}{4}  \left(E_4 \left(\underline{\alpha} \right) - 4\Omega \underline{\alpha} \right) \cdot \underline{\alpha} - 2 \slashed{div} \left(\underline{\alpha} \cdot \underline{\beta}\right) \Big) \nonumber
\end{align}
We conclude that this estimate requires  $p\leq 2$ and $p>8$ to make both volume terms positive. However, by the usual argument of adding multiples of the previous estimates, only $p\leq \min\left(2,p_3\right)$ is relevant. We summarize this as
\begin{proposition} \label{loinf}
For  a boundary $0$-admissible tuple $\left(0,p_2,p_3,p_4,p_5,0,0\right)$ we have, for sufficiently large $R$,
\begin{align}
\overline{\mathbb{E}}^0_{\left(0,p_2,p_3,p_4,p_5, 0, 0\right)} \left[W\right] \left(N_{out} \left(\tau_2,R\right)\right) \nonumber \\ + \overline{\mathbb{I}}^0_{\left(0,p_2-1,p_2-1,p_3-1,p_4-1,p_5-1, 0\right)} \left[W\right] \left(\mathcal{D}\left(\tau_1,\tau_2\right)\right) \nonumber \\ \leq
\overline{\mathbb{E}}^0_{\left(0,p_2,p_3,p_4,p_5, 0,0\right)} \left[W\right] \left(N_{out} \left(\tau_1,R\right)\right) \nonumber \\ + B \cdot \mathbb{I}^{0,deg}_{R-M<r<R+M}\left[W\right]\left(\mathcal{M}\left({\tau_1,\tau_2}\right)\right) + B \cdot Err_{\left(0, p_2, p_3, p_4, p_5, 0,0\right)}
\end{align}
where
\begin{align}
Err_{\left(0, p_2, p_3, p_4, p_5, 0, 0\right)} =
\Big|  \int_{\mathcal{D}^{\tau_2}_{\tau_1}} dt^\star dr d\omega \, e_{p_2} \left[\alpha,\beta \right] \Big|  + \Big|  \int_{\mathcal{D}^{\tau_2}_{\tau_1}} dt^\star dr d\omega  \, e_{p_3} \left[\beta, \left(\rho,\sigma\right) \right] \Big| \nonumber \\  + \Big|  \int_{\mathcal{D}^{\tau_2}_{\tau_1}} dt^\star dr d\omega  \, e_{p_4} \left[\left(\rho,\sigma\right),\underline{\beta} \right] \Big| + \Big|  \int_{\mathcal{D}^{\tau_2}_{\tau_1}} dt^\star dr d\omega  \, e_{p_5} \left[\underline{\beta},\underline{\alpha}\right]  \Big| \nonumber
\end{align}
\end{proposition}
\begin{proof}
Add up the estimates (\ref{st1}), (\ref{st2}),  (\ref{st3}) and  (\ref{st4}).
\end{proof}
For this proposition to be useful, we need to estimate the error-term in terms of the bulk-term on the left hand side. Clearly, this may impose further constraints on the admissible $p$'s. We now show it doesn't:

\begin{proposition} \label{erlosto}
For a boundary $0$-admissible tuple $\left(0,p_2,p_3,p_4,p_5,0,0\right)$ we have
\begin{align}
Err_{\left(0, p_2, p_3, p_4, p_5, 0, 0\right)} \leq \nonumber \\
\sqrt{\|\rho^2 r^6 \|_{L^\infty} \cdot I^{0} \left[\mathfrak{R}\right] \left(\mathcal{D}^{\tau_2}_{\tau_1}\right)} \sqrt{\overline{\mathbb{I}}^0_{\left(0,p_2-1,p_2-1,p_3-1,p_4-1,p_5-1, 0\right)} \left[W\right] \left(\mathcal{D}\left(\tau_1,\tau_2\right)\right)} \nonumber \\ + \Big\|r^2 \left(\mathfrak{R}^{main}, \widehat{\underline{H}} \frac{1}{r}\right)\Big\|_{L^\infty} \cdot  \overline{\mathbb{I}}^0_{\left(0,p_2-1,p_2-1,p_3-1,p_4-1,p_5-1, 0\right)} \left[W\right] \left(\mathcal{D}\left(\tau_1,\tau_2\right)\right) \nonumber \, .
\end{align}
\end{proposition}
Note that in view of the smallness assumptions on the Ricci rotation coefficients given in  Definition \ref{RRCapproach}, the estimate of Proposition \ref{loinf} closes.
\begin{proof}
For the $\alpha$-part of the error of estimate (\ref{st1}) we have
\begin{align}
\int r^{p_2}  \left(E_3\left(\alpha\right) - 4\underline{\Omega}\alpha\right)\alpha \nonumber \\ 
\leq \sqrt{\int r^{p_2-1}|\alpha|^2 dt^\star dr d\omega }\sqrt{\int r^{p_2+1}| \left(E_3\left(\alpha\right) - 4\underline{\Omega}\right)|^2 dt^\star dr d\omega} \, .
\end{align}
We note that
\begin{align}
\int r^{p_2+1}| \left(E_3\left(\alpha\right) - 4\underline{\Omega}\right)|^2 dt^\star dr d\omega\nonumber \\ \leq C \int \left[\|\widehat{H}\|^2 \left(\rho^2+\sigma^2\right) + \left(\|V\|^2 +\|Z\|^2\right) |\beta|^2 \right] r^{p_2+1} dt^\star dr d\omega \nonumber \\
\leq C ||\rho^2 r^6||_{L^\infty} \int \| \widehat{H}\|^2 r^{p_2-5} dt^\star dr d\omega + C ||r^4 \left(\|V\|^2 +\|Z\|^2\right)||_{L^\infty}  \int r^{p_2-3}|\beta|^2 dt^\star dr d\omega \nonumber
\end{align}
from which we obtain the constraint $p_2<7-\delta$ (otherwise the spacetime-term involving $\widehat{H}^2$ is not small).

For the $\beta$-part, recalling that in our frame $Y=0$,
\begin{align}
\int r^{p_2+1}| \left(E_4\left(\beta\right) + 2{\Omega}\beta\right)|^2 dt^\star dr d\omega \leq C \int \left(\|V\|^2 +\|\underline{Z}\|^2\right) |\alpha|^2 r^{p_2+1} dt^\star dr d\omega \nonumber \\
\leq C ||r^4 \left(V^2 + \|\underline{Z}\|^2\right)||_{L^\infty}  \int r^{p_2-3}|\alpha|^2 dt^\star dr d\omega \nonumber
\end{align}
which does not impose further constraints on $p_2$.

Turning to $e_{p_3} \left[\beta, \left(\rho,\sigma\right)\right]$ we have
\begin{align}
\int r^{p_3} \widehat{H} \cdot \underline{\beta} \cdot \beta  \ dt^\star dr d\omega \leq \sqrt{\int r^{p_2-1} \|\beta\|^2 dt^\star dr d\omega} \sqrt{\|\widehat{H}^2 r^4\|_{L^\infty} \int r^{2p_3-p_2-3} \|\underline{\beta}\|^2 dt^\star dr d\omega} \nonumber
\end{align}
from which we read off the constraint $2p_3-p_2-3 \leq p_4 -1$,
\begin{align}
\int r^{p_3} \underline{Y} \cdot \alpha \cdot \beta  \,  dt^\star dr d\omega \leq \sqrt{\int r^{p_2-1} \|\beta\|^2 dt^\star dr d\omega} \sqrt{\|\underline{Y}^2 r^4\|_{L^\infty} \int r^{2p_3-p_2-3} \|\alpha \|^2 dt^\star dr d\omega} \nonumber
\end{align}
from which we read off the constraint $p_3 \leq p_2 + 2$, 
\begin{align}
\int r^{p_3} \rho Z  \cdot \beta  \,  dt^\star dr d\omega \leq \sqrt{\int r^{p_2-1} \|\beta\|^2 dt^\star dr d\omega} \sqrt{\|\rho^2 r^6\|_{L^\infty} \int r^{2p_3-p_2-5} \|Z \|^2 dt^\star dr d\omega} \nonumber
\end{align}
from which we read off the constraint $2p_3 \leq p_2 + 7-\delta$ in order for the $\|Z\|^2$ spacetime integral to produce a smallness factor. A similar estimate is obtained for the term ${}^\star Z \sigma$. Turning to $\hat{E}_4\left(\rho\right)$ we estimate
\begin{align}
\int r^{p_3}\underline{\widehat{H}}  \cdot \alpha \, \widehat{\rho}  \,  dt^\star dr d\omega \leq \sqrt{\int r^{p_3-1}  \|\hat{\rho} \|^2 dt^\star dr d\omega} \sqrt{\|\underline{\widehat{H}}^2 r^2\|_{L^\infty} \int r^{p_3-1}\|\alpha\|^2  dt^\star   dr d\omega} \nonumber
\end{align}
which imposes $p_3<p_2$. The remaining terms do not create new constraints and we can turn to $e_{p_4} \left[\left(\rho,\sigma\right),\underline{\beta} \right]$, for which the most difficult terms are
\begin{align}
\int r^{p_4}\widehat{H}  \cdot \underline{\alpha} \, \widehat{\rho}  \,  dt^\star dr d\omega \leq \sqrt{\int r^{p_3-1}  \|\hat{\rho} \|^2 dt^\star dr d\omega} \sqrt{\|\widehat{H}^2 r^4\|_{L^\infty} \int r^{2p_4-p_3-3}\|\underline{\alpha}\|^2  dt^\star   dr d\omega} \nonumber
\end{align}
imposing $2p_4-p_3 < p_5+2$ and
\begin{align}
\int r^{p_4} \rho \underline{Z}  \cdot \underline{\beta}  \,  dt^\star dr d\omega \leq \sqrt{\int r^{p_4-1} \|\underline{\beta}\|^2 dt^\star dr d\omega} \sqrt{\|\rho^2 r^6\|_{L^\infty} \int r^{p_4-5} \|\underline{Z} \|^2 dt^\star dr d\omega} \nonumber
\end{align}
is fine since $p_4<4$ anyway. Finally, for $e_{p_5} \left[\underline{\beta},\underline{\alpha}\right]$ we note
\begin{align}
\int r^{p_5} \left(V,Z\right)  \cdot \underline{\alpha} \, \underline{\beta}  \,  dt^\star dr d\omega \leq \sqrt{\int r^{p_4-1}  \|\underline{\beta} \|^2 dt^\star dr d\omega} \sqrt{\|\left(V,Z\right)^2 r^4\|_{L^\infty} \int r^{2p_5-p_4-3}\|\underline{\alpha}\|^2  dt^\star   dr d\omega} \nonumber
\end{align}
revealing $p_5<2+p_4$, which is already satisfied as $p_5<2$.
Inspecting all the constraints on the $p$'s obtained, we observe that they are all implied by the assumption that the $p$-tuple is boundary admissible.
\end{proof}
This establishes Theorem \ref{maintheoinf} for $n=0$.
\subsection{Estimates for the commuted equations}
To prove Theorem \ref{maintheoinf} we will first prove the higher derivative version of Proposition \ref{loinf}. The focus here is on the main-terms and their $r$-weights, while all errors are simply collected on the right hand side. The latter will be estimated separately in Proposition \ref{errhigh} in the next section. Combining the two Proposition will imply Theorem \ref{maintheoinf}.
\begin{proposition} \label{hiinf}
For $P$ a boundary admissible matrix with second row being the tuple $p=\left(p_1,p_2,..., p_6, 0\right)$ we have the estimate
\begin{align}
\overline{\mathbb{E}}^{m+1}_P \left[W\right] \left(N_{out} \left(\tau_2,R\right)\right) + \overline{\mathbb{I}}^{m+1}_{\tilde{P}} \left[W\right] \left(\mathcal{D}\left(\tau_1,\tau_2\right)\right)  \leq
\overline{\mathbb{E}}^{m+1}_{P} \left[W\right] \left(N_{out} \left(\tau_1,R\right)\right) \nonumber \\ + B \cdot \overline{\mathbb{I}}^{m+1,deg}_{R-M<r<R+M}\left[W\right]\left(\tilde{\mathcal{M}}\left({\tau_1,\tau_2}\right)\right) + B \cdot \overline{Err}^{m+1}_P + B \cdot Low^{m+1}_P \, .
\end{align}
Here the error is given by (all integrals are taken over $\mathcal{D}^{\tau_2}_{\tau_1}$)
\begin{align}
\overline{Err}^{m+1}_{P} = \sum_{j=0}^m {Err}^{j+1}_{P} = 
\sum_{j=0}^k \sum_{l=0}^j \sum_{n_l} \Bigg( \Big| \int e_{p_1+2l_2}\left[\Omega_i^{j-l} \alpha_{4n_l}, \Omega_i^{j-l} \beta_{4n_l}\right] \Big| + \nonumber \\ \Big| \int e_{p_2+2l_2}\left[\Omega_i^{j-l}\beta_{4n_l}, \Omega_i^{j-l}\left(\rho_{4n_l},\sigma_{4n_l}\right)\right]  \Big|  
+ \Big| \int  e_{p_2+2l_2}\left[\Omega_i^{j-l}\beta_{4n_l}, \Omega_i^{j-l}\alpha_{3n_l} \right] \Big| +\nonumber \\
 \Big|  \int  e_{p_3+2l_2}\left[ \Omega_i^{j-l}\left(\rho_{4n_l}, \sigma_{4n_l}\right), \Omega_i^{j-l}\beta_{3n_l} \right] \Big| +  \Big|\int e_{p_4+2l_2}\left[\Omega_i^{j-l}\underline{\beta}_{4n_l}, \Omega_i^{j-l}\underline{\alpha}_{4n_l}\right] \Big| \nonumber \\
+ \Big| \int e_{p_5+2l_2}\left[\Omega_i^{j-l}\underline{\alpha}_{4n_l}, \Omega_i^{j-l}\underline{\beta}_{3n_l}\right] \Big| + \Big|\int e_{p_6+2l_2}\left[\Omega_i^{j-l}\underline{\beta}_{3n_l}, \Omega_i^{j-l}\underline{\alpha}_{3n_l}\right] \Big| \Bigg) \nonumber
\end{align}
\begin{align}
\overline{Low}^{m+1}_P =  \sum_{j=0}^m {Low}^{j+1}_P = \sum_{j=0}^m \sum_{l=0}^j \sum_{n_l} \int_{N_{out}} dv d\omega \,  r^{2l_2} Low^{\Omega_i^{m-l}}_{n_l}\left(p_2, p_3, p_4, p_5, p_6\right) \nonumber \\
+  \sum_{j=0}^m \sum_{l=0}^j \sum_{n_l} \int_{\mathcal{D}^{\tau_2}_{\tau_1}} dt^\star dr d\omega \, r^{2l_2} Low^{\Omega_i^{m-l}}_{n_l}\left(\tilde{p}_1,\tilde{p}_2, \tilde{p}_3, \tilde{p}_4, \tilde{p}_5\right) \nonumber 
\end{align}
where $n_l$ denotes a $l$-tuple of $3$'s and $4$'s with $l_1$ being the number of $3$'s and $l_2$ being the number of $4$'s. The integral is short hand for $\int = \int \sqrt{g} dt^\star dr d\omega$ and the integration region is always  $\mathcal{D}^{\tau_2}_{\tau_1}$. The precise definitions of the $e_p \left[..., ...\right]$ are given in the following section, while $Low^{\Omega_i^{m-l}}_{n_l}\left(...\right)$ is defined in Lemma \ref{allfrom4}.
\end{proposition}

\subsubsection{Derivatives in the $4$-direction} \label{4der}
We start with the Bianchi equations for $\alpha_{34n_l}$ (derived from (\ref{a3nk})) and $\beta_{44n_l}$ (derived from (\ref{b4nk})). First we  relate $\alpha_{34n_l}$ to $\alpha_{4n_l3}$ using the commutation formulae. Since the resulting equation is of the same form as the original Bianchi equation, we proceed as for the non-commuted equations, i.e.~we multiply the equation for $\alpha_{4n_l3}$ by $r^{p-2} \left(1-\mu\right)^q \alpha_{4n_l}$ and integrate by parts. This yields the estimate
\begin{align} \label{4ab}
 \int_{N_{out}\left(S^2_{\tau_2,R}\right)} dv \, d\omega_2 \, r^{p} \|\alpha_{4n_l}\|^2 + \int_{\mathcal{I}^{\tau_2-R}_{\tau_1-R}} du \, d\omega_2 \, r^{p} \|\beta_{4n_l}\|^2 \nonumber \\
+ \int_{\mathcal{D}^{\tau_2}_{\tau_1}} dt^\star dr d\omega \, \|\alpha_{4n_l}\|^2 \, r^{{p}-1} \Big[-\frac{{p}}{2}\slashed{D}_3 r + \vartheta^-\left(\alpha_{4n_l}\right) tr \left(\underline{H}\right) r + \mathcal{O} \left(\frac{1}{r}\right)\Big] \nonumber \\ 
+ \int_{\mathcal{D}^{\tau_2}_{\tau_1}} dt^\star dr d\omega \, \|\beta_{4n_l}\|^2  r^{p-1} \Big[2\vartheta^+\left(\beta_{4n_l}\right) \,r \cdot tr H -p \slashed{D}_4 r + \mu \Big]  \nonumber \\ \leq B \Big| \int_{\mathcal{D}^{\tau_2}_{\tau_1}} e_{p}\left[\alpha_{4n_l}, \beta_{4n_l}\right] \Big| + B \cdot \mathbb{I}^{l+1,deg}_{R-M<r<R+M}\left[W\right]\left(\mathcal{M}\left({\tau_1,\tau_2}\right)\right) \nonumber \\
+2 \int_{N_{out}\left(S^2_{\tau_1,R}\right)} dv \, d\omega_2 \, r^{p} \|\alpha_{4n_l}\|^2 
\end{align}
with the error-term
\begin{align}
e_{p}\left[\alpha_{4n_l}, \beta_{4n_l}\right] = r^{p-2} \left(1-\mu\right)^q \left(\hat{E}_3^{4n_l}\left(\alpha\right) \alpha_{4n_l} + 2 \tilde{E}_4^{4n_l} \left(\beta\right) \cdot \beta_{4n_l} + 2\slashed{div} \left[\alpha_{4n_l} \cdot \beta_{4n_l} \right] \right) \nonumber
\end{align}
One reads off that the left hand side is positive for $r>R$ sufficiently large if
\begin{equation}
2\left(4+l-s\left(\alpha_{4n_l}\right) \right)  < {p} < 2\left(4+l+s\left(\beta_{4n_l}\right)\right) 
\end{equation}
holds, which is equivalent to
\begin{equation}  \label{cond4ab}
2 + 4 \left(\textrm{number of $3$'s in $n_l$}\right) < {p} \leq 12 + 4 \left(\textrm{number of $4$'s in $n_l$}\right) \, .
\end{equation}
From the Bianchi equations for $\beta_{34n_l}$ (derived from (\ref{b3nk})) and $\left(\rho,\sigma\right)_{44n_l}$ (derived from (\ref{r4nk}), (\ref{s4nk})) we obtain, after relating $\beta_{34n_l}$ to $\beta_{4n_l3}$:
\begin{align} \label{4brs}
\frac{1}{2} \int_{N_{out}\left(S^2_{\tau_2,R}\right)} dv \, d\omega_2 \, r^{p} \|\beta_{4n_l}\|^2 + \frac{1}{2} \int_{\mathcal{I}^{\tau_2-R}_{\tau_1-R}} du \, d\omega_2 \, r^{{p}} \left( \rho_{4n_l}^2 +\sigma_{4n_l}^2\right) \nonumber \\
+ \int_{\mathcal{D}^{\tau_2}_{\tau_1}} dt^\star dr d\omega \,  r^{{p}-1} \|\beta_{4n_l}\|^2 \, \Big[-\frac{{p}}{2}\slashed{D}_3 r + \vartheta^-\left(\beta_{4n_l}\right) r \, tr \underline{H} + \mathcal{O} \left(\frac{1}{r}\right) \Big] \nonumber \\     
+ \frac{1}{2}\int_{\mathcal{D}^{\tau_2}_{\tau_1}}dt^\star dr d\omega \, r^{p-1}  \left( \rho_{4n_l}^2 +\sigma_{4n_l}^2\right) \left[-p \slashed{D}_4 r + 2\vartheta^+\left(\rho_{4n_l}\right) r \, tr H + \mu \right]    \nonumber \\ \leq \Bigg|\int_{\mathcal{D}^{\tau_2}_{\tau_1}} e_{p}\left[\beta_{4n_l}, \left(\rho_{4n_l}, \sigma_{4n_l}\right) \right] \Bigg|+ B \cdot \mathbb{I}^{l+1,deg}_{R-M<r<R+M}\left[W\right]\left(\mathcal{M}\left({\tau_1,\tau_2}\right)\right)  \nonumber \\ + 2 \int_{N_{out}\left(S^2_{\tau_1,R}\right)} dv \, d\omega_2 \, r^{p} \|\beta_{4n_l}\|^2
\end{align}
with the error-term
\begin{align}
e_{p}\left[\beta_{4n_l}, \left(\rho_{4n_l}, \sigma_{4n_l}\right) \right] = r^{p-2} \left(1-\mu\right)^q \Big( \hat{E}_3^{4n_l} \left(\beta\right) \cdot \beta_{4n_l} + \tilde{E}_4^{4n_l} \left(\rho\right) \cdot \rho_{4n_l} \nonumber \\ +\tilde{E}_4^{4n_l} \left(\sigma\right) \cdot \sigma_{4n_l} + \slashed{div} \left[ \left(-\rho_{4n_l} ,\sigma_{4n_l} \right) \cdot \beta_{4n_l}\right] \Big) \, . \nonumber
\end{align}
The left hand side of (\ref{4brs}) is positive for $r>R$ sufficiently large if
\begin{equation}  \label{cond4brs}
4 + 4 \left(\textrm{number of $3$'s in $n_l$}\right) < {p} \leq 10 + 4 \left(\textrm{number of $4$'s in $n_l$}\right) \, .
\end{equation}
Note that only the upper bound is relevant, as we already control the $\alpha_{4n_l}$ spacetime term from the previous step.

Thirdly, combining the Bianchi equations for $\left(\rho,\sigma\right)_{34n_l}$ and $\underline{\beta}_{44n_l}$
\begin{align}  \label{4rsbb}
\frac{1}{2}\int_{N_{out}\left(S^2_{\tau_2,R}\right)} dv \, d\omega_2  r^{p} \left( \rho_{4n_l}^2 +\sigma_{4n_l}^2\right) \, + \frac{1}{2}\int_{\mathcal{I}^{\tau_2-R}_{\tau_1-R}} du \, d\omega_2 \,  r^{p} \|\underline{\beta}_{4n_l}\|^2  \nonumber \\
+ \int_{\mathcal{D}^{\tau_2}_{\tau_1}}  dt^\star dr d\omega \,  r^{{p}-1} \left( \rho_{4n_l}^2 +\sigma_{4n_l}^2\right) \,  \Big[-\frac{{p}}{2}\slashed{D}_3 r + \vartheta^-\left(\rho_{4n_l}\right) r \, tr \underline{H}+ \mathcal{O} \left(\frac{1}{r}\right) \Big] \nonumber \\     
+ \frac{1}{2}\int_{\mathcal{D}^{\tau_2}_{\tau_1}} dt^\star dr d\omega \, r^{p-1} \|\underline{\beta}_{4n_l}\|^2  \left[p \slashed{D}_4 r + 2\vartheta^+\left(\underline{\beta}_{4n_l}\right) r \,  tr H + \mu\right]   \leq  \nonumber \\  \Bigg|\int_{\mathcal{D}^{\tau_2}_{\tau_1}} e_{p}\left[\left(\rho_{4n_l}, \sigma_{4n_l}, \underline{\beta}_{4n_l}\right) \right] \Bigg| + B \cdot \mathbb{I}^{l+1,deg}_{R-M<r<R+M}\left[W\right]\left(\mathcal{M}\left({\tau_1,\tau_2}\right)\right) \nonumber \\ + \frac{1}{2}\int_{N_{out}\left(S^2_{\tau_1,R}\right)} dv \, d\omega_2  \left( \rho_{4n_l}^2 +\sigma_{4n_l}^2\right)\, ,
\end{align}
with error
\begin{align}
e_{p}\left[\left(\rho_{4n_l}, \sigma_{4n_l}\right), \underline{\beta}_{4n_l} \right] = r^{p-2} \left(1-\mu\right)^q \Big( \hat{E}_3^{4n_l}\left(\rho\right) \cdot \rho_{4n_l}
+\hat{E}_3^{4n_l}\left(\sigma\right)\cdot \sigma_{4n_l} \nonumber \\
+ \tilde{E}_4^{4n_l}\left(\underline{\beta}\right) \cdot \underline{\beta}_{4n_l} + \slashed{div} \left[ \left(\rho_{4n_l} ,\sigma_{4n_l} \right) \cdot \underline{\beta}_{4n_l} \right] \Big)
\nonumber
\end{align}
and a good left hand side for 
\begin{align}  \label{cond4rsbb}
6 + 4 \left(\textrm{number of $3$'s in $n_l$}\right) < {p} \leq 8 + 4 \left(\textrm{number of $4$'s in $n_l$}\right) \, .
\end{align}
Finally, using $\underline{\beta}_{34n_l}$ and $\underline{\alpha}_{44n_l}$
\begin{align} \label{4abbb}
\frac{1}{2} \int_{N_{out}\left(S^2_{\tau_2,R}\right)} dv \, d\omega_2 \, r^{p} \|\underline{\beta}_{4n_l}\|^2 + \frac{1}{2} \int_{\mathcal{I}^{\tau_2-R}_{\tau_1-R}} du \, d\omega_2 \, r^{{p}} \| \underline{\alpha}_{4n_l}\|^2 \nonumber \\
+ \int_{\mathcal{D}^{\tau_2}_{\tau_1}} dt^\star dr d\omega \,  \|\underline{\beta}_{4n_l}\|^2 \, r^{{p}-1} \Big[-\frac{{p}}{2}\slashed{D}_3 r + \vartheta^-\left(\underline{\beta}_{4n_l}\right) r \, tr \underline{H}+ \mathcal{O} \left(\frac{1}{r}\right) \Big] \nonumber \\     
+ \frac{1}{2}\int_{\mathcal{D}^{\tau_2}_{\tau_1}} dt^\star dr d\omega \, r^{p-1} \| \underline{\alpha}_{4n_l}\|^2  \left[-p \slashed{D}_4 r + 2\vartheta^+\left(\underline{\alpha}_{4n_l}\right)r \,  tr H + \mu\right]  \leq \nonumber \\ \Bigg|\int_{\mathcal{D}^{\tau_2}_{\tau_1}}  e_{p}\left[\underline{\beta}_{4n_l}, \underline{\alpha}_{4n_l} \right] \Bigg| + B \cdot \mathbb{I}^{l+1,deg}_{R-M<r<R+M}\left[W\right]\left(\mathcal{M}\left({\tau_1,\tau_2}\right)\right) \nonumber \\ + 2 \int_{N_{out}\left(S^2_{\tau_1,R}\right)} dv \, d\omega_2 \, r^{p} \|\underline{\beta}_{4n_l}\|^2
\end{align}
with error
\begin{align}
e_{p}\left[\underline{\beta}_{4n_l}, \underline{\alpha}_{4n_l} \right] = r^{p-2} \left(1-\mu\right)^q \left( \hat{E}_3^{4n_l}\left(\underline{\beta}\right) \cdot \underline{\beta}_{4n_l} + \frac{1}{2} \tilde{E}_4^{4n_l}\left(\underline{\alpha}\right)\cdot \underline{\alpha}_{4n_l} -2 \slashed{div} \left[\underline{\alpha}_{4n_l} \cdot \underline{\beta}_{4n_l} \right] \right) \nonumber
\end{align}
and a good left hand side for 
\begin{equation}  \label{cond4abbb}
8 + 4 \left(\textrm{number of $3$'s in $n_l$}\right) < {p} \leq 6 + 4 \left(\textrm{number of $4$'s in $n_l$}\right) .
\end{equation}
\subsubsection{The missing derivatives in the $3$-direction} \label{mixedder}
Controlling all $4$-derivatives is not quite sufficient to control all derivatives from the Bianchi equations: We will also need $\underline{\beta}_{3n_l}$ and $\underline{\alpha}_{3n_l}$. The first can be obtained by 
starting from the Bianchi equation for $\underline{\alpha}_{43n_l}$ and then using $\underline{\beta}_{34n_l}$, which yields
\begin{align} \label{4ab3bb}
\frac{1}{2} \int_{N_{out}\left(S^2_{\tau_2,R}\right)} dv \, d\omega_2 \, r^{p} \|\underline{\alpha}_{4n_l}\|^2 + \int_{\mathcal{I}^{\tau_2-R}_{\tau_1-R}} du \, d\omega_2 \, r^{p}\|\underline{\beta}_{3n_l}\|^2 \nonumber \\
+ \int_{\mathcal{D}^{\tau_2}_{\tau_1}} dt^\star dr d\omega \,  \, r^{{p}-1}\|\underline{\alpha}_{4n_l}\|^2 \Big[-\frac{{p}}{2}\slashed{D}_3 r + \vartheta^-\left(\underline{\alpha}_{3n_l}\right) tr \left(\underline{H}\right) r+ \mathcal{O} \left(\frac{1}{r}\right) \Big] \nonumber \\ 
+ \int_{\mathcal{D}^{\tau_2}_{\tau_1}}dt^\star dr d\omega \, r^{p-1} \|\underline{\beta}_{3n_l}\|^2 \Big[2\vartheta^+\left(\underline{\beta}_{4n_l}\right)r \, tr H  -p \slashed{D}_4 r  + \mu \Big] \leq \nonumber \\ \Bigg|\int_{\mathcal{D}^{\tau_2}_{\tau_1}} e_{p}\left[\underline{\alpha}_{4n_l}, \underline{\beta}_{3n_l}\right] \Bigg| + B \cdot \mathbb{I}^{l+1,deg}_{R-M<r<R+M}\left[W\right]\left(\mathcal{M}\left({\tau_1,\tau_2}\right)\right) \nonumber \\ + \frac{1}{2} \int_{N_{out}\left(S^2_{\tau_1,R}\right)} dv \, d\omega_2 \, r^{p} \|\underline{\alpha}_{4n_l}\|^2 
\end{align}
with error
\begin{align}
e_{p}\left[\underline{\alpha}_{4n_l}, \underline{\beta}_{3n_l}\right] = r^{p-2} \left(1-\mu\right)^q \Big(\tilde{E}_4^{3n_l} \left(\underline{\alpha}\right) \cdot \underline{\alpha}_{4n_l} + 2   \, \tilde{E}_3^{4n_l} \cdot \underline{\beta}_{3n_l} -2 \slashed{div} \left[ \underline{\alpha}_{4n_l} \cdot  \underline{\beta}_{3n_l} \right] \Big) \nonumber 
\end{align}
and positive left hand side for
\begin{equation}  \label{cond4ab3bb}
10 + 4 \left(\textrm{number of $3$'s in $n_l$}\right) < {p} \leq 4 + 4 \left(\textrm{number of $4$'s in $n_l$}\right) \, .
\end{equation}
Note however that only the upper bound is going to be relevant in view of the previous estimate (\ref{4abbb}) for the spacetime term of $\|\underline{\alpha}_{4n_l}\|^2$.
Finally, from the Bianchi equations for $\underline{\beta}_{33n_l}$ and $\underline{\alpha}_{43n_l}$
\begin{align} \label{3abbb}
\frac{1}{2} \int_{N_{out}\left(S^2_{\tau_2,R}\right)} r^{p} \|\underline{\beta}_{3n_l}\|^2  dv \, d\omega_2 \, +\frac{1}{4} \int_{\mathcal{I}^{\tau_2-R}_{\tau_1-R}}  du \, d\omega_2 \,  r^{p} \|\underline{\alpha}_{3n_l}\|^2 \nonumber \\
+ \int_{\mathcal{D}^{\tau_2}_{\tau_1}} dt^\star dr d\omega \,  \|\underline{\beta}_{3n_l}\|^2 \, r^{{p}-1} \Big[-\frac{{p}}{2}\slashed{D}_3 r + \vartheta^-\left(\underline{\beta}_{3n_l}\right) tr \left(\underline{H}\right) r + \mathcal{O} \left(\frac{1}{r}\right) \Big] \nonumber \\ 
+ \frac{1}{2} \int_{\mathcal{D}^{\tau_2}_{\tau_1}} dt^\star dr d\omega \,  \|\underline{\alpha}_{3n_l}\|^2  r^{p-1} \Big[2\vartheta^+\left(\underline{\alpha}_{3n_l}\right) r \, tr H  -p \slashed{D}_4 r  + \mu \Big] \nonumber \\  \leq \Big| \int_{\mathcal{D}^{\tau_2}_{\tau_1}} e_p\left[\underline{\beta}_{3n_l}, \underline{\alpha}_{3n_l}\right] \Big| + B \cdot \mathbb{I}^{l+1,deg}_{R-M<r<R+M}\left[W\right]\left(\mathcal{M}\left({\tau_1,\tau_2}\right)\right) \nonumber \\  + \frac{1}{2} \int_{N_{out}\left(S^2_{\tau_1,R}\right)} r^{p} \|\underline{\beta}_{3n_l}\|^2  dv \, d\omega_2
\end{align}
\begin{align}
e_{p}\left[\underline{\beta}_{3n_l}, \underline{\alpha}_{3n_l}\right] = r^{p} \left(1-\mu\right)^q \left( \tilde{E}_3^{3n_l}\left(\underline{\beta}\right) \cdot \underline{\beta}_{3n_l} + \frac{1}{2} \hat{E}_4^{3n_l}\left(\underline{\alpha}\right)\cdot  \underline{\alpha}_{3n_l} - \slashed{div} \left[ \underline{\alpha}_{3n_l} \cdot  \underline{\beta}_{3n_l} \right] \right)
\nonumber 
\end{align}
with the left hand side being positive for $r>R$ sufficiently large if
\begin{equation}  \label{cond3abbb}
12 + 4 \left(\textrm{number of $3$'s in $n_l$}\right) < p \leq 2 + 4 \left(\textrm{number of $4$'s in $n_l$}\right) \, .
\end{equation}
As usual, in view of (\ref{4ab3bb}) only the upper bound is going to be relevant.
\subsubsection{Angular momentum operators} \label{angremark}
The derivatives obtained so far are in principle sufficient to control all derivatives using the Bianchi equations only, just as we saw in section \ref{redsection}. However, for the optimal $r$-weights, we need to commute with angular momentum operators as well.

We first note that we can derive the identical $4$+$2$ estimates of sections \ref{4der} and \ref{mixedder} for the $\Omega_i$-commuted quantities. We will see the same constraints on the admissible $p$'s (as commutation with the $\Omega_i$ doesn't change the signature of the curvature components, cf.~section \ref{abrsabbb}), while the errorterms are replaced by
\begin{align} 
e_{p}\left[\Omega_i^j \alpha_{4n_l}, \Omega_i^j \beta_{4n_l}\right] = r^{p} \Big(\hat{E}_3^{4n_l\Omega_i^j}\left(\alpha\right) \cdot \Omega_i^j \alpha_{4n_l} + 2\tilde{E}_4^{4n_l \Omega_i^j} \left(\beta\right) \cdot \Omega_i^j \beta_{4n_l} \nonumber \\ 
+ 2\slashed{div} \left[ \Omega_i^j \alpha_{4n_l} \cdot  \Omega_i^j \beta_{4n_l} \right] \Big) \nonumber
\end{align}
and analogously for the other components. We also note
\begin{lemma} \label{pcc}
Given a covariant tensor $f$ tangent to the two-spheres $S^2_{t^\star,u}$ on $\mathcal{R}$ there exists a constant $c_0$ close to $1$ such that
\begin{equation}
 \frac{1}{c_0} \int_{S^2} r^2 | \slashed{\nabla} f|^2 d\mu_{\gamma} \leq \int_{S^2} | \slashed{\mathcal{L}}_{\Omega} f|^2 d\mu_{\gamma} \leq c_ 0 \int_{S^2} \left( |f|^2 + r^2 | \slashed{\nabla} f|^2 \right) d\mu_{\gamma}
\end{equation}
where we denote $| \slashed{\mathcal{L}}_{\Omega} f|^2=\sum_i |  \slashed{\mathcal{L}}_{\Omega_i} f|^2$.
Moreover, for $f$ a scalar we have
\begin{equation}
 \frac{1}{c_0} \int_{S^2} r^2 | \slashed{\nabla} f|^2 d\mu_{\gamma} \leq \int_{S^2} | \slashed{\mathcal{L}}_{\Omega} f|^2 d\mu_{\gamma} \leq c_ 0 \int_{S^2} \left(r^2 | \slashed{\nabla} f|^2 \right) d\mu_{\gamma}
\end{equation}
while for $f$ a one-form or a symmetric two-covariant traceless tensor tangent to the surfaces $S^2$ 
\begin{equation}
 \frac{1}{c_0} \int_{S^2} |f|^2 d\mu_{\gamma} \leq \int_{S^2} | \mathcal{L}_{\Omega} f|^2 d\mu_{\gamma} \, .
\end{equation}
\end{lemma}
This lemma allows us to gain improved decay in $r$ for the $\slashed{\nabla}$-derivatives as exhibited by the $\mathbb{E}\left[W\right]$-energies (cf.~(\ref{ennorm})).
\subsubsection{Retrieving the remaining derivatives}
Controlling the derivatives appearing in the estimates of sections \ref{4der} and \ref{mixedder}  suffices to control all remaining derivatives from the Bianchi equation (Cf.~Lemma \ref{allfrom3}):
\begin{lemma} \label{allfrom4}
For $j\geq 0$
\begin{align}
\int_{S^2} \Big( r^{q_2}\|\slashed{\nabla} \slashed{\mathcal{L}}_{\Omega_i}^j \alpha_{n_l}\|^2 + r^{q_3}\|\slashed{\nabla}  \slashed{\mathcal{L}}_{\Omega_i}^j \beta_{n_l} \|^2 + r^{q_4}|\slashed{\nabla} \slashed{\mathcal{L}}_{\Omega_i}^j \left( \rho,\sigma\right)_{n_l}|^2 + r^{q_5}\|\slashed{\nabla}  \slashed{\mathcal{L}}_{\Omega_i}^j\underline{\beta}_{n_l} \|^2 \nonumber \\ + r^{q_6}\|\slashed{\nabla}  \slashed{\mathcal{L}}_{\Omega_i}^j \underline{\alpha}_{n_l}\|^2 + r^{q_3}\| \slashed{\mathcal{L}}_{\Omega_i}^j \alpha_{3n_l}\|^2 + r^{q_4} \| \slashed{\mathcal{L}}_{\Omega_i}^j \beta_{3n_l}\|^2 + r^{q_5} |\slashed{\mathcal{L}}_{\Omega_i}^j\left(  \rho, \sigma\right)_{3n_l}|^2\Big) \leq 
\nonumber \\ 
B  \int_{S^2} \Big(r^{q_1}\|\slashed{\mathcal{L}}_{\Omega_i}^j{\alpha}_{4n_l}\|^2 + r^{q_2}\|\slashed{\mathcal{L}}_{\Omega_i}^j{\beta}_{4n_l}\|^2  + r^{q_3}\|\slashed{\mathcal{L}}_{\Omega_i}^j\left(\rho,\sigma\right)_{4n_l}\|^2  +r^{q_4} \|\slashed{\mathcal{L}}_{\Omega_i}^j\underline{\beta}_{4n_l}\|^2  \nonumber \\  + r^{q_5}\|\slashed{\mathcal{L}}_{\Omega_i}^j\underline{\alpha}_{4n_l}\|^2  +r^{q_6}  \|\slashed{\mathcal{L}}_{\Omega_i}^j\underline{\beta}_{3n_l}\|^2  + r^{q_7}\|\slashed{\mathcal{L}}_{\Omega_i}^j\underline{\alpha}_{3n_l}\|^2   \Big)  + B \int_{S^2} Low^{\Omega_i^j}_{n_l} \left(q_2, ... , q_6\right) \nonumber 
\end{align}
with
\begin{align}
Low_{n_l}^{\Omega_i^j} \left(q_2, ... , q_6\right) = 
 \Big[r^{q_2}  \| E_4^{n_l\Omega_i^j}\left(\beta\right)\|^2   + r^{q_3} \left(\| E_3^{n_l\Omega_i^j}\left(\alpha\right)\|^2  + \| E_4^{n_l\Omega_i^j}\left(\rho\right)\|^2 \right) \nonumber \\
+ r^{q_4} \left(\| E_3^{n_l\Omega_i^j}\left(\beta\right)\|^2 + \| E_4^{n_l\Omega_i^j}\left(\underline{\beta}\right)\|^2\right) \nonumber \\+ r^{q_5} \left(\| E_3^{n_l\Omega_i^j}\left(\rho,\sigma\right)\|^2 + \| E_4^{n_l\Omega_i^j}\left(\underline{\alpha}\right)\|^2 \right) + r^{q_6} \|E_3^{n_l\Omega_i^j}\left(\underline{\beta}\right)\|^2 \Big] \nonumber
\end{align}
\end{lemma}
\begin{proof}
This follows by integrating the commuted Bianchi equations over $S^2$ and integrating by parts. For instance, in the case $j=0$, from the bound on ${\rho}_{4n_l}$:
\begin{align}
\int_{S^2} \|\slashed{\nabla} \beta_{n_l}\|^2 = -\int_{S^2} \beta_{n_l} \slashed{\nabla}^2 \beta_{n_l} = \int_{S^2}  \beta_{n_l}  \mathcal{D}^\star_1 \mathcal{D}_1 \beta_{n_l} - K \| \beta_{n_l} \|^2  \nonumber \\ \leq \int_{S^2} \|\mathcal{D}_1 \beta_{n_l}\|^2 \leq \int_{S^2} \|\left(\rho,\sigma\right)_{4n_l}\|^2 +\int_{S^2} \|E_4^{n_l}\left(\rho\right)\|^2 + \|E_4^{n_l}\left(\sigma\right)\|^2
\end{align}
where $K$ is the Gauss curvature of the $S^2$. This in turn controls
\begin{align}
 \int_{S^2} \|\alpha_{3n_l}\|^2 &\leq 8 \int_{S^2} \|\mathcal{D}^\star_2 \beta_{n_l}\|^2 + 2 \int_{S^2} \|E_3^{n_l}\left(\alpha\right)\|^2 \nonumber \\ &\leq 8 \int_{S^2}  \beta_{n_l}\mathcal{D}_2 \mathcal{D}^\star_2 \beta_{n_l}+ 2 \int_{S^2} \|E_3^{n_l}\left(\alpha\right)\|^2 + 2 \int_{S^2} \|E_3^{n_l}\left(\alpha\right)\|^2 \nonumber \\
 &\leq -4 \int_{S^2} \beta_{n_l} \slashed{\Delta} \beta_{n_l} - 4 K \|\beta_{n_l}\|^2 + 2 \int_{S^2} \|E_3^{n_l}\left(\alpha\right)\|^2 \nonumber \\ & \leq  4 \int_{S^2} \|\slashed{\nabla} \beta_{n_l}\|^2 + 2 \int_{S^2} \|E_3^{n_l}\left(\alpha\right)\|^2 
\end{align}
In other words,
\begin{align}
 \int_{S^2} \|\alpha_{3n_l}\|^2 + \|\slashed{\nabla} \beta_{n_l}\|^2 \leq \nonumber \\
 8\Big(  \int_{S^2} \|\left(\rho,\sigma\right)_{4n_l}\|^2 + \|E_4^{n_l}\left(\rho\right)\|^2 + \|E_4^{n_l}\left(\sigma\right)\|^2 + \|E_3^{n_l}\left(\alpha\right)\|^2 \Big)
 \end{align}
 Similarly, we estimate the pair $\slashed{\nabla} \left(\rho_{n_l}, \sigma_{n_l}\right)$ and $\beta_{3n_l}$ from $\underline{\beta}_{4n_l}$, the pair $\slashed{\nabla} \underline{\beta}$ and $\left(\rho_{n_l3}, \sigma_{n_l3}\right)$ from $\underline{\alpha}_{4n_l}$, and finally
\begin{align}
\int_{S^2} \|\slashed{\nabla} \underline{\alpha}_{n_l}\|^2 = -\int_{S^2} \underline{\alpha}_{n_l} \slashed{\nabla}^2 \underline{\alpha}_{n_l} = \int_{S^2}  \underline{\alpha}_{n_l}  \mathcal{D}^\star_2 \mathcal{D}_2 \underline{\alpha}_{n_l} - 2K \| \underline{\alpha}_{n_l} \|^2  \nonumber \\ \leq \int_{S^2} \|\mathcal{D}_2 \underline{\alpha}_{n_l}\|^2 \leq \int_{S^2} \|\underline{\beta}_{3n_l}\|^2 +\int_{S^2} \|E_3^{n_l}\left(\underline{\beta}\right)\|^2
\end{align}
\end{proof}

\subsubsection{The summed asymptotic estimate}
The claim is that we can add up the asymptotic estimates to estimate all derivatives at order $k+1$ with the weights appearing in the $\mathbb{E}$-energy. Before we do that let us note the following lemma, which establishes that permutations of the $n_l$ are lower order.
\begin{lemma} \label{permcommute}
Let $u \in \{\alpha,\beta, \widehat{\rho}, \sigma, \underline{\beta}, \underline{\alpha}\}$ be a curvature component and $n_p$ a $p$-tuple of $3$'s and $4$'s. We have
\begin{align}
\|u_{n_p}\|^2 \leq \|u_{perm(n_p)}\|^2 + \frac{B}{r^4} \sum_{n_{p-2}} \left[\|\slashed{D}_3 u_{n_{p-2}}\|^2 + \|\slashed{D}_4 u_{n_{p-2}}\|^2  + \|\slashed{\nabla} u_{n_{p-2}}\|^2 + \|u_{n_{p-2}}\|^2\right] \nonumber
\end{align}
where the sum is over all $p-2$ tuples containing one less $3$ and one less $4$ then the original $n_p$. In addition
\begin{align}
\|u_{n_p4}\|^2 \leq \|\slashed{D}_4 u_{n_p}\|^2 + \frac{B}{r^2} \| u_{n_p}\|^2 \textrm{ \ \ \ and \ \ \ } \|u_{n_p3}\|^2 \leq \|\slashed{D}_3 u_{n_p}\|^2 + \frac{B}{r^2} \| u_{n_p}\|^2 \nonumber
\end{align}
\end{lemma}
\begin{proof}
The second statement is immediate from the definition.
The first statement follows by induction using Lemma \ref{commutelemma} and the pointwise decay of the Ricci-coefficients. 
\end{proof}
Proposition \ref{hiinf} now follows easily from the next Proposition, which reduces the problem at order $l+1$ to the problem at order $l$.

\begin{proposition} \label{zwischri}
Let $p= \left(p_1, ... , p_6, 0\right)$ be a boundary admissible tuple with associated bulk admissible tuple $\tilde{p} = \left(\tilde{p}_1, ... , \tilde{p}_7\right)$. For $l\geq 0$ we have
\begin{align}
\mathbb{E}^{l+1}_{p} \left[W\right] \left(N_{out}\left(S^2_{\tau_2,R}\right)\right) + \mathbb{I}^{l+1, deg}_{\tilde{p}} \left[W\right] \left(\mathcal{D}^{\tau_2}_{\tau_1}\right)  \leq B \cdot  \mathbb{E}^{l+1}_{p}  \left[W\right] \left(N_{out}\left(S^2_{\tau_1,R}\right)\right) \nonumber \\ + B \cdot \overline{\mathbb{I}}^{l+1,deg}_{R-M<r<R+M}\left[W\right]\left(\mathcal{M}\left({\tau_1,\tau_2}\right)\right) 
 + B \cdot Err^{l+1}_{P}  + B \cdot Low^{l+1}_P
\nonumber \\
+ B \sum_{i=1}^k  \left(\mathbb{E}^{i}_{p} \left[W\right] \left(N_{out}\left(S^2_{\tau_2,R}\right)\right) + \mathbb{I}^{i, deg}_{\tilde{p}} \left[W\right] \left(\mathcal{D}^{\tau_2}_{\tau_1}\right)\right) \nonumber \, .
\end{align}
\end{proposition}
We remark that the last line estimates lower order terms arising from commutation of the ordering of derivatives. It is absent for $l=0$.
\begin{proof}
We start with the tuple of length $l$ (denoted $n_l$, with $l_1$ being the number of $3$'s, $l_2$ the number of $4$'s) consisting of all $4$'s (i.e.~$l_1=0$ and $l_2=l$) and apply the estimates as follows: 
\begin{itemize}
\item (\ref{4ab}) with $p=p_1+2l_2$
\item (\ref{4brs}) with $p=p_2+2l_2$
\item (\ref{4rsbb})with $p=p_3+2l_2$  and -- in case that $p_3=8$ and $l_2=0$ -- in addition with $p=8-\delta$
\item  (\ref{4abbb}) with $p=p_4+2l_2$ and -- in case that $p_4=6$ and $l_2=0$ -- in addition with $p=6-\delta$
\item(\ref{4ab3bb}) with $p=p_5+2l_2$ and -- in case that $p_5=4$ and $l_2=0$ -- in addition with $p=4-\delta$
\item (\ref{3abbb}) applied with $p=p_6+2l_2$ and -- in case that $p_6=2$ and $l_2=0$ -- in addition with $p=2-\delta$
\end{itemize}
{\bf Remark: } The reason for the additional application with a $\delta$-loss in the last four estimates is caused by the fact that, as we saw, in this ``extreme" case, the spacetime term loses an additional power compared to the boundary term, i.e.~instead of the usual weight $p_4-1$ for the spacetime term, it will only admit $p_4-2$. Applying the estimate also with a delta loss ensures that we get the spacetime term with an improved weight of $p_4-1-\delta$. (We also obtain weaker $r^{6-\delta}$-weighted positive boundary terms in this process, which we simply discard.)

We next add the estimates above, so that all spacetime and all boundary terms on the left hand side are positive: This is already automatic for the first three estimates in view of the constraints (\ref{cond4ab}), (\ref{cond4brs}), (\ref{cond4rsbb}) being satisfied. For estimate (\ref{4abbb}) we observe that the possibly negative (if $n_l$ is the zero tuple) signed spacetime term containing $\|\underline{\beta}_{4n_l}\|^2$ can be absorbed by adding enough of the previous estimate (\ref{4rsbb}) which controls this term. This uses that the tuple is admissible: The spacetime term of the second quantity of the $i^{th}$ estimate has weight $r^{p_i+2l_2-1}$ ($r^{p_i-1-\delta}$ in the extreme case) while the spacetime-term of the first quantity in the $\left(i+1\right)^{th}$ estimate has weight $r^{p_{i+1}+2 l_2-1}$. Since $p_i \geq p_{i+1}$ holds for an admissible tuple ($p_i \geq p_{i+1} +2$ in the extremal case), one can always control the first spacetime term of the $\left(i+1\right)^{th}$ estimate from the second of the $i^{th}$.

Next, an application of Lemma \ref{allfrom4} allows us to control $\|Du_{n_l}\|^2$ where $u$ is any curvature component, $D$ is either a $3$-, $4$- or an angular derivative and $n_l$ is the tuple consisting of only $4$-derivatives. In particular, the spacetime term of $\|\alpha_{34...4}\|^2$ is controlled with the same weight as $\|\rho_{44...4}\|^2$, namely $r^{p_2+2l-1}$. Modulo lower order terms, which are covered by Lemma \ref{permcommute}, this is equivalent to controlling $\|\alpha_{4\tilde{n}_l}\|^2$ for tuples $\tilde{n}_l$ which have $l_2=l-1$ and $l_1=1$. Hence we can iterate the procedure and start with the estimate (\ref{4ab}) again. Now only the upper bound in (\ref{cond4ab}) is relevant, as we already control the spacetime-term of $\|\alpha_{4\tilde{n}_l}\|^2$ from the previous step. Since for an admissible tuple $p_1+2l_2-1 \leq p_2 + 2\left(l_2+1\right)-1$, we can apply (\ref{4ab}) with $p=p_1+2l_2$ and then step down: (\ref{4brs}) is applied with $p=p_2+2l_2$, (\ref{4rsbb}) applied with $p=p_3+2l_2$, (\ref{4abbb}) applied with $p=p_4+2l_2$, (\ref{4ab3bb}) applied with $p=p_5+2l_2$ and (\ref{3abbb}) applied with $p=p_6+2l_2$. Moreover, we can add these estimates so that all spacetime and boundary terms are positive. Using once again Lemma \ref{allfrom4}, we now control all derivatives of the form $Du_{n_l}$, where now $n_l$ is a tuple consisting of either all $4$'s, or all but one entries being $4$'s. In particular, we control the spacetime term of $\|\alpha_{334...4}\|^2$, which is, modulo lower order terms, equivalent to controlling $\|\alpha_{4\hat{n}_l}\|^2$ for tuples $\hat{n}_l$ satisfying $l_2=l-2$, $l_1=2$. We hence reiterate our estimates....

Using this stepping down procedure we will eventually control all derivatives of the form $Du_{n_l}$ where $n_l$ is \emph{any} tuple of $3$'s and $4$'s. At this step one can bring in the wave character of the Bianchi equations ($u_{34} - \slashed{\nabla}^2 u = \textrm{l.o.t}$ holds for any curvature component $u$) to estimate \emph{all} derivatives. However, there is a drawback: This procedure will only improve the $r$-weight in the energy by a power of one per angular derivative, since a $34$ derivative pair gains only $r^{0+2}$ in terms of weights. To prove the full decay, one has to bring in angular momentum operators: We apply the estimates of section \ref{4der} and \ref{mixedder} to the $\Omega_i$ commuted equations (for which they are valid by the remarks of section \ref{angremark}). More precisely, to any arbitrary fixed tuple of length $j$ we apply $l-j$ $\Omega_i$-derivatives and redo the algorithm we outlined above. This produces precisely the error-terms collected in $Err_P^{l+1}$.
\end{proof}
\subsection{Controlling the error-terms} \label{caworst}
We now prove the analogue of the error-estimate of Proposition \ref{erlosto}:
\begin{proposition} \label{errhigh}
For a boundary-admissible matrix $P$ with second line being the tuple $p=\left(p_1,p_2, ... , p_6, 0\right)$ we have, for any $\lambda>0$ and $m\geq 0$
\begin{align}
\overline{Err}^{m+1}_{P} \leq 
 \lambda \cdot \overline{\mathbb{I}}^{m+1}_{\tilde{P}} \left[W\right] \left(\mathcal{D}\left(\tau_1,\tau_2\right)\right)
+ B \cdot \overline{\mathbb{I}}^{m,deg}_{R-M<r<R+M}\left[W\right]  \left(\tilde{\mathcal{M}} \left(\tau_1,\tau_2\right)\right)\nonumber \\
 + B_{\lambda} \left[
\sup_{\tau \in \left(\tau_1,\tau_2\right)}\overline{\mathbb{E}}^{m+1} \left[\mathfrak{R}\right] \left(N_{out}\left(\tau,R\right)\right)  + \overline{\mathbb{I}}^{m+1} \left[\mathfrak{R}\right] \left(\mathcal{D} \left(\tau_1,\tau_2\right)\right) \right] 
\end{align}
and
\begin{align} \label{sekt}
\overline{Low}^{m+1}_P \leq B_k \sum_{i=1}^{max(1,m)} \left(\mathbb{E}^{i}_{p} \left[W\right] \left(N_{out}\left(S^2_{\tau_2,R}\right)\right) + \mathbb{I}^{i, deg}_{\tilde{p}} \left[W\right] \left(\mathcal{D}^{\tau_2}_{\tau_1}\right)\right) \nonumber \\ + B  \cdot \left(\overline{\mathbb{E}}^{m} \left[\mathfrak{R}\right] \left(N_{out}\left(S^2_{\tau_2,R}\right)\right) + \overline{\mathbb{I}}^{m, deg} \left[\mathfrak{R}\right] \left(\mathcal{D}^{\tau_2}_{\tau_1}\right)\right) \, ,
\end{align}
where in case that $m=0$, $B_0=\epsilon_R$ is a small constant (which gets smaller as $R$ is chosen larger).
Generally, the constants $B$ depend on the mass and lower order energies of the Weyl-curvature, which are bounded by the ultimately Schwarzschildean property.
\end{proposition}

Before we turn to the proof let us outline the main idea. Inspecting the structure of $Err_P^{m+1}$ we see that the error-terms consist of derivatives of the inhomogeneities $E_{3,4}\left(\alpha,...,\underline{\alpha}\right)$ in the Bianchi equations. As is manifest in our energies used, taking a $4$-derivative or an angular derivative improves the (pointwise) decay in $r$ by one, while a $3$-derivative does not change the $r$-weight. Hence the higher derivative estimates for the error would be a completely straightforward generalization of Proposition \ref{erlosto}, if it was not for the fact that we are actually claiming stronger decay in $r$ for some higher derivatives than the naive improvement above provides. For instance, we applied the estimate (\ref{st1}) with $p_2=7-\delta$, while we applied (\ref{4ab}) with $p=10-\delta$ instead of the naive $9-\delta$ which would follow from the above reasoning. In order for this to work, there has to be a cancellation of the slowest decaying terms, which we are going to unravel in the proof. Note that this special improvement only occurs from the $0^{th}$ to the $1^{st}$ order, as for the second order, we claim no additional improvement. This is immediate from the structure of the decay matrices: All tuples following the second line are all identical to the second.

The cancellations are apparent by doing the following computation:
\begin{lemma}
In the standard null-frame for which $\boxed{Y=0}$, we have the following expression for the error-terms
\begin{align} 
E_3^4 \left(\alpha\right) = 4\underline{\Omega} \alpha_4 + 4 \left(\slashed{D}_4 \underline{\Omega}\right) \alpha - 4 \widehat{H} \rho_4 - 4 {}^\star \widehat{H} \sigma_4 + \widehat{H} E_4\left(\rho\right) + {}^\star \widehat{H} E_4\left(\sigma\right) \nonumber \\
+ 2\widehat{s}\left(F_4 \left(\beta\right)\right) - 3 \rho \left(-2\Omega \widehat{H} - \alpha\right) + \left(V+4Z\right) \widehat{\otimes} \beta_4 \nonumber \\
\left(-2\slashed{\nabla} tr H + \frac{1}{2} tr H \left(V+4Z\right) + \slashed{D}_4 \left(V+4Z\right)\right) \widehat{\otimes} \beta 
\end{align}
\begin{align} \label{hate34}
\hat{E}_3^4\left(\alpha\right) = E_3^4 \left(\alpha\right) + F_{34}\left(\alpha\right) + \alpha \Big[tr H tr \underline{H} + \Omega tr H - 5\underline{\Omega} tr H - 4\rho \nonumber \\ + 2 \widehat{H}\cdot \widehat{\underline{H}} + \slashed{div}  \underline{Z} - 5 \slashed{div} Z + \underline{Z} \cdot \underline{Z} - 5 Z \cdot Z \Big]
\end{align}
\begin{align}
E_3^4 \left({\beta}\right) = 2\widehat{H} \cdot \underline{\beta}_4 + 2\underline{\Omega} \underline{\beta}_4 + \underline{Y} \alpha_4 + 3 Z \rho_4 + 3{}^\star Z \sigma_4 +\left[\slashed{D}_4 \underline{Y} - \frac{1}{2} tr H \underline{Y} \right] \alpha \nonumber \\
+ \underline{\beta} \left[-4\Omega \widehat{H} - 2\alpha   \right] + 2 \slashed{D}_4 \underline{\Omega} \beta
 + 3\left[\slashed{D}_4 {}^\star Z - \frac{1}{2} tr H {}^\star Z\right] \sigma + {}^\star F_3\left(\sigma\right)  \nonumber \\ + 3\left[\slashed{D}_4 Z - \frac{1}{2} tr H Z \right] \rho + F_3\left(\rho\right)  - \frac{3}{2} \slashed{\nabla} tr H \rho - \widehat{H} \slashed{\nabla} \rho - \frac{3}{2} {}^\star \slashed{\nabla} tr H \sigma - {}^\star\widehat{H}  \slashed{\nabla} \sigma \nonumber
 \end{align}
 \begin{align}
E_4^4 \left(\beta\right) = -2\Omega \beta_4 - \left(2\slashed{D}_4 \Omega +2\Omega tr  H\right) \beta + \alpha \slashed{D}_4 \left(2V+\underline{Z}\right) + \left(2V+\underline{Z}\right) \alpha_4  \nonumber \\ + \frac{1}{2} tr H \left(2V + \underline{Z}\right) \alpha  - \frac{5}{2} \slashed{\nabla} tr H \cdot \alpha - \widehat{H} \cdot \slashed{\nabla} \alpha + tr F_4\left(\alpha\right)
\end{align}
\begin{align}
E_3^4\left(\rho\right) = -\frac{1}{2}  \widehat{H} \underline{\alpha}_4 - \frac{1}{2} \left(-2\Omega \widehat{H} - \alpha \right) \underline{\alpha} + \underline{\beta} \cdot \slashed{D}_4 \left(V-2Z\right) + 2\beta \cdot \slashed{D}_4 \underline{Y} \nonumber \\ + \left(V-2Z\right) \cdot \underline{\beta}_4 + 2\underline{Y} \beta_4 + \frac{1}{2} tr H \left(V-2Z\right) \cdot \underline{\beta} - tr H \underline{Y} \beta \nonumber \\ + tr H \underline{\beta} \cdot \slashed{\nabla} tr H + \widehat{H} \slashed{\nabla} \underline{\beta} - tr F_4\left(\underline{\beta}\right)
\end{align}
 \begin{align}
E_4^4 \left(\rho\right) = -\frac{1}{2} \underline{H} \alpha_4 - \frac{1}{2} \slashed{D}_4 \widehat{\underline{H}} + \left(V+2\underline{Z}\right) \cdot \beta_4 \nonumber \\ 
+ \left[ \slashed{D}_4 \left(V+2\underline{Z}\right) + \frac{1}{2}\left(V+2\underline{Z}\right)\right] \beta - 2 \beta \slashed{\nabla} tr H - \widehat{H} \cdot \slashed{\nabla} \beta + tr F_4 \left(\beta\right)
\end{align}
\begin{align}
E_4^4\left(\underline{\beta}\right) = 2\Omega \underline{\beta}_4 + 2\widehat{\underline{H}}\beta_4 + 2\beta \slashed{D}_4 \widehat{\underline{H}}   + 2\underline{\beta}\left(\slashed{D}_4 \Omega + tr H \Omega\right) \nonumber \\ -3 \underline{Z} \rho_4 - 3 {}^\star \underline{Z} \sigma_4  -3\rho \slashed{D}_4 \underline{Z} - 3 \sigma \slashed{D}_4 {}^\star \underline{Z} + \frac{1}{2} tr H \left(-3\underline{Z} \rho - 3 {}^\star \underline{Z} \sigma \right) \nonumber \\ +\frac{3}{2} \rho \slashed{\nabla} \left(tr H\right) - \frac{3}{2} \sigma^\star \slashed{\nabla} \left(tr H\right) + \widehat{H} \cdot \slashed{\nabla}\rho -  {}^\star \widehat{H} \cdot \slashed{\nabla}\sigma - F_4\left(\rho\right) - {}^\star F_4\left(\sigma\right)
\end{align}
\begin{align}
E_4^4 \left(\underline{\alpha}\right) = 4\Omega \underline{\alpha}_4 - 3\widehat{\underline{H}}\rho_4 - 3{}^\star \widehat{\underline{H}} \sigma_4 + \left(V-4\underline{Z}\right) \widehat{\otimes} \underline{\beta}_4 -3\rho \slashed{D}_4 \widehat{\underline{H}} -3\sigma \slashed{D}_4 {}^\star\widehat{\underline{H}}\nonumber \\ + \underline{\alpha} \left(4\slashed{D}_4 \Omega + 4 tr H \Omega\right) + \left[\slashed{D}_4 \left(V-4Z\right) \right] \widehat{\otimes} \underline{\beta} - \frac{1}{2} tr H \left(V-4Z\right) \widehat{\otimes} \underline{\beta} \nonumber \\ \widehat{H} \slashed{div} \underline{\beta}  - {}^\star \widehat{H} \slashed{curl} \underline{\beta} + \left(\slashed{\nabla} tr H\right) \widehat{\otimes} \underline{\beta} - 2 \widehat{s}\left(F_4 \left(\underline{\beta}\right)\right)
\end{align}
\begin{align}
E_4^3 \left(\underline{\alpha}\right) = 4\Omega \underline{\alpha}_3 - 3\widehat{\underline{H}}\rho_3 - 3{}^\star \widehat{\underline{H}} \sigma_3 + \left(V-4\underline{Z}\right) \widehat{\otimes} \underline{\beta}_3  + 4\underline{\alpha} \slashed{D}_3 \Omega  - 2 \widehat{s}\left(F_3 \left(\underline{\beta}\right)\right)  \nonumber \\ -3 \left[\slashed{D}_3 {}^\star \widehat{\underline{H}} + tr \underline{H}{}^\star \widehat{\underline{H}} \right]  \sigma  + \left[\slashed{D}_3 \left(V-4\underline{Z}\right) + \frac{1}{2} tr \underline{H} \left(V-4\underline{Z}\right) \right] \widehat{\otimes} \underline{\beta} \nonumber \\ -3 \left[\slashed{D}_3 \widehat{\underline{H}} + tr \underline{H}\widehat{\underline{H}} \right] \rho + \underline{\widehat{H}} \slashed{div} \underline{\beta}  - {}^\star \underline{\widehat{H}} \slashed{curl} \underline{\beta} + 2\left(\slashed{\nabla} tr \underline{H}\right) \widehat{\otimes} \underline{\beta}
\end{align}
\end{lemma}
\begin{proof}
A lengthy but straightforward computation. 
\end{proof}

\begin{remark}
The remarkable thing is that inserting the structure equation for $\widehat{H}$, (\ref{H4}), which in our gauge ($Y=0$) reads
\begin{align}
\slashed{D}_4 \widehat{H}_{AB} + tr H \widehat{H}_{AB} = -2 \Omega \widehat{H}_{AB} - \alpha_{AB},
\end{align}
always cancels the terms with the worst decay. For instance, $E_3\left(\rho\right)$ only decays like $r^{-3}$ pointwise, while $E_3^4\left(\rho\right)$ decays like $r^{-5}$, while the naive improvement would only suggest $r^{-4}$. Similarly, while $E_3\left(\alpha\right)$ decays only like $r^{-5}$ (the $\rho \hat{H}$-term), $E^4_3\left(\alpha\right)$ decays like $r^{-\frac{13}{2}}$.
\end{remark}

\begin{proof}[Proof of Proposition \ref{errhigh}]
We turn to Proposition \ref{hiinf} to estimate the expression for $Err_P^{m+1}$.
For $j<m$ and some $0 \leq l \leq j$, the first term we have to estimate is (recall that $l_2$ denotes the number of $4$'s in a tuple of lenght $l$)
\begin{align} \label{hula}
e_{p_1+2l_2}\left[\Omega_i^{j-l} \alpha_{4n_l}, \Omega_i^{j-l}\beta_{4n_l}\right] \nonumber \\ 
= r^{p_1+2l_2} \left(\hat{E}_3^{4n_l\Omega_i^{j-l}}\left(\alpha\right) \cdot \slashed{\mathcal{L}}_{\Omega_i}^{j-l}\alpha_{4n_l} + 2 \tilde{E}_4^{4n_l\Omega_i^{j-l}} \left(\beta\right) \cdot \slashed{\mathcal{L}}_{\Omega_i}^{j-l} \beta_{4n_l}  \right) \, .
\end{align}
We conclude that it suffices to understand the case $j=0, l=0$. This is a consequence of the structure of the errorterms $E$, $\tilde{E}$, $\hat{E}$ and the fact that we only have to worry about the number of derivatives involved and about the amount of decay in $r$. Since the $E$ are \emph{quadratic} in Ricci-coefficients and curvature, the number of derivatives is not a problem. Moreover, since taking a $4$-derivative improves the $r$-decay of both Ricci coefficients and curvature components by a power of $2$ in the energy (i.e.~wherever the derivative falls on), while taking a three or an $\Omega_i$-derivative does not change the decay, the additional weight factor of $r^{l_2}$ in (\ref{hula}) for each additional $4$-derivative is naturally incorporated.

We hence turn to the case $j=0, l=0$, starting with the first of $7$ terms, $e_{p_1}\left[\alpha_{4n_k}, \beta_{4n_k}\right]$, which is in turn is written out explicitly below (\ref{4ab}). For a boundary admissible matrix $P$ with second row $p=(p_1,p_2, ...., p_6, 0)$ and associated bulk admissible matrix $\tilde{P}$ with second row $\tilde{p} =(p_1-1,p_1-1,p_2-1, ...., p_5-1, p_6-1)$:
\begin{align}
\int_\mathcal{D} \left(\|E_3^4 \left(\alpha\right)\|^2 + \| F_{34}\left(\alpha\right)\|^2\right) r^{p_1+1} \leq |r^4 \underline{\Omega}^2 |  \int_\mathcal{D} dt^\star dr d\omega \, \|\alpha_4\|^2 r^{p_1-3} \nonumber \\ +  \|r^4\left( \widehat{H}, V, ... \right)^2\|_{L^\infty} \int_\mathcal{D} dt^\star dr d\omega \Big[\left( \|\beta_4\|^2 + |\left(\rho,\sigma\right)_4|^2\right) r^{p_1-3}  \nonumber \\ + \left( E_4\left(\rho\right)^2 + E_4 \left(\sigma\right)^2 \right) r^{p_1-3} \Big] + \|\rho^2 \Omega^2 r^{10} \|_{L^\infty}  \int_{\mathcal{D}} dt^\star dr d\omega \| \widehat{H}\|^2 r^{p_1-9} \nonumber \\  + \|\rho^2 r^{6} \|_{L^\infty}  \int_\mathcal{D} dt^\star dr d\omega \, \|\alpha\|^2 r^{p_1-7} + \nonumber \\
+ \| r^7 | \beta|^2 \|_{L^\infty}  \int_\mathcal{D}dt^\star dr d\omega \, r^{p_1-6}  \left(\|\slashed{\nabla} tr H\|^2 + \|\slashed{D}_4 \left(V+4Z\right)\|^2 \right) \nonumber \\  \leq
\epsilon_R \cdot \overline{\mathbb{I}}^{k+1}_{\tilde{P}} \left[W\right] \left(\mathcal{D}\right) + 
\epsilon \cdot \overline{\mathbb{I}}^{1} \left[\mathfrak{R}\right] \left(\mathcal{D}\right) + B\cdot {\mathbb{I}}^{0} \left[\mathfrak{R}\right] \left(\mathcal{D}\right)  \nonumber
\end{align}
the last step following as long as $p_1 \leq 10-\delta$.
Remarkably, the last term in the expression
\begin{align}
\int_\mathcal{D} \hat{E}_3^4 \left(\alpha\right) \cdot \alpha_4 r^{p_1} = \int_{\mathcal{D}} \left(E_3^4 \left(\alpha\right) + F_{34}\left(\alpha\right) + \alpha\left[...\right] \right)  \cdot \alpha_4 r^{p_1} \, ,
\end{align}
where $[...]$ denotes the square bracket in expression (\ref{hate34}), actually generates a good lower order term. The point is that the leading order term in $[...]$ is $-\frac{4}{r^2}$ at infinity (arising from the $tr H tr \underline{H}$-term), while all the others decay like $\frac{1}{r^3}$ and are easily estimated as before. We write
\begin{align}
\alpha\left[...\right]  \cdot \alpha_4 \, r^{p_1} = \frac{1}{2} \slashed{D}_4 \left(r^{p_1} \|\alpha\|^2 \left[...\right]  \right) + \|\alpha\|^2 \left(-\frac{1}{2} \slashed{D}_4 \left(\left[ ... \right] r^{p_1}\right) + \frac{5}{2} tr H \left[ ... \right] r^{p_1} \right) \nonumber 
\end{align}
and find, using the structure equations and the pointwise estimates for the Ricci-coefficients that a negative boundary and spacetime-term is generated, as long as $p_1<12$ (which consequently have a good sign when brought to the left hand side). Besides an improved weight for $\alpha$ on $\mathcal{I}$, we gain a spacetime term $\int dt^\star dr d\omega r^{p_1-3} \|\alpha\|^2$, which in view of $p_1<10-\delta$ considerably improves the estimate from the uncommuted equations (where we obtained only control over $\int dt^\star dr d\omega r^{6-\delta} \|\alpha\|^2$).\footnote{An even stronger improvement follows by a different argument. Namely, since we control the component $|\beta_4|^2$ in the main term of this estimate (\ref{4ab}), we also control $|\slashed{\nabla} \alpha|^2$ with the same $r$-weight from the Bianchi equation (cf.~Lemma \ref{allfrom4}). The Poincare inequality $\int_{S^2} \|\alpha\|^2 \leq \int_{S^2} r^2 \| \slashed{\nabla} \alpha\|^2$ then generates the improved decay of $\alpha$.} We summarize
\begin{align}
\int_\mathcal{D} \hat{E}_3^4 \left(\alpha\right) \cdot \alpha_4 r^{p_1} \leq - b \int_\mathcal{D} dt^\star dr d\omega r^{p_1-3} \|\alpha\|^2 + B \cdot \overline{\mathbb{I}}^{0,deg}_{R-M<r<R+M}\left[W\right]\left(\mathcal{M}\left({\tau_1,\tau_2}\right)\right)  \nonumber \\ \epsilon_R \cdot {\mathbb{I}}^{1}_{\tilde{p}} \left[W\right] \left(\mathcal{D}\right) + 
\epsilon \cdot \overline{\mathbb{I}}^{1} \left[\mathfrak{R}\right] \left(\mathcal{D}\right) + B\cdot {\mathbb{I}}^{0} \left[\mathfrak{R}\right] \left(\mathcal{D}\right) \nonumber \, .
\end{align}
Similarly (note $\tilde{E}_4^4\left(\beta\right)=E_4^4\left(\beta\right)$, in view of $\mathcal{C}_{44}\left[\beta\right]=0$), 
\begin{align}
\int_\mathcal{D} \tilde{E}_4^4 \left(\beta\right) \cdot \beta_4 r^{p_1} \leq  \epsilon_R \cdot {\mathbb{I}}^{1}_{\tilde{p}} \left[W\right] \left(\mathcal{D}\right)  + 
\epsilon \cdot \overline{\mathbb{I}}^{1} \left[\mathfrak{R}\right] \left(\mathcal{D}\right) + B\cdot {\mathbb{I}}^{0} \left[\mathfrak{R}\right] \left(\mathcal{D}\right) \nonumber \, .
\end{align}
We remark that it is crucial here that $Y=0$ and hence $F_4$ does not introduce $3$-derivatives (cf. Lemma \ref{commutelemma}).
The other error-terms are dealt with analogously with no further difficulties involved. To give one more (interesting) example, we consider the worst part of the error-term
\begin{align} \label{rhoexmple}
\int dt^\star dr d\omega \ r^{p_3} \hat{E}^4_3 \left(\rho\right) \cdot \rho_4 + r^{p_3} \hat{E}^4_3 \left(\sigma\right) \cdot \sigma_4 + r^{p_3} \tilde{E}^4_4 \left(\underline{\beta}\right) \underline{\beta}_4 \, ,
\end{align}
which arises from the $\left(\left(\rho, \sigma\right)_4,\underline{\beta}_4\right)$-pair of estimates (cf.~(\ref{4rsbb})).
The leading order contribution from the first term is estimated (provided $p_4$ is not equal to $6$)
\begin{align}
\int dt^\star dr d\omega \ r^{p_3} \hat{H} \underline{\alpha}_4 \rho _4 \leq B_{\lambda} \int dt^\star dr d\omega \|\hat{H}\|^2 r^4 \|\underline{\alpha}_4\|^2 r^{p_4-1} \nonumber \\ + \lambda \int dt^\star dr d\omega |\rho_4|^2 r^{2p_3-p_4-3} \leq 
\left(B_{\lambda} \epsilon + \lambda\right) \cdot \mathbb{I}^1_{\tilde{p}} \left[W\right] \left(\mathcal{D}\right) \, ,
\nonumber
\end{align}
with the last step following provided $2p_3-p_4-3 \leq p_2-1$, which holds for a boundary admissible tuple. If $p_4=6$, a similar computation leads to the condition $\left(p_3-p_2\right) + p_3 \leq 8 - \delta$, which is also valid. For the $\beta$-term in (\ref{rhoexmple}), a typical term is estimated (say $p_3 \neq 8$)
\begin{align}
\int dt^\star dr d\omega \ r^{p_3} \underline{\hat{H}} {\beta}_4 \underline{\beta} _4 \leq B_{\lambda} \int dt^\star dr d\omega \|\hat{\underline{H}}\|^2 r^2 \|\underline{\beta}_4\|^2 r^{p_3-1} \nonumber \\ + \lambda \int dt^\star dr d\omega |\beta_4|^2 r^{p_3-1} \leq 
\left(B_{\lambda} \epsilon + \lambda\right) \cdot \mathbb{I}^1_{\tilde{p}} \left[W\right] \left(\mathcal{D}\right)
\nonumber
\end{align}
requiring $p_3 \leq p_1$, which is true for a boundary admissible tuple. For $p_3=8$ we obtain $8 + \delta \leq p_1$. 

The estimate for (\ref{sekt}) is straightforward. Again one can convince oneself that it suffices to understand $j=0$ and $l=0$. In this case, for the first term (cf.~Lemma \ref{allfrom4})
\begin{align}
\int_{N_{out}\left(S^2_{\tau_2,R}\right)} \left(E_4 \left(\beta\right)^2\right) r^{p_2} dv d\omega \leq | 4\Omega^2 r^4| \  \int_{N_{out}\left(S^2_{\tau_2,R}\right)} |\beta|^2 r^{p_2-4} dv d\omega \nonumber \\
+  | \left(V,\underline{Z}\right) r^4| \  \int_{N_{out}\left(S^2_{\tau_2,R}\right)} |\alpha|^2 r^{p_2-4} dv d\omega \leq \epsilon_R \cdot E^1_{p} \left[W\right]  \left(N_{out}\left(S^2_{\tau_2,R}\right)\right) \, ,
\end{align}
and similarly for the spacetime term. The other $7$ terms are dealt with similarly.
\end{proof}

\section{A spacetime estimate for the Weyl-tensor} \label{XestWeyltensor}
Recall the energy (\ref{degXall}). It turns out that we have control of this spacetime energy in the interior, provided we can control the $\rho$ and the $\sigma$ component of the Weyl-tensor under consideration. To establish this result (Proposition \ref{spinreduce}) it is convenient  to introduce the concept of spin-reduction.

\subsection{Spin reduction} \label{spired}
Consider the \emph{Maxwell pseudo-tensor}
\begin{equation}
 \mathcal{F}_{\alpha \beta} = r \cdot \mathcal{W}_{\alpha \beta 3 4}  \, .
\end{equation}
It is clearly antisymmetric. Its null components are
\begin{equation}
 \mathcal{F}_{A3} = 2r \underline{\beta}_A \textrm{ \ \ \ \ , \ \ \ \ \ } \mathcal{F}_{A4} = 2r \beta_A \, ,
\end{equation}
\begin{equation}
 \mathcal{F}_{34} = 4r \rho \textrm{ \ \ \ \ , \ \ \ \ \ } \mathcal{F}_{AB} = 2r \sigma \epsilon_{AB} \, .
\end{equation}
Note that in the Schwarzschild case, $\mathcal{F}$ corresponds to the electromagnetic field created by a point charge of strength proportional to $M$.

As is readily computed using the formulae in the appendix, if $\nabla^a \mathcal{W}_{abcd} = \mathcal{J}_{bcd}$, then $\mathcal{F}$ satisfies the inhomogeneous Maxwell's equations 
\begin{align} \label{ihoma}
\nabla^\alpha \mathcal{F}_{\alpha \beta}=\mathcal{J}^{spin1}_{\beta} := r \mathcal{J}_{\beta 34} + r {j}_{\beta}
\end{align}
with $\mathcal{J}_{\beta 34}$ a null-component of the right hand side of the Bianchi equation and\begin{align} \label{con}
{j}_3 = \rho \left(tr \underline{H} - 2 \slashed{D}_3 r\right) + \underline{\alpha} \cdot H - 2\underline{Y} \beta + 2 Z \underline{\beta} \, , \nonumber \\
{j}_4 = \rho \left(tr {H} - 2 \slashed{D}_4 r\right) + {\alpha} \cdot \underline{H} + 2Y \underline{\beta} - 2 \underline{Z} {\beta} \, , \nonumber \\
j_B =  \frac{1}{2} \underline{\beta}  \left(tr {H} - 2 \slashed{D}_4 r\right) + \frac{1}{2} {\beta}  \left(tr {\underline{H}} - 2 \slashed{D}_3 r\right) -\epsilon_{AB} {}^\star \underline{\beta}_F \hat{H}_{AF} \, ,\nonumber \\ 
-\epsilon_{AB} {}^\star {\beta}_F \underline{\hat{H}}_{AF} - \underline{Y} \alpha + Y \underline{\alpha} - \rho \left(Z - \underline{Z}\right) + \sigma  \left({}^\star Z + {}^\star \underline{Z}\right) \, .
\end{align}
If the background in exactly Schwarzschild, the right hand side of (\ref{ihoma}) vanishes, which is the origin of the name ``spin reduction". Note the ``linear" (non-quadratically decaying) error-terms that arise in $j$ in $j_3$ and $j_4$ that occur at the lowest order (when $\rho=\rho\left(W\right)$) but that all terms are decaying quadratically in case that $\mathcal{W} = \hat{\mathcal{L}}^i_T W$, $i\geq 1$.

The associated pseudo energy momentum tensor 
\begin{align}
\tilde{Q}_{\alpha \beta} = \mathcal{F}_{\alpha \gamma} F^{\phantom{\beta}\gamma}_{\beta}  + \left({}^\star \mathcal{F}\right)_{\alpha \gamma} \left({}^\star \mathcal{F}\right)^{\phantom{\beta}\gamma}_{\beta} = 2\mathcal{F}_{\alpha \gamma} \mathcal{F}^{\phantom{\beta}\gamma}_{\beta} - \frac{1}{2} g_{\alpha \beta} \mathcal{F}^2
\end{align}
has null components 
\begin{equation}
 \tilde{Q}_{33} = 8 r^2 |\underline{\beta}|^2 \textrm{ \ \ \ \ \ , \ \ \ \ }  \tilde{Q}_{44} = 8 r^2 |{\beta}|^2 \textrm{ \ \ \ \ , \ \ \ \ \ }  \tilde{Q}_{34} = 8 r^2 \left(\rho^2 +\sigma^2\right) 
\end{equation}
\begin{align}
 \tilde{Q}_{AB} = 4r^2 \left(\delta_{AB} \beta \cdot \underline{\beta} - \beta_A \underline{\beta}_B - \beta_B \underline{\beta}_A\right) + 4 \delta_{AB} r^2 \left(\rho^2+\sigma^2\right) \, ,\nonumber \\
 \tilde{Q}_{3A} = -8r^2 \left(\rho \underline{\beta} + \sigma {}^\star \underline{\beta} \right) \textrm{ \ \ \ \ \ \ \ , \ \ \ \ \ }  \tilde{Q}_{4A} =  8r^2 \left(\rho {\beta} - \sigma {}^\star {\beta} \right) \, .
\end{align}
Note that $\tilde{Q}$ is traceless, $\gamma^{AB}\tilde{Q}_{AB} = \tilde{Q}_{34}$. Finally, we have the spin1-energy identity arising from a vectorfield $X$,
\begin{equation} \label{Fmi}
\nabla^{\alpha} \left(J_\alpha^X\left[\mathcal{F}\right]\right) = \nabla^{\alpha} \left(\tilde{Q}_{\alpha \beta} X^\beta\right)  = K_1^X\left[\mathcal{F}\right] + K_2^X\left[\mathcal{F}\right] \, ,
\end{equation}
\begin{align}
K_1^X\left[\mathcal{F}\right] = \tilde{Q}_{\alpha \beta} {}^{(X)} \pi^{\alpha \beta} \textrm{ \ \ \ and \ \ \ } K_2^X\left[\mathcal{F}\right] = \left(\mathcal{J}_{\delta}^{spin1}\mathcal{F}^{\phantom{\gamma}\delta}_{\gamma} +  \left(\mathcal{J}^{spin1}{}^\star\right)_\delta \left({}^\star \mathcal{F}\right)^{\phantom{\gamma}\delta}_{\gamma} \right) X^\gamma \, . \nonumber
\end{align}
These formulae should be compared with the spin2-identity (\ref{mainid}). For the error one observes that 
\begin{align}
K_2^{e_3}\left[\mathcal{F}\right] \Big|_{old} = \left(\mathcal{J}_{\delta34}\mathcal{F}^{\phantom{c}\delta}_{\gamma} +  \left(\mathcal{J}{}^\star\right)_{\delta34} \left({}^\star \mathcal{F}\right)^{\phantom{c}\delta}_{\gamma} \right) \left(e_3\right)^\gamma =  r^2 K_2^{e_3 e_3 e_4} \left[\mathcal{W}\right]
\end{align}
\begin{align}
K_2^{e_4}\left[\mathcal{F}\right] \Big|_{old}   = \left(\mathcal{J}_{\delta34}\mathcal{F}^{\phantom{c}\delta}_{\gamma} +  \left(\mathcal{J}{}^\star\right)_{\delta34} \left({}^\star \mathcal{F}\right)^{\phantom{c}\delta}_{\gamma} \right) \left(e_4\right)^\gamma =  r^2 K_2^{e_3 e_4 e_4} \left[\mathcal{W}\right]
\end{align}
\begin{align}
K_2^{e_3}\left[\mathcal{F}\right] \Big|_{new} = \left(r j_{\delta}\mathcal{F}^{\phantom{c}\delta}_{\gamma} +  r \left(j{}^\star\right)_{\delta} \left({}^\star \mathcal{F}\right)^{\phantom{c}\delta}_{\gamma} \right) \left(e_3\right)^\gamma
\end{align}
\begin{align}
K_2^{e_4}\left[\mathcal{F}\right] \Big|_{new} = \left(r j_{\delta}\mathcal{F}^{\phantom{c}\delta}_{\gamma} +  r \left(j{}^\star\right)_{\delta} \left({}^\star \mathcal{F}\right)^{\phantom{c}\delta}_{\gamma} \right) \left(e_4\right)^\gamma
\end{align}
and recalls that the first two terms have been estimated before (however, without the additional $r^2$ weight) in section 7. The additional $r^2$-weight does not alter the treatment of the intricate ``linear" error-terms, in view of the strong decay in $r$ of the component $\rho\sim \frac{1}{r^3}$. However, the additional weight at infinity requires that an $\epsilon$ of an $r$-weighted spacetime term needs to be added to Proposition \ref{K12summary}. For this we recall (\ref{Prho}) and state
\begin{proposition} \label{K2Fsum}
The estimate of Proposition \ref{K12summary} holds for the expression
\begin{align}
\int_{\tilde{\mathcal{M}}} K_2^{\mathcal{X}} \left[\mathcal{F}=\left(\hat{\mathcal{L}}^{n+1}_T W\right)_{\alpha \beta 3 4} \right]
\end{align}
(with $\mathcal{X}=p_\mathcal{X} e_3 + q_\mathcal{X} e_4$ as in Proposition \ref{K12summary}) as well, provided the spacetime-term $\epsilon \cdot \overline{I}^{n+1}_{\tilde{P}_\rho}\left[\tilde{W}\right] \left(\tilde{\mathcal{M}}\right)$ is added to the right hand side.
\end{proposition}
\begin{proof}
As mentioned, this is clear for the terms proportional to the non-decaying component $\rho$. It remains to look at the behavior at infinity. For the terms in $\mathcal{J}^{spin1}_d$ arising from the original $\mathcal{J}_{d34}$, one performs a null-decomposition as in \cite{ChristKlei} and inspects each term individually. For instance,
\begin{align}
\int_{r\geq R} dt^\star dr d\omega \, r^4 tr \chi {}^{(T)} \mathbf{i} \cdot \underline{\alpha} \cdot \rho \left(\mathcal{L}_T W\right) \nonumber \\
\leq \epsilon \int_{r\geq R} dt^\star dr d\omega r^{0} |\underline{\alpha}|^2  + \epsilon \int_{r\geq R} dt^\star dr d\omega r^{4} |\rho \left(\mathcal{L}_T W\right)|^2 \leq \epsilon \cdot \overline{I}^{n+1}_{\tilde{P}_{\rho}}\left[\tilde{W}\right] \, . \nonumber
\end{align}
For the terms arising from $j$ we have to look at the formulae (\ref{con}). The estimates for these terms are straightforward and analogous to the above with one important exception:  It is crucial that in our gauge $Y=0$, as otherwise the term $r^4 Y \cdot \underline{\alpha}$ could give rise to an $r^{-\delta}$ divergence). 
\end{proof}

\subsection{The main result}
\begin{proposition} \label{spinreduce}
Let $\mathcal{W}$ be a Weyl-field satisfying the inhomogeneous Bianchi equations, $\mathcal{F}$ the associated pseudo Maxwell-tensor and recall (\ref{K2}).  We have the integrated decay estimates
 \begin{align} \label{alldeg}
\sup_{(\tau_1,\tau_2)} E_{spin1} \left[\mathcal{W}\right]\left(\tilde{\Sigma}_{\tau}\right) + E \left[\mathcal{W}\right]\left(\mathcal{H}\left(\tau_1,\tau_2\right)\right)  + I^{deg} \left[\mathcal{W}\right] \left(\tilde{\mathcal{M}}\left({\tau_1},{\tau_2}\right)\right) \leq  B \cdot E_{spin1} \left[\mathcal{W}\right] \left( \tilde{\Sigma}_{\tau_1}\right)  \nonumber \\
 + B \int _{\tau_1}^{\tau_2} \int_{r_0}^{R} \int dt^\star  dr \, r^2 \, d\omega \frac{\left(r-3M\right)^2}{r^3} \left[|\rho|^2+|\sigma|^2\right] + B \cdot Err_{Mora} \left[\mathcal{W}\right]
 \end{align}
  \begin{align} \label{allnondeg}
\sup_{(\tau_1,\tau_2)} E_{spin1} \left[\mathcal{W}\right]\left(\tilde{\Sigma}_{\tau}\right) + E \left[\mathcal{W}\right]\left(\mathcal{H}\left(\tau_1,\tau_2\right)\right) + I^{nondeg}  \left[\mathcal{W}\right] \left(\tilde{\mathcal{M}}\left({\tau_1},{\tau_2}\right)\right)     \leq \nonumber \\
B \cdot  E_{spin1} \left[\mathcal{W}\right] \left( \tilde{\Sigma}_{\tau_1}\right)
+ B \int _{\tau_1}^{\tau_2} \int_{r_0}^{R} \int dt^\star  dr \, r^2 \, d\omega  \left[|\rho|^2+|\sigma|^2\right]  
 + B \cdot  Err_{Mora} \left[\mathcal{W}\right]
 \end{align}
 with 
 \begin{align} Err_{Mora} \left[\mathcal{W}\right]=  \sum_{a,b,c=3,4} \Big| \int_{\tilde{\mathcal{M}}\left(\tau_1,\tau_2\right)} K_2^{e_a e_b e_c}\left[\mathcal{W}\right]  \Big| +\Big| \int_{\tilde{\mathcal{M}}\left(\tau_1,\tau_2\right)} r^2 K_2^{e_a}\left[\mathcal{F}\right]  \Big|
 \end{align}
\end{proposition} 
\begin{remark}
In applications $\mathcal{W}$ will be the $n$-times Lie-$T$ commuted Weyl-tensor, whose $K_2\left[\mathcal{W}\right]$ error-term was estimated in section \ref{errorterms}, while the error-term arising from $K_2\left[\mathcal{F}\right]$ was estimated in Proposition \ref{K2Fsum}. 
\end{remark}

\begin{proof}[Proof of Proposition \ref{spinreduce}]
We will apply the energy identity for the Bel-Robinson tensor  (\ref{mainid}) with the vectorfields 
\begin{equation}
X=2X_0 + 2 C \cdot T = f\left(r\right)\left(-p e_3 + q e_4\right) + 2C \cdot T \ \ \ \  \textrm{ and} \ \ \ \ \  2T= p e_3 + q e_4 \, .
\end{equation}
Here $f$ is a bounded function behaving like $f^\prime \sim r^{-1-h}$ near infinity for a small constant $h \geq 0$, while $C=2\sup |f|$ is a constant chosen to make $X$ timelike everywhere in $\mathcal{R}$ except for an $\epsilon$-small region close to the horizon. For the following computations, we will use the ``Schwarzschildean" normal $2\tilde{n}_{\Sigma} = \sqrt{k_\chi} e_3+\frac{1}{\sqrt{k_\chi}} e_4$, which is $C^1$ close to $n_{\Sigma}$, and then argue by stability of the estimates (cf.~(\ref{bndyssrel})).
We compute
\begin{align}
J_{\mu}^{X_0TT} \left[\mathcal{W}\right]  \tilde{n}^\mu_{\Sigma} := Q\left(X_0,T,T, \tilde{n}^\mu_{\Sigma} \right) = \frac{1}{16} \Bigg\{ 2|\underline{\alpha}|^2 \left[-f p^3 \sqrt{k_+} \right] + 2|\alpha|^2 q^3 \frac{f}{\sqrt{k_+}} \nonumber \\
+ 4|\underline{\beta}|^2 \left(-p^2 f \sqrt{k_\chi} - p^3 f  \frac{1}{\sqrt{k_\chi}}\right) 
+ 4|\beta|^2 \left( p q^2 f \frac{1}{\sqrt{k_\chi}} - q^3 f  \sqrt{k_\chi} \right) \nonumber \\
+ 4\left(\rho^2+\sigma^2\right) \left( -p^2 q f \frac{1}{\sqrt{k_\chi}} + p q^2 f  \sqrt{k_\chi}\right) \, . \nonumber
\end{align}
It is clear that by adding $C\cdot T$ to $X_0$ we obtain the bounds
\begin{equation} \label{obsb2}
 Q\left(X,T,T, {n}^\mu_{\Sigma} \right) \geq b \cdot Q\left(T,T,T, n^{\mu}_{\Sigma} \right) - \epsilon \cdot Q\left(N,N,N, n^{\mu}_{\Sigma} \right) \textrm{ \ \ \ and \ \ }
\end{equation}
\begin{equation}
 Q\left(X,T,T, {n}^\mu_{\Sigma} \right) \leq B \cdot Q\left(T,T,T, n^{\mu}_{\Sigma} \right) + \epsilon \cdot Q\left(N,N,N, n^{\mu}_{\Sigma} \right) \, .\nonumber
\end{equation}

For the bulk term $K_1\left[\mathcal{W}\right]$ we first compute the components of the deformation in case the metric is Schwarzschild 
\begin{align}
{}^{(2X_0)}\pi_{SS}^{34} = -\frac{1}{2} \left(1-\mu\right) f_{,r} - \frac{M}{r^2}f  \textrm{ \ \ , \ \ } {}^{(2X_0)}\pi_{SS}^{33} = \frac{1}{2}\left(k_{\chi}^-\right)^2 f_{,r} \, , \nonumber \\   {}^{(2X_0)}\pi_{SS}^{44} = \frac{1}{2} \left(\frac{k_{\chi}^+}{k_\chi}\right)^2 f_{,r} \textrm{ \ \ , \ \ }  {}^{(2X_0)}\pi_{SS}^{AB}  = \frac{f}{r} \left(1-\mu\right) \delta^{AB}
\end{align}
and then decompose
\begin{align} \label{bulkX2}
 K_1^{XTT} \left[\mathcal{W}\right] =K_{1a}^{XTT} \left[\mathcal{W}\right] + K_{1b}^{XTT} \left[\mathcal{W}\right] + K^{XTT}_{1,error}\left[\mathcal{W}\right] \, ,
\end{align}
where 
\begin{align}
 K_{1a}^{XTT} \left[\mathcal{W}\right] = \frac{1}{4}  f_{,r} \Bigg[\frac{1}{2} \left(k_{\chi}^-\right)^4 \|\underline{\alpha}\|^2 + \left(\frac{k_{\chi}^+}{k_\chi}\right)^4\frac{1}{2} \|\alpha\|^2 \Bigg] \nonumber \\ + f \left(\frac{1}{r} - \frac{3}{2}\frac{\mu}{r}\right)\Bigg[  \left(k_{\chi}^-\right)^2 \|\underline{\beta}\|^2 + \left(\frac{k_{\chi}^+}{k_\chi}\right)^2 \|\beta\|^2 \Bigg] \, ,
\end{align}
\begin{align}
 K_{1b}^{XTT} \left[\mathcal{W}\right] = \left(1-\mu\right)f \left(\frac{2}{r} - 3\frac{\mu}{r}\right)\left[\rho^2 + \sigma^2 \right] -\frac{1}{8} \left(1-\mu\right)^3 f_{,r} \left[\rho^2 + \sigma^2 \right] \nonumber
\end{align}
are the expressions one would obtain if the components of the deformation tensor of $X_0$ took on exactly their Schwarzschildean values (and if $T$ was replaced by the Schwarzschildean $\tilde{T} = \frac{1}{2} k_{\chi}^- e_3 + \frac{k_{\chi}^+}{2k_\chi} e_4$) and
\begin{align}
K^{XTT}_{1,error}  \left[\mathcal{W}\right] = \left({}^{(X)}\pi - {}^{(X)}\pi_{SS}\right) Q \left(T,T\right) +2 \left({}^{(T)}\pi \right) Q \left(X,T\right)  \nonumber \\
+ 2\left({}^{(X)} \pi \right) Q\left(T-\tilde{T},T\right)+ \left({}^{(X)} \pi \right) Q\left(T-\tilde{T},T-\tilde{T}\right) \, ,
\end{align}
is the error introduced. One easily sees that there is no way to make $K_{1a}$ and $K_{1b}$ globally positive at the same time, in view of the components $\rho$ and $\sigma$ having the opposite sign compared with the other curvature components.

To avoid degeneration of control near the horizon we will add to the energy identity of $XTT$ a little bit of $NNN$ (we drop the $ \left[\mathcal{W}\right] $-argument for the moment to simplify the notation) producing
\begin{align} \label{mit}
\int_{\tilde{\mathcal{M}}} \left[ K^{XTT}_{1a} + K^{XTT}_{1b} + K^{XTT}_{1,error} + K^{XTT}_2 + e K_1^{NNN} + e K_2^{NNN} \right]  \nonumber \\  + \int_{\Sigma} J^{XTT}_\mu n^\mu_{\Sigma_{\tau}} + e J^{NNN}_\mu n^\mu_{\Sigma_{\tau}} + \int_{\mathcal{H}} e J^{NNN}_\mu n^\mu_{\mathcal{H}} + J^{XTT}_\mu n^\mu_{\mathcal{H}}  \nonumber \\= \int_{\Sigma_0} J^{XTT}_\mu n^\mu_{\Sigma_0} + e J^{NNN}_\mu n^\mu_{\Sigma_{0}}
\end{align}
for $e>0$ a small constant. We choose $e$ so small that in particular
\begin{equation}
 e |K_1^{NNN}| \leq \frac{1}{2} K_{1a}^{XTT} + \frac{1}{16} \left(1-\mu\right)^3 f_{,r} \left[\rho^2 + \sigma^2\right] \textrm{ \ \ \ \ in $\left[r_Y,\frac{r_Y-2M}{2}\right]$} \nonumber
\end{equation}
holds. Observe that $K_{1b}^{XTT}$ has a positive sign close to the horizon and near infinity so that we only have to estimate this term between two fixed constant $r$ hypersurfaces. Note also that the boundary terms in the second line of (\ref{mit}) have a sign and control the non-degenerate energy.

To obtain the estimate (\ref{alldeg}), we would like to choose $f\left(3M\right)=0$ and $f_{,r}$  to be a smooth function which is always non-negative and vanishes quadratically at $r=3M$. This almost gives estimate (\ref{alldeg}), except that the $\beta$-term now degenerates like $\left(r-3M\right)^4$. Moreover, the same problem occurs when we try to prove (\ref{allnondeg}). In this case we would like to choose $f$ to change sign at $3M$ and $f_{,r}$ everywhere positive. Clearly, the $\beta$-term would still degenerate quadratically at $3M$. We resolve this problem using the $X$-estimate for the Maxwell-pseudo tensor: Applying (\ref{Fmi}) with the same vectorfield $X$, we find 
\begin{align} \label{mit2}
\int_{\tilde{\mathcal{M}}} K^{X}_{1} \left[\mathcal{F}\right] + K^{X}_2 \left[\mathcal{F}\right]    + \int_{\Sigma} J^{X}_\mu \left[\mathcal{F}\right] n^\mu_{\Sigma_{\tau}} + \int_{\mathcal{H}} J^{X}_\mu \left[\mathcal{F}\right] n^\mu_{\mathcal{H}} = \int_{\Sigma_0} J^{X}_\mu\left[\mathcal{F}\right]  n^\mu_{\Sigma_0}  \, .
\end{align}
For the boundary term  $J^{X}_\mu n^\mu_{\Sigma_{\tau}} = \tilde{Q}\left(X,n_{n_{\Sigma}}\right)$ we have the analogue of (\ref{obsb2}):
\begin{align}\label{obsb3}
 b \cdot \tilde{Q}\left(T,n_{n_{\Sigma}}\right) - \epsilon \cdot \tilde{Q}\left(N,n_{n_{\Sigma}}\right) \leq \tilde{Q}\left(X,n_{n_{\Sigma}}\right) \leq B \cdot \tilde{Q}\left(T,n_{n_{\Sigma}}\right) + \epsilon \cdot \tilde{Q}\left(N,n_{n_{\Sigma}}\right) \, .
\end{align}
As in the spin2 case, we decompose the bulk term as
\begin{align}
K_1^X \left[\mathcal{F}\right] = K^X_{1a} \left[\mathcal{F}\right] + K^X_{1b} \left[\mathcal{F}\right] + K^X_{1,error} \left[\mathcal{F}\right] \, ,
\end{align}
where now
\begin{align}
K^X_{1a} \left[\mathcal{F}\right] = r^2 \frac{\left(1-\mu\right)^2}{4} f_{,r} \left[8\|\underline{\beta}\|^2 + 8\|{\beta}\|^2\right] \, ,
\end{align}
\begin{align}
 K^X_{1b} \left[\mathcal{F}\right] = - r^2 \frac{\left(1-\mu\right)^2}{4} f_{,r} \left[\frac{1}{1-\mu} 16 \left(\rho^2+\sigma^2 \right)\right]  + r^2 f \left(\frac{1}{r} - \frac{3\mu}{2r}\right) \left[8\left(\rho^2 +\sigma^2\right)\right] \nonumber
\end{align}
and
\begin{equation}
 K^X_{1,error}  \left[\mathcal{F}\right] = \left({}^{(X)}\pi - {}^{(X)}\pi_{SS}\right) \cdot \tilde{Q} + C \cdot {}^{(T)}\pi \cdot \tilde{Q} \, .
\end{equation}
We remark that adding to $X$ a little bit of the $N$ vectorfield, one can eliminate the degenerating weights close to the horizon in the standard fashion (it is unnecessary at this point though, as we already dealt with this problem in the spin2 estimate). 

Finally, to obtain the estimate (\ref{alldeg}), we choose $f\left(3M\right)=0$ and $f_{,r}$  to be a smooth function which is always positive and vanishes quadratically at $3M$. In this case the desired spacetime integral is controlled by
\begin{align} \label{igy}
 I^{deg} \left[\mathcal{W}\right] \left( \tilde{\mathcal{M}}\left({\tau_1},{\tau_2}\right)\right)  \leq B \int dt^\star \int_{r_0}^{R} dr r^2 \int d\omega \frac{\left(r-3M\right)^2}{r^3} \left[|\rho|^2+ |\sigma|^2\right]  \nonumber \\ + B \int_{\tilde{\mathcal{M}}} \left[ K^{XTT}_{1a} \left[\mathcal{W}\right] + K^{XTT}_{1b} \left[\mathcal{W}\right] + e K_1^{NNN} \left[\mathcal{W}\right] + K^X_{1a}  \left[\mathcal{F}\right]+ K^X_{1b} \left[\mathcal{F}\right] \right]  \, .
\end{align}
The right hand side is now estimated by inserting the identities (\ref{mit}) and (\ref{mit2}): The resulting boundary terms have a good sign on the horizon and are otherwise estimated by (\ref{obsb2}) and (\ref{obsb3}). For the error-term we have
\begin{align} \label{eabs}
\Bigg| \int_{\tilde{\mathcal{M}}\left(\tau_1,\tau_2\right)} K^X_{1,error}  \left[\mathcal{F}\right]  + \int_{\tilde{\mathcal{M}}\left(\tau_1,\tau_2\right)} K^{XTT}_{1,error}  \left[\mathcal{W}\right] \Bigg| \leq \nonumber \\
\epsilon \cdot I^{deg} \left[\mathcal{W}\right] \left(\tilde{\mathcal{M}}\left({\tau_1},{\tau_2}\right)\right)  + \epsilon \cdot \sup_{\tau} \, E_{spin1} \left[\mathcal{W}\right] \left( \tilde{\Sigma}_{\tau}\right)  \, ,
\end{align}
as an consequence of the ultimately Schwarzschildean assumption.
Hence the estimate (\ref{alldeg}) follows after absorbing the terms on the right hand side of (\ref{eabs}) by the main terms on the left.

Similarly, to obtain (\ref{allnondeg}), we choose $f$ to vanish linearly at $3M$ and $f_{,r}$ everywhere positive. Note that since $f_{,r}$ no longer degenerates at $3M$, we now need non-degenerate control of $\rho$ and $\sigma$ on the right hand side. The estimate (\ref{igy}) now holds without the degenerating weight on the right hand side and the non-degenerate energy on the left. The second line is estimated as previously, using the energy identities (\ref{mit}) and (\ref{mit2}).
\end{proof}

\subsection{The estimate at the lowest order}
For $\mathcal{W}=W$ the original Weyl-field,  the integrated decay estimate of the previous section is not useful, as the right hand side is expected to grow in $t$ in view of the component $\rho$ appearing. However, using multipliers directly on the renormalized null-Bianchi equations, just as we did in section \ref{normalizedenergy}, we obtain
\begin{corollary} \label{coloo}
For the original Weyl field $W$ we have the following estimate for the renormalized components: For any $h > 0$
  \begin{align}
\sup_{(\tau_1,\tau_2)} \int_{\tilde{\Sigma}_{\tau}} \|\tilde{W}\|^2 +  \int_{\mathcal{H}\left(\tau_1,\tau_2\right)} \|\tilde{W}\|^2 + \int_{\tilde{\mathcal{M}}\left(\tau_1,\tau_2\right)} \frac{1}{r^{1+h}} \|\tilde{W}\|^2    \leq B_h  \int_{\tilde{\Sigma}_{\tau_1}} \|\tilde{W}\|^2  \nonumber \\
+ B \int dt^\star \int_{r_0}^{R} dr r^2 \int d\omega  \left[|\hat{\rho}|^2+|\sigma|^2\right] 
 + B_h \|\rho^2\|_{L^\infty} \int_{\tilde{\mathcal{M}}\left(\tau_1,\tau_2\right)} \|\mathfrak{R}-\mathfrak{R}_{SS}\|^2 \nonumber
 \end{align}
\end{corollary}
\section{The components $\rho$ and $\sigma$} \label{rssection}
As we have seen in the previous section, the components $\rho$ and $\sigma$ play a special role in the decay. In this section, we derive the wave equations for $\rho$ and $\sigma$ and their $T$-commuted higher order analogues, define appropriate energies for them and prove multiplier estimates for the latter. 

As mentioned previously, the natural variables to look at are the renormalized quantities
\begin{align} \label{renormdef}
\phi_n =  r^3 \left[\rho \left(\widehat{\mathcal{L}}^n_T W\right) + \frac{2M}{r^3} \delta^n_0\right] \textrm{ \ \ \ \ and \ \ \ \ } \psi_n = r^3 \sigma\left(\widehat{\mathcal{L}}^n_T W\right) \, ,
\end{align}
for which we defined the energies $E\left[\mathcal{D}\phi_n\right]\left( \tilde{\Sigma}_\tau\right)$ and $I_{deg}  \left[\mathcal{D}\phi_n\right] \left( {\mathcal{U}}\right)$ (and similarly with $\psi_n$) in (\ref{strhoe}) and (\ref{strho}). Note in this context that the relation $E\left[\mathcal{D}\phi_n\right]\left( \tilde{\Sigma}_\tau\right) \leq B\cdot\overline{\mathbb{E}}^{n+1}_{P_\rho}\left[W\right]\left(\tilde{\Sigma}_\tau\right)$ holds for the boundary admissible decay matrix defined in (\ref{Prho}). 

The quantities (\ref{renormdef}) will be seen to satisfy an inhomogeneous Regge-Wheeler-type equation (section \ref{scalrhoeq}). In section \ref{multestrho}, we derive an integrated decay estimate for the homogeneous part of this equation. Such type of estimates are known \cite{BlueSoffer, BlueSoffer2}, in the context of the Regge-Wheeler equation on an exact Schwarzschild background. The crucial point here is that the lower order terms, which are typically the source of difficulty in deriving such an estimate, are a-priori controlled from the ultimately Schwarzschildean assumption. The second crucial ingredient consists in controlling the inhomogeneity which is present in our case. This is particularly intricate in view of the additional $r^3$-weight introduced by the renormalization (\ref{renormdef}).  The analysis is the content of section \ref{structureRHS}. Finally, we remark that all error-estimates are carried out  for the $\rho$-equation only, since the $\sigma$-analogues are easily inferred (section \ref{scalsigmaeq}).

\subsection{The main result} 
We start by stating the main result of this section, an integrated decay estimate for the quantities $\phi_n$ and $\psi_n$, which follows from combining the decay estimates for the Regge-Wheeler operator with that for the inhomogeneity.

\begin{proposition} \label{rhosigmaest}
Recall the boundary admissible decay matrix $P_{\rho}$ (cf.~(\ref{Prho})) and its bulk associated matrix $\tilde{P}_{\rho}$, as well as  (\ref{renormdef}). For $1 \leq n \leq k-1$
\begin{align}
\sup_{(\tau_1,\tau_2)} E \left[\mathcal{D}\phi_{n}\right]\left(\tilde{\Sigma}_{\tau} \right) + E \left[\mathcal{D}\phi_{n}\right]\left(\mathcal{H}\left(\tau_1,\tau_2\right) \right)  + I_{deg}  \left[\mathcal{D}\phi_{n}\right] \left(\tilde{\mathcal{M}}\left(\tau_1,\tau_2\right)\right) \nonumber \\
\leq B \cdot E \left[\mathcal{D}\phi_{n}\right]\left(\tilde{\Sigma}_{\tau_1} \right)
+ \epsilon  \left[ \overline{\mathbb{I}}^{n+1,deg}_{\tilde{P}_{\rho}} \left[W\right] \left(\tilde{\mathcal{M}}\left(\tau_1,\tau_2\right)\right) + \sup_{\tau} \overline{\mathbb{E}}^{n+1} \left[W\right]  \left(\tilde{\Sigma}_{\tau}\right)\right] \nonumber \\ + \lambda \cdot \mathbb{D}^{n+1} \left[\mathfrak{R}\right] \left(\tau_1,\tau_2\right) + B_{\lambda} \cdot \mathbb{D}^{n} \left[\mathfrak{R}\right] \left(\tau_1,\tau_2\right) \nonumber \\
B_\lambda \cdot I_{deg} \left[\mathcal{D}\phi_{n-1}\right]\left( \tilde{\mathcal{M}}\left(\tau_1,\tau_2\right)\right) + B \cdot E\left[\mathcal{D}\phi_{n-1}\right]\left( \tilde{\Sigma}_{\tau_1,\tau_2}, \mathcal{H}^+\right) \nonumber
\end{align}
for any $\lambda>0$. For $n=0$, the same estimate holds with the last two lines replaced by the term $B \cdot {\mathbb{D}}^{1} \left[\mathfrak{R}\right]$. Moreover, rhe same estimates hold replacing $\phi_n$ by $\psi_n$.
\end{proposition}
\begin{proof}
Proposition \ref{rhosigmaest} will follow from applying Proposition \ref{finXest} and estimating the error-term via Propositions \ref{errorestrho} and \ref{errorestsigma}. See section \ref{mnrs}.
\end{proof}
\begin{remark}
Note that the energy $I_{deg} \left[\mathcal{D}\phi_{n}\right]$ provides at the same time a non-degenerate integrated decay estimate for $\phi_{n-1}$. 
\end{remark}

The above result should be understood as follows: It is possible to control a spacetime integral for the $T$-commuted renormalized $\rho$ and $\sigma$ components from the $L^2$-energy at the expense of 
\begin{itemize}
\item borrowing an $\epsilon$ from the spacetime integral of a weighted energy (decay matrix $\tilde{P}_{\rho}$)  of all the other curvature components (second line)
\item borrowing a small amount ($\lambda$) of the highest energy on the Ricci-coefficients (third line)
\item and, provided the analogous estimate is available one order lower (last line).
\end{itemize}

\begin{remark}
It is crucial to note that the analogous degenerate estimate for \underline{all} curvature components requires the full $\mathbb{D}^n \left[\mathfrak{R}\right]$ instead of $\mathbb{D}^n_\lambda \left[\mathfrak{R}\right]$. The point here is that in proving a \underline{non-degenerate} estimate for all curvature components one only loses derivatives for the quantities $\rho$ and $\sigma$, i.e.~one only needs a higher order estimate for $\rho$ and $\sigma$, not for the other curvature components (cf.~Proposition \ref{spinreduce}). For these scalar equations, however, one can exploit the fact that one gains non-degenerate integrated decay for a particular radial derivative, which allows the gain of a small $\lambda$ when controlling the error-terms.
\end{remark}
Clearly, one can eliminate the last line of the estimate of Proposition \ref{rhosigmaest} by iterating the estimate:
\begin{corollary} \label{rhosigmaestcor}
For $1 \leq n \leq k-1$
\begin{align}
\sup_{(\tau_1,\tau_2)} E \left[\mathcal{D}\phi_{n}\right]\left(\tilde{\Sigma}_{\tau} \right) + E \left[\mathcal{D}\phi_{n}\right]\left(\mathcal{H}\left(\tau_1,\tau_2\right) \right)  + I_{deg}  \left[\mathcal{D}\phi_{n}\right] \left(\tilde{\mathcal{M}}\left(\tau_1,\tau_2\right)\right) \nonumber \\
\leq B \sum_{i=0}^n E \left[\mathcal{D}\phi_{i}\right]\left(\tilde{\Sigma}_{\tau_1} \right)
+ \epsilon  \left[ \overline{\mathbb{I}}^{n+1,deg}_{\tilde{P}_{\rho}} \left[W\right] \left(\tilde{\mathcal{M}}\left(\tau_1,\tau_2\right)\right) + \sup_{\tau} \overline{\mathbb{E}}^{n+1} \left[W\right]  \left(\tilde{\Sigma}_{\tau}\right)\right] \nonumber \\ + \lambda \cdot \mathbb{D}^{n+1} \left[\mathfrak{R}\right] \left(\tau_1,\tau_2\right) + B_{\lambda} \cdot \mathbb{D}^{n} \left[\mathfrak{R}\right] \left(\tau_1,\tau_2\right) \nonumber
\end{align}
for any $\lambda>0$. For $n=0$ the same estimate holds with the last line replaced by the term $B \cdot {\mathbb{D}}^{1} \left[\mathfrak{R}\right]$. The same estimates hold replacing $\phi_n$ by $\psi_n$.
\end{corollary}

\subsection{The scalar wave equation for $\rho$} \label{scalrhoeq}
We want to derive the equation for the $\rho$ and $\sigma$ components of the commuted Bianchi equation
\begin{equation}
D^\alpha \left(\widehat{\mathcal{L}}^n_T W\right)_{\alpha \beta \gamma \delta} = \mathfrak{J}^n_{\beta \gamma \delta} \, ,
\end{equation}
where for the inhomogeneous term we have by (\ref{nXWeylcommute})
\begin{equation} 
\mathfrak{J}^n_{\beta \gamma \delta} = \widehat{\mathcal{L}}^{n-1}_T J_{\beta \gamma \delta} \left(T,W\right) + \sum_{i=0}^{n-2} \widehat{\mathcal{L}}^i_T J_{\beta \gamma \delta} \left(T, \widehat{\mathcal{L}}^{n-1-i}_T \widehat{\mathcal{L}}_T W\right) \, .
\end{equation}
Recall the notation for the null-decomposition of the inhomogeneity introduced in Proposition \ref{GNDK}, $4\Lambda \left[\mathfrak{J}^n\right] = \mathfrak{J}^n_{434}$ etc. and also the notation for the inhomogeneities in the null-decomposed Bianchi equations, $E_4 \left(\rho\left(\widehat{\mathcal{L}}^n_T W\right) \right)$, cf.~(\ref{Bianchi6}).
\begin{proposition} \label{tre}
The quantity $\rho \left(\widehat{\mathcal{L}}^n_T W\right)$ for $n \geq 0$ satisfies
\begin{align}
\left(\slashed{D}_4 \slashed{D}_3 + \slashed{D}_3 \slashed{D}_4 \right) \rho\left(\widehat{\mathcal{L}}^n_T W\right) - 2\slashed{\nabla}^2 \rho \left(\widehat{\mathcal{L}}^n_T W\right)  \nonumber \\
+ \frac{3}{2} \left(2tr H - \frac{4}{3}\Omega \right) \slashed{D}_3 \rho\left(\widehat{\mathcal{L}}^n_T W\right) 
+ \frac{3}{2} \left(2tr \underline{H} - \frac{4}{3}\underline{\Omega} \right) \slashed{D}_4\rho\left(\widehat{\mathcal{L}}^n_T W\right)  \nonumber \\ 
+ \rho\left(\widehat{\mathcal{L}}^n_T W\right) \left(3 tr \underline{H} tr H +12\rho \left(W\right)+ \left(\textrm{dec. RRC}\right)^2 \right) = \sum_{i=1}^2 \mathcal{F}^n_i + \underline{\mathcal{F}}^n_i =:\mathcal{F}^n \, , \nonumber
\end{align}
where the right hand side has the following structure
\begin{align}
\mathcal{F}^n_1 = \slashed{D}_3 E_4\left( \rho\left(\widehat{\mathcal{L}}^n_T W\right) \right) + \frac{3}{2} tr \underline{H} E_4\left(\rho\left(\widehat{\mathcal{L}}^n_T W\right)\right) - \beta\left(\widehat{\mathcal{L}}^n_T W\right) \cdot \slashed{\nabla} tr \underline{H} \nonumber \\
- \widehat{\underline{H}} \cdot \slashed{\nabla}{\beta}\left(\widehat{\mathcal{L}}^n_T W\right) + tr F_3\left(\beta\left(\widehat{\mathcal{L}}^n_T W\right)\right)  - 2 \underline{\Omega} E_4\left(\rho\left(\widehat{\mathcal{L}}^n_T W\right)\right)  \nonumber \\ + 2 \slashed{\nabla}^A \underline{\Omega} \beta\left(\widehat{\mathcal{L}}^n_T W\right) 
+ \slashed{div} \left(2 \widehat{H} \cdot \underline{\beta}\left(\widehat{\mathcal{L}}^n_T W\right) + \underline{Y} \cdot \alpha\left(\widehat{\mathcal{L}}^n_T W\right)\right) \nonumber \\ + 3 \slashed{curl} Z \sigma\left(\widehat{\mathcal{L}}^n_T W\right) + 3 Z^A \slashed{\nabla}_A \rho\left(\widehat{\mathcal{L}}^n_T W\right) + 3 \left(Z^\star\right)^A \slashed{\nabla}_A \sigma\left(\widehat{\mathcal{L}}^n_T W\right)
\end{align}
\begin{align} \label{four}
\mathcal{F}^n_2 = -2\left(\Lambda\left[\mathfrak{J}^n\right]\right)_3 + 4\underline{\Omega}\left(\Lambda\left[\mathfrak{J}^n\right]\right) - 2 \slashed{div} \left(I\left[\mathfrak{J}^n\right]\right) +3   \rho_0 \cdot \rho\left(\widehat{\mathcal{L}}^n_T W\right)
\end{align}
and $\underline{\mathcal{F}}_i^n$ being the conjugate quantities. 
\end{proposition}
\begin{remark}
The expression on the left hand side is essentially (see below) the Regge Wheeler operator of the metric $g$ acting on $ \rho\left(\widehat{\mathcal{L}}^n_T W\right)$. If the right hand side exhibited only quadratic error-terms, proving decay would be straightforward. However, this is only manifestly true for $\mathcal{F}_1^n$ and $\underline{\mathcal{F}}_1^n$, while for $\mathcal{F}_2^n$ and  $\underline{\mathcal{F}}_2^n$ there are in fact terms proportional to the non-decaying component $\rho$. The structure of the error-terms will be discussed in section \ref{structureRHS}.
\end{remark}

The last term in (\ref{four}) will be seen to cancel by contributions from the other terms in case $n\geq1$. For $n=0$, however, $\mathfrak{J}^0=0$ and $\mathcal{F}^0_2 + \underline{\mathcal{F}}^0_2=3\rho^2+3\rho^2$ and hence
\begin{corollary} \label{conz}
For $n=0$ we obtain
\begin{align}
\left(\slashed{D}_4 \slashed{D}_3 + \slashed{D}_3 \slashed{D}_4 \right) \rho - 2\slashed{\nabla}^2 \rho 
+ \frac{3}{2} \left(2tr H - \frac{4}{3}\Omega \right) \slashed{D}_3 \rho  \nonumber \\
+ \frac{3}{2} \left(2tr \underline{H} - \frac{4}{3}\underline{\Omega} \right) \slashed{D}_4\rho 
+ \rho \left(3 tr \underline{H} tr H +6\rho + \left(\textrm{dec. RRC}\right)^2 \right) =  \mathcal{F}^0_1 + \underline{\mathcal{F}}^0_1\nonumber \, .
\end{align}
Note that the right hand side decays quadratically.\footnote{In particular, the equation for $\rho$ itself decouples modulo quadratically decaying error-terms. However, $\rho$ does \emph{not} satisfy a Regge-Wheeler equation and only the renormalized quantity $\rho+\frac{2M}{r^3}$ can be seen to satisfy an equation with the Regge-Wheeler operator on the right hand side. Unfortunately, this renormalization also introduces ``linear" error-terms of the form ``$\rho \times \mathcal{D}\left(\mathfrak{R}-\mathfrak{R}_{SS}\right)$", in view of the fact that $\rho$ itself does not decay. A more careful choice of the renormalizing function may remedy this problem, cf.~the remarks at the end of this section.}
\end{corollary}

For the following proof we will, to avoid overloading the notation, denote by $\alpha, \underline{\alpha}, ..., \rho, \sigma$ the null-components of the $n$-commuted Weyl tensor $\widehat{\mathcal{L}}^n_T W$, which satisfies the $n$ times $T$-commuted inhomogeneous Bianchi equations. Moreover, we denote by $\rho_0$ the (non-decaying) $\rho$-component of the original Weyl curvature field, i.e.~$\rho_0=\rho\left(W\right)$. 
\begin{proof}[Proof of Proposition \ref{tre}]
Start from the scalar Bianchi equation for $\rho = \rho\left(\widehat{\mathcal{L}}^n_T W\right)$:
\begin{equation}
 \rho_4 = \slashed{D}_4 \rho + \frac{3}{2} tr \left(H\right) \rho = \slashed{div} \beta + E_4\left(\rho\right) - \frac{1}{2} \mathfrak{J}_{434} \, .
\end{equation}
Taking a $3$-derivative we obtain
\begin{align} \label{stp1}
 \rho_{43} - \slashed{\nabla}^2 \rho  = \slashed{div} \left(E_3\left(\beta\right) + \mathfrak{J}_{3A4}\right) + \slashed{D}_3 E_4\left(\rho\right) + \frac{3}{2} tr \underline{H} E_4\left(\rho\right) \nonumber \\ - \beta \cdot \slashed{\nabla} tr \underline{H} - \widehat{\underline{H}} \cdot \slashed{\nabla}{\beta} + tr F_3\left(\beta\right) - \frac{1}{2} \slashed{D}_3 \mathfrak{J}_{434} - \frac{3}{4} tr \left(\underline{H}\right) \mathfrak{J}_{434}  \, ,
\end{align}
where
\begin{equation}
E_3\left(\beta\right) = 2\widehat{H} \cdot \underline{\beta} + 2\underline{\Omega} \beta + \underline{Y} \alpha + 3 \left(Z \rho + {}^\star Z \sigma\right) \, .
\end{equation}
We first deal with the left hand side. Using the null structure equation
\begin{align} \label{nstrH}
 \slashed{D}_3 tr H = - \frac{1}{2} tr \underline{H} tr H - 2 \slashed{div} Z + 2 \underline{\Omega} tr H - \widehat{H} \cdot \underline{\widehat{H}} \nonumber \\ + 2 \left(\underline{Y} \cdot Y + Z \cdot Z\right) + 2\rho_0
\end{align}
with $\rho_0$ denoting the $\rho$ component of the original Weyl-tensor, we derive
\begin{align} \label{rhoeq2}
LHS = \rho_{43} - \slashed{\nabla}^2 \rho  =  \slashed{D}_3 \slashed{D}_4 \rho + \frac{3}{2} \left(tr H\right) \slashed{D}_3 \rho + \frac{3}{2} \left(tr \underline{H} \right) \slashed{D}_4\rho + \slashed{\nabla}^2 \rho \nonumber \\ + \frac{3}{2}\rho \left(tr \underline{H} tr H + 2 \slashed{div} Z + 2 \underline{\Omega} tr H - \widehat{H} \cdot \underline{\widehat{H}} + 2 \left(\underline{Y} \cdot Y + Z \cdot Z\right) + 2\rho_0 \right) \, .
\end{align}
Turning to the right hand side of (\ref{stp1}) we decompose  
\begin{equation}
RHS_1 =   - \frac{1}{2} \left(\mathfrak{J}_{434}\right)_3 + \slashed{div} \left(\mathfrak{J}_{3A4}\right) = - \frac{1}{2} \slashed{D}_3 \mathfrak{J}_{434} - \frac{3}{4} tr \left(\underline{H}\right) \mathfrak{J}_{434} + \slashed{div} \left(\mathfrak{J}_{3A4}\right) \, , \nonumber
\end{equation}
\begin{equation}
RHS_2 =  \slashed{D}_3 E_4\left(\rho\right) + \slashed{div} E_3 \left(\beta\right) + \frac{3}{2} tr \underline{H} E_4\left(\rho\right) - \beta \cdot \slashed{\nabla} tr \underline{H} - \widehat{\underline{H}} \cdot \slashed{\nabla}{\beta} + tr F_3\left(\beta\right)  \, . \nonumber
\end{equation}
Using again the Bianchi equations we compute
\begin{align}
RHS_2 = \slashed{D}_3 E_4\left(\rho\right) + \frac{3}{2} tr \underline{H} E_4\left(\rho\right) - \beta \cdot \slashed{\nabla} tr \underline{H} - \widehat{\underline{H}} \cdot \slashed{\nabla}{\beta} + tr F_3\left(\beta\right)\nonumber \\  2 \underline{\Omega} \left(\slashed{D}_4 \rho + \frac{3}{2} tr H \, \rho - E_4\left(\rho\right) + \frac{1}{2} \mathfrak{J}_{434} \right) + 2 \slashed{\nabla}^A \underline{\Omega} \beta_A \nonumber \\ 
+ \slashed{div} \left(2 \widehat{H} \cdot \underline{\beta} + \underline{Y} \cdot \alpha\right) + 3 \slashed{div} Z \rho + 3 \slashed{curl} Z \sigma + 3 Z^A \slashed{\nabla}_A \rho +  3 \left(Z^\star\right)^A \slashed{\nabla}_A \sigma
\end{align}
and observe that all terms containing $\rho$ without a derivative are cancelled by terms in $LHS$. The term $\underline{\Omega}\slashed{D}_4 \rho$ is moved to the left hand side. Repeat the same computation starting from $\rho_{3}$ and applying a $4$-derivative. This yields the conjugate terms.
\end{proof}

To obtain the energy naturally associated with the equation, we will have to rescale the equations of Lemma \ref{tre} by $r^3$. Recall that $\phi_n = r^3 \rho \left(\hat{\mathcal{L}}^n_T W\right)$ for $n\geq 1$. 
Define he vectorfield
\begin{align}
L = -\frac{1}{2} \left(2-p\right) e_3 + \frac{1}{2}q e_4 \, ,
\end{align}
\glossary{
name={$L$},
description={$2L = - \left(2-p\right) e_3 + q e_4$ ($L=\partial_r$ in Schwarzschild in $(t^\star,r)$-coords)}
}
which equals $\partial_r$ in Schwarzschild with respect to the regular $\left(t^\star,r\right)$-coordinates. Recall $2T=pe_3 + q e_4$ and $2T^\perp = -p e_3 + q e_4$ and hence 
$\left(1-p\right)T + p L = T^\perp$.

\begin{corollary} \label{vecwrite}
For  $n\geq 1$, we can write
\begin{align}  \label{phicommute} 
 \left(2-p\right) T^2 \left(\phi_n\right) - L \left(pL\phi_n\right) + \left(L \left(p\right)\right) T \left(\phi_n\right) - 2 \left(1-p\right) T L \phi_n \nonumber \\ 
- q \slashed{\nabla}^2 \phi_n - q \frac{6M + \mathcal{E}}{r^3} \phi_n = \frac{q}{2} \ r^3 \mathcal{F}^n +  \left(\textrm{decaying RRC} \ , \ {}^{(T)}\pi \right) \cdot \mathcal{D} \phi_n
\end{align}
where $|\mathcal{E}| \leq \epsilon$ is a small error ($C^2$ in the metric, independent of $n$) arising from the fact that the metric is close to Schwarzschild, and the last term denotes a collection of terms of the form ``any derivative of $\phi_n$ multiplying any null-component of the deformation tensor of $T$, or an expression $\mathfrak{R}-\mathfrak{R}_{SS}$." \newline
The renormalized quantity $\phi_0$ also satisfies (\ref{phicommute}), except that terms of the form $\rho_0 \cdot \left(\textrm{decaying RRC} \ , \ {}^{(T)}\pi \right)$ need to be added to the right hand side. 
\end{corollary}
\begin{proof}
For $n\geq1$, we first write the equation of Proposition \ref{tre} as 
\begin{align}
\left(\slashed{D}_4 \slashed{D}_3 + \slashed{D}_3 \slashed{D}_4 \right) \phi_n  - 2\slashed{\nabla}^2  \phi_n 
- 2 \Omega  \slashed{D}_3  \phi_n 
-2 \underline{\Omega}  \slashed{D}_4  \phi_n
+  \phi_n \left(6\rho \left(W\right) \right) = r^3 \mathcal{F}^n \nonumber \, ,
\end{align}
which holds modulo the error specified in the Corollary. Secondly, we write
\begin{align}
\left(\slashed{D}_4 \slashed{D}_3 + \slashed{D}_3 \slashed{D}_4 \right) \phi_n - 2 \Omega  \slashed{D}_3  \phi_n 
-2 \underline{\Omega}  \slashed{D}_4  \phi_n = \nonumber \\ 
\frac{2}{q} \left[\left(2-p\right)T^2 \left(\phi_n\right)  - L \left(pL\phi_n\right) + \left(L \left(p\right)\right) T \left(\phi_n\right) - 2 \left(1-p\right) T L \phi_n \right] \, ,
\end{align}
which again holds modulo the error specified. For $n=0$, we start from Corollary \ref{conz} directly and do the same computation.
\end{proof}
Clearly, the exact form of the last term on the right hand side of (\ref{phicommute}) will not matter, since any such term will be easily estimated in view of the strong pointwise decay available for the expression in brackets (cf.~Definition \ref{RRCapproach}). 

We finally remark that, in case that the metric is exactly Schwarzschild and Edington-Finkelstein coordinates are used, the operator on the left of (\ref{phicommute}) reads 
\begin{align}
- \partial_u \partial_v \phi_n + \frac{1-\mu}{r^2}{\Delta_{S^2}} \phi_n + \left(1-\mu\right) \frac{6M}{r^3} \phi_n \, ,
\end{align}
which is the Regge-Wheeler operator in its perhaps most familiar form. 
\subsection{The scalar wave equation for $\sigma$} \label{scalsigmaeq}
The computations for $\sigma$ are similar. In particular, we have the analogue of Proposition \ref{tre} (this can either be computed directly or inferred from duality):
\begin{proposition} \label{tse}
The quantity $\sigma \left(\widehat{\mathcal{L}}^n_T W\right)$ satisfies, for $n \geq 0$
\begin{align}
\left(\slashed{D}_4 \slashed{D}_3 + \slashed{D}_3 \slashed{D}_4 \right) \sigma\left(\widehat{\mathcal{L}}^n_T W\right) - 2\slashed{\nabla}^2 \sigma \left(\widehat{\mathcal{L}}^n_T W\right)  \nonumber \\
+ \frac{3}{2} \left(2tr H - \frac{4}{3}\Omega \right) \slashed{D}_3 \sigma\left(\widehat{\mathcal{L}}^n_T W\right) 
+ \frac{3}{2} \left(2tr \underline{H} - \frac{4}{3}\underline{\Omega} \right) \slashed{D}_4\sigma\left(\widehat{\mathcal{L}}^n_T W\right)  \nonumber \\ 
+ \sigma\left(\widehat{\mathcal{L}}^n_T W\right) \left(3 tr \underline{H} tr H +12\rho \left(W\right)+ \left(\textrm{dec. RRC}\right)^2 \right) = \sum_{i=1}^2 \mathcal{G}^n_i + \underline{\mathcal{G}}^n_i\nonumber \, ,
\end{align}
where the right hand side has the following structure
\begin{align}
\mathcal{G}^n_1 = \slashed{D}_3 E_4\left(\sigma \left(\widehat{\mathcal{L}}^n_T W\right)\right) + \frac{3}{2} tr \underline{H} E_4\left(\sigma \left(\widehat{\mathcal{L}}^n_T W\right)\right) + {}^\star \beta \left(\widehat{\mathcal{L}}^n_T W\right)  \cdot \slashed{\nabla} tr \underline{H} \nonumber \\
+ \widehat{\underline{H}} \cdot \slashed{\nabla}{}^\star \beta \left(\widehat{\mathcal{L}}^n_T W\right)  - tr F_3 \left({}^\star \beta \left(\widehat{\mathcal{L}}^n_T W\right)\right) 
  - 2 \underline{\Omega} E_4\left(\sigma \left(\widehat{\mathcal{L}}^n_T W\right)\right)  \nonumber \\ + 2 \slashed{\nabla}^\star \underline{\Omega} \beta \left(\widehat{\mathcal{L}}^n_T W\right)
-\slashed{curl} \left(2 \widehat{H} \cdot \underline{\beta}  \left(\widehat{\mathcal{L}}^n_T W\right)+ \underline{Y} \cdot \alpha \left(\widehat{\mathcal{L}}^n_T W\right) \right) \nonumber \\  -3 \slashed{curl}Z \rho \left(\widehat{\mathcal{L}}^n_T W\right) + 3Z \slashed{\nabla}\sigma \left(\widehat{\mathcal{L}}^n_T W\right) - 3 {}^\star Z \slashed{\nabla} \rho  \left(\widehat{\mathcal{L}}^n_T W\right) \, ,
\end{align}
\begin{align}
\mathcal{G}^n_2 = -2\left(K\left[\mathfrak{J}^n\right]\right)_3 + 4\underline{\Omega}\left(K\left[\mathfrak{J}^n\right]\right) + 2 \slashed{curl} \left(I\left[\mathfrak{J}^n\right]\right) +3   \rho_0 \cdot \sigma\left(\widehat{\mathcal{L}}^n_T W\right) \, ,
\end{align}
and $\underline{\mathcal{G}}_i^n$ being the conjugate quantities. 
\end{proposition}
For the $n=0$ case, one can use the null-structure equations $\slashed{curl} Z = - \frac{1}{2} \hat{H} \wedge \underline{\hat{H}} +  Y \wedge \underline{Y} + \sigma$ plus its conjugate to conclude
\begin{corollary} \label{sigma0}
For $n=0$, one obtains
\begin{align}
\left(\slashed{D}_4 \slashed{D}_3 + \slashed{D}_3 \slashed{D}_4 \right) \sigma - 2\slashed{\nabla}^2 \sigma  
+ \frac{3}{2} \left(2tr H - \frac{4}{3}\Omega \right) \slashed{D}_3 \sigma + \frac{3}{2} \left(2tr \underline{H} - \frac{4}{3}\underline{\Omega} \right) \slashed{D}_4\sigma
\nonumber \\ 
+ \sigma  \left(3 tr \underline{H} tr H +12\rho \left(W\right)+ \left(\textrm{dec. RRC}\right)^2 \right) = \mathcal{G}^0_1 + \underline{\mathcal{G}}^0_1 + \hat{H} \wedge \underline{\hat{H}} + 2 Y \wedge \underline{Y} \, . \nonumber
\end{align}
Note that the right hand side decays quadratically.
\end{corollary}
\subsection{The integrated decay estimate} \label{multestrho}
We apply general multiplier identities to equation (\ref{phicommute}) to derive energy estimates.  As previously mentioned, decay for the (homogeneous) Regge-Wheeler equation on Schwarzschild has been shown (\cite{BlueSoffer, BlueSoffer2}). In this section we present an alternative multiplier approach based on ideas of \cite{DafRod2}. The main result is Proposition \ref{finXest} which together with the error-estimate of Proposition \ref{errorestrho} will finally imply Proposition \ref{rhosigmaest}.

Multiplying (\ref{phicommute}) by the multiplier
\begin{equation}
X = \left[\mathfrak{a} \left(r\right) T + \mathfrak{b}\left(r\right) L + \mathfrak{c}\left(r\right)\right] \phi_k
\end{equation}
with $\mathfrak{a},\mathfrak{b}, \mathfrak{c}$ being $C^1$ bounded functions, one derives (let us momentarily drop the subscript $n$)
\begin{lemma} \label{mulemma}
We have the following basic multiplier identity
\begin{align}
e^X \left[\mathcal{D}\phi\right] + K^X \left[\mathcal{D} \phi\right] = Err\left[\mathcal{D} \phi\right] + \left[r^3 \mathcal{F}^k  \right] \left(X \phi_k\right) \, ,
\end{align}
where
\begin{align}
e^X \left[\mathcal{D}\phi\right] = \nabla_T \Bigg[ \frac{1}{2} \mathfrak{a} \left(2-p\right) |T\phi|^2 + \left(\frac{1}{2} \mathfrak{a} p - \mathfrak{b} \left(1-p\right)\right) | L \phi|^2 \nonumber \\
+ \frac{1}{2} \mathfrak{a} q  |\slashed{\nabla} \phi|^2 + \frac{1}{2} \left( -\frac{6M + \mathcal{E}}{r^3} \mathfrak{a} q + \mathfrak{c} \ L \left(p\right) \right) \phi^2 
+\mathfrak{b} \left(2-p\right)  T\left(\phi\right) L \left(\phi\right) \nonumber \\ + \mathfrak{c} \left(2-p\right)\  T\left(\phi\right) \cdot \phi -2\mathfrak{c} \left(1-p\right) L\left(\phi\right) \cdot \phi 
 \Bigg] \nonumber \\
+ \nabla_L \Bigg[ \left(-\mathfrak{a}\left(1-p\right) - \frac{1}{2} \mathfrak{b} \left(2-p\right) \right) |T\phi|^2 - \frac{1}{2}\mathfrak{b} p |L\phi|^2 + \frac{1}{2} \mathfrak{b} q | \slashed{\nabla} \phi|^2 \nonumber \\ + \frac{1}{2} \left( p \cdot L\left(\mathfrak{c}\right) - \mathfrak{b}q \frac{6M + \mathcal{E}}{r^3} \right) \phi^2 -\mathfrak{a}p L\left(\phi\right) T \left(\phi\right) - p \mathfrak{c} \left(L\phi\right) \phi \Bigg] \nonumber \\
+ \slashed{\nabla} \Bigg[ -\mathfrak{a} q \slashed{\nabla} \phi \cdot T \phi - \mathfrak{b} q  \slashed{\nabla} \phi \cdot L \phi - \mathfrak{c} q \left(\slashed{\nabla} \phi\right) \phi \Bigg] \, ,
\end{align}
\begin{align}
K^X  \left[\mathcal{D} \phi\right] = \left[ \left(1-p\right) L \left(\mathfrak{a}\right) + \frac{1}{2} L \left(\mathfrak{b} \left(2-p\right)\right) - \mathfrak{c} \left(2-p\right)  \right] |T\phi|^2 \nonumber \\
+  \left[ \frac{1}{2} p L \left(\mathfrak{b}\right) - \frac{1}{2} \mathfrak{b} L \left(p\right) + p \mathfrak{c} \right] |L\phi|^2 + \left[ -\frac{1}{2} L \left(\mathfrak{b} q\right) + \frac{\mathfrak{b}q}{r} + \mathfrak{c} q \right] |\slashed{\nabla} \phi|^2 \nonumber \\
+ T\left(\phi\right) \cdot L \left(\phi\right) \left[p L \left(\mathfrak{a}\right) + \mathfrak{b} L \left(p\right) + 2 \mathfrak{c} \left(1-p\right)  \right] \nonumber \\
+ \left[\frac{1}{2} L \left(\mathfrak{b}q \frac{6M + \mathcal{E}}{r^3} \right) - \frac{1}{2} L \left( p L\left(\mathfrak{c}\right)\right) - q \mathfrak{c} \frac{6M+\mathcal{E}}{r^3} \right] \phi^2  \nonumber 
\end{align}
and the error-term is cubically decaying, schematically of the form
\begin{align}
Err\left[\mathcal{D} \phi\right] = \left(\textrm{decaying RRC, ${}^{(T)}\pi$} \right) \cdot \mathcal{D}\phi \cdot \left(\mathcal{D} \phi \textrm{\ \ or \ } \phi\right) \, .
\end{align}
\end{lemma}
Multiplying this identity by $\frac{\sqrt{\slashed{g}}}{r^2}$ and then integrating over $\tilde{\mathcal{M}}\left(\tau_1,\tau_2\right)$ with respect to the measure $dt^\star dr d\theta d\phi$, we realize that $e^X \left[\mathcal{D}\phi\right]$ is a boundary term modulo terms of the form $Err\left[\mathcal{D}\phi\right]$.

With the Lemma at hand, we can prove an integrated decay estimate for the Regge-Wheeler-type equation under consideration. In view of the trapping, the best one can hope for is to generate a spacetime-integral term which provides non-degenerate decay for a particuar radial derivative and degenerate (at the photon sphere) integrated decay for the remaining derivatives of $\phi$. In the context of the wave equation, obtaining positivity for the zeroth order terms in $\phi$ is intricate and typically involves the use of a Poincar\'e inequality exploiting the positivity of the angular-derivative term. In our case, the situation is more favorable, as we are assuming convergence to Schwarzschild. In particular, the assumption of being ultimately Schwarzschildean to order $k+1$ implies that the integrated, non-degenerate spacetime energy of $k-1$ derivatives of curvature is bounded. This means that in the putative degenerate integrated decay estimate for \emph{derivatives} of $\phi_{k-1}$, the lowest order terms $\phi_{k-1}$ are a-priori controlled from the ultimately Schwarzschildean assumption. Moreover, we can exploit the fact that the degenerate integrated decay estimate for  derivatives of $\phi_{k-1}$ will provide control over a particular radial derivative of $\phi_{k-1}$ near the photon sphere (cf.~the energy (\ref{strho})). Indeed, this allows to turn non-degenerate decay into degenerate decay plus a small amount of the non-degenerate derivative term, as is spelled out in the  following Lemma:

\begin{lemma} \label{borrowhigh}
For any $\lambda_3>0$ and $n\geq 1$
\begin{align} \label{lemst}
\int_{\tilde{\mathcal{M}}\left(\tau_1,\tau_2\right)} dt^\star dr d\omega_2 \frac{1}{r^{3+\delta}} \phi_n^2 \leq \lambda_3 \cdot I_{deg} \left[\mathcal{D}\phi_n\right] \left(\tilde{\mathcal{M}}\left(\tau_1,\tau_2\right)\right) \nonumber \\ +\frac{B}{\lambda_3} \cdot I_{deg} \left[\mathcal{D}\phi_{n-1}\right]\left( \tilde{\mathcal{M}}\left(\tau_1,\tau_2\right)\right) + \epsilon  \sup_{\tau} E \left[\mathcal{D}\phi_n\right]\left(\tilde{\Sigma}_\tau\right)
\end{align}
\end{lemma}
\begin{proof}
Decompose
\begin{align}
\phi_k = \phi_{k}^{photon} + \phi_{k}^{away} = \chi \left(r\right) \cdot \phi_k + \left(1-\chi \left(r\right)\right) \phi_k
\end{align}
where $\chi\left(r\right)$ is a function which is equal to $1$ in $\left[3M-\frac{1}{8}M, 3M+\frac{1}{8}M\right]$ and equal to zero outside of $\left[3M-\frac{1}{4}M, 3M+\frac{1}{4}M\right]$.
For $\phi_k^{away}$, the estimate (\ref{lemst}) holds trivially requiring only the second term on the right hand side. For $\phi_k^{photon}$ one has the estimate (recall that $T^\perp = \frac{2M}{r} \partial_{t^\star} + \left(1-\mu\right) \partial_r$ in Schwarzschild)
\begin{align}
\chi^2 \phi^2 \leq 6\chi^2 \phi^2 \left( T^\perp \left(r-3M\right)\right) \nonumber \\
= 6 T^\perp \left(\chi^2 \phi^2 \left(r-3M\right)\right) - 12 \chi^2 \left(r-3M\right) \phi \left(T^\perp \phi\right) -12 \chi T^\perp \left(\chi\right) \left(r-3M\right) \phi^2 \nonumber \\
\leq 6 T^\perp \left(\chi^2 \phi^2 \left(r-3M\right)\right) +  \frac{B}{\lambda_3} \left(r-3M\right)^2 \phi^2 + \lambda_3  \left(T^\perp \phi\right)^2 \, , \nonumber
\end{align}
since $T^\perp \left(\chi\right)=0$ in $\left[3M-\frac{1}{8}M, 3M+\frac{1}{8}M\right]$. Integrating the estimate with respect to $\int \frac{1}{1-\mu} dt^\star dr d\omega$ gives the estimate of the Lemma. The contribution from the boundary term on $\tilde{\Sigma}_{\tau}$ is $\epsilon$-small in view of the ultimately Schwarzschildean assumption ($g\left(T^\perp, \partial_{t^\star}\right) \sim \epsilon$).
\end{proof}

\begin{proposition} \label{finXest}
We have for $n\geq 1$
\begin{align} \label{finXesteq}
E\left[\mathcal{D}\phi_n\right]\left( \tilde{\Sigma}_{\tau_2}\right) +  E\left[\mathcal{D}\phi_n\right]\left( \mathcal{H}\left(\tau_1,\tau_2\right) \right) + I_{deg}  \left[\mathcal{D}\phi_n\right] \left( \tilde{\mathcal{M}}\left(\tau_1,\tau_2\right)\right) \leq \nonumber \\ 
 B \cdot E\left[\mathcal{D}\phi_n\right]\left( \tilde{\Sigma}_{\tau_1}\right) + B\Big|\int dt^\star dr d\omega r^3 \mathcal{F}^n \left(X\phi_n\right)  \Big| \nonumber \\
B \cdot I_{deg} \left[\mathcal{D}\phi_{n-1}\right]\left( \tilde{\mathcal{M}}\left(\tau_1,\tau_2\right)\right) + B \cdot E\left[\mathcal{D}\phi_{n-1}\right]\left( \tilde{\Sigma}_{\tau}, \mathcal{H}^+\right)  \, .
\end{align}
\end{proposition}

\begin{proof}
Consider the following multiplier:
\begin{align} \label{Xmulti}
X = B_1 X_1 + b_2 X_2 + X_3 + b_4 X_4
\end{align}
for $B_1$ a large and $b_2, b_4$ small constants depending only on $M$ (to be chosen) and
\begin{align}
X_1 =T \textrm{ \  , \ } X_2 =  -\gamma\left(r\right) L \textrm{ \  ,  \ } X_3 =\left(1-p\right) f \left(r\right) T + f\left(r\right)p \ L + \frac{1}{2} p \ L \left(f\left(r\right)\right) \nonumber
\end{align}
where we recall the redshift multiplier $\gamma \left(r\right)$ defined in (\ref{gammdefo}) and 
\begin{align} \label{fnchoice}
f\left(r\right) = \left(1-\frac{3M}{r}\right)\left(1+\frac{M}{r}\right) \, .
\end{align}
Finally,
\begin{align}
X_4 = \frac{\left(r-3M\right)^2}{r^4} \tilde{\xi} \left(r\right)  \, .
\end{align}
where $\tilde{\xi}\left(r\right)$ is a cut-off function supported for $r \geq r_Y - \frac{r_Y-2M}{2} > 2M$ and equal to $1$ for $r \geq r_Y$. From Lemma \ref{mulemma} we compute, the integration being over $ \tilde{\mathcal{M}}\left(\tau_1,\tau_2\right)$),
\begin{align}
\int K^{X_3}\left[\mathcal{D}\phi\right]  = \int dt^\star dr d\omega \Bigg[ L\left(f\right) \left( \left(1-p\right)T\left(\phi\right) + p L \left(\phi\right) \right)^2 \nonumber \\
- \frac{1}{4} \left(\phi\right)^2 \cdot L \left(p L \left(p L \left(f\right) \right) \right) - \frac{1}{2} f r^2 L \left(\frac{ pq}{r^2}\right) |\slashed{\nabla} \phi|^2 + L \left[\frac{3M+ \mathcal{E}}{r^3}pq\right] f |\phi|^2  \Bigg] \, . \nonumber
\end{align}
Since $f$ is monotonically increasing, the first term is already non-negative. Moreover, as $f$ changes sign at $3M$, the angular term is seen to be non-negative, too (modulo a term of the form $Err\left[\mathcal{D} \phi\right]$; in Schwarzschild $L\left(\frac{pq}{r^2}\right) = -\frac{2}{r^3}\left(1-\frac{3M}{r}\right)$). Finally, the lower order terms in $\phi$ itself are treated as error-terms and will eventually be controlled with an application of Lemma \ref{borrowhigh} and the ultimately Schwarzschildean assumption. Hence up to now we have shown the estimate
\begin{align}
B \int_{ \tilde{\mathcal{M}}\left(\tau_1,\tau_2\right)} \frac{\phi^2}{r^4} dt^\star dr d\omega + \int_{ \tilde{\mathcal{M}}\left(\tau_1,\tau_2\right)} K^{X_3}\left[\mathcal{D}\phi\right] \nonumber \\
\geq b \int_{ \tilde{\mathcal{M}}\left(\tau_1,\tau_2\right)} dt dr^\star d\omega \Bigg[ \frac{1}{r^2}\left( \left(1-p\right)T\left(\phi\right) + p L \left(\phi\right) \right)^2 +  \frac{\left(r-3M\right)^2}{r^3} |\slashed{\nabla} \phi|^2 \Bigg] \, , \nonumber
\end{align}
which controls all derivatives except the $\partial_{t^\star}$-derivative. The latter is retrieved by adding a little bit of $X_4$. It's not hard to see that for our choice of $X_4$ (note that $X_4$ is supported away from the horizon) we have
\begin{align}
B \int_{ \tilde{\mathcal{M}}\left(\tau_1,\tau_2\right)} \frac{\phi^2}{r^4} dt^\star dr d\omega + \int_{ \tilde{\mathcal{M}}} K^{X_4} \left[\mathcal{D}\phi\right] + B \int_{ \tilde{\mathcal{M}}\left(\tau_1,\tau_2\right)} K^{X_3}\left[\mathcal{D}\phi\right] \nonumber \\ \geq b \int_{ \tilde{\mathcal{M}} \cap \{ r\geq r_Y \}} dt^\star dr d\omega \Bigg[  \frac{1}{r^2}\left( T^\perp \left(\phi\right)  \right)^2 +  \frac{\left(r-3M\right)^2}{r^3} \left(|\slashed{\nabla} \phi|^2 +\left(T\phi \right)^2 \right) \Bigg]   \, . \nonumber
\end{align}
Turning to the redshift vectorfield $X_2$, we observe
\begin{align}
 \int_{ \tilde{\mathcal{M}} }K^{X_2} \left[\mathcal{D}\phi\right] +  B \int_{ \tilde{\mathcal{M}} \cap \{ \frac{5}{2}M \geq r\geq r_Y \}} dt^\star dr d\omega \Bigg[ |\phi|^2 + \left(\mathcal{D}\phi\right)^2 \Bigg] \nonumber \\
 \geq b \int_{ \tilde{\mathcal{M}} \cap \{ r\leq r_Y \}} dt^\star dr d\omega \Bigg[ |\phi|^2   +|\slashed{\nabla} \phi|^2 +\left(T \phi \right)^2 +  \left(\partial_r \phi\right)^2 \Bigg]  \, . \nonumber
\end{align}
Consequently, we can find small constants $b_2, b_4$ such that (reinstalling the subindex $n$)
\begin{align} \label{gos2}
\int_{ \tilde{\mathcal{M}}\left(\tau_1,\tau_2\right)} K^{b_2 X_2+ X_3 b_4 X_4} \left[\mathcal{D}\phi_n\right]\geq b \cdot I_{deg}^{\delta=1}  \left[\mathcal{D}\phi_n\right] \left( \tilde{\mathcal{M}}\left(\tau_1,\tau_2\right)\right) \nonumber \\
-B \int_{ \tilde{\mathcal{M}}\left(\tau_1,\tau_2\right)} \frac{\phi_n^2}{r^4} dt^\star dr d\omega \, ,
\end{align}
where $I_{deg}^{\delta=1}  \left[\mathcal{D}\phi_n\right] \left( \tilde{\mathcal{M}}\left(\tau_1,\tau_2\right)\right)$ denotes the energy (\ref{strho}) with $\delta=1$. The last term is estimated by Lemma \ref{borrowhigh}.

Before we optimize the weight at infinity to the $\delta=\frac{1}{100}$ we specified, let us note that by choosing $B_1$ large enough, we have for the boundary term
\begin{align} \label{gos1}
\int_{\tilde{\Sigma}_\tau, \mathcal{H}^+} e^{X} \left[\phi_n\right] \geq b \cdot E\left[\mathcal{D}\phi_n\right]\left( \tilde{\Sigma}_{\tau}, \mathcal{H}^+\right) -B \cdot E\left[\mathcal{D}\phi_{n-1}\right]\left( \tilde{\Sigma}_{\tau}, \mathcal{H}^+\right) \, .
\end{align}
This follows because $X$ is timelike for sufficiently large $B_1$ (noting that both $X_3$ and $X_4$ degenerate at the horizon) and that moreover, we can estimate the zeroth order term by the lower order energy.
Finally,
\begin{align}
\Big| \int_{ \tilde{\mathcal{M}}\left(\tau_1,\tau_2\right)} Err \left[\mathcal{D}\phi_n\right] \Big| \leq \epsilon \sup_{\tau}  E\left[\mathcal{D}\phi_n\right]\left( \tilde{\Sigma}_{\tau}\right) \, .
\end{align}

In summary, we have established Proposition \ref{finXest}, except that the spacetime-energy does not have the correct $r$-weights at infinity yet. As this optimization is standard, it is omitted.\footnote{Here one can simply redo the argument with an $f$ in (\ref{fnchoice}) which is bounded, monotonically increasing, changing sign at the photon sphere and decaying like $f_{,r} \sim \frac{1}{r^{1+\delta}}$ near infinity.}
\end{proof}

\subsection{Error-terms: The structure of the inhomogeneity} \label{structureRHS}
It is apparent from Proposition \ref{finXest} that we finally need to understand the error-term arising from the inhomogeneity. There are two important regions: The interior region, where we need to have sufficient decay in $t$ and be careful about the degeneration of the estimates near $r=3M$, and the exterior where we do not encounter the problem of degeneration but, on the other hand, need to worry about the weights in $r$ for the error-terms that arise, especially in view of the fact that the renormalized quantities have introduced a large $r^3$-weight for these terms. 

\begin{proposition} \label{errorestrho}
Let $n\geq 1$ and $X$ be the multiplier used in (\ref{Xmulti}). For any $\lambda>0$ we have
\begin{align}
\sum_{i=0}^2 \int_{\tilde{\mathcal{M}}\left(\tau_1,\tau_2\right)} dt^\star dr d\omega_2 r^3 \left(\mathcal{F}^n_i \right) \left(X\phi_n \right) \nonumber \\
\leq \lambda \cdot I_{deg} \left[\mathcal{D}\phi_n\right] \left(\tilde{\mathcal{M}}\left(\tau_1,\tau_2\right)\right) + \lambda \cdot \sup_\tau \, E \left[\mathcal{D}\phi_n\right]\left( \tilde{\Sigma}_\tau\right) + \lambda \cdot E \left[\mathcal{D}\phi_n\right] \left(\mathcal{H}\right)\nonumber \\ 
+ \epsilon \Big[\overline{\mathbb{I}}^{n+1,deg}_{\tilde{P}_{\rho}} \left[W\right] \left(\tilde{\mathcal{M}}\left(\tau_1,\tau_2\right)\right) +  \sup_{\tau} \overline{\mathbb{E}}^{n+1} \left[W\right] \left(\tilde{\Sigma}_{\tau}\right) \nonumber \\
+ \sup_{\tau} \overline{\mathbb{E}}^{n+1} \left[\mathfrak{R}\right] \left(\tilde{\Sigma}_{\tau}\right) + \overline{\mathbb{I}}^{n} \left[\mathfrak{R}\right] \left(\tilde{\mathcal{M}}\left(\tau_1,\tau_2\right)\right)   \Big] \nonumber \\
+ B_\lambda \cdot I_{deg} \left[\mathcal{D}\phi_{n-1}\right] \left( \tilde{\mathcal{M}}\left(\tau_1,\tau_2\right)\right) + \lambda \cdot \mathbb{D}^{n+1} \left[\mathfrak{R}\right] \left(\tau_1,\tau_2\right) + B_{\lambda} \cdot \mathbb{D}^{n} \left[\mathfrak{R}\right] \left(\tau_1,\tau_2\right) \, . \nonumber
\end{align}
\end{proposition}
\begin{remark}
Note that the terms in the first line will be absorbed by the left hand side of (\ref{finXesteq}) for sufficiently small $\lambda$.
\end{remark}
\begin{proof}
We begin with the term $\boxed{\mathcal{F}_1^n}$. It is obviously an expression which decays at least quadratically and hence improves upon taking derivatives. For $r\leq R$ we can exploit the pointwise decay of the Ricci coefficients (and first derivatives thereof):
\begin{align}
\int_{\mathcal{M}\left(\tau_1,\tau_2\right), r\leq R} dt^\star dr d\omega_2 r^3 \left(\mathcal{F}^n_1 \right) \left(X\phi \right) \nonumber \\
\leq \epsilon \cdot \sup_{\tau} \sqrt{E\left(\mathcal{D} \phi_n, \Sigma_{\tau, r\leq R}\right)} \cdot \sup_{\tau} \sqrt{\overline{\mathbb{E}}^{n+1}\left[W\right] \left(\Sigma_{\tau, r\leq R}\right)} \, .
\end{align}
For $r \geq R$ we can borrow a bit from the weighted energy for the Weyl field and estimate the $X\phi_n$-term in a spacetime integral, as we are far away from the trapped region:
\begin{align}
\int_{\tilde{\mathcal{M}}\left(\tau_1,\tau_2\right), r\geq R} dt^\star dr d\omega_2 r^3 \left(\mathcal{F}^n_1 \right) \left(X\phi_n \right) \nonumber \\
\leq \lambda_1 \cdot I_{deg} \left[\mathcal{D}\phi_n\right] \left( \tilde{\mathcal{M}}\left(\tau_1,\tau_2\right)\right) + \epsilon \cdot B_{\lambda_1} \cdot \overline{\mathbb{I}}^{n+1,deg}_{\tilde{P}_{\rho}} \left[W\right] \left(\tilde{\mathcal{M}}\left(\tau_1,\tau_2\right)\right)
\end{align}
with the $\epsilon$ arising again from the pointwise bounds on the Ricci-coefficients. The same estimate holds for the conjugate quantity $\underline{\mathcal{F}}^n_1$.

Turning to $\boxed{\mathcal{F}_2^n}$ we have to understand the structure of $\Lambda=J_{434}$ and $I = J_{3A4}$ and Lie-$T$-derivatives thereof, in particular the asymptotic decay in $r$ on the one hand and the terms not decaying quadratically on the other. Indeed, it suffices to understand 
$\Lambda$ and $I$ as arising from one commutation in $T$, since the null-decomposition almost commutes with the $T$-derivative:
\begin{align}
\Lambda\left[\mathfrak{J}\right] = \widehat{\mathcal{L}}^{n-1}_T \left[\Lambda \left(J\left(T,W\right) \right)\right] + \sum_{i=0}^{n-2} \widehat{\mathcal{L}}^i_T \Lambda \left(J \left(T, \widehat{\mathcal{L}}^{n-1-i}_T \widehat{\mathcal{L}}_T W\right)\right) + \mathfrak{J}^{error}_{434} \, ,
\end{align}
where the last two terms on the right hand side decay quadratically. We first split
\begin{align}
\mathcal{F}_2^n = \mathcal{F}_{2a}^n + \mathcal{F}_{2b}^n \, ,
\end{align}
where $ \mathcal{F}_{2a}^n$ collects the terms in which at least two factors decay and  $\mathcal{F}_{2b}^n$ are the potentially dangerous ``linear" terms, i.e.~those proportional to $\rho_0$.

For  $\boxed{\mathcal{F}_{2a}^n}$ we decompose
\begin{align} \label{split}
 r^3 \mathcal{F}^n_{2a} \cdot \left(X\phi \right) = r^3 \mathcal{F}^n_{2a} \cdot \left(X\phi \right) \cdot \chi\left(r\right) 
+ \mathcal{F}^n_{2a} \cdot \left(X\phi \right)  \cdot \left(1-\chi\left(r\right)\right) = \mathfrak{A} + \mathfrak{B}
\end{align}
where $\chi\left(r\right)$ is an interpolating function which is equal to $1$ for $r \leq R-M$ and $0$ for $r \geq R$. Note that $\mathfrak{A}$ is supported for $r\leq R$ only, while $\mathfrak{B}$ is supported in $r\geq R-M$. For $\mathfrak{B}$ we have
\begin{align}
\int_{\tilde{\mathcal{M}}\left(\tau_1,\tau_2\right)} dt^\star dr d\omega \ r^3 \mathcal{F}^n_{2a} \cdot \left(X\phi_n \right) \cdot \chi\left(r\right) \leq \lambda_1 \cdot I_{deg} \left[\mathcal{D}\phi\right] \left(\tilde{\mathcal{M}}\left(\tau_1,\tau_2\right) \right) \nonumber \\
+ \frac{1}{\lambda_1} \int_{\tilde{\mathcal{M}}\left(\tau_1,\tau_2\right) \cap \{r\geq R-M\}} r^{7+\delta} \chi^2 \left( \mathcal{F}_{2a}^n \right)^2 \, .
\end{align}
For the second term we show
\begin{lemma} \label{infstudy}
\begin{align}
\int_{\tilde{\mathcal{M}}\left(\tau_1,\tau_2\right) \cap \{r\geq R-M\}} r^{7+\delta} \left( \mathcal{F}_{2a}^n \right)^2 \leq \epsilon  \cdot \overline{\mathbb{I}}^{n+1}_{\tilde{P}_{\rho}}\left[W\right] \left(\mathcal{D}^{\tau_2}_{\tau_1}\right) + \epsilon \cdot \overline{\mathbb{I}}^{n} \left[\mathfrak{R}\right] \left(\tilde{\mathcal{M}}\left(\tau_1,\tau_2\right) \right)  \nonumber\, .
\end{align}
\end{lemma} 
\begin{proof}
To estimate these quadratic error-terms in the far away region, we only have to be careful about the weights in $r$. To reveal the structure which is present for the error-terms (cancellation of worst terms, cf.~section \ref{caworst}), we are going to think about the wave equation for $\rho$ in yet another way. At the zeroth order we can write
\begin{align} \label{rholod}
\rho_{34} + \rho_{43} - 2 \slashed{\nabla}^2 \rho = - \slashed{div} E_4 \left(\underline{\beta}\right)  + E_3^4 \left(\rho\right) + \slashed{div} E_3 \left({\beta}\right)  + E_4^3 \left(\rho\right) \, .
\end{align}
We now imagine commuting this equation with the Lie-$T$-derivative directly to obtain the higher derivative version of the equation. Note that this commutation will not produce any additional difficult terms at infinity, but simply Lie-$T$-derivatives of the right hand side of (\ref{rholod}).\footnote{It may produce terms proportional to $\rho$ from commuting $T$ past the derivatives on the left hand side. However, these decay strongly in $r$.} However, the decay in $r$ of the right hand side is not altered by application of $T$. It follows that it suffices to understand the decay in $r$ of the right hand side of (\ref{rholod}), more precisely, to establish that it decays better than $r^{-4-\delta}$ pointwise. We note in passing that the gauge condition $Y=0$ is not needed here and that we will hence carry out the estimates including this term. 

Let $\equiv$ denote equality up to terms which decay like $r^{-\frac{9}{2}}$ or stronger. We observe that
\begin{align}
- \slashed{div} E_4 \left(\underline{\beta}\right) + \slashed{div} E_3 \left({\beta}\right) \equiv \slashed{div} \left(Y \underline{\alpha}\right) 
\end{align}
and, noting that in view of the null structure equation (\ref{Y3}) $\slashed{D}_3 Y$ actually decays like $\frac{1}{r^3}$ (instead of the naive $\frac{1}{r^2}$) also
\begin{align}
E_3^4 \left(\rho\right) + E_4^3 \left(\rho\right) \equiv \slashed{\mathcal{D}}_2^\star Y \cdot \underline{\alpha} + Y^A \slashed{D}_3 \underline{\beta} -2Y {D}_3 \underline{\beta} \, .
\end{align}
Adding up these terms one sees that the worst terms cancel by the Bianchi equation (\ref{Bianchi4}), producing terms decaying like $r^{-5}$.
\end{proof}
We emphasize that it is absolutely crucial to exploit the cancellation of the worst error-terms as otherwise a $r^{-\delta}$-divergence would arise.

For the term $\mathfrak{A}$ in (\ref{split}) we do not have to worry about $r$-weights at all, only about the order of differentiability in view of the degeneration at the photon sphere. However, all the terms in the expression for $\mathcal{F}_{2a}^n$ decay quadratically. This makes the error-term cubic and the estimates are straightforward, putting the highest terms in $L^2$ and the lowest in $L^\infty$. Using Sobolev embedding, we obtain
\begin{align}
\int_{\tilde{\mathcal{M}}\left(\tau_1,\tau_2\right), r \leq R}r^3 \mathcal{F}^n_{2a} \Big|_{\mathfrak{A}} \cdot \left(X\phi_n \right) \leq \left(\tau_2-\tau_1\right) \sup_{\tau} \sqrt{E\left[\mathcal{D} \phi_n\right]\left(\tilde{\Sigma}_{\tau}\right)} \Bigg[ \nonumber \\ \sum_{i=0}^{\lfloor \frac{n-1}{2} \rfloor} \sup_{\tau} \sqrt{\overline{\mathbb{E}}^{i+2} \left[\mathfrak{R}\right] \left(\tilde{\Sigma}_{\tau}\right)} \sup_{\tau} \sqrt{\overline{\mathbb{E}}^{n+1-i} \left[W\right] \left(\tilde{\Sigma}_{\tau}\right)}   \nonumber \\ +
\sum_{i=\lfloor \frac{n-1}{2} \rfloor+1}^{n+1} \sup_{\tau} \sqrt{\overline{\mathbb{E}}^{n+3-i} \left[W\right] \left(\tilde{\Sigma}_{\tau}\right)} \sup_{\tau} \sqrt{\overline{\mathbb{E}}^{i} \left[\mathfrak{R}\right] \left(\tilde{\Sigma}_{\tau}\right)} \Bigg] \leq \nonumber \\ 
\lambda \cdot E\left[\mathcal{D} \phi_n\right] \left( \tilde{\Sigma}_{\tau}\right) + B_\lambda \left(\tau_2-\tau_1\right)^2   \sup_\tau \overline{\mathbb{E}}^{\lfloor \frac{n-1}{2}\rfloor+2} \left[\mathfrak{R}\right] \left(\tilde{\Sigma}_{\tau}\right) \cdot \sup_{\tau} \overline{\mathbb{E}}^{n+1} \left[W\right] \left(\tilde{\Sigma}_{\tau}\right) \nonumber \\ + B \left(\tau_2-\tau_1\right)^2\sup_\tau \overline{\mathbb{E}}^{\lfloor \frac{n-1}{2}\rfloor+1} \left[W\right] \left(\tilde{\Sigma}_{\tau}\right) \cdot \sup_{\tau} \overline{\mathbb{E}}^{n+1} \left[\mathfrak{R}\right] \left(\tilde{\Sigma}_{\tau}\right) \, .\nonumber 
\end{align}
and in view of the decay assumptions on the lower oder energies (recall the ultimately Schwarzschildean assumption and $n\leq k-1$), the $\tau$-weights can be absorbed, producing the terms on the right hand side of the proposition. 

It remains to estimate $\boxed{\mathcal{F}^n_{2b}}$, i.e.~the terms proportional to $\rho_0$. These terms are particularly difficult near the photon sphere $r=3M$ (in fact, near infinity they are harmless, in view of the strong decay of $\rho_0$). We first note that it suffices to understand the case $k=1$, since higher $k$ are just $T$-Lie-derivatives of this expression (and commutation by $\mathcal{L}_T$ introduces additional decay factors). Hence we want to compute the expression
\begin{align}
 \mathcal{F}^1_{2b} = -2 \Bigg[ \slashed{\nabla}_3 \Lambda_0 + \frac{3}{2} tr \underline{\chi} \Lambda_0 - 2\underline{\Omega} \Lambda_0 + \slashed{div} I_0 \nonumber \\ + \slashed{\nabla}_4 \underline{\Lambda}_0 + \frac{3}{2} tr {\chi} \underline{\Lambda}_0 - 2{\Omega} \underline{\Lambda}_0 + \slashed{div} \underline{I}_0 - 3 \rho\left(W\right) \cdot \rho \left(\hat{\mathcal{L}}_T W\right) \Bigg] \nonumber \, ,
\end{align}
where 
\begin{align}
\Lambda_0 = \frac{3}{4} \rho_0 \Big[ -\frac{1}{2}\slashed{D}_4 \left(tr \pi\right) - \slashed{\nabla}^A \mathcal{P}_A +2 p \rho - \frac{1}{2} tr H \left(2\slashed{\nabla}_3 p - 4 \underline{\Omega} p\right) \nonumber \\ -tr \chi \pi_{34} - \frac{1}{2} tr \underline{H} \pi_{44} - \frac{1}{2} tr H \delta^{AB} \pi_{AB} \Big]  \nonumber \, ,
\end{align}
\begin{align}
I_0 = \frac{3}{4} \rho_0 \left[ 2q \beta_A - \psi_{34} Z_A - 2\underline{\Omega} \mathcal{P}_A - tr \underline{\chi} \hat{\pi}^{3}_{\phantom{3}A} - 2 tr \chi  \hat{\pi}^{4}_{\phantom{4}A} + \slashed{\nabla}_3 \mathcal{P}_A + \frac{1}{2} \slashed{\nabla}_A \left(tr \pi \right) \right] \nonumber 
\end{align}
(and similarly for the conjugate\footnote{The conjugate quantity of $\beta$ is $-\underline{\beta}$, while all other quantities just receive an underline and $p$ becomes $q$.} quantities $\underline{\Lambda}_0$, $\underline{I}_0$) were computed at the end of the proof of Lemma \ref{ncLT}.

In fact, we will only need to do evaluate $\mathcal{F}_{2b}^1$ modulo quadratic error-terms (as the latter immediately transfer to $\mathcal{F}_{2a}^n$). We make the a-priori convention that {\bf for the computations within the context of the errorterm $\mathcal{F}_{2b}^n$, ``$=$" means ``equality up to quadratically decaying terms".}

\begin{lemma} \label{Tchoice}
If the vectorfield $T$ is chosen such that its deformation tensor ${}^{(T)}\pi$ is traceless (which it is according to Definition \ref{ultS}), then ($n\geq 1$)
\begin{align} \label{esh}
\mathcal{F}^n_{2b}  = \rho_0 \cdot \mathcal{D} \mathcal{L}_T^{n-1} \left(\pi, \textrm{dec.~RRC}\right) + \textrm{l.o.t.} \, 
\end{align} 
\end{lemma}
\begin{proof}
Direct computation, using that
\begin{align}
2p \slashed{\nabla}_3 \rho + 2q \slashed{\nabla}_4 \rho - 2q \slashed{\nabla}_4 \rho = 2 T \left(\rho\right)  - 2q \slashed{\nabla}_4 \rho \nonumber \\
= 2 \rho \left(\hat{\mathcal{L}}_T W\right) + \frac{1}{4} tr \pi \rho + 3q tr \chi \rho -2q \slashed{div} \beta \, 
\end{align}
and 
\begin{align}
\slashed{div} I_0 = \frac{3}{4} \rho_0 \Big[ 2q \slashed{div} \beta_A - \psi_{34} \slashed{div} Z_A - 2\underline{\Omega} \slashed{div} \mathcal{P}_A - tr \underline{\chi} \slashed{div} \hat{\pi}^{3}_{\phantom{3}A}  \nonumber \\ - 2 tr \chi  \slashed{div} \hat{\pi}^{4}_{\phantom{4}A} + \slashed{div} \slashed{\nabla}_3 \mathcal{P}_A + \frac{1}{2} \slashed{\nabla}^2 \left(tr \pi \right) \Big] \, . \nonumber
\end{align}
\end{proof}
\begin{remark}
Note that the Lemma once more exploits a cancellation of the worst terms: Naively, (\ref{esh}) only holds with $n$ instead of $n-1$.
\end{remark}
\begin{remark} \label{trpremark}
If the deformation tensor is not traceless, $\mathcal{F}^n_{2b}$ contains an additional higher order term of the form $\rho_0 \mathcal{L}_T^{n-1} \left(\slashed{D}_3 \slashed{D}_4 + \slashed{D}_4 \slashed{D}_3- 2\slashed{\nabla}^2 \right) tr \pi$. However, the operator in brackets is the principal part of the Regge-Wheeler-operator. Hence via several integration by parts, moving the operators to $\phi_n$ in the spacetime-error-term arising from $\mathcal{F}^n_{2b}$, this term can also be seen to be of lower order. It is here where the technical assumption $tr {}^{(T)}=0$ (discussed below Definition \ref{ultS}) helps us to reduce the number of error-terms.
\end{remark}
With the Lemma at hand, we pick a cut-off function $\chi_\star \left(r\right)$, which is equal to 1 for $r\leq 10M$ and equal to $0$ for $r \geq 11M$ and decompose
\begin{align}
 \mathcal{F}^n_{2b}  =  \mathcal{F}^n_{2b} \cdot \chi_\star\left(r\right) +  \mathcal{F}^n_{2b}\left[1-\chi_\star\left(r\right)\right] \, .
\end{align}
We can then estimate away from the photon sphere
\begin{align}
 \int_{\tilde{\mathcal{M}}} dt^\star dr d\omega_2 r^3 \mathcal{F}^n_{2b} \left[1-\chi_\star\left(r\right)\right] \cdot \left(X\phi_n \right)  \leq \lambda \cdot I_{deg} \left[\mathcal{D}\phi_n\right] \left(\tilde{\mathcal{M}}\left(\tau_1,\tau_2\right)\right) + \frac{B}{\lambda} \mathbb{D}^{n} \left[\mathfrak{R}\right] \, , \nonumber
\end{align}
while for  the other term we integrate by parts, writing a shorthand $\int$ to denote $\int_{\tilde{\mathcal{M}}\left(\tau_1,\tau_2\right) \cap \{r \leq 11M \}}dt^\star dr d\omega$ and noting that weights in $r$ are bounded in the region under consideration
\begin{align}
\int r^3 \mathcal{F}^n_{2b} \cdot \chi_\star\left(r\right) \cdot \left(X\phi_n \right)
=  \int  r^3 \rho_0 \left(\mathcal{L}_T^{n-1} \mathcal{D} \left(\pi, \textrm{dec.~RRC}\right)\right) \cdot \left(X\phi_n \right)
\nonumber \\ 
\leq 
+\lambda_2 \cdot \sup_{\tau} E \left[\mathcal{D}\phi_n\right]\left(\tilde{\Sigma}_\tau,\mathcal{H}\right) + B_{{\lambda}_2} \cdot \sup_\tau \overline{\mathbb{E}}^n\left[\mathfrak{R}\right] \left(\tilde{\Sigma}_{\tau}, \mathcal{H} \right) \nonumber \\
\lambda_2 \int  \left(\mathcal{D}^2\mathcal{L}_T^{n-1} \left(\pi, \textrm{dec.~RRC}\right)\right)^2 + \frac{1}{ \lambda_2} \int  \phi_n^2+ B \int \left(\mathcal{D}\mathcal{L}_T^{n-1} \left(\pi, \textrm{dec.~RRC}\right)\right)^2  \nonumber \\   \leq  \frac{1}{ \lambda_2} \int   \phi_n^2 +\mathbb{D}^{n+1}_{\lambda_2} \left[\mathfrak{R}\right] + \lambda_2 \cdot \sup_{\tau} E \left[\mathcal{D}\phi_n\right]\left(\tilde{\Sigma}_\tau,\mathcal{H}\right) \nonumber
\end{align}
for any (small) number $\lambda_2>0$. Applying Lemma \ref{borrowhigh}  with a $\lambda_3$ much smaller than $\lambda_2$ to the first term in the last line finally yields
\begin{align}
\int r^3 \mathcal{F}^n_{2b} \cdot \chi_\star\left(r\right) \cdot \left(X\phi_n \right) \leq \lambda \cdot I_{deg} \left[\mathcal{D}\phi_n\right] \left(\tilde{\mathcal{M}}\left(\tau_1,\tau_2\right)\right) \nonumber \\ + \left(\lambda + \epsilon\right) \cdot \sup_{\tau} E \left[\mathcal{D}\phi_n\right]\left(\tilde{\Sigma}_\tau,\mathcal{H}\right) + {B}_{\lambda} \cdot I_{deg} \left[\mathcal{D}\phi_{n-1}\right]\left( \tilde{\mathcal{M}}\left(\tau_1,\tau_2\right)\right)  +\mathbb{D}^{n+1}_{\lambda} \left[\mathfrak{R}\right] \nonumber \, ,
 \end{align}
for any small $\lambda$. Clearly, the first two terms on the right hand side will be absorbed by the terms available on the left for sufficiently small $\lambda$ (depending only on $M$).
Note the importance of Lemma \ref{borrowhigh} to estimate this error-term, as it allows one to turn the non-degenerate into a degerate energy. 

Collecting our estimates for $\mathcal{F}^n_{1,2a,2b}$, $\underline{\mathcal{F}}^n_{1,2a,2b}$, Proposition \ref{errorestrho} is proven.
\end{proof}

Finally, we note that for the wave equation satisfied by $\sigma$ (and Lie-T-derivatives thereof), 
the terms $\mathcal{G}_1^n$ and $\mathcal{G}_2^n$ can be estimated in exactly the same fashion as $\mathcal{F}_1^n$ and $\mathcal{F}_2^n$ were for the $\rho$-equation. In fact, the estimates are much easier, as the ``linear" error-terms on the right hand side all cancel. This fact is also easily inferred from commuting the scalar equation of Corollary \ref{sigma0} directly with $T$ and using that $\sigma$ decays. We conclude

\begin{proposition} \label{errorestsigma}
The estimate of Proposition \ref{errorestrho} also holds replacing $\mathcal{F}$ by $\mathcal{G}$ everywhere.
\end{proposition}
\subsection{Proof of Proposition \ref{rhosigmaest}} \label{mnrs}
Combining Proposition \ref{finXest} with the error-estimate \ref{errorestrho} yields the estimate of Proposition \ref{rhosigmaest} for the case $n\geq 1$. For the case $n=0$ we run the argument of Proposition \ref{finXest} again, since $\phi_0$ satisfies the same equation by Corollary \ref{vecwrite}. The only difference is that Lemma \ref{borrowhigh} is now not available. Instead, we simply control the zeroth order curvature term from the energies on the Ricci-coefficients:
\begin{align}
\int_{\tilde{\mathcal{M}}\left(\tau_1,\tau_2\right)} \frac{\phi_0^2}{r^{3+\delta}} dt^\star dr d\omega \leq B \cdot \mathbb{D}^1 \left[\mathfrak{R}\right] \left(\tau_1,\tau_2\right) \, ,
\end{align}
as a consequence of the null-structure equations.
\section{Proofs of the main theorems}

\subsection{Uniform boundedness: Theorem \ref{theo1}}
Uniform boundedness for the top-order derivatives can be proven without proving a (degenerate) integrated decay estimate at this order. This observation was first made in the context of the wave equation on black hole backgrounds in \cite{DafRodKerr} (see also \cite{Mihalisnotes}) and is adapted here to the spin2-setting. We are going to use the slices $\Sigma$ (no ``ears") is this section.

\begin{proof}[Proof of Theorem \ref{theo1}]
To lighten the notation we write $\Sigma_i = \Sigma_{\tau_i}$, $\Sigma_i^{+} = \Sigma_{\tau_i} \cap \{r \geq r_Y-3\frac{r_Y-2M}{4}\}$ and $\Sigma_i^{-} = \Sigma_{\tau_i} \cap \{r \leq r_Y-3\frac{r_Y-2M}{4}\}$, $\mathcal{H}^2_1=\mathcal{H}\left(\tau_1,\tau_2\right)$, $\mathcal{M}^2_1=\mathcal{M}\left(\tau_1,\tau_2\right)$. The identity (\ref{mainid}) applied with three $T$-vectorfields yields, in view of Lemma \ref{Tboundaryterm} 
\begin{align}
E^0 \left[\tilde{W}\right]\left( \Sigma_{2}^{+} \right) + \sum_{j=1}^n E\left[\widehat{\mathcal{L}}_T^j W\right]\left(\Sigma_{2}^{+}  \right) \leq 
 B \cdot E^0 \left[\tilde{W}\right]\left( \Sigma_{1}^{} \right) + B \cdot \sum_{j=1}^n E\left[\widehat{\mathcal{L}}_T^j W\right]\left( \Sigma_{1}^{}  \right) \nonumber \\ 
 + \epsilon \cdot E^0 \left[\tilde{W}\right] \left(\Sigma_{2}^{-}, \Sigma_1^{-} \right) + \epsilon \sum_{j=1}^n E\left[\widehat{\mathcal{L}}_T^j W\right] \left(\Sigma_{2}^{-}, \Sigma_1^{-}\right)
 \nonumber \\ + B \sum_{j=1}^n \int_{\mathcal{M}_1^2} \left(K_1^{TTT} \left[\widehat{\mathcal{L}}_T^j W\right] + K_2^{TTT} \left[\widehat{\mathcal{L}}_T^j W\right] \right) \, . \nonumber
\end{align}
Using the elliptic estimate in the interior, Proposition \ref{ellipticint}, we obtain
\begin{align} \label{hole}
\overline{\mathbb{E}}_{r \geq r_Y-\frac{5}{8}\frac{r_Y-2M}{2}}^n\left[W\right] \left(\Sigma_2 \right) \leq B \cdot \overline{\mathbb{E}}^n\left[W\right] \left(\Sigma_1 \right) 
+ \epsilon \cdot \overline{\mathbb{E}}_{r \leq r_Y}^n\left[W\right] \left(\Sigma_2, \Sigma_1 \right) \nonumber \\
+ B \sum_{j=1}^n \int_{\mathcal{M}_1^2} \sum_{i=1,2} K_i^{TTT} \left[\widehat{\mathcal{L}}_T^j W\right] + B \sup_\tau \overline{E}^{n-1}\left[\mathfrak{R}\right] \left(\Sigma_{\tau}\right) \, .
\end{align}
Adding the redshift-estimate of Theorem \ref{redshiftcomplete} and the expression \newline $\int_{\mathcal{M}^2_1 \cap \{ r \geq r_Y \} } dt^\star \  \overline{\mathbb{E}}^n\left[W\right] \left(\Sigma_{t^\star} \right)$ to both sides yields
\begin{align}
\overline{\mathbb{E}}^n\left[W\right] \left(\Sigma_2 \right) + \overline{\mathbb{E}}^n\left[W\right] \left(\mathcal{H}^2_1 \right) + \int_{\mathcal{M}^2_1} \ dt^\star \overline{\mathbb{E}}^n\left[W\right] \left(\Sigma_{t^\star} \right) \leq B \cdot \overline{\mathbb{E}}^n\left[W\right] \left(\Sigma_1 \right) \nonumber \\
+ B \sum_{j=1}^n \int_{\mathcal{M}_1^2} \left(K_1^{TTT} \left[\widehat{\mathcal{L}}_T^j W\right] + K_2^{TTT} \left[\widehat{\mathcal{L}}_T^j W\right] \right) +B \cdot Err^n_{hoz} \left[\mathfrak{R}\right] \left(\tau_1,\tau_2\right) \nonumber \\
+ B \sup_\tau \overline{E}^{n-1}\left[\mathfrak{R}\right] \left(\Sigma_{\tau}\right) + B \int_{\mathcal{M}^2_1 \cap \{ r \geq r_Y \} } dt^\star \  \overline{\mathbb{E}}^n\left[W\right] \left(\Sigma_{t^\star} \right) \, .\nonumber
\end{align}
Note that the left hand side controls in particular the time-integrated $L^2$-based energy. 
For the errorterms $K^{TTT}_1$ and $K^{TTT}_2$ we use the following Corollary of Proposition \ref{K12summary} (for the untilded region $\mathcal{M}\left(\tau_1,\tau_2\right)$):
\begin{corollary} \label{K12explicit}
If $\mathcal{X}=\mathcal{Y}=\mathcal{Z}=T$ in Proposition \ref{K12summary} we have, for any (small) $\lambda>0$ and $0\leq n \leq k-1$ the estimate
\begin{align}
\Big| \int_{\mathcal{M}\left(\tau_1,\tau_2\right)}  K_{2}^{TTT}\left[\widehat{\mathcal{L}}^{n+1}_T W\right]\Big| \leq  \lambda \sup_{\tau \in \left(\tau_1, \tau_2\right)} \ \overline{\mathbb{E}}^{n+1} \left[W\right]\left(\Sigma_{\tau}\right)
+ B_\lambda \cdot \mathbb{C}^{n+1} \left[\mathfrak{R}\right] \left(\tau_1,\tau_2\right) \, . \nonumber
\end{align}
where the $\mathbb{C}^n\left[\mathfrak{R}\right]$-energy was defined in Theorem \ref{theo1}. Moreover, in view of Proposition \ref{K1est}, the same estimate holds for $K_{1}^{TTT}\left[\widehat{\mathcal{L}}^{n+1}_T W\right]$ (in fact, without the last term).
\end{corollary}
Applying also the estimate for $Err_{hoz}^n\left[\mathfrak{R}\right]\left(\tau_1,\tau_2\right)$ given in Proposition \ref{redshiftcomplete}, we obtain
\begin{align}
\overline{\mathbb{E}}^n\left[W\right] \left(\Sigma_2 \right) + \overline{\mathbb{E}}^n\left[W\right] \left(\mathcal{H}^2_1 \right) + \int_{\mathcal{M}^2_1} \ dt^\star \overline{\mathbb{E}}^n\left[W\right] \left(\Sigma_{t^\star} \right) \leq B \cdot \overline{\mathbb{E}}^n\left[W\right] \left(\Sigma_1 \right) \nonumber \\
+ B \cdot \mathbb{C}^n\left[\mathfrak{R}\right] \left(\tau_1,\tau_2\right) + B \int_{\mathcal{M}^2_1 \cap \{ r \geq r_Y \} } dt^\star \  \overline{\mathbb{E}}^n\left[W\right] \left(\Sigma_{t^\star} \right) \, .\nonumber
\end{align} 

We finally estimate the last term of the above estimate using the estimate (\ref{hole}) integrated in time.\footnote{It is very important that (\ref{hole}) does not have an error-term on the horizon for this argument to go through. This is the underlying reason for the assumption that $T$ is null on the horizon: It ensures both that the energy flux through the horizon has a sign and that highest order error-terms are absent on the horizon in Proposition \ref{K12summary} .}

Defining $f\left(\tau\right) = \overline{\mathbb{E}}^n\left[W\right] \left(\Sigma_\tau \right)$, the resulting estimate can be written in the following form:  For any $\lambda>0$ and $\tau_2\geq \tau_1\geq \tau_0$
\begin{align} \label{bndest}
f\left(\tau_2\right) + \int_{\tau_1}^{\tau_2} dt f\left(t\right) \leq B_{\lambda} \cdot \left(\mathbb{C}^{n}\left[\mathfrak{R}\right] \left(\tau_0\right) + f\left(\tau_0\right) \right) \cdot \max \left[1, \left(\tau_2-\tau_1\right)\right] \nonumber \\ + \lambda \max \left[1, \left(\tau_2-\tau_1\right)\right]  \sup_{\tau \in \left(\tau_1,\tau_2\right)} f\left(\tau\right) + B \cdot f\left(\tau_1\right) \, .
\end{align}
In particular, we can achieve that $\lambda < \frac{1}{16B^2}$. Boundedness of $f\left(\tau\right)$ now follows from Lemma \ref{bndboot}. 
\end{proof}

\begin{lemma} \label{bndboot}
Suppose $f\left(\tau_0\right) <\infty$ and that $f\left(\tau\right)$ satisfies the estimate (\ref{bndest}) with $\lambda < \frac{1}{16B^2}<\frac{1}{16}$ and $B_{\lambda} = \tilde{B}>1$ fixed. Then $\sup_{\left(\tau_0, \infty\right)} f\left(\tau\right) < C$ holds for a constant $C=2\cdot B \cdot f\left(\tau_0\right)  + 9 \cdot 8 \cdot B^2 \cdot \tilde{B}  \left( \mathbb{C}^{n}\left[\mathfrak{R}\right] \left(\tau_0\right) + f\left(\tau_0\right) \right)$.
\end{lemma}
\begin{proof}
We will write $D= \mathbb{C}^{n}\left[\mathfrak{R}\right] \left(\tau_0\right) + f\left(\tau_0\right)$.
The proof is a bootstrap argument. Consider the region $A\left(\tau\right)=\mathcal{M}\left(\tau_0, \tau\right)$ for all $\tau$ such that $f\left(\tau\right)<C$ holds in $A\left(\tau\right)$. This region is clearly non-empty and open as a subset of the black hole exterior $\mathcal{M}\left(\tau_0, \infty\right)$. We now show it is also closed using the estimate (\ref{bndest}). For this, let $L=4B^2$ and define a sequence of slices $\tau_i = \tau_0 + i \cdot L$. Applying the estimate in a region between two slices $\tau_i, \tau_{i+1}$ located in the region $A\left(\tau\right)$ yields a ``good" slice $\tau_i^g$ satisfying
\begin{align}
f\left(\tau_i^{g}\right) = \inf_{\left(\tau_i,\tau_{i+1}\right)} f\left(\tau\right) \leq \tilde{B} \cdot D + \lambda \cdot C + \frac{B}{L} \cdot C \, .
\end{align}
Having found a good slice in each region $\mathcal{M}\left(\tau_i, \tau_{i+1}\right)$, we can apply the estimate (\ref{bndest}) again, this time from the good slice to any $\tau$-slice in $\mathcal{M}\left(\tau^{g}_i, \tau_{i+2}\right)$  (and, for the last good slice, in $\mathcal{M}\left(\tau^{g}_i, \tau\right)$). Note that the distance between the two slices considered is at most $2L$. Hence the estimate yields
\begin{align}
\sup_{\left(\tau^{g}_i, \tau_{i+2}\right)} f\left(\tau\right) &\leq \tilde{B} \cdot D \cdot 2L + \lambda \cdot 2L \cdot C + B \cdot \left(\tilde{B} \cdot D +\lambda \cdot C + \frac{B}{L} \cdot C\right) \nonumber \\
&\leq 8 B^2 \cdot \tilde{B} \cdot D + \frac{1}{2} \cdot C + B^2 \ \tilde{B} \cdot D + \frac{1}{16}C  + \frac{1}{4} C < \frac{15}{16}C \nonumber \, .
\end{align}
This uniform bound has been established for all slices in $A\left(\tau\right)$ north of the first good slice, $\tau_0^{g}$. However, for the region up to $\tau^g_0$ our estimate tells us that
\begin{align}
\sup_{\left(\tau_0, \tau_0^g\right)} f\left(\tau\right) \leq B \cdot f\left(\tau_0\right) + 4\cdot B^2 \ \tilde{B} \cdot D + \frac{1}{4} C < \frac{3}{4}C \, .
\end{align}
It follows that $A\left(\tau\right)$ is closed and that the uniform bound we established is valid for the entire black hole exterior.
\end{proof}

\subsection{Integrated decay: Theorem \ref{maintheointdec1}}
\begin{proof}
[Proof of Theorem \ref{maintheointdec1}] Apply Proposition \ref{spinreduce} with $\mathcal{W}_i = \mathcal{L}^i_T W$ for $1 \leq i \leq k$ and add the estimate of Corollary \ref{coloo} at the lowest order. 
Apply Theorem \ref{rhosigmaest} to estimate the $\rho$ and $\sigma$-term on the right hand side.
This yields, for $1 \leq n \leq k$
 \begin{align} 
 \sup_{(\tau_1,\tau_2)} \overline{\mathbb{E}}^0\left[W\right] \left(\tilde{\Sigma}_{\tau}\right) + \sum_{i=1}^n \Big\{ \sup_{(\tau_1,\tau_2)} E_{spin1} \left[\mathcal{W}_i\right]\left(\tilde{\Sigma}_{\tau}\right) + \sup_{(\tau_1,\tau_2)} E \left[\mathcal{D}\mathbf{\Phi}_{i-1}\right]\left(\tilde{\Sigma}_{\tau} \right) \nonumber \\
+ E \left[\mathcal{W}_i\right]\left(\mathcal{H}\left(\tau_1,\tau_2\right)\right)  + I^{deg} \left[\mathcal{W}_i\right] \left(\tilde{\mathcal{M}}\left({\tau_1},{\tau_2}\right)\right) \Big\}  + \overline{\mathbb{I}}^0 \left[W\right]  \left(\tilde{\mathcal{M}}\left({\tau_1},{\tau_2}\right)\right) \nonumber \\
\leq B \sum_{i=1}^n \Big\{  E_{spin1} \left[\mathcal{W}_i\right]\left(\tilde{\Sigma}_{\tau_1}\right) + E \left[\mathcal{D}\mathbf{\Phi}_{i-1}\right]\left(\tilde{\Sigma}_{\tau_1} \right) \Big\} + B \cdot  \overline{\mathbb{E}}^0\left[W\right] \left(\tilde{\Sigma}_{\tau_1}\right)  \nonumber \\
+   \epsilon  \left[ \overline{\mathbb{I}}^{n,deg}_{\tilde{P}_{\rho}} \left[W\right] \left(\tilde{\mathcal{M}}\left(\tau_1,\tau_2\right)\right) + \sup_{\tau} \overline{\mathbb{E}}^{n} \left[W\right]  \left(\tilde{\Sigma}_{\tau}\right)\right] + B \cdot \mathbb{D}_\lambda^{n} \left[\mathfrak{R}\right] \left({\tau_1},{\tau_2}\right) \nonumber \\ + B  \sum_{i=1}^n Err_{Mora} \left[\mathcal{W}_i\right]  \nonumber
 \end{align}
 where we wrote $\mathbf{\Phi}_i$ to indicate that we are adding both the estimate for $\phi_i$ and $\psi_i$ of Theorem \ref{rhosigmaest}.
For the non-degenerate case we have
 \begin{align} \label{nds}
 \sup_{(\tau_1,\tau_2)} \overline{\mathbb{E}}^0\left[W\right] \left(\tilde{\Sigma}_{\tau}\right) + \sum_{i=1}^{n-1} \Big\{ \sup_{(\tau_1,\tau_2)} E_{spin1} \left[\mathcal{W}_i\right]\left(\tilde{\Sigma}_{\tau}\right) + \sup_{(\tau_1,\tau_2)} E \left[\mathcal{D}\mathbf{\Phi}_{i-1}\right]\left(\tilde{\Sigma}_{\tau} \right) \nonumber \\
+ E \left[\mathcal{W}_i\right]\left(\mathcal{H}\left(\tau_1,\tau_2\right)\right)  + I^{nondeg} \left[\mathcal{W}_i\right] \left(\tilde{\mathcal{M}}\left({\tau_1},{\tau_2}\right)\right) \Big\}  + \overline{\mathbb{I}}^0 \left[W\right]  \left(\tilde{\mathcal{M}}\left({\tau_1},{\tau_2}\right)\right) \nonumber \\
\leq B \sum_{i=1}^{n-1}  E_{spin1} \left[\mathcal{W}_i\right]\left(\tilde{\Sigma}_{\tau_1}\right) + \sum_{i=1}^n E \left[\mathcal{D}\mathbf{\Phi}_{i-1}\right]\left(\tilde{\Sigma}_{\tau_1} \right)+ B \cdot  \overline{\mathbb{E}}^0\left[W\right] \left(\tilde{\Sigma}_{\tau_1}\right)  \nonumber \\
+   \epsilon  \left[ \overline{\mathbb{I}}^{n,deg}_{\tilde{P}_{\rho}} \left[W\right] \left(\tilde{\mathcal{M}}\left(\tau_1,\tau_2\right)\right) + \sup_{\tau} \overline{\mathbb{E}}^{n} \left[W\right]  \left(\tilde{\Sigma}_{\tau}\right)\right] + B \cdot \mathbb{D}_\lambda^{n} \left[\mathfrak{R}\right] \left({\tau_1},{\tau_2}\right) \nonumber \\
+ B  \sum_{i=1}^{n-1} Err_{Mora} \left[\mathcal{W}_i\right]  \, .
 \end{align}
 It is important to note that we only need the $n$-derivative energy of the components $\rho$ and $\sigma$ (and an $\epsilon$ of the $n$-derivative energy of all components) to obtain a non-degenerate integrated decay estimate for $n-1$-derivatives of all components of the Weyl-tensor.  As a consequence, this estimate requires only $Err_{Mora}\left[\mathcal{W}_{n-1}\right]$. Note also that for $n=1$ in (\ref{nds}) there is no $Err_{Mora}$-term.

We next turn to Propositions \ref{K12summary} and \ref{K2Fsum} to estimate $Err_{Mora}\left[\mathcal{W}_i\right]$. We observe that for $\lambda$ sufficiently small, the error-term on the horizon (the last term in (\ref{harderrorest}),) can be absorbed by the horizon-term available on the left of (\ref{nds}). (In particular, the assumption that $T$ is null one the horizon is not necessary for this argument to work.) We conclude, therefore,
 \begin{align} 
 \sup_{(\tau_1,\tau_2)} \overline{\mathbb{E}}^0\left[W\right] \left(\tilde{\Sigma}_{\tau}\right) + \sum_{i=1}^n \Big\{ \sup_{(\tau_1,\tau_2)} E_{spin1} \left[\mathcal{W}_i\right]\left(\tilde{\Sigma}_{\tau}\right) + \sup_{(\tau_1,\tau_2)} E \left[\mathcal{D}\mathbf{\Phi}_{i-1}\right]\left(\tilde{\Sigma}_{\tau} \right) \nonumber \\
+ E \left[\mathcal{W}_i\right]\left(\mathcal{H}\left(\tau_1,\tau_2\right)\right)  + I^{deg} \left[\mathcal{W}_i\right] \left(\tilde{\mathcal{M}}\left({\tau_1},{\tau_2}\right)\right) \Big\}  + \overline{\mathbb{I}}^0 \left[W\right]  \left(\tilde{\mathcal{M}}\left({\tau_1},{\tau_2}\right)\right) \nonumber \\
\leq B \sum_{i=1}^n \Big\{  E_{spin1} \left[\mathcal{W}_i\right]\left(\tilde{\Sigma}_{\tau_1}\right) + E \left[\mathcal{D}\mathbf{\Phi}_{i-1}\right]\left(\tilde{\Sigma}_{\tau_1} \right) \Big\} + B \cdot  \overline{\mathbb{E}}^0\left[W\right] \left(\tilde{\Sigma}_{\tau_1}\right)  \nonumber \\
+ \left( \lambda + \epsilon \right) \left[ \overline{\mathbb{I}}^{n,deg}_{\tilde{P}_{\rho}} \left[W\right] \left(\tilde{\mathcal{M}}\left(\tau_1,\tau_2\right)\right) + \sup_{\tau} \overline{\mathbb{E}}^{n} \left[W\right]  \left(\tilde{\Sigma}_{\tau}\right)\right] + B \cdot \mathbb{D}^{n} \left[\mathfrak{R}\right]  \left(\tau_1,\tau_2\right) \, . \nonumber
 \end{align}
Note in particular that we need $\mathbb{D}^{n} \left[\mathfrak{R}\right]$ and not only $\mathbb{D}^{n}_\lambda \left[\mathfrak{R}\right]$. On the contrary, because the non-degenerate estimate needs only $Err_{Mora}\left[\mathcal{W}_{n-1}\right]$, we have
\begin{align}
 \sup_{(\tau_1,\tau_2)} \overline{\mathbb{E}}^0\left[W\right] \left(\tilde{\Sigma}_{\tau}\right) + \sum_{i=1}^{n-1} \Big\{ \sup_{(\tau_1,\tau_2)} E_{spin1} \left[\mathcal{W}_i\right]\left(\tilde{\Sigma}_{\tau}\right) + \sup_{(\tau_1,\tau_2)} E \left[\mathcal{D}\mathbf{\Phi}_{i-1}\right]\left(\tilde{\Sigma}_{\tau} \right) \nonumber \\
+ E \left[\mathcal{W}_i\right]\left(\mathcal{H}\left(\tau_1,\tau_2\right)\right)  + I^{nondeg} \left[\mathcal{W}_i\right] \left(\tilde{\mathcal{M}}\left({\tau_1},{\tau_2}\right)\right) \Big\}  + \overline{\mathbb{I}}^0 \left[W\right]  \left(\tilde{\mathcal{M}}\left({\tau_1},{\tau_2}\right)\right) \nonumber \\
\leq B \sum_{i=1}^{n-1}  E_{spin1} \left[\mathcal{W}_i\right]\left(\tilde{\Sigma}_{\tau_1}\right) + \sum_{i=1}^n E \left[\mathcal{D}\mathbf{\Phi}_{i-1}\right]\left(\tilde{\Sigma}_{\tau_1} \right)+ B \cdot  \overline{\mathbb{E}}^0\left[W\right] \left(\tilde{\Sigma}_{\tau_1}\right)  \nonumber \\
+   \epsilon  \left[ \overline{\mathbb{I}}^{n,deg}_{\tilde{P}_{\rho}} \left[W\right] \left(\tilde{\mathcal{M}}\left(\tau_1,\tau_2\right)\right) + \sup_{\tau} \overline{\mathbb{E}}^{n} \left[W\right]  \left(\tilde{\Sigma}_{\tau}\right)\right] + B \cdot \mathbb{D}_\lambda^{n} \left[\mathfrak{R}\right]\left({\tau_1},{\tau_2}\right) \, , \nonumber
 \end{align}
which requires only $\mathbb{D}^{n}_\lambda \left[\mathfrak{R}\right]$ on the right hand side.

These estimates provide control over the Lie-$T$ derivatives of all components. In view of the elliptic estimates of section \ref{elliptic} this is sufficient to control all derivatives away from the horizon and from infinity ($r \leq R+M$). In particular, defining\footnote{The reason that $\tilde{\Sigma}_{\tau,int}$ only reaches up to $R$ is that $\tilde{\Sigma}$ starts being null for $r \geq R$, and standard elliptic estimates are no longer available.}
\begin{equation}
\tilde{\Sigma}_{\tau,int} = \tilde{\Sigma}_\tau \cap \Big\{r_Y-\frac{r_Y-2M}{2} \leq r \leq R \Big\} \, ,
\end{equation}
\begin{equation}
\tilde{\mathcal{M}}_{int} \left(\tau_1,\tau_2\right)  = \tilde{\mathcal{M}} \left(\tau_1,\tau_2\right) \cap \Big\{r_Y-\frac{r_Y-2M}{2} \leq r \leq R+M \Big\} \, ,
\end{equation}
we have
\begin{align} \label{degty}
\sup_{(\tau_1,\tau_2)} \overline{\mathbb{E}}^n \left[W\right]
\left(\tilde{\Sigma}_{\tau, int}  \right) + \overline{\mathbb{I}}^{n,deg} \left[W\right] \left(\tilde{\mathcal{M}}_{int} \left(\tau_1,\tau_2\right) \right) \nonumber \\ \leq  B \sum_{i=1}^n \Big\{  E_{spin1} \left[\mathcal{W}_i\right]\left(\tilde{\Sigma}_{\tau_1}\right) + E \left[\mathcal{D}\mathbf{\Phi}_{i-1}\right]\left(\tilde{\Sigma}_{\tau_1} \right) \Big\} + B \cdot  \overline{\mathbb{E}}^0\left[W\right] \left(\tilde{\Sigma}_{\tau_1}\right)  \nonumber \\
+   \epsilon  \left[ \overline{\mathbb{I}}^{n,deg}_{\tilde{P}_{\rho}} \left[W\right] \left(\tilde{\mathcal{M}}\left(\tau_1,\tau_2\right)\right) + \sup_{\tau} \overline{\mathbb{E}}^{n} \left[W\right]  \left(\tilde{\Sigma}_{\tau}\right)\right] + B \cdot \mathbb{D}^{n} \left[\mathfrak{R}\right] \left({\tau_1},{\tau_2}\right)
\end{align}
and
\begin{align}
 \overline{\mathbb{I}}^{n-1,nondeg} \left[W\right] \left(\tilde{\mathcal{M}}_{int} \left(\tau_1,\tau_2\right) \right)
 \nonumber \\
\leq B \sum_{i=1}^{n-1}  E_{spin1} \left[\mathcal{W}_i\right]\left(\tilde{\Sigma}_{\tau_1}\right) + \sum_{i=1}^n E \left[\mathcal{D}\mathbf{\Phi}_{i-1}\right]\left(\tilde{\Sigma}_{\tau_1} \right)+ B \cdot  \overline{\mathbb{E}}^0\left[W\right] \left(\tilde{\Sigma}_{\tau_1}\right)  \nonumber \\
+   \epsilon  \left[ \overline{\mathbb{I}}^{n,deg}_{\tilde{P}_{\rho}} \left[W\right] \left(\tilde{\mathcal{M}}\left(\tau_1,\tau_2\right)\right) + \sup_{\tau} \overline{\mathbb{E}}^{n} \left[W\right]  \left(\tilde{\Sigma}_{\tau}\right)\right] + B \cdot \mathbb{D}_\lambda^{n} \left[\mathfrak{R}\right] \left({\tau_1},{\tau_2}\right) \, .\nonumber
 \end{align}
We now couple the degenerate estimate (\ref{degty}) with the redshift estimate of Proposition \ref{redshiftcomplete} and with the estimate near infinity, Theorem \ref{maintheoinf}, to obtain Theorem \ref{maintheointdec1}.
\end{proof}

\subsection{The decay iteration: Theorem \ref{maintheorem}}
\begin{proposition} \label{auxdecprop}
Recall the decay matrices $P_j$ defined in (\ref{P1})-(\ref{P3}). For $j=0, ... ,2$, we have for any $i = 0, ... , k-2$ the estimates
\begin{align} \label{uly1}
 \overline{\mathbb{E}}^{i+1}_{P_{j+1}} \left[W\right] \left(\tilde{\Sigma}_{\tau}\right) \leq
\frac{B}{\tau}  \Big[ \sup_{\tau^\prime \geq \tau} \, \overline{\mathbb{E}}^{i+2}_{P_j} \left[W\right] \left(\tilde{\Sigma}_{\tau^\prime}\right) +
B \cdot \mathbb{D}_\lambda^{i+2} \left[\mathfrak{R}\right] \left(\tau,\infty\right)   \Big]
\end{align}

\begin{align} \label{uly2}
\overline{\mathbb{I}}^{i+1,deg}_{\tilde{P}_{j+1}} \left[W\right] \left(\tilde{\mathcal{M}}\left(\tau, \infty \right) \right) \leq   \nonumber \\   B \cdot \mathbb{D}^{i+1} \left[\mathfrak{R}\right] \left(\tau\right) +\frac{B}{\tau}  \Big[ \sup_{\tau^\prime \geq \tau} \, \overline{\mathbb{E}}^{i+2}_{P_j} \left[W\right] \left(\tilde{\Sigma}_{\tau^\prime}\right) +
B \cdot \mathbb{D}_\lambda^{i+2} \left[\mathfrak{R}\right] \left(\tau,\infty\right)   \Big] 
\end{align}
and
\begin{align} \label{uly3}
 \overline{\mathbb{I}}^{i+1,nondeg}_{\tilde{P}_{j+1}} \left[W\right] \left(\tilde{\mathcal{M}}\left(\tau, \infty \right) \right)
   \leq \nonumber \\ +  B \cdot \mathbb{D}_{\lambda}^{i+1} \left[\mathfrak{R}\right] \left(\tau\right) +\frac{B}{\tau}  \Big[ \sup_{\tau^\prime \geq \tau} \, \overline{\mathbb{E}}^{i+2}_{P_j} \left[W\right] \left(\tilde{\Sigma}_{\tau^\prime}\right) +
B \cdot \mathbb{D}_\lambda^{i+2} \left[\mathfrak{R}\right] \left(\tau,\infty\right)   \Big] 
\end{align}
\end{proposition}
Before we prove the above Proposition let us note how it implies Theorem \ref{maintheorem}:
\begin{proof}[Proof of Theorem \ref{maintheorem}]
Theorem \ref{maintheorem} follows from repeatedly (depending on the decay matrix) inserting the estimate (\ref{uly1}) into (\ref{uly2}) and (\ref{uly3}).
\end{proof}

\begin{proof} [Proof of Proposition \ref{auxdecprop}]
At order $i+1$, we apply the estimate (\ref{mestI2}) with decay matrix $P_0$, cf.~(\ref{P0}),
which has the associated bulk admissible matrix
 \begin{equation*} 
\tilde{\mathbf{P}}_0 = \left( \begin{array}{ccccccc}
0 & 6-\delta & 6-\delta & 5-\delta & 3-\delta & 1-\delta & 0\\ 
9-\delta & 9-\delta & 8-\delta & 7-\delta & 5-\delta & 3-\delta & 1-\delta \\
\vdots & \vdots & \vdots & \vdots & \vdots & \vdots & \vdots
\end{array} \right) \, .
 \end{equation*}
From the resulting estimate we derive
\begin{align} \label{hlu}
  \overline{\mathbb{E}}^{i+1}_{P_1} \left[W\right] \left(\tilde{\Sigma}_{\tau_{n+1}}\right) + \int_{\tau_n}^{\tau_{n+1}} dt^\star \, \overline{\mathbb{E}}^{i+1}_{P_1} \left[W\right] \left(\tilde{\Sigma}_{\tau}\right) \leq B \cdot \overline{\mathbb{E}}^{i+2}_{P_{0}} \left[W\right] \left(\tilde{\Sigma}_{\tau_n}\right) \nonumber \\
+ B \cdot \mathbb{D}_\lambda^{i+2} \left[\mathfrak{R}\right] \left(\tau_n, \tau_{n+1}\right)  \, ,
\end{align}
for the boundary admissible decay matrix $P_1$ defined in (\ref{P1}), where we only 
used that $P_1 \leq \tilde{P}_0 < P_0$ (recall the ordering relation defined in section \ref{decmatrix}). We conclude
\begin{align} \label{fiits}
 \overline{\mathbb{E}}^{i+1}_{P_{1}} \left[W\right] \left(\tilde{\Sigma}_{\tau}\right) \leq
\frac{B}{\tau}  \Big[ \sup_{\tau^\prime \geq \tau} \, \overline{\mathbb{E}}^{i+2}_{P_0} \left[W\right] \left(\tilde{\Sigma}_{\tau^\prime}\right) +
B \cdot \mathbb{D}_\lambda^{i+2} \left[\mathfrak{R}\right] \left(\tau,\infty\right)   \Big] \, , 
\end{align}
which is (\ref{uly1}) for $j=0$. Since the square bracket is bounded (by applying (\ref{mestI}) from data with $P_0$), we have established $\frac{1}{t}$ decay of a weighted energy. In particular, since $P_{\rho}<P_1$ we have shown this decay for the $P_\rho$-weighted energy at order $i+1$.
Applying (\ref{mestI}) with $P_1$ yields, after inserting (\ref{fiits}), the estimate (\ref{uly2}) for $j=0$. In the same way, applying (\ref{mestI2}) produces (\ref{uly3}) for $j=0$.

It is clear that this procedure can be iterated. We next apply our estimate (\ref{mestI}) with the boundary admissible decay matrix $P_1$ and its associated bulk admissible matrix
 \begin{equation*} 
\tilde{\mathbf{P}}_1 = \left( \begin{array}{ccccccc}
0 & 5-\delta & 5-\delta & 5-\delta & 3-\delta & 1-\delta & 0\\ 
8-\delta & 8-\delta & 8-\delta & 7-\delta & 5-\delta & 3-\delta & 1-\delta \\
\vdots & \vdots & \vdots & \vdots & \vdots & \vdots & \vdots
\end{array} \right) \, .
 \end{equation*}
 Noting that $P_2 < \tilde{P}_1 < P_1$ for $P_2$ as in (\ref{P2}), we obtain, after adding the non-degenerate integrated decay estimate, the analogue of (\ref{hlu})
 \begin{align} \label{hlu2}
  \overline{\mathbb{E}}^{i+1}_{P_2} \left[W\right] \left(\tilde{\Sigma}_{\tau_{n+1}}\right) + \int_{\tau_n}^{\tau_{n+1}} dt^\star \, \overline{\mathbb{E}}^{i+1}_{P_2} \left[W\right] \left(\tilde{\Sigma}_{\tau}\right) \nonumber \\ \leq  B \cdot \, \overline{\mathbb{E}}^{i+2}_{P_{1}} \left[W\right] \left(\tilde{\Sigma}_{\tau_n}\right) 
+ B \cdot \mathbb{D}_\lambda^{i+2} \left[\mathfrak{R}\right] \left(\tau_n,\tau_{n+1}\right) \, ,
 \end{align}
from which we conclude
\begin{align} 
 \overline{\mathbb{E}}^{i+1}_{P_2} \left[W\right] \left(\tilde{\Sigma}_{\tau}\right)  \leq
\frac{B}{\tau}  \Big[ \sup_{\tau^\prime \geq \tau} \, \overline{\mathbb{E}}^{i+2}_{P_1} \left[W\right] \left(\tilde{\Sigma}_{\tau^\prime}\right) +
B \cdot \mathbb{D}_\lambda^{i+2} \left[\mathfrak{R}\right] \left(\tau,\infty\right)   \Big] \, .
\end{align}
Another iteration, applying (\ref{mestI}) with $P_2$, will provide the estimate
\begin{align} 
 \overline{\mathbb{E}}^{i+1}_{P_3} \left[W\right] \left(\tilde{\Sigma}_{\tau}\right)  \leq
\frac{B}{\tau}  \Big[ \sup_{\tau^\prime \geq \tau} \, \overline{\mathbb{E}}^{i+2}_{P_2} \left[W\right] \left(\tilde{\Sigma}_{\tau^\prime}\right) +
B \cdot \mathbb{D}_\lambda^{i+2} \left[\mathfrak{R}\right] \left(\tau,\infty\right)   \Big] \, ,
\end{align}
and hence an $r$-weighted energy which decays stronger than $\frac{1}{\tau^2}$. This is enough to close the argument. We remark that, of course, there is no reason to stop the iteration at this point. To obtain the optimal result in terms of the decay in $t$, one can define additional admissible decay matrices and continue the iteration.
%
%
%
%
%
%
\end{proof}
\section{Acknowledgements}
I would like to thank Mihalis Dafermos for suggesting to study this problem and Igor Rodnianski for numerous insightful discussions and comments contributing to this paper.
\appendix
\section{Useful formulae} \label{UF}
Let $(a)$ denote tetrad (=frame) indices, while Greek letters denote spacetime indices. 
The Ricci coefficients are defined as
\begin{equation}
\gamma_{(c)(a)(b)} = e_{(a) \, \kappa ;  \nu}  e_{(b)}^\nu  e_{(c)}^\kappa
\end{equation}
with the null-components
\begin{align}
\gamma_{(B)(3)(A)} = \underline{H}_{AB}    \textrm{ \ \ \ \ , \ \ \ \ \ \ } \gamma_{(B)(4)(A)} = H_{AB} \, ,
\end{align}
\begin{align}
\gamma_{(A)(3)(3)} = 2\underline{Y}_A   \textrm{ \ \ \ \ , \ \ \ \ \ \ } \gamma_{(A)(4)(4)} = 2Y_A \, ,
\end{align}
\begin{align}
\gamma_{(A)(3)(4)} = 2 \underline{Z}_{A}   \textrm{ \ \ \ \  , \ \ \ \ \ \ } \gamma_{(A)(4)(3)} = 2 Z_A \, ,
\end{align}
\begin{align}
\gamma_{(4)(3)(3)} = 4\underline{\Omega} \textrm{ \ \ \ \  , \ \ \ \ \ \ } \gamma_{(3)(4)(4)} = 4 \Omega \, ,
\end{align}
\begin{align}
\gamma_{(3)(4)(A)} = 2V_A \, .
\end{align}
Due to the antisymmetry in the first two indices we can write $e_{(a)\phantom{\mu};\kappa}^{\phantom{(a)}\mu} = -\gamma_{(a) \phantom{\mu} \kappa}^{\phantom{(a)}\mu}$.

For the Weyl-tensor we note the relations
\begin{equation}
W_{(a)(b)(c)(d)} = e_{(a)}^\mu e_{(b)}^\nu e_{(c)}^\sigma e_{(d)}^\tau  W_{\mu \nu \sigma \tau} 
\end{equation} 
and
\begin{align}
D_{(a)} W_{(b)(c)(d)(f)} = D_{(a)} \left(W_{(b)(c)(d)(f)}\right) - \gamma_{(g)(b)(a)} W^{(g)}_{\phantom{(g)} (c)(d)(f)} \nonumber \\ - \gamma_{(g)(c)(a)} W^{\phantom{(b)}(g)}_{(b)\phantom{(g)}(d)(f)} - \gamma_{(g)(d)(a)} W^{\phantom{(b)(c)}(g)}_{(b)(c)\phantom{(g)}(f)} - \gamma_{(g)(f)(a)} W^{\phantom{(b)(c)(d)}(g)}_{(b)(c)(d)\phantom{(g)}} \, .
\end{align}
\printglossary
\bibliographystyle{utphys}
\bibliography{thesisrefs}
\end{document}